\pdfoutput=1 

\PassOptionsToPackage{unicode}{hyperref}
\PassOptionsToPackage{hyphens}{url}
\PassOptionsToPackage{hypertexnames=false}{hyperref}
\documentclass[12pt,letterpaper]{article}

\makeatletter
\let\tavie@std@subsection\subsection
\let\tavie@std@subsubsection\subsubsection
\let\tavie@std@paragraph\paragraph
\let\tavie@std@subparagraph\subparagraph
\makeatother
\usepackage{titlesec}
\titleformat{\section}{\normalfont\Large\bfseries}{\thesection}{1em}{}
\titlespacing*{\section}{0pt}{*0}{*0}   
\makeatletter
\let\tavie@ttl@subsection\subsection
\let\tavie@ttl@subsubsection\subsubsection
\let\tavie@ttl@paragraph\paragraph
\let\tavie@ttl@subparagraph\subparagraph
\let\subsection\tavie@std@subsection
\let\subsubsection\tavie@std@subsubsection
\let\paragraph\tavie@std@paragraph
\let\subparagraph\tavie@std@subparagraph
\makeatother

\makeatletter
\def\section{\@startsection{section}{1}{\z@}
  {-2.0ex plus -0.5ex minus -.2ex}
  {1.0ex plus 0.2ex}
  {\normalfont\Large\bfseries}}
\makeatother

\usepackage{blindtext}
\usepackage{amsmath,amsfonts,amssymb,amsthm}
\usepackage{mathbbol}

\makeatletter
\def\tavie@mathsizes{5,6,7,8,9,10,10.95,12,14.4,17.28,20.74,24.88}
\@for\tavie@s:=\tavie@mathsizes\do{%
  \expandafter\let\csname tavie@std@S@\tavie@s\expandafter\endcsname
                  \csname S@\tavie@s\endcsname}
\newmuskip\tavie@std@thinmu  \tavie@std@thinmu=\thinmuskip
\newmuskip\tavie@std@medmu   \tavie@std@medmu=\medmuskip
\newmuskip\tavie@std@thickmu \tavie@std@thickmu=\thickmuskip
\let\tavie@std@boldmath\boldmath
\let\tavie@std@jmath\jmath
\let\tavie@std@amalg\amalg
\let\tavie@std@coprod\coprod
\let\tavie@std@rmdefault\rmdefault
\let\tavie@AtBeginDocument\AtBeginDocument
\renewcommand\AtBeginDocument[1]{}
\makeatother
\usepackage{mathptmx}
\makeatletter
\let\AtBeginDocument\tavie@AtBeginDocument
\let\rmdefault\tavie@std@rmdefault
\@for\tavie@s:=\tavie@mathsizes\do{%
  \expandafter\let\csname tavie@ptm@S@\tavie@s\expandafter\endcsname
                  \csname S@\tavie@s\endcsname
  \expandafter\let\csname S@\tavie@s\expandafter\endcsname
                  \csname tavie@std@S@\tavie@s\endcsname}
\newmuskip\tavie@ptm@thinmu  \tavie@ptm@thinmu=\thinmuskip
\newmuskip\tavie@ptm@medmu   \tavie@ptm@medmu=\medmuskip
\newmuskip\tavie@ptm@thickmu \tavie@ptm@thickmu=\thickmuskip
\thinmuskip=\tavie@std@thinmu \medmuskip=\tavie@std@medmu \thickmuskip=\tavie@std@thickmu
\let\tavie@ptm@boldmath\boldmath \let\boldmath\tavie@std@boldmath
\let\tavie@ptm@jmath\jmath       \let\jmath\tavie@std@jmath
\let\tavie@ptm@amalg\amalg       \let\amalg\tavie@std@amalg
\let\tavie@ptm@coprod\coprod     \let\coprod\tavie@std@coprod
\makeatother
\usepackage{tikz}
\usepackage{authblk}
\usetikzlibrary{bayesnet}
\usepackage{graphicx}
\usepackage{caption,subcaption}
\usepackage{dsfont}
\usepackage[dvipsnames]{xcolor}
\usepackage{mathrsfs}
\usepackage{circuitikz}
\usepackage{booktabs}
\usepackage{tocloft} 
\usepackage{array}
\usepackage{float}
\usepackage{pifont}
\usepackage[ruled,vlined]{algorithm2e}
\usepackage{bm}
\usepackage{multirow, makecell}

\usepackage{siunitx}
\usepackage{iftex}
\ifPDFTeX
  \usepackage[T1]{fontenc}
  \usepackage[utf8]{inputenc}
  \usepackage{textcomp}
\else
  \usepackage{unicode-math}
  \defaultfontfeatures{Scale=MatchLowercase}
  \defaultfontfeatures[\rmfamily]{Ligatures=TeX,Scale=1}
\fi
\usepackage{lmodern}

\IfFileExists{upquote.sty}{\usepackage{upquote}}{}
\IfFileExists{microtype.sty}{
  \usepackage{microtype}
  \UseMicrotypeSet[protrusion]{basicmath}
}{}

\makeatletter
\@ifundefined{KOMAClassName}{
  \IfFileExists{parskip.sty}{
    \usepackage{parskip}
  }{
    \setlength{\parindent}{0pt}
    \setlength{\parskip}{6pt plus 2pt minus 1pt}}
}{
  \KOMAoptions{parskip=half}}
\makeatother

\makeatletter
\ifx\paragraph\undefined\else
  \let\oldparagraph\paragraph
  \renewcommand{\paragraph}{
    \@ifstar
      \xxxParagraphStar
      \xxxParagraphNoStar
  }
  \newcommand{\xxxParagraphStar}[1]{\oldparagraph*{#1}\mbox{}}
  \newcommand{\xxxParagraphNoStar}[1]{\oldparagraph{#1}\mbox{}}
\fi
\ifx\subparagraph\undefined\else
  \let\oldsubparagraph\subparagraph
  \renewcommand{\subparagraph}{
    \@ifstar
      \xxxSubParagraphStar
      \xxxSubParagraphNoStar
  }
  \newcommand{\xxxSubParagraphStar}[1]{\oldsubparagraph*{#1}\mbox{}}
  \newcommand{\xxxSubParagraphNoStar}[1]{\oldsubparagraph{#1}\mbox{}}
\fi
\makeatother

\usepackage{longtable}
\usepackage{calc}
\usepackage{etoolbox}
\makeatletter
\patchcmd\longtable{\par}{\if@noskipsec\mbox{}\fi\par}{}{}
\makeatother

\IfFileExists{footnotehyper.sty}{\usepackage{footnotehyper}}{\usepackage{footnote}}
\makesavenoteenv{longtable}

\makeatletter
\def\maxwidth{\ifdim\Gin@nat@width>\linewidth\linewidth\else\Gin@nat@width\fi}
\def\maxheight{\ifdim\Gin@nat@height>\textheight\textheight\else\Gin@nat@height\fi}
\makeatother
\setkeys{Gin}{width=\maxwidth,height=\maxheight,keepaspectratio}

\makeatletter
\def\fps@figure{htbp}
\makeatother

\makeatletter\edef\tavie@std@textwidth{\the\textwidth}\edef\tavie@std@evensidemargin{\the\evensidemargin}\makeatother
\makeatletter\edef\tavie@std@secpenalty{\the\@secpenalty}\makeatother


\makeatletter
\AtBeginDocument{%
\ifdefined\contentsname
  \renewcommand*\contentsname{Table of contents}
\else
  \newcommand\contentsname{Table of contents}
\fi
\ifdefined\listfigurename
  \renewcommand*\listfigurename{List of Figures}
\else
  \newcommand\listfigurename{List of Figures}
\fi
\ifdefined\listtablename
  \renewcommand*\listtablename{List of Tables}
\else
  \newcommand\listtablename{List of Tables}
\fi
\ifdefined\figurename
  \renewcommand*\figurename{Figure}
\else
  \newcommand\figurename{Figure}
\fi
\ifdefined\tablename
  \renewcommand*\tablename{Table}
\else
  \newcommand\tablename{Table}
\fi
}
\floatstyle{ruled}
\@ifundefined{c@chapter}{\newfloat{codelisting}{h}{lop}}{\newfloat{codelisting}{h}{lop}[chapter]}
\floatname{codelisting}{Listing}

\makeatother

\ifLuaTeX
  \usepackage{selnolig}
\fi

\usepackage{natbib}
\usepackage{bibunits}
\defaultbibliography{bibliography}
\defaultbibliographystyle{apalike}
\makeatletter
\renewcommand\NAT@wrout[5]{%
\if@filesw
      {\let\protect\noexpand\let~\relax
       \immediate
       \write\@auxout{\string\bibcite{#5\@extra@b@citeb}{{#1}{#2}{{#3}}{{#4}}}}}\fi
\ignorespaces}
\makeatother
\usepackage{bookmark}

\IfFileExists{xurl.sty}{\usepackage{xurl}}{}
\hypersetup{
  colorlinks=true,
  linkcolor={Blue},
  filecolor={Blue},
  citecolor={Blue},
  urlcolor={Blue}
}

\newtheorem{theorem}{Theorem}
\newtheorem{lemma}{Lemma} 
\newtheorem{proposition}{Proposition} 
\newtheorem{remark}{Remark}

\newtheorem{definition}{Definition}
\newtheorem{assumption}{Assumption}

\newtheorem{claim}{Claim} 

\numberwithin{equation}{section}

\newcommand{\ssg}{\mathsf{SSG}}
\newcommand{\tssg}{\mathsf{TAVIE}\text{-}\mathsf{SSG}}

\makeatletter
\newcounter{suppdivision}

\newcommand{\l@suppdivision}[2]{%
  \addpenalty{-\@highpenalty}%
  \addvspace{1.2em}%
  \noindent
  \parbox{\textwidth}{%
    \centering
    \bfseries
    #1%
  }%
  \par
  \addvspace{0.35em}%
}
\makeatother

\makeatletter
\newcommand\tavieHideMainFromTOC{%
  \if@filesw
    \immediate\write\@auxout{\string\@writefile{toc}{\string\setcounter{tocdepth}{-10}}}%
  \fi}
\newcommand\tavieStartSupplement{%
  \clearpage
  \bookmark[page=\the\c@page,level=0]{Supplementary Materials}%
  \setcounter{page}{1}%
  \thispdfpagelabel{S\thepage}%
  \AddToHook{shipout/before}{\thispdfpagelabel{S\thepage}}%
  \@for\tavie@c:=part,section,subsection,subsubsection,paragraph,subparagraph,%
      equation,figure,table,footnote,mpfootnote,subfigure,subtable,codelisting,%
      example,theorem,lemma,proposition,remark,corollary,definition,assumption,%
      conjecture,claim,axiom,condition,suppdivision\do{%
    \@ifundefined{c@\tavie@c}{}{\setcounter{\tavie@c}{0}}}%
  \setcounter{algocf}{1}
  \counterwithin{lemma}{section}%
  \counterwithin{remark}{section}%
  \counterwithin{definition}{section}%
  \evensidemargin=\tavie@std@evensidemargin\relax \addtolength{\evensidemargin}{-.1in}%
  \textwidth=\tavie@std@textwidth\relax \addtolength{\textwidth}{1in}%
  \columnwidth\textwidth \hsize\textwidth \linewidth\textwidth
  \global\@secpenalty=\tavie@std@secpenalty\relax
  \begingroup
    \fontencoding{T1}\fontfamily{ptm}\fontseries{m}\fontshape{n}\fontsize{12}{14.5}\selectfont
    \global\emergencystretch=3em
    \global\bibhang=1em
  \endgroup
  \let\subsection\tavie@ttl@subsection
  \let\subsubsection\tavie@ttl@subsubsection
  \let\oldparagraph\tavie@ttl@paragraph
  \let\oldsubparagraph\tavie@ttl@subparagraph
  \thinmuskip=\tavie@ptm@thinmu \medmuskip=\tavie@ptm@medmu \thickmuskip=\tavie@ptm@thickmu
  \@for\tavie@s:=\tavie@mathsizes\do{%
    \expandafter\let\csname S@\tavie@s\expandafter\endcsname
                    \csname tavie@ptm@S@\tavie@s\endcsname}%
  \let\boldmath\tavie@ptm@boldmath \let\jmath\tavie@ptm@jmath
  \let\amalg\tavie@ptm@amalg \let\coprod\tavie@ptm@coprod
  \DeclareFontFamily{U}{msa}{}%
  \DeclareFontShape{U}{msa}{m}{n}{<->msam10}{}%
  \DeclareFontFamily{U}{msb}{}%
  \DeclareFontShape{U}{msb}{m}{n}{<->msbm10}{}%
  \DeclareFontFamily{U}{euf}{}%
  \DeclareFontShape{U}{euf}{m}{n}{<-6>eufm5<6-8>eufm7<8->eufm10}{}%
  \DeclareFontShape{U}{euf}{b}{n}{<-6>eufb5<6-8>eufb7<8->eufb10}{}%
  \@for\tavie@s:=\tavie@mathsizes,7.4\do{%
    \@for\tavie@f:=U/msa/m/n,U/msb/m/n,U/euf/m/n,U/euf/b/n\do{%
      \expandafter\global\expandafter\let
        \csname\tavie@f/\tavie@s\endcsname\relax}}%
  \mathversion{normal}%
  \hypersetup{linkcolor={BrickRed},filecolor={BrickRed},citecolor={RoyalBlue},urlcolor={Blue}}%
  \if@filesw
    \immediate\write\@auxout{\string\@writefile{toc}{\string\setcounter{tocdepth}{\the\c@tocdepth}}}%
  \fi
  \def\@extra@b@citeb{@supp}%
}
\makeatother

\allowdisplaybreaks

\begin{document}

\AddToHookNext{shipout/foreground}{%
  \put(0,0){\raisebox{-0.42in}[0pt][0pt]{\makebox[\paperwidth][c]{%
    \parbox[t]{\paperwidth}{\centering\normalfont\def\baselinestretch{1}\footnotesize
      Accepted in the \emph{Journal of the American Statistical Association} (Theory \& Methods)\\[1pt]
      DOI: \url{https://doi.org/10.1080/01621459.2026.2749698}}}}}}
\tavieHideMainFromTOC
\bookmark[page=1,level=0]{Main Paper}

\begin{bibunit}
\def\spacingset#1{\renewcommand{\baselinestretch}%
{#1}\small\normalsize} \spacingset{1}


\newcommand{\anon}{0} 
\if0\anon
{
  \title{\bf A Generalized Tangent Approximation based Variational Inference Framework for Strongly Super-Gaussian Likelihoods}
  \date{}
  \renewcommand\thefootnote{\fnsymbol{footnote}}
  \author[1,*,\textdagger]{Somjit Roy}
  \author[1,*]{Pritam Dey}
  \author[2]{Debdeep Pati}
  \author[1]{Bani K. Mallick}

  \affil[1]{\small Department of Statistics, Texas A\&M University, College Station, TX 77843}
  \affil[2]{\small Department of Statistics, University of Wisconsin-Madison, Madison, WI 53706}

  \footnotetext[1]{
  These authors contributed equally.
  \textsuperscript{\textdagger}Corresponding author, e-mail: \href{mailto:sroy\_123@tamu.edu}{sroy\_123@tamu.edu}.
  }
  \maketitle
} \fi

\if1\anon
{
  \bigskip
  \bigskip
  \bigskip
  \begin{center}
    {\LARGE\bf A Generalized Tangent Approximation based Variational Inference Framework for Strongly Super-Gaussian Likelihoods}
\end{center}
  \medskip
} \fi

\renewcommand\thefootnote{\arabic{footnote}}


\bigskip
\begin{abstract}
\noindent{
Variational inference, as an alternative to Markov chain Monte Carlo sampling, has played a transformative role in enabling scalable computation for complex Bayesian models. Nevertheless, existing approaches often depend on either rigid model-specific formulations or stochastic black-box optimization routines.  Tangent approximation is a principled class of structured variational methods that exploit the geometry of the underlying probability model. However, its utility has largely been confined to logistic regression and related modeling regimes. In this article, we propose a novel variational framework based on tangent transformation for a broad class of probability models characterized by strongly super-Gaussian likelihoods. Our method leverages convex duality to construct tangent minorants of the log-likelihood, thereby inducing conjugacy with Gaussian priors over model parameters in an otherwise intractable setup. Under mild assumptions on the data-generating mechanism, we establish algorithmic convergence guarantees, a contribution that stands in contrast to the limited theoretical assurances typically available for black-box variational methods. Additionally, we derive near-parametric variational risk bounds. Superior performance of our proposed methodology is illustrated on simulated and real-data scenarios that challenge state-of-the-art variational algorithms in terms of scalability and their ability to consistently capture complex underlying data structure.
}
\end{abstract}


\noindent%
\vfill

\newpage
\spacingset{1.8}


\section{Introduction}\label{main-section:introduction}

\subsection{Variational Inference for Scalable Bayesian Computation}

Monte Carlo (MC) methods have long served as the backbone of Bayesian computation, enabling asymptotically exact posterior sampling~\citep{Gelfand-MCMC-JASA,Neal2011HMC}. However, the advent of large and complex datasets has spurred the development of approximate Bayesian techniques~\citep{Rockova2017ParticleEM,RossellAbrilBhattacharya2021} that trade some statistical accuracy for substantial computational gains, among which variational inference (VI) has emerged as the most prominent one. Typically, VI leverages deterministic optimization, minimizing a divergence measure (such as Kullback-Leibler (KL) divergence) between the intractable target posterior and a family of tractable distributions~\citep{katsevich2024approximation}, regarded as the variational family. Recasting Bayesian inference as an optimization problem allows VI to scale to modern large-data applications spanning graphical models~\citep{Jordan1999-graphical-VI,Wainwright-Jordan-Graphical-VI}, hidden Markov models~\citep{MacKay-technical-report}, latent variable models~\citep{Blei-LDA-JMLR}, and neural networks~\citep{Graves-NN-VI}. For a comprehensive review on various aspects of VI, see~\cite{Blei-VI-review-JASA}.

A substantial body of VI research centers on formulations that impose structural assumptions within the variational family~\citep{Hoffman-Blei-SSVI-2015}. A canonical instance is mean-field variational inference (MFVI), which, owing to its conceptual and computational simplicity, has been widely employed in statistical physics~\citep{parisi1988statistical} and later in Bayesian statistics~\citep{Jordan1999-graphical-VI}. The mean-field approximation enables an efficient block-wise optimization routine, coordinate ascent VI (CAVI), with well-established convergence guarantees and optimality properties~\citep{Wang03072019,AlquierRidgway2020,alpha-VB,Pati-Annals-CAVI}. While this structural independence ensures scalability, it often yields overconfident posteriors that underestimate true uncertainty~\citep{Blei-VI-review-JASA}. Complementing MFVI, expectation propagation (EP)~\citep{minka-EP} refines local approximations via moment matching, yielding richer posterior representations. Like EP, non-conjugate variational message passing (NCVMP)~\citep{NCVMP} relies on local message updates, extending them to non-conjugate exponential family models~\citep{Nott-GLMMs-VI}.
{
\renewcommand{\baselinestretch}{1.0}\normalsize
\begin{figure}[t!]
  \centering
  \includegraphics[width=\linewidth]{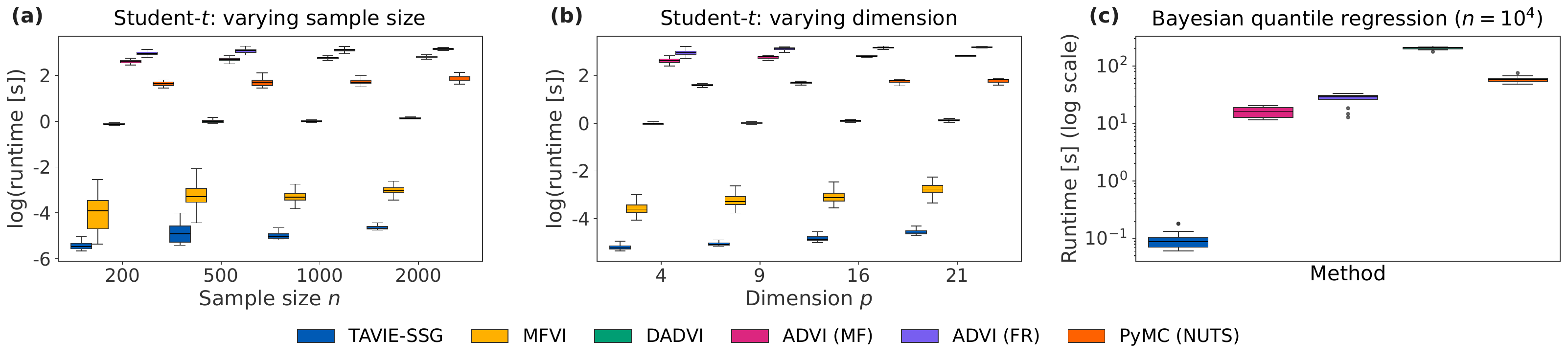}
  \caption{\footnotesize{Runtimes (in $\log$-scale) of $\tssg$ and competitors.
  (a)-(b) Student's-$t$ Type I $\ssg$ likelihood with
  $\nu=5$ across $100$ data repetitions in
  Section~\ref{main-subsec:sim-exp-student}, under
  varying sample sizes $n$ and feature dimensions $p$.
  (c) Bayesian quantile regression (BQR) on the sub-sampled
  ($n=10^{4}$) U.S. Census 2000 dataset in
  Section~\ref{main-subsec:Census-data-study}.}}
  \label{main-fig:runtime_student_BQR}
\end{figure}
}

More recently, black-box VI (BBVI) algorithms~\citep{ranganath2014bbvi,Blei-VI-review-JASA} have emerged as a flexible alternative to structured VI methods, extending applicability to a broad class of probabilistic models beyond the scope of MFVI and EP by operating under minimal structural assumptions. Their flexibility makes them well-suited for complex, non-conjugate, and high-dimensional Bayesian models~\citep{kingma2014vae}.
In particular, automatic differentiation VI (ADVI)~\citep{advi} is a widely celebrated BBVI algorithm that employs stochastic gradient estimators of the variational objective, providing scalable and easily deployable inference routines. Further,~\cite{dadvi} proposed deterministic ADVI (DADVI), which replaces stochastic gradient estimation by fixed MC approximations. These BBVI algorithms act as the major workhorse for inference in several probabilistic programming frameworks~\citep{gal2016dropout,tran2016edward,bingham2019pyro}. However, the generality of BBVI methods frequently comes at the expense of computational efficiency because of their reliance on MC approximations.  In certain modeling regimes, as shown in Figure~\ref{main-fig:runtime_student_BQR} and supported by the numerical studies in Section \ref{main-sec:TAVIE-SSG-applications}, exact MC sampling algorithms can achieve better scalability than BBVI methods, reducing the extent to which these VI methods fulfill their intended computational gains.

Bridging the preceding methodological divide, tangent approximation (or tangent transformation) based VI offers a principled middle ground~\citep{Jaakkola-PhD-thesis,Jaakkola-Jordan-2000}. By exploiting convex duality~\citep{Rockafellar1970} of the $\log$-likelihood, tangent approximation provides a deterministic and fully explainable optimization framework, while being independent of MC approximations akin to BBVI. Such characteristics not only enhance scalability but also allow for provable convergence guarantees and statistical optimality~\citep{Ghosh-JMLR-VI}. Collectively, these insights motivate a systematic treatment of tangent approximation as the foundation for a unified VI framework applicable to a wide class of probabilistic models.

\subsection{Related Work and Our Contributions}
\label{main-subsec:related-works-contribution}

We begin by reviewing the literature on tangent approximation. This technique, introduced by~\cite{Jaakkola-Jordan-2000}, employs convex duality to construct tangent minorants of the logistic $\log$-likelihood, thereby inducing conjugacy under a Gaussian prior endowed over the regression parameters in an otherwise intractable setup. It has found widespread applications in graphical models~\citep{Jordan1999-graphical-VI}, low-rank approximations~\citep{low-rank-aaprox-VI}, non-conjugate latent Gaussian models~\citep{Khan-Seeger},
and sparse kernel machines~\citep{Shi-Yu}. Also,~\cite{durante2019conditionally} showed that tangent approximation can be reformulated as MFVI for conditionally conjugate inference in Bayesian logistic models through P\'{o}lya-Gamma data augmentation. Further,~\cite{Ghosh-JMLR-VI} studied the optimality and algorithmic stability of tangent approximation for logit and multinomial logit models. Despite these advances, the adoption of tangent approximation beyond logit models remains elusive. To address this limitation, we utilize the notion of strong super-Gaussianity with respect to likelihoods.

Strong super-Gaussianity~\citep{palmer-2005-paper,palmer-book-chapter} enables parametric minorization of probability densities. Notably,~\cite{seeger2011} considered the factorization of the prior distribution into a product of strongly super-Gaussian ($\ssg$) potentials to develop a scalable variational framework for image learning under the sparse linear model regime. Our approach is fundamentally different, where we exploit the $\ssg$ property intrinsic to a broad class of likelihood functions. 

Building on this idea, we develop $\mathsf{T}$angent $\mathsf{A}$pproximation based $\mathsf{V}$ariational $\mathsf{I}$nferenc$\mathsf{E}$ for $\ssg$ likelihoods ($\tssg$), a framework that leverages the geometry of the underlying $\ssg$ likelihood to construct a tangent minorant and uses the resulting lower bound as the working likelihood. { In principle, this general construction provides a variational recipe that can lead to tractable inference across a wide class of Bayesian models. However, in this article we focus on two key classes of likelihoods, where specification of priors conjugate to the minorant enables an efficient VI framework. The first comprises linear regression models with non-Gaussian (heavy-tailed) errors~\citep{RossellRubio2018} arising in robust regression~\citep{huber-robust-regression}. By invoking the $\ssg$ property of heavy-tailed likelihoods, we provide a novel and scalable VI approach, with applications spanning financial risk modeling~\citep{graphical-modeling-heavy-tailed-markets} and skewed, heavy-tailed models such as Bayesian quantile regression~\citep{Yu-Moeed-Bayesian-QR}. The second class includes discrete-response models, encompassing both the logistic regression formulation of~\cite{Jaakkola-Jordan-2000} and the Negative-Binomial regression model, which is popular in biostatistics~\citep{He2021NEBULA}.}

$\tssg$ works under an $\alpha$-fractional likelihood setup, where the original likelihood is raised to an exponent $\alpha \in (0,1]$. This tempering step aligns our formulation with the literature on Bayesian fractional posteriors~\citep{FrielPettitt2008,Bayesian-fractional-posterior} and related modeling settings~\citep{chernozhukov2003mcmc,jiang2008gibbs}. Minorizing the obtained tempered likelihood and combining it with a conjugate prior, we derive the corresponding $\alpha$-fractional variational posteriors (or power variational posteriors) for the considered $\ssg$ models and devise an efficient variational Expectation-Maximization (EM) algorithm for the optimization of the ensuing variational parameters. With $n$ being the sample size, $\tssg$'s construction facilitates the decomposition of the $n$-variate optimization into $n$ independent univariate closed-form updates per iteration, allowing the algorithm to scale linearly in $n$ and rendering it embarrassingly parallelizable. This scalability advantage is strongly corroborated by numerical comparisons with competing VI and MC methods in Figure~\ref{main-fig:runtime_student_BQR} and Section \ref{main-sec:TAVIE-SSG-applications}.

From the theoretical standpoint, a primary technical contribution of our work is to establish convergence of the $\tssg$ algorithm with associated rates starting from any arbitrary initialization, despite the nonconvex nature of the objective function (see Figure~\ref{main-fig:ELBO_landscape}). Specifically, we exploit the Kurdyka-\L ojasiewicz property~\citep{lojasiewicz1963propriete,kurdyka1998gradients} to guarantee that the sequence of $\tssg$ iterates converges to a fixed point of the underlying map of the variational EM updates. While such arguments are standard in the broader nonconvex optimization literature~\citep{attouch2010proximal,bolte2014proximal}, they have not previously been leveraged to study convergence of variational algorithms. Furthermore, adopting a frequentist perspective to assess the statistical accuracy of the variational proxy, we derive near-parametric bounds for the integrated variational risk under both $\alpha$-R\'{e}nyi divergence (for the fractional likelihood) and Hellinger distance (for the standard likelihood). These methodological and theoretical advances are substantiated with extensive empirical analyses demonstrating the superior scalability, accuracy, and robustness of $\tssg$ relative to state-of-the-art VI and MC algorithms across both simulated and real-world data scenarios.

The remainder of the article is organized as follows. Section \ref{main-section:notation} lays out the notational conventions used throughout. Section \ref{main-subsec:SSG-likelihood-models} introduces the classes of $\ssg$ likelihood models. Section \ref{main-subsec:prior-posterior} develops the $\alpha$-fractional formulation and derives the corresponding fractional variational posterior distributions under conjugate priors. Section \ref{main-subsec:tavie-algo} presents the $\tssg$ variational EM algorithm. Section~\ref{main-subsec:calibration-alpha} provides a practical recommendation for choosing $\alpha$. Section \ref{main-section:theory} contains theoretical results, including convergence guarantees of the $\tssg$ iterates and near-parametric bounds for the integrated variational risk under both fractional and standard likelihood settings. Section \ref{main-sec:TAVIE-SSG-applications} demonstrates the empirical success of $\tssg$, while Section \ref{main-section:discussion} wraps up with a discussion.


\subsection{Notation}\label{main-section:notation}

Let $\mathbb{R}$ and $\mathbb{N}$ denote the real and natural numbers, and let $\mathbb{R}^{+}$ and $\mathbb{R}_{0}^{+}$ denote the positive and nonnegative reals, respectively. For $n,p\ge 1$, $\mathbb{R}^{n}$, $\mathbb{R}^{n}_{+}$, and $\mathbb{R}^{n\times p}$ denote Euclidean vector and matrix spaces, and $\mathbb{Z}_{+,0}^{n}$ denotes the set of $n$-dimensional vectors of nonnegative integers. For  $x \in \mathbb{R}^{n}$, $\lVert x\rVert_{2}$ is its $\ell_2$ norm and $\mathrm{diag}(x)$ is a diagonal matrix of order $n$ with $i$th diagonal entry $x_i$. We use $[n] = \{1, \ldots, n\}$ for index sets, $\equiv$ for definitional equivalence, $|A|$ for the determinant of a square matrix $A$, and $I_p$ for the identity matrix of order $p$. For sequences $a_n, b_n$, $a_n = \mathcal{O}(b_n)$ signifies $|a_n| \leq C |b_n|$ for some constant $C$ and $a_n \asymp b_n$ denotes equivalence up to universal positive constants. For a real-valued function $f(x)$, $\nabla_{x} f(x)$ is its gradient with respect to $x$. The indicator function of the event $A$ is $\mathds{1}(A)$. We write $\mathcal{N}_{p}(\mu, \Sigma)$ and $\mathcal{N}_{p}(x\mid \mu, \Sigma)$ for the Gaussian distribution and its evaluation at $x\in \mathbb{R}^{p}$, and analogously $\mathcal{NG}_{p}(\mu, \Sigma, a, b)$ and $\mathcal{NG}_{p}(x, y \mid \mu, \Sigma, a, b)$ for the Normal-Gamma distribution and its evaluation at $(x, y) \in \mathbb{R}^{p}\times \mathbb{R}^{+}$. Independently and identically distributed is denoted by i.i.d. 
Finally, $\mathbb{P}_{\theta_0}$ is the probability measure induced by the data-generating distribution under the true parameter value $\theta_0$.


\section{The \texorpdfstring{$\tssg$}{TAVIE-SSG} Framework}\label{main-section:methods}

At the outset, we define the class of {strongly super-Gaussian} ($\ssg$) {density functions}. 

\begin{definition}[Strongly super-Gaussian ($\ssg$) density function~\citep{palmer-2005-paper}]
\label{main-def:SSG-density}
A probability density function $p:\mathbb{R} \to \mathbb{R}^{+}$ which is symmetric about $0$ is $\ssg$ if the mapping $g:\mathbb{R}^{+}_0 \to \mathbb{R}$ defined by $g(s) := \log p(\sqrt{s})$ is convex and monotonically decreasing.
\end{definition}

Prominent examples of $\ssg$ density functions include the {generalized Gaussian} family ($p(z) \propto \exp\{-\tau |{z}|^{d}\}$, for $\tau \in \mathbb{R}^+$, $0 < d \leq 2$, and $z \in \mathbb{R}$) and any {scale mixture of Gaussians} ($p(z) \propto \int_{0}^{\infty} \exp\{-sz^2\}k(s)ds$, with $k$ being any density supported on $\mathbb{R}^{+}$ and $z \in \mathbb{R}$). Motivated by the goal of encompassing a broad class of probabilistic models, we generalize the notion of strong super-Gaussianity to likelihood functions in Definition \ref{main-def:SSG}.


\subsection{Strongly Super-Gaussian Likelihood Models}\label{main-subsec:SSG-likelihood-models}

Let a collection of observed data units be $\mathcal{D}_n = \{(\mathbf{x}_i, y_i): i\in [n]\}$ with design matrix $\mathbf{X} = (\mathbf{x}_1, \ldots, \mathbf{x}_n)^{\top} \in \mathbb{R}^{n \times p}$ and response vector $y = (y_1, \ldots, y_n)^{\top} \in \mathbb{R}^n$. Definition \ref{main-def:SSG} below characterizes the class of models with $\ssg$ likelihoods.

\begin{definition}[Strongly super-Gaussian ($\mathsf{SSG}$) likelihood function]
\label{main-def:SSG}
The conditional likelihood of $y_i$ given $\mathbf{x}_i$ is said to be an {$\ssg$} likelihood function if:
\begin{equation}\label{main-eq:SSG-defintion}
p(y_i \mid \mathbf{x}_i, \theta) \propto r_i \, \exp \left\{\,s_i \zeta_i + t_i\,h(\zeta_i^2)\right\},\quad \zeta_i := u_i \mathbf{x}_i^{\top}\beta + v_i,
\end{equation}
independently for $i\in [n]$, where $\theta\in \Theta$ contains the model parameters including the regression coefficient vector $\beta \in \mathbb{R}^{p}$ and (possibly) dispersion parameter $\tau$;  $h:\mathbb{R}^{+}_{0} \to \mathbb{R}$ is strictly convex, monotonically decreasing, and at least twice continuously differentiable on $\mathbb{R}^{+}$; and $r_i, t_i \in \mathbb{R}^{+}, s_i, u_i, v_i\in \mathbb{R}$ are constants (possibly) depending on $y_i$ and $\tau$. 
\end{definition}
Building on this general definition, we next present two key classes of $\ssg$ likelihoods. 

\textbf{Type I $\ssg$ likelihoods (Heavy-tailed families)}. Consider the linear regression model, $y_i = \mathbf{x}_i^{\top}\beta + \epsilon_i$, with i.i.d. $\ssg$ error $\epsilon_i$. The likelihood then takes the form:
\begin{equation}
\label{main-eq:SSG-typeI-likelihood}
p(y_i \mid \mathbf{x}_i, \theta) \propto \tau \exp\left\{h\left(\tau^2(y_i - \mathbf{x}_i^{\top}\beta)^2\right)\right\}, \quad i\in [n],
\end{equation}
with $\theta = (\beta^{\top}, \tau^2)^{\top} \in \mathbb{R}^p \times \mathbb{R}^{+}$ and $h$ being convex, decreasing, and twice differentiable on $\mathbb{R}^{+}$. Notable examples include the {Laplace} family ($h(t) =-\sqrt{t}$) and the {Student's-$t$} family with fixed $\nu \in \mathbb{N}$ degrees of freedom ($h(t) =-(\nu + 1) \log (1 + t/\nu)/2$). In general, error distributions representable as scale mixtures of Gaussians fall under this category in \eqref{main-eq:SSG-typeI-likelihood}.

\textbf{Type II $\ssg$ likelihoods (Bernoulli-type models)}.
We consider discrete-response models induced by a sequence of {Bernoulli} trials having the following likelihood form:
\begin{equation}
\label{main-eq:SSG-typeII-likelihood}
p(y_i \mid \mathbf{x}_i, \beta) \propto 
\frac{\exp\{a_i \mathbf{x}_i^{\top}\beta\}}{\left[1 + \exp\{\mathbf{x}_i^{\top}\beta\}\right]^{b_i}} 
= \exp\left\{\left(a_i-\frac{b_i}{2}\right)\mathbf{x}_i^{\top}\beta + b_i h\left((\mathbf{x}_i^{\top}\beta)^2\right)\right\},
\end{equation}
independently for $i \in [n]$, where $a_i$ and $b_i$ may depend on $y_i$, $\beta \in \mathbb{R}^p$, and $h(t) = -\log\left[2\cosh(\sqrt{t}/2)\right]$. This class consists of the {Binomial} ($y_i\mid \mathbf{x}_i, \beta \sim \mathrm{Bin}(m_i, p_i)$, $m_{i}\in \mathbb{N}$, $a_{i}=y_{i}$, $b_{i}=m_{i}$) and the {Negative-Binomial} ($y_i\mid \mathbf{x}_i, \beta \sim \mathrm{NB}(m_i, p_i)$, $m_{i}\in \mathbb{R}^{+}$, $a_{i}=m_{i}$, $b_{i}=m_{i}+y_{i}$) regression models, with size parameter $m_{i}$ and $p_i = (1+ \exp\{-\mathbf{x}_i^{\top}\beta\})^{-1}$, for $i\in [n]$, and subsumes the {logistic} regression case \citep{Jaakkola-Jordan-2000}. The parameterizations of the aforementioned $\ssg$ likelihoods appear in Section \ref{sec:parameterization-details} of Supplementary Materials and the corresponding $(r_i, s_i, t_i, u_i, v_i, \zeta_i)$ specifications are summarized in Table \ref{main-tab:ssg-types}.

{
\renewcommand{\baselinestretch}{1.0}\normalsize
\begin{table}[H]
\centering
\caption{\small Specification of $(r_i, s_i, t_i, u_i, v_i, \zeta_i)$ in Definition \ref{main-def:SSG} for Type I and Type II $\ssg$ likelihoods.}
\label{main-tab:ssg-types}
\begin{tabular}{lccccccc}
\toprule
\toprule
Likelihood Type & $r_i$ & $s_i$ & $t_i$ & $u_i$ & $v_i$ & $\zeta_i$ & $\theta$\\
\midrule
Type~I $\ssg$ & $\tau$ &
$0$ &
$1$ &
$-\tau$ &
$\tau y_i$ &
$\tau\,(y_i - \mathbf{x}_i^{\top}\beta)$ &
$(\beta^{\top}, \tau^2)^{\top}$\\
Type~II $\ssg$ & $1$ &
$a_i - \tfrac{b_i}{2}$ &
$b_i$ &
$1$ &
$0$ &
$\mathbf{x}_i^{\top}\beta$ &
$\beta$\\
\midrule
\bottomrule
\end{tabular}
\end{table}
}
We now observe that the $\ssg$ likelihood  in \eqref{main-eq:SSG-defintion} admits the following {tangent minorizer}~\citep{palmer-2005-paper}, which forms the backbone of the $\tssg$ framework. 
\begin{proposition}[Tangent minorizer]\label{main-prop:tangent-lower-bound}
For {any} $\xi_i \in \mathbb{R}^{+}$, the likelihood in \eqref{main-eq:SSG-defintion} has a minorizer:
\begin{equation}\label{main-eq:general-minorizer}
\varphi(y_i \mid \mathbf{x}_i, \theta, \xi_i) \propto r_i \, \exp \left\{\,s_i \zeta_i + t_i\,A(\xi_i)\zeta_i^2 + t_i\,\gamma(\xi_i)\right\} \leq p(y_i \mid \mathbf{x}_i, \theta),
\end{equation}
for $i\in [n]$, where $\gamma(t) = h(t^2) - t^2h'(t^2)$ and $A(t) = h'(t^2)$. Equality in \eqref{main-eq:general-minorizer} holds if and only if $|\zeta_i| = \xi_i$, i.e., $\varphi(y_i \mid \mathbf{x}_i, \theta, \xi_i)$ is tangent to $p(y_i \mid \mathbf{x}_i, \theta)$ at $|\zeta_i| = \xi_i$.
\end{proposition}
\begin{proof}
By strict convexity of $h$, the tangent minorization in~\eqref{main-eq:general-minorizer} follows from the supporting hyperplane inequality, $h(\zeta_i^2)\ge h(\xi_i^2)+h'(\xi_i^2)(\zeta_i^2-\xi_i^2)$, with equality if and only if $|\zeta_i|=\xi_i$.
\end{proof}
Our $\tssg$ framework entails three key components: (i) {tangent minorization} of the likelihood via Proposition \ref{main-prop:tangent-lower-bound}, (ii) conjugate prior specification for the resulting minorized likelihood, and (iii) {optimization} of the ensuing variational posterior with respect to the variational parameters {$\{\xi_i\}_{i \in [n]}$}. We now develop these components for the two main $\ssg$ types, noting that the construction can be adapted to any $\ssg$ likelihood in Definition~\ref{main-def:SSG}.

For the Type I and Type II $\ssg$ likelihoods in \eqref{main-eq:SSG-typeI-likelihood}-\eqref{main-eq:SSG-typeII-likelihood}, the tangent minorizers are:
\begin{align}\label{main-eq:minorizers-type-I-II}
    \varphi(y_i \mid \mathbf{x}_i, \theta, \xi_i) \propto \begin{cases}
        \tau \exp\left\{A(\xi_i)\tau^2(y_i - \mathbf{x}_i^{\top}\beta)^2 + \gamma(\xi_i)\right\}, & \text{Type I $\ssg$},\\
        \exp\left\{\left(a_i-\frac{b_i}{2}\right)\mathbf{x}_i^{\top}\beta + b_i A(\xi_i)(\mathbf{x}_i^{\top}\beta)^2 + b_i\gamma(\xi_i)\right\}, & \text{Type II $\ssg$},
    \end{cases}
\end{align}
where $\theta = (\beta^{\top}, \tau^2)^{\top}$ and $\theta = \beta$ for Type I and Type II $\ssg$ likelihoods, respectively, along with $A(t)$ and $\gamma(t)$ as in Proposition \ref{main-prop:tangent-lower-bound}.

The tangent minorizer $\varphi(y_i\mid \mathbf x_i, \theta, \xi_i)$ in~\eqref{main-eq:general-minorizer} is an adaptive exponential-quadratic approximation in $\zeta_i$ to the true $\ssg$ likelihood in~\eqref{main-eq:SSG-defintion}. Unlike a simple Gaussian approximation to the response distribution~\citep{breslow1993approximate}, $\tssg$ continues to target the original non-Gaussian model and uses this local surrogate only to obtain tractable variational updates in $\xi_i$. To illustrate, under linear regression with Student's-$t$ errors, the true likelihood is replaced by a tangent minorizer that, although Gaussian in $\zeta_i$, has an adaptive curvature $-A(\xi_i) = (\nu+1)/[2(\nu+\xi_i^2)]$, which decreases with $|\xi_i|$ and hence downweights large residuals. In this way, the surrogate remains locally faithful to the true likelihood near $\xi_i$ (see panel (a) of Figure~\ref{main-fig:minorizer_plot}), while retaining the robustness of the Student's-$t$ model. The same intuition applies to other Type I $\ssg$ likelihoods, including Laplace regression with adaptive curvature $(2\xi_i)^{-1}$. For Type II $\ssg$ likelihoods, the tangent minorizer likewise provides a local approximation around $\xi_i$ (see panel (c) of Figure~\ref{main-fig:minorizer_plot}), but without replacing the discrete model by a Gaussian approximation, thereby preserving the mean-variance structure and discreteness of the true likelihood. This minorization induces the Jensen's gap:
\begin{align}
\label{main-eq:jensen's-gap}
\Delta_{2, i} := \log p(y_{i}\mid \mathbf{x}_i, \theta) - \log \varphi(y_i \mid \mathbf{x}_i, \theta, \xi_i) = t_i[h(\zeta_i^2) - h(\xi_i^{2}) - h'(\xi_i^2)(\zeta_i^2 - \xi_i^2)],
\end{align}
for $i\in [n]$, where $\zeta_i$ and $t_i$ for Type I and Type II $\ssg$ likelihoods are as in Table~\ref{main-tab:ssg-types}. Thus, $\Delta_{2, i}$ is precisely the vertical discrepancy between the true $\log$-likelihood and its tangent minorizer. Panels (b) and (d) of Figure~\ref{main-fig:minorizer_plot} illustrate that near the optimized variational parameter $\xi^{\star}$, the tangent minorizer closely tracks the true $\log$-likelihood. A detailed empirical study of this gap is provided in Section \ref{subsec:jensen's-gap} of Supplementary Materials. 

{
\renewcommand{\baselinestretch}{1.0}\normalsize
\begin{figure}[!t]
    \centering
    \includegraphics[width=\linewidth]{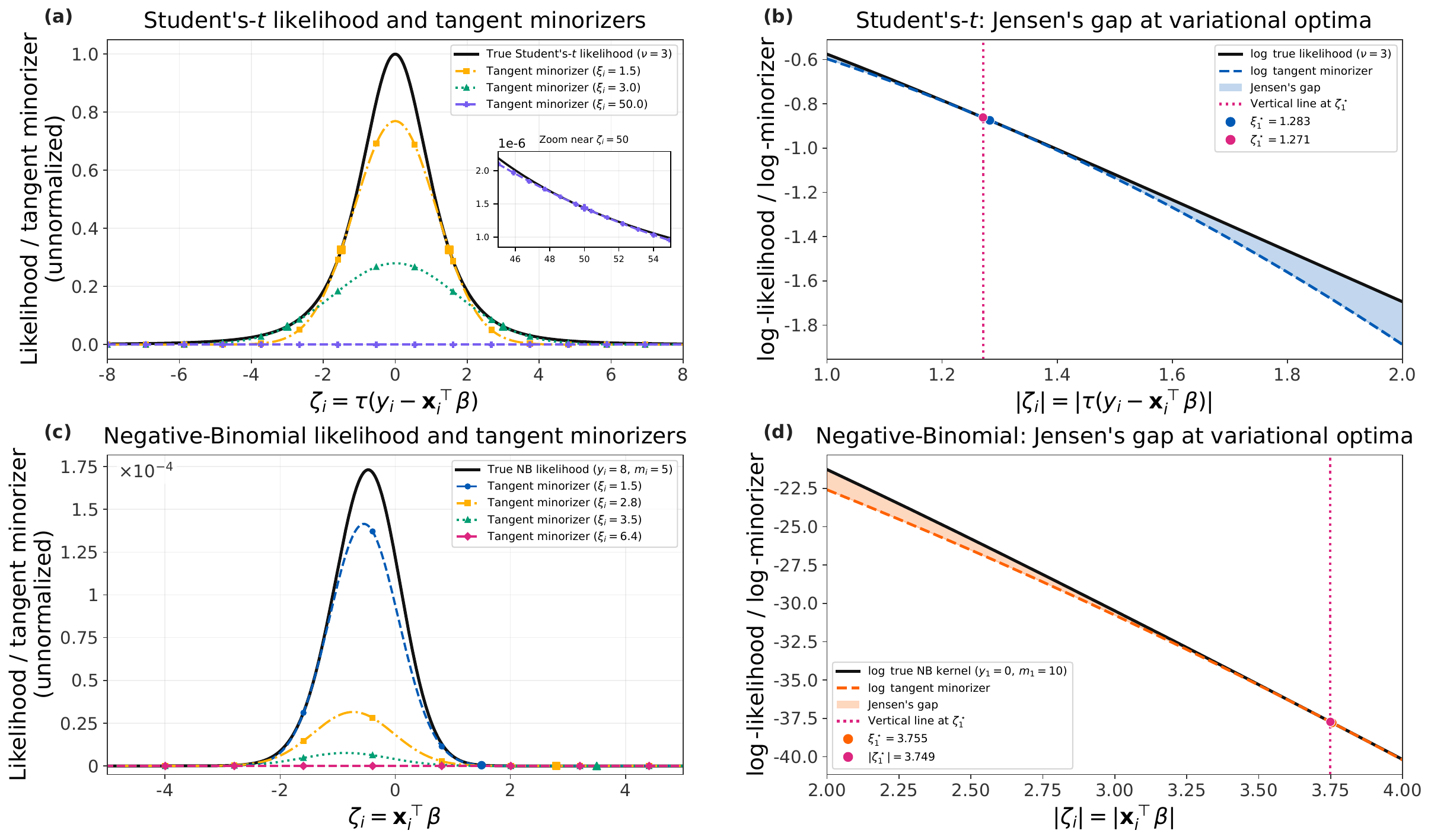}
    \caption{\footnotesize{Illustration of the tangent minorization mechanism underlying $\tssg$ for Type I and Type II $\ssg$ likelihoods. \emph{Panels} (a), (c): Adaptive exponential-quadratic surrogates constructed at different $\xi_i$'s. \emph{Panels} (b), (d): Jensen's gap near the optimized variational parameter $\xi_{1}^{\star}$ and the fitted latent value $\zeta_1^{\star}$.}}
    \label{main-fig:minorizer_plot}
\end{figure}
}


\subsection{Priors and \texorpdfstring{$\alpha$-}{ }Fractional Variational Posteriors}\label{main-subsec:prior-posterior}

Consider the {true} $\alpha$-fractional likelihood, $p_{\alpha}(y \mid \mathbf{X}, \theta) = \{\prod_{i \in [n]} p(y_i \mid \mathbf{x}_i, \theta)\}^\alpha$ and its tangent minorizer, $\varphi_{\alpha}(y \mid \mathbf{X}, \theta, \xi) = \left\{\prod_{i\in [n]} \varphi(y_i \mid \mathbf{x}_i, \theta, \xi_i)\right\}^{\alpha}$, for any fixed choice of $\alpha \in (0, 1]$. Here, $\xi = (\xi_1, \ldots, \xi_n)^{\top}\in \mathbb{R}_{+}^{n}$ is the $n$-vector of variational parameters. Note that $\alpha=1$ recovers the standard tangent minorizer as in~\eqref{main-eq:general-minorizer}. With $\varphi_{\alpha}(y\mid \mathbf{X}, \theta, \xi)$ as the {working likelihood} and prior distribution $\pi(\theta)$ over the model parameters in $\theta$, the tractable $\tssg$ $\alpha$-fractional variational posterior surrogate distribution is:
\begin{equation}\label{main-eq:alpha-frac-posterior}
\pi_{\alpha}(\theta \mid \mathcal{D}_n, \xi) = \frac{\varphi_{\alpha}(y \mid \mathbf{X}, \theta, \xi) \pi(\theta)}{\varphi_{\alpha}(y \mid \mathbf{X}, \xi)},
\qquad \varphi_{\alpha}(y \mid \mathbf{X}, \xi) = \int_{\theta \in \Theta} \varphi_{\alpha}(y \mid \mathbf{X}, \theta, \xi) \pi(\theta) d\theta.
\end{equation}
In \eqref{main-eq:alpha-frac-posterior}, $\varphi_{\alpha}(y \mid \mathbf{X}, \xi)$ admits a closed form if $\pi(\theta)$ is conjugate to the tangent minorizer $\varphi_{\alpha}(y \mid \mathbf{X}, \theta, \xi)$. 
For Type I $\ssg$ likelihoods, we consider the {multivariate Normal-Gamma} prior distribution, $\pi(\theta) \equiv \pi(\beta, \tau^2) \equiv \mathcal{NG}_p(\beta, \tau^2 \mid \mu, \Sigma, a,b)$, for which the $\alpha$-fractional variational posterior distribution is $\pi_{\alpha}(\theta \mid \mathcal{D}_n, \xi) \equiv \mathcal{NG}_p(\beta, \tau^2 \mid \mu_{\alpha}(\xi), \Sigma_{\alpha}(\xi), a_{\alpha},b_{\alpha}(\xi))$ having hyperparameters:
\begin{equation}\label{main-eq:typeI-variational-hyperparameters}
\begin{split}
&\Sigma_{\alpha}(\xi) = \left[\Sigma^{-1} -2\alpha\mathbf{X}^{\top}\mathcal{A}(\xi)\mathbf{X}\right]^{-1},\quad \mu_{\alpha}(\xi) = \Sigma_{\alpha}(\xi)\left[\Sigma^{-1}\mu - 2\alpha\mathbf{X}^{\top}\mathcal{A}(\xi)y\right],\\
&a_{\alpha} = a + n\alpha,\quad b_{\alpha}(\xi) = b - 2\alpha y^{\top}\mathcal{A}(\xi)y + \mu^{\top} \Sigma^{-1} \mu - \mu_{\alpha}(\xi)^{\top} \Sigma^{-1}_{\alpha}(\xi) \mu_{\alpha}(\xi),
\end{split}    
\end{equation}
where $\mathcal{A}(\xi) = \mathrm{diag}(A(\xi_1), \ldots, A(\xi_n))$. Similarly, for Type II $\ssg$ likelihoods, a {multivariate Gaussian} prior, $\pi(\theta) \equiv \pi(\beta) \equiv \mathcal{N}_p(\beta \mid \mu, \Sigma)$, yields the $\alpha$-fractional variational posterior distribution, $\pi_{\alpha}(\beta \mid \mathcal{D}_n, \xi) \equiv \mathcal{N}_p(\beta\mid \mu_{\alpha}(\xi), \Sigma_{\alpha}(\xi))$, having hyperparameters:
\begin{equation}\label{main-eq:typeII-variational-hyperparameters}
\begin{split}
&\Sigma_{\alpha}(\xi) = \left[\Sigma^{-1} - 2\alpha \mathbf{X}^{\top}\mathcal{A}(\xi)\mathrm{diag}(\mathbf{b})\mathbf{X}\right]^{-1},\quad
\mu_{\alpha}(\xi) = \Sigma_{\alpha}(\xi)\left[\Sigma^{-1}\mu + \alpha \mathbf{X}^{\top}\left(\mathbf{a} - \frac{\mathbf{b}}{2}\right)\right],
\end{split}    
\end{equation}
where $\mathbf{a} = (a_1, \dots, a_n)^{\top}$, $\mathbf{b} = (b_1, \dots, b_n)^{\top}$, and $\mathcal{A}(\xi)$ have the same form as in \eqref{main-eq:typeI-variational-hyperparameters}.

Customarily, VI is performed by maximizing the evidence lower bound ($\mathsf{ELBO}$)~\citep{Blei-VI-review-JASA}, $\mathcal{L}_{\mathrm{true}}(q) = \int_{\theta\in \Theta}\log \frac{p_{\alpha}(y\mid \mathbf X, \theta)\pi(\theta)}{q(\theta)}q(\theta)d\theta$.
We treat $\varphi_{\alpha}(y\mid \mathbf{X}, \theta, \xi)$ as the working likelihood, analogously yielding the $\tssg$ $\mathsf{ELBO}$ as:
\begin{align}\label{main-eq:ELBO-original}
    \mathcal{L}(q, \xi) = \int_{\theta \in \Theta}\log \frac{\varphi_{\alpha}(y\mid \mathbf{X}, \theta, \xi)\pi(\theta)}{q(\theta)}q(\theta)d\theta\;\leq\;\mathcal{L}_{\mathrm{true}}(q).
\end{align}
For any fixed $\xi \in \mathbb{R}_{+}^{n}$, 
$\mathcal{L}(q, \xi)$ is maximized by the {exact} variational posterior $q_{\xi} \equiv \pi_{\alpha}(\cdot \mid \mathcal{D}_n, \xi)$ over the space of all densities $\mathcal{P}_{\Theta}$ supported on $\Theta$.
Substituting $q \equiv q_{\xi}$ in \eqref{main-eq:ELBO-original} reduces $\mathcal{L}(q, \xi)$ to the profile $\tssg$ $\mathsf{ELBO}$ $\mathsf{L}(\xi) = \mathcal{L}(q_{\xi}, \xi) = \log\varphi_{\alpha}(y\mid \mathbf{X}, \xi)$, which is maximized with respect to $\xi$ to obtain the optimal variational parameter $\xi^{\star} \in \mathbb{R}_{+}^{n}$.
The explicit forms of $\mathsf{L}(\xi)$ (up to additive constants) for Type I and Type II $\ssg$ likelihoods are:
\begin{equation}\label{main-eq:ELBO-general}
\begin{split}
\mathsf{L}(\xi) = \begin{cases}
        -\frac{a_{\alpha}}{2}\log b_{\alpha}(\xi) + \frac{1}{2}\log |\Sigma_{\alpha}(\xi)| + \alpha\sum_{i\in[n]}\gamma(\xi_i), & \text{Type I $\ssg$},\\
        \frac{1}{2}\mu_{\alpha}(\xi)^{\top}\Sigma_{\alpha}^{-1}(\xi) \mu_{\alpha}(\xi) + \frac{1}{2}\log |\Sigma_{\alpha}(\xi)| + \alpha\sum_{i\in[n]}b_i\gamma(\xi_i), & \text{Type II $\ssg$}.
    \end{cases}
\end{split}
\end{equation}
Before presenting the optimization algorithm for $\mathsf{L}(\xi)$ in Section~\ref{main-subsec:tavie-algo}, we note that replacing $\mathcal L_{\mathrm{true}}(q)$ with $\mathcal L(q,\xi)$, evaluated at $q\equiv q_\xi$, introduces an additional approximation cost given by the $q_\xi$-average of the total Jensen’s gap $\Delta_2=\sum_{i\in[n]}\Delta_{2,i}$ scaled by $\alpha$. Thus, beyond the usual variational discrepancy (Kullback-Leibler (KL) divergence) between the variational and true $\alpha$-fractional posteriors, $\tssg$ incurs a minorization-induced discrepancy arising from the use of the tangent surrogate in the $\mathsf{ELBO}$. These two sources of approximation error are examined empirically in Sections \ref{subsec:ELBO-gap}-\ref{subsec:empirical-true-variational-posterior} of Supplementary Materials.


\subsection{The \texorpdfstring{$\tssg$}{TAVIE-SSG} EM Algorithm}\label{main-subsec:tavie-algo}

Although direct optimization of $\mathsf{L}(\xi)$ over $\xi$ is feasible, it is typically nonconvex, incurring $\mathcal{O}(n^2)$ computational cost per iteration, thus limiting the scalability of our framework even for moderate $n$. To mitigate this issue, a more principled and computationally efficient approach emerges by reformulating this task as a maximum likelihood estimation problem in $\xi$, where $\theta$ is treated as a latent variable. This perspective naturally leads to an elegant {Expectation-Maximization} (EM) procedure~\citep{EM-Dempster-Laird-Rubin}, outlined as follows.

\textbf{Complete-data likelihood and surrogate construction}. Let $\mathbb{E}_{\xi}$ denote the expectation with respect to the variational posterior distribution, $\pi_{\alpha}(\theta \mid \mathcal{D}_n, \xi)$, for any $\xi \in \mathbb{R}^n_{+}$. Then, the general EM surrogate function for the $(t+1)$th step is:
\begin{align}\label{main-eq:EM-surrogate}
\begin{split}
\mathcal{Q}(\xi^{\dagger} \mid \xi^{(t)}) 
&= \mathbb{E}_{\xi^{(t)}} \left[\log \pi(\theta)\right] + \alpha \sum_{i\in[n]} \mathbb{E}_{\xi^{(t)}} \left[\log \varphi(y_i \mid \mathbf{x}_i, \theta, \xi^{\dagger}_i) \right] \\
&= \alpha \sum_{i\in[n]} t_i\left\{\,A(\xi_i^{\dagger}) \kappa_i(\xi^{(t)}) + \gamma(\xi_i^{\dagger})\right\} + \mathfrak{C}(\xi^{(t)}),
\end{split}
\end{align}
where $\mathfrak{C}(\xi^{(t)})$ collects terms free of $\xi^{\dagger}$ and $\kappa_i(\xi) = \mathbb{E}_{\xi} \left[\zeta_i^2\right]$ is given by:
\begin{align}\label{main-eq:kappa-Type-I-II}
    \kappa_i(\xi) 
    = \begin{cases}
        \mathbf{x}_i^{\top}\Sigma_{\alpha}(\xi) \mathbf{x}_i + \frac{a_\alpha}{b_{\alpha}(\xi)}\left(y_i - \mathbf{x}_{i}^{\top}\mu_{\alpha}(\xi)\right)^2, &\text{Type I } \mathsf{SSG},\\
        \mathbf{x}_i^{\top}\Sigma_{\alpha}(\xi)\mathbf{x}_i + \left(\mathbf{x}_i^{\top}\mu_{\alpha}(\xi)\right)^{2}, & \text{Type II } \mathsf{SSG},
    \end{cases}
\end{align}
for $i\in [n]$.
Next, we present the optimization of $\mathcal{Q}(\xi^{\dagger}\mid \xi^{(t)})$ in \eqref{main-eq:EM-surrogate} with respect to $\xi^{\dagger}$.

\textbf{Maximization of the surrogate}. Since the surrogate function in \eqref{main-eq:EM-surrogate} splits into a {coordinate-wise sum} over $i\in [n]$, the iterative EM update for the $(t+1)$th step involves $n$ {univariate} optimizations given by:
\begin{equation}
\label{main-eq:TAVIE-xi-update}
{\xi}_i^{(t+1)}
= \arg\max_{{\xi}_i^{\dagger} > 0}\;\mathbb{E}_{\xi^{(t)}} \left[\log \varphi(y_i \mid \mathbf{x}_i, \theta, {\xi}_i^{\dagger}) \right] = \sqrt{\kappa_i(\xi^{(t)})},\quad i\in [n],
\end{equation}
where the last equality in \eqref{main-eq:TAVIE-xi-update} above follows from $\gamma'(s) = -s^2 A'(s)$ for $s\in \mathbb{R}^{+}$, which uses the definitions of $\gamma(s)$ and $A(s)$ in Proposition \ref{main-prop:tangent-lower-bound}.
This splitting of the multivariate optimization step into $n$ univariate optimization problems with simple closed-form solutions in \eqref{main-eq:TAVIE-xi-update} makes the $\xi$-update linear in $n$ and embarrassingly parallelizable. Each $\tssg$ iteration, however, also updates the variational posterior hyperparameters, yielding an overall per-iteration complexity $\mathcal{O}(np^{2} + p^{3})$; see Section \ref{sec:time-complexity-analysis-TAVIE-SSG} of Supplementary Materials for details. Thus, $\tssg$ is computationally efficient in large $n$ and small-to-moderate $p$ settings, its primary target regime, as evidenced by the empirical results in Section \ref{main-sec:TAVIE-SSG-applications}.
We conclude by formalizing the $\tssg$ EM algorithm in Algorithm \ref{main-alg:tavie-em}.
{
\renewcommand{\baselinestretch}{1.0}\normalsize
\begin{algorithm}[t!]
\caption{The $\tssg$ EM Algorithm}
\label{main-alg:tavie-em}
\DontPrintSemicolon

\KwIn{Data $\mathcal{D}_n$, prior hyperparameters, tempering parameter $\alpha$, tolerance $\texttt{tol}$.}
\KwOut{Variational parameters $\xi^\star$ and variational posterior hyperparameters.}

\textbf{Initialize}: Set $t \gets 0$ and initialize $\xi^{(0)} \in \mathbb{R}_{+}^{n}$. \;

\Repeat{$\lVert \xi^{(t)} - \xi^{(t-1)} \rVert_2 \le$ \texttt{tol}}{
  \tcc{Update variational posterior hyperparameters}
  \uIf{Type I $\ssg$ likelihoods}{
    update $(\mu_{\alpha}(\xi^{(t)}),\, \Sigma_{\alpha}(\xi^{(t)}),\, a_{\alpha},\, b_{\alpha}(\xi^{(t)}))$ via~\eqref{main-eq:typeI-variational-hyperparameters}\;
  }
  \ElseIf{Type II $\mathsf{SSG}$ likelihoods}{
    update $(\mu_{\alpha}(\xi^{(t)}),\, \Sigma_{\alpha}(\xi^{(t)}))$ via~\eqref{main-eq:typeII-variational-hyperparameters}\;
  }

  \tcc{Update variational parameters (coordinate-wise)}
  \For{$i \in [n]$}{
    $\xi_i^{(t+1)} \gets \sqrt{\kappa_i(\xi^{(t)})}$ \quad where $\kappa_i(\xi)=\mathbb{E}_{\xi} \left[\zeta_i^2\right]$ is defined in~\eqref{main-eq:kappa-Type-I-II}\;
  }

  $t \gets t+1$\;
}

\end{algorithm}
}


\subsection{Calibration of \texorpdfstring{$\alpha$}{the tempering parameter}}\label{main-subsec:calibration-alpha}

The likelihood tempering parameter $\alpha\in(0,1]$ in Algorithm~\ref{main-alg:tavie-em} controls the strength of posterior updating in $\tssg$. As shown in the sensitivity analysis in Section \ref{sec:sensitivity-alpha} of Supplementary Materials, larger $\alpha$ yields sharper posteriors and more accurate point estimates of $\theta$, whereas smaller $\alpha$ produces more diffuse posteriors and conservative uncertainty quantification (UQ), inducing an accuracy-coverage trade-off.

Accordingly, our practical recommendation is to choose $\alpha$ through coverage calibration~\citep{SyringMartin2019}, so as to balance posterior concentration against interval estimation. Concretely, we select $\alpha$ so that the resulting highest posterior density credible region attains the target nominal frequentist coverage level; see Section \ref{sec:calibration-alpha} of Supplementary Materials for illustration. More computationally involved approaches with improved UQ for variational methods include the variational weighted likelihood bootstrap of~\cite{han2019statisticalinferencemeanfieldvariational}.


\section{Theoretical Results}\label{main-section:theory}

Theoretical investigation of VI typically proceeds along two complementary directions:
(i) convergence properties of the iterative optimization algorithm 
\citep{ZhangZhou2017, Pati-Annals-CAVI} and
(ii) frequentist properties of the resulting posterior or point estimates, quantified through the variational risk~\citep{Wang03072019,AlquierRidgway2020}. Building on this general perspective, we develop the theoretical underpinnings of $\tssg$. Proofs of all theoretical results are provided in Sections \ref{app:TAVIE-convergence-proof} and \ref{app:variational-risk-bounds} of Supplementary Materials.


\subsection{Convergence Guarantee for the \texorpdfstring{$\tssg$}{TAVIE-SSG} EM Algorithm}\label{main-subsec:algorithmic-convergence}

A key early theoretical contribution by~\cite{wu1983convergence} established the monotonicity of the EM-generated sequence of likelihood values and stationarity of its limit points under mild regularity conditions. Modern convergence theory of EM algorithms largely follows two routes. One line studies local contractions of the EM operator near a well-separated maximizer, typically requiring strong curvature conditions such as local strong concavity of the surrogate $\mathcal Q$, or equivalently a Polyak-{\L}ojasiewicz (PL) inequality~\citep{Polyak1963,lojasiewicz1963propriete} for the population likelihood. These assumptions yield linear convergence rates and underlie the analyses of \cite{balakrishnan2017statistical} and related refinements.  In contrast, the $\tssg$ variational parameter vector $\xi\in \mathbb{R}_{+}^{n}$ behaves as a latent variable from an algorithmic perspective, precluding population-level analyses.

A second line of work views EM as block-coordinate or alternating
optimization and leverages analytic or semi-algebraic structure of the
objective. Here, the Kurdyka-{\L}ojasiewicz (K\L) property~\citep{lojasiewicz1963propriete,kurdyka1998gradients} is central in proving convergence of EM-type iterates under substantially
weaker curvature conditions. More recent geometric analyses recast EM as alternating maximization over distributional and parametric spaces, using functional inequalities ($\log$-Sobolev or Wasserstein contraction), and establish convergence under global joint PL conditions for the EM objective~\citep{caprio2025fast}.

In this article, we adopt the latter strategy by viewing the $\tssg$ EM Algorithm~\ref{main-alg:tavie-em} as an alternating maximization of $\mathcal L(q,\xi)$ in \eqref{main-eq:ELBO-original} over $(q, \xi) \in \mathcal{P}_{\Theta} \times \mathbb{R}^{n}_{+}$. Specifically, each iteration maximizes the objective over $q$ with $\xi$ fixed and then over $\xi$ with $q$ fixed:
\begin{equation*}
q^{(t+1)}
=
\arg\max_{q\in\mathcal P_\Theta}
\mathcal L(q,\xi^{(t)}),
\quad
\xi^{(t+1)}
=
\arg\max_{\xi\in\mathbb R_+^n}
\mathcal L(q^{(t+1)},\xi).
\end{equation*}
As established in Section~\ref{main-subsec:prior-posterior}, the first maximization yields the exact variational posterior $q^{(t+1)} \equiv q_{\xi^{(t)}} \equiv \pi_{\alpha}(\cdot \mid \mathcal D_n, \xi^{(t)})$, whereas the second leads to the M-step update in \eqref{main-eq:TAVIE-xi-update}. Maximizing out $q$ gives the profile objective $\mathsf{L}(\xi) =
\max_{q\in\mathcal P_\Theta}\mathcal L(q,\xi)
= \mathcal L(q_\xi,\xi)$. Under mild differentiability and uniqueness assumptions on the inner maximizer $q_\xi$, stationarity of the joint objective in $(q,\xi)$ is equivalent to stationarity of $\mathsf{L}$, as solidified in Proposition \ref{main-proposition:stationarity-equivalence} below.
\begin{proposition}[Stationarity equivalence]
\label{main-proposition:stationarity-equivalence}
Let $\mathcal L(q,\xi)$ be differentiable in $\xi$ and suppose that for each $\xi$ in an open neighborhood of $\xi^\star$, the inner maximizer $q_\xi=\arg\max_{q\in\mathcal{P}_{\Theta}}\mathcal L(q,\xi)$ exists, is unique, and depends smoothly (or continuously) on $\xi$ so that the envelope theorem~\citep{Danskin1967} applies.  Then the following are equivalent:
\begin{enumerate}
  \item $(q^\star,\xi^\star)$ satisfies the first-order necessary conditions for a stationary point of $\mathcal L(q,\xi)$, i.e., the functional derivative in the $q$-direction vanishes at $q^\star$ and $\nabla_\xi\mathcal L(q^\star,\xi^\star)=0$.
  \item $\xi^\star$ satisfies $\nabla_{\xi} \mathsf{L}(\xi^\star)=0$ and $q^\star=q_{\xi^\star}$.
\end{enumerate}
\end{proposition}

Note that a global joint PL condition for $(q,\xi)\mapsto -\mathcal L(q,\xi)$ is typically too strong and hard to verify when $\mathsf{L}$ depends on $\xi$ in a highly nonconvex way (for instance, see Figure~\ref{main-fig:ELBO_landscape}).
Instead, if the profile objective $\mathsf{L}$ is real-analytic on a compact set containing the iterates, then it satisfies the K\L{} property at every point in that set.  Combining monotone ascent of the alternating updates, boundedness of iterates, and the local K\L{} property, Theorem \ref{main-theorem:convergence} yields convergence of the iterates to a critical point of $\mathsf{L}$ (hence to a stationary pair $(q_{\xi^{\star}}, \xi^{\star})$ of the alternating maximization problem), even when a global PL inequality is unavailable.

{
\renewcommand{\baselinestretch}{1.0}\normalsize
\begin{figure}[t!]
    \centering

    \begin{subfigure}{0.52\linewidth}
        \centering
        \includegraphics[width=\linewidth]{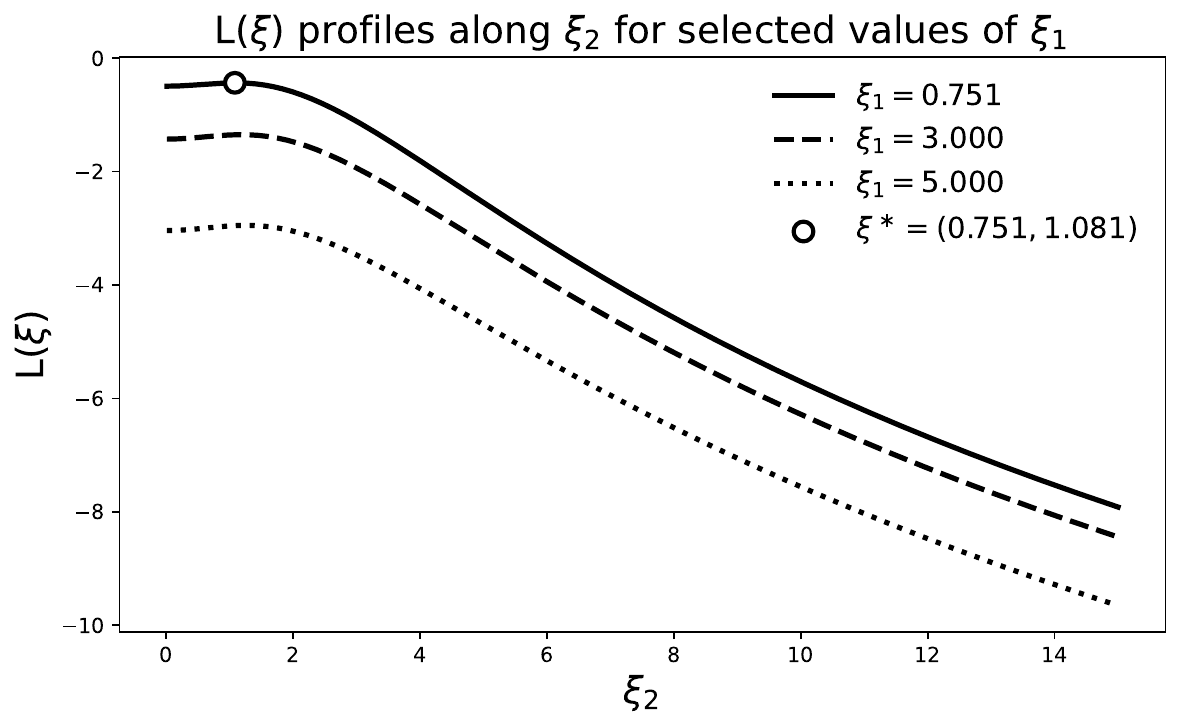}
    \end{subfigure}
    \hfill
    \begin{subfigure}{0.45\linewidth}
        \centering
        \includegraphics[width=\linewidth]{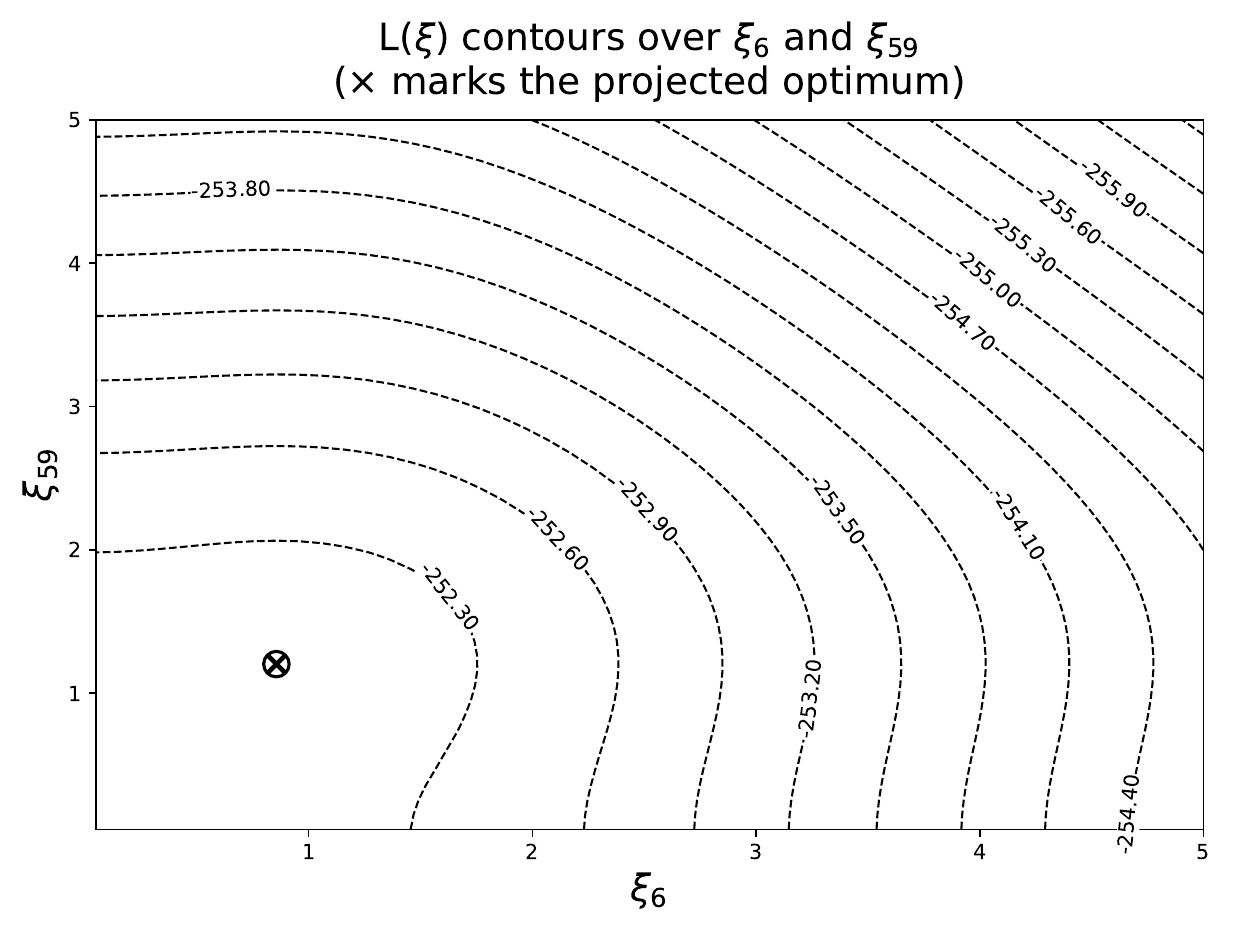}
    \end{subfigure}

    \caption{\footnotesize{Overview of the nonconvex landscape of $\mathsf{L}(\xi)$. Data $\mathcal{D}_n$ is generated from the Student's-$t$ $\ssg$ likelihood ($\nu=5, \tau_0^2=3$) with true parameter $\beta_0\sim \mathcal{N}_{p}(0, 0.5^2 I_{p})$ and covariates $x_{ij}\sim \mathcal{N}_{1}(0, 1)$, i.i.d. for $j\in[p]$. \emph{Left}: $\tssg$ converges in $21$ iterations for $(n, p) = (2, 2)$ with optimal $\xi^{\star} = (0.751, 1.081)$. \emph{Right}: $\tssg$ converges in $69$ iterations for $(n, p) = (100, 50)$; the contour plot shows a randomly selected two-dimensional slice through the optimal iterate at $(\xi_{6}^{\star}, \xi_{59}^{\star}) = (0.855,1.203)$.}}
    \label{main-fig:ELBO_landscape}
\end{figure}
}

We begin by first introducing in Assumptions \ref{main-ass:1} and \ref{main-ass:2}, a set of standard regularity conditions on the $\ssg$ likelihoods as well as on the design matrix $\mathbf{X}$.
\begin{assumption}[Regularity on $h$]
\label{main-ass:1}
The function $h: \mathbb{R}^{+}_{0}\to \mathbb{R}$ in Definition \ref{main-def:SSG} satisfies the following regularity conditions: { (i) $h$ is thrice differentiable on $\mathbb{R}^{+}$ and all derivatives up to order $3$ are continuous on $\mathbb{R}^{+}$}; (ii) $h$ is strictly decreasing with $\lim_{t \rightarrow \infty} h(t) = -\infty$; (iii) $h''(t) > 0$; and (iv) there exists a constant $K\in \mathbb{R}^{+}$ such that the function $t \mapsto h(t^2)$ is $K$-Lipschitz.
\end{assumption}

\begin{assumption}[Design regularity]
\label{main-ass:2}
The rows $\{\mathbf{x}_i: i \in [n]\}$ of $\mathbf{X}$ are non-zero.
\end{assumption}

Detailed discussion on the significance and validity of Assumptions \ref{main-ass:1} and \ref{main-ass:2} is provided in Section \ref{sec:discussion-assumptions} of Supplementary Materials. We now state our main convergence result.

\begin{theorem}[Convergence of the $\tssg$ EM algorithm]
\label{main-theorem:convergence}
Under Assumptions \ref{main-ass:1} and \ref{main-ass:2}, the sequence of $\tssg$ iterates $\{\xi^{(t)}: t\geq 0\}$ from Algorithm \ref{main-alg:tavie-em}, for both Type I and Type II $\ssg$ likelihoods, converges to a fixed point $\xi^{\star} \in \mathbb{R}^{n}_{+}$ for any initialization $\xi^{(0)} \in \mathbb{R}_{+}^{n}$. Moreover, any such fixed point is a critical point of $\mathsf{L}(\xi)$, i.e., $\nabla_{\xi}\mathsf{L}(\xi^\star) = 0$.
\end{theorem}
In conjunction with Theorem \ref{main-theorem:convergence}, the real analyticity of $\mathsf{L}$ for Type I and II $\ssg$ likelihoods near its critical points ensures the K\L{} exponent property~\citep{KL-SIAM-1,bolte2014proximal}, i.e., there exist an exponent $\Omega \in (0, 1), c_{\mathrm{K\text{\L}}} > 0$, and a neighborhood $\mathcal{U}$ of $\xi^{\star}$, such that:
\begin{equation}\label{main-eq:KL-exponent}
    \lVert \nabla_{\xi}\mathsf{L}(\xi)\rVert_{2} \geq c_{\mathrm{K\text{\L}}}|\mathsf{L}(\xi^{\star}) - \mathsf{L}(\xi)|^{\Omega},\quad \text{for all }\xi\in \mathcal{U}.
\end{equation}
Contingent on $\Omega$ and $\mathcal{U}$, we formalize a local convergence rate result for the $\tssg$ iterates in Theorem \ref{main-theorem:convergence-rate} below.

\begin{theorem}[Convergence rate of the $\tssg$ EM algorithm]
\label{main-theorem:convergence-rate}
Define the value gap $s_t := \mathsf{L}(\xi^{\star}) - \mathsf{L}(\xi^{(t)})$. Then under Assumptions \ref{main-ass:1} and \ref{main-ass:2}, there exist $T\in \mathbb{N}$ and $C>0$, such that for all $t \geq T$, we have $\xi^{(t)}\in \mathcal{U}$ and:
\begin{enumerate}
    \item [i.] linear rate for $\Omega \in (0, \tfrac{1}{2}]$: there exists $q\in (0, 1)$, $s_t \leq C q^{t-T}$, and $\lVert \xi^{(t)} - \xi^{\star}\rVert_{2} \leq C q^{\tfrac{t-T}{2}}$.

    \item [ii.] sub-linear rate for $\Omega \in (\tfrac{1}{2}, 1)$: $s_t \leq C(t-T)^{-\tfrac{1}{2\Omega -1}}$ and $\lVert \xi^{(t)}-\xi^{\star}\rVert_2 \leq C(t-T)^{-\tfrac{1-\Omega}{2\Omega-1}}$.
\end{enumerate}
\end{theorem}
As $\Omega$ is unknown, the algorithm may exhibit only sub-linear convergence. Achieving linear convergence would require a PL inequality, which in turn demands positive definiteness of the Hessian of $-\mathsf{L}$ at its critical points and thus additional restrictions on the prior and design. Hence, we do not pursue this direction.


\subsection{Variational Risk Bounds}\label{main-subsec:var-risk}

We now derive $\alpha$-dependent risk bounds for the variational posterior $\pi_{\alpha}(\theta\mid\mathcal D_n,\xi^\star)$. For $\alpha\in(0,1)$, risk is measured through the discrepancy between the variational posterior and the true parameter $\theta_0\in\Theta$ using the $\alpha$-R\'enyi divergence defined as:
\begin{align}\label{main-eq:alpha-renyi-divergence}
    D_{\alpha}(\theta, \theta_0) := \frac{1}{n(\alpha-1)}\log \int\left\{p(y \mid \mathbf{X}, \theta)\right\}^{\alpha}\left\{p(y \mid \mathbf{X}, \theta_0)\right\}^{1-\alpha}dy,
\end{align}
where the factor $n^{-1}$ is used to measure the average discrepancy per observation. See~\cite{Bayesian-fractional-posterior} for detailed insights into $\alpha$-R\'{e}nyi posterior risk bounds. To present the variational risk bound for the Type I $\ssg$ likelihoods, in addition to Assumption \ref{main-ass:1}, we impose the following mild {moment condition}, which is satisfied by both the Laplace and Student’s-$t$ (for $\nu > 2$) families; see Section \ref{sec:discussion-assumptions} of Supplementary Materials.

\begin{assumption}[Existence of moments for Type I $\mathsf{SSG}$ likelihoods]
\label{main-ass:3}
The $k$th moment exists and $\mathscr{E}_k = \mathcal{Z}_h^{-1}\int_{\mathbb{R}} |s|^k\exp\left\{h(s^2)\right\}ds  < \infty$, for all $k\in [2]$, where $\mathcal{Z}_h = \int_{\mathbb{R}} \exp\left\{h(s^2)\right\}ds$.
\end{assumption}

Theorems~\ref{main-theorem-alpha-less-1-variational-risk-bound-Type-I} and~\ref{main-theorem-alpha-less-1-variational-risk-bound-Type-II} bound the variational risk, obtained by averaging the $\alpha$-R\'enyi divergence in~\eqref{main-eq:alpha-renyi-divergence} under the optimal variational posterior, for Type I and Type II $\ssg$ likelihoods.

\begin{theorem}[$\alpha$-R\'{e}nyi risk bound for Type I $\ssg$ likelihoods]
\label{main-theorem-alpha-less-1-variational-risk-bound-Type-I}
Under $p(y\mid \mathbf{X}, \theta_0)$ in \eqref{main-eq:SSG-typeI-likelihood}; Assumptions \ref{main-ass:1}(i), \ref{main-ass:1}(iv), and \ref{main-ass:3}; and any $\varepsilon_n \in \left(0, \frac{1}{3}\right)$, with $\mathbb{P}_{\theta_0}$-probability at least $\left(1-3\varepsilon_n - [(D-1)^2 n \varepsilon_n^2]^{-1}\right)$, the variational risk with respect to $D_{\alpha}(\theta, \theta_0)$ satisfies:
\begin{align}\label{main-eq:type-1-risk-bound-alpha-less-1}
    \begin{split}
    (1-\alpha)\int_{\theta\in \Theta}D_{\alpha}(\theta, \theta_0)\pi_{\alpha}(\theta\mid \mathcal{D}_n, \xi^{\star})d\theta
        \leq D\alpha \varepsilon_n^2 + \frac{p\log p}{n} + C_1\frac{p}{n}\log\left(\frac{1}{\varepsilon_n}\right),
    \end{split}
\end{align}
for any $\alpha\in (0,1)$, arbitrary $D>1$, and constant $C_1 \in \mathbb{R}^{+}$ depending only on prior hyperparameters $(\mu, \Sigma, a, b)$, design matrix $\mathbf{X}$, and true parameter $\theta_0$.
\end{theorem}

\begin{theorem}[$\alpha$-R\'{e}nyi risk bound for Type II $\ssg$ likelihoods]
\label{main-theorem-alpha-less-1-variational-risk-bound-Type-II}
For some $\bar{m}\in \mathbb{R}^{+}$, suppose the size parameters $m_i$ of the Type II $\ssg$ likelihoods satisfy $m_{i}\leq \bar{m}$ for all $i\in [n]$. Then, under $p(y\mid \mathbf{X}, \beta_0)$ in \eqref{main-eq:SSG-typeII-likelihood} and any $\varepsilon_n\in \left(0, \frac{1}{3}\right)$, with $\mathbb{P}_{\beta_0}$-probability at least $\left(1-2\varepsilon_n - [(D-1)^2n\varepsilon_n^2]^{-1}\right)$, the variational risk with respect to $D_{\alpha}(\beta, \beta_0)$ satisfies:
\begin{equation}\label{main-eq:type-2-risk-bound-alpha-less-1}
\begin{split}
(1-\alpha)\int_{\beta\in \mathbb{R}^p} D_{\alpha}(\beta, \beta_0)\pi_\alpha(\beta \mid \mathcal{D}_n, \xi^{\star})d\beta \leq D\alpha \varepsilon_n^2 + \frac{p\log p}{n} + C_2\frac{p}{n}\log\left(\frac{1}{\varepsilon_n}\right),
\end{split}
\end{equation}
for any $\alpha\in (0, 1)$, arbitrary $D>1$, and constant $C_2\in \mathbb{R}^{+}$ depending only on prior hyperparameters $(\mu, \Sigma)$, design matrix $\mathbf{X}$, true parameter $\beta_0$, and $\bar{m}$.
\end{theorem}
Having presented the variational risk bounds, we highlight two key observations below.

\begin{remark}
\label{main-remark:near-minimaxity}
Consider the fixed-dimensional regime where $p$ is fixed as $n\to \infty$. Section \ref{subsec:fixed-p-constant-control} of Supplementary Materials shows that $C_1$ and $C_2$ in \eqref{main-eq:type-1-risk-bound-alpha-less-1} and \eqref{main-eq:type-2-risk-bound-alpha-less-1}, respectively, are of $\mathcal{O}(1)$ with respect to $n$. Consequently, taking $\varepsilon^{2}_{n} \asymp n^{-1} p \log n$, the resulting risk bounds in \eqref{main-eq:type-1-risk-bound-alpha-less-1} and \eqref{main-eq:type-2-risk-bound-alpha-less-1} are $\mathcal{O}(n^{-1})$, up to logarithmic factors, yielding near-parametric optimality (see the PAC-Bayes formulation of the Bayesian oracle inequality in~\cite{Bayesian-fractional-posterior}).
\end{remark}

\begin{remark}
\label{main-remark:alpha-dependence}
The tempering parameter $\alpha\in (0, 1)$ in  $\pi_{\alpha}(\theta \mid \mathcal{D}_n, \xi^{\star})$ and $D_{\alpha}(\theta, \theta_0)$ in Theorems \ref{main-theorem-alpha-less-1-variational-risk-bound-Type-I} and \ref{main-theorem-alpha-less-1-variational-risk-bound-Type-II} need not coincide. Using the inequality $\tfrac{\alpha_1(1-\alpha_2)}{\alpha_2(1-\alpha_1)} D_{\alpha_2} \le D_{\alpha_1} \le D_{\alpha_2}$ for $0 < \alpha_1 \le \alpha_2 < 1$~\citep{van2014renyi}, one can extend the above results to any $D_{\delta}$ with $\delta \in (0, 1)$.
\end{remark}

We next treat the case $\alpha=1$ under a compact parameter space $\Theta$. This restriction is rather mild, as our prior choices in Section~\ref{main-subsec:prior-posterior} are proper, so $\Theta$ can be chosen as a sufficiently large compact set that carries all but a negligible amount of prior mass. Section \ref{subsec:compact-parameter-space} of Supplementary Materials shows that $\tssg$ extends directly to such restrictions. In this setting, variational risk is measured by the average squared Hellinger distance defined as:
\begin{align}\label{main-eq:Hellinger-divergence}
    \mathcal{H}^{2}(\theta \parallel \theta_0) := \frac{1}{n}\sum_{i\in [n]}\int \left(\sqrt{p(y_i\mid \mathbf{x}_i, \theta)} - \sqrt{p(y_i\mid \mathbf{x}_i, \theta_0)}\right)^2dy_i.
\end{align}
We conclude our theoretical exposition with the counterparts of Theorems~\ref{main-theorem-alpha-less-1-variational-risk-bound-Type-I} and~\ref{main-theorem-alpha-less-1-variational-risk-bound-Type-II} for $\alpha = 1$. 
\begin{theorem}[$\mathcal{H}^{2}(\theta \parallel \theta_0)$ risk bound for Type I $\mathsf{SSG}$ likelihoods]
\label{main-theorem-alpha-equals-1-variational-risk-bound-Type-I}
Under $p(y\mid \mathbf{X}, \theta_0)$ in \eqref{main-eq:SSG-typeI-likelihood}; Assumptions \ref{main-ass:1}(i), \ref{main-ass:1}(iv), and \ref{main-ass:3}; and any $\varepsilon _n\in \left(0, \frac{1}{3}\right)$, with $\mathbb{P}_{\theta_0}$-probability at least $(1 - 3\varepsilon_n - [(D-1)^2n\varepsilon_n^2]^{-1} - 2e^{-\tfrac{c n \varepsilon_n^2}{2}})$, the variational risk with respect to $\mathcal{H}^{2}(\theta \parallel \theta_0)$ satisfies:
\begin{align}\label{main-eq:type-1-risk-bound-alpha-1}
    \begin{split}
    c\int_{\theta \in \Theta}\mathcal{H}^{2}(\theta \parallel \theta_0) \pi_{1}(\theta\mid \mathcal{D}_n, \xi^{\star})d\theta \leq \left(D + \frac{3c}{2}\right)\varepsilon_n^2 + \frac{p\log p}{n} + C_{3}\frac{p}{n}\log\left(\frac{1}{\varepsilon_n}\right),
    \end{split}
\end{align}
for arbitrary $D>1$, some $c\in \mathbb{R}^{+}$, and $C_3 \in \mathbb{R}^{+}$ depending only on prior hyperparameters $(\mu, \Sigma, a, b)$, design matrix $\mathbf{X}$, and true parameter $\theta_0$.
\end{theorem}

\begin{theorem}[$\mathcal{H}^{2}(\beta \parallel \beta_0)$ risk bound for Type II $\mathsf{SSG}$ likelihoods]
\label{main-theorem-alpha-equals-1-variational-risk-bound-Type-II}
For some $\bar{m}\in \mathbb{R}^{+}$, suppose the size parameters $m_i$ of the Type II $\ssg$ likelihoods satisfy $m_{i}\leq \bar{m}$ for all $i\in [n]$. Then, under $p(y\mid \mathbf{X}, \beta_0)$ in \eqref{main-eq:SSG-typeII-likelihood} and any $\varepsilon_n \in \left(0, \frac{1}{3}\right)$, with $\mathbb{P}_{\beta_0}$-probability at least $(1 - 2\varepsilon_n - [(D-1)^2n\varepsilon_n^2]^{-1} - 2e^{-\tfrac{c n \varepsilon_n^2}{2}})$, the variational risk with respect to $\mathcal{H}^{2}(\beta \parallel \beta_0)$ satisfies:
\begin{align}\label{main-eq:type-2-risk-bound-alpha-1}
    \begin{split}
    c\int_{\beta\in \Theta}\mathcal{H}^{2}(\beta \parallel \beta_0)\pi_{1}(\beta\mid \mathcal{D}_n, \xi^{\star})d\beta
        \leq \left(D + \frac{3c}{2}\right)\varepsilon_n^2 + \frac{p\log p}{n} + C_4\frac{p}{n}\log\left(\frac{1}{\varepsilon_n}\right),
    \end{split}
\end{align}
for arbitrary $D>1$, some $c \in \mathbb{R}^{+}$, and $C_4\in \mathbb{R}^{+}$ depending only on prior hyperparameters $(\mu, \Sigma)$, design matrix $\mathbf{X}$, true parameter $\beta_0$, and $\bar{m}$.
\end{theorem}

Following Remark \ref{main-remark:near-minimaxity} and taking $\varepsilon_n^2 \asymp n^{-1} p\log n$, the variational risk bounds in Theorems \ref{main-theorem-alpha-equals-1-variational-risk-bound-Type-I} and \ref{main-theorem-alpha-equals-1-variational-risk-bound-Type-II} retain near-parametric optimality. All bounds in Theorems~\ref{main-theorem-alpha-less-1-variational-risk-bound-Type-I}-\ref{main-theorem-alpha-equals-1-variational-risk-bound-Type-II} are {non-asymptotic}. The associated constants appear in the theorem proofs. Furthermore, Section \ref{sec:empirical-evidence-risk-bounds} of Supplementary Materials evidences the empirical validation of the risk bounds.


\section{\texorpdfstring{$\tssg$}{TAVIE-SSG} in Action}\label{main-sec:TAVIE-SSG-applications}

We showcase $\tssg$'s empirical performance across multiple simulated experiments and challenging real-world applications including Bayesian quantile regression (BQR) on the U.S. Census 2000 dataset~\citep{yang-FastQR-pmlr} and spatial count regression for spatial transcriptomics~\citep{STARmap}. It is benchmarked against state-of-the-art VI methods, including ADVI MF and FR (mean-field and full-rank)~\citep{advi}, DADVI~\citep{dadvi}, and MFVI~\citep{mfvi-student}, along with an efficient MC sampler, viz., \texttt{PyMC} {No-U-Turn Sampler} (NUTS)~\citep{pymc,Hoffman-Gelman-NUTS,pymc2023peerj}. Across all applications, $\tssg$ demonstrates strong statistical accuracy and computational efficiency, while competing methods exhibit less stable and less consistent performance.


\subsection{\texorpdfstring{Student's-$t$ Type I $\ssg$ Likelihood}{Student's-t Type I SSG Likelihood}}
\label{main-subsec:sim-exp-student}

For $\nu \in \mathbb{N}$, the Student's-$t$ model, $y_i\mid \theta, \nu\stackrel{\mathrm{ind.}}{\sim} t_{\nu}(\mu_i = \mathbf{x}_i^{\top}\beta, \tau)$, follows from \eqref{main-eq:SSG-typeI-likelihood}.
Data $\mathcal{D}_n$ is simulated using $\nu=5$, $\tau^2_0=3$, $\beta_0\sim \mathcal{N}_{p}(0, I_{p})$, and $\mathbf{X}=(\mathbf{x}_1, \ldots, \mathbf{x}_{n})^{\top}\in \mathbb{R}^{n\times p}$ with $x_{i1}=1$ (intercept) and $x_{ij} \stackrel{\mathrm{i.i.d.}}{\sim} \mathcal{N}_1(0, 1)$ for $j=2,\ldots,p$. We consider two experiments: (E1) $n\in \{200, 500, 1000, 2000\}$ with fixed $p=9$ and (E2) $p\in \{4, 9, 16, 21\}$ with fixed $n=1000$. Under both, $\tssg$ is implemented with $\alpha=1$, prior hyperparameters $(\mu, \Sigma, a, b) = (0, I_p, 0.05, 0.05)$, and convergence tolerance of $\texttt{tol}=10^{-9}$; competing methods are executed as described in Section \ref{app-competing-methods-details} of Supplementary Materials. Performance in both experiments is assessed through {mean-squared errors} (MSEs) of $\widehat{\beta}$ and $\widehat{\tau}^{2}$, defined as $p^{-1}\lVert \beta_0 - \widehat{\beta}\rVert_{2}^{2}$ and $(\tau_0^2 - \widehat{\tau}^{2})^2$, as well as runtime. For experiment E1, we also estimate the sliced Wasserstein (SW) distance~\citep{bonneel2015sliced,kolouri2016sliced} between the variational posterior and the true posterior approximated using \texttt{PyMC} (NUTS).
Numerical results for MSE and runtime based on $100$ replications are visualized in Figures~\hyperref[main-fig:student_singlep_multin]{\ref{main-fig:student_singlep_multin}a},~\hyperref[main-fig:student_singlep_multin]{\ref{main-fig:student_singlep_multin}b}, and~\hyperref[main-fig:runtime_student_BQR]{\ref{main-fig:runtime_student_BQR}a}  for experiment E1 and Figures~\ref{main-fig:student_singlen_multip} and ~\hyperref[main-fig:runtime_student_BQR]{\ref{main-fig:runtime_student_BQR}b} for experiment E2. Additionally, for experiment E1, Figure~\hyperref[main-fig:student_singlep_multin]{\ref{main-fig:student_singlep_multin}c} reports the estimated SW distance based on the first $10$ of these $100$ repetitions. The corresponding numerical summaries are tabulated in Sections \ref{app-avg-metrics-student-singlep-multin} and \ref{app-avg-metrics-student-singlen-multip} of Supplementary Materials.

{
\renewcommand{\baselinestretch}{1.0}\normalsize
\begin{figure}[!t]
    \centering
    \includegraphics[width=\linewidth]{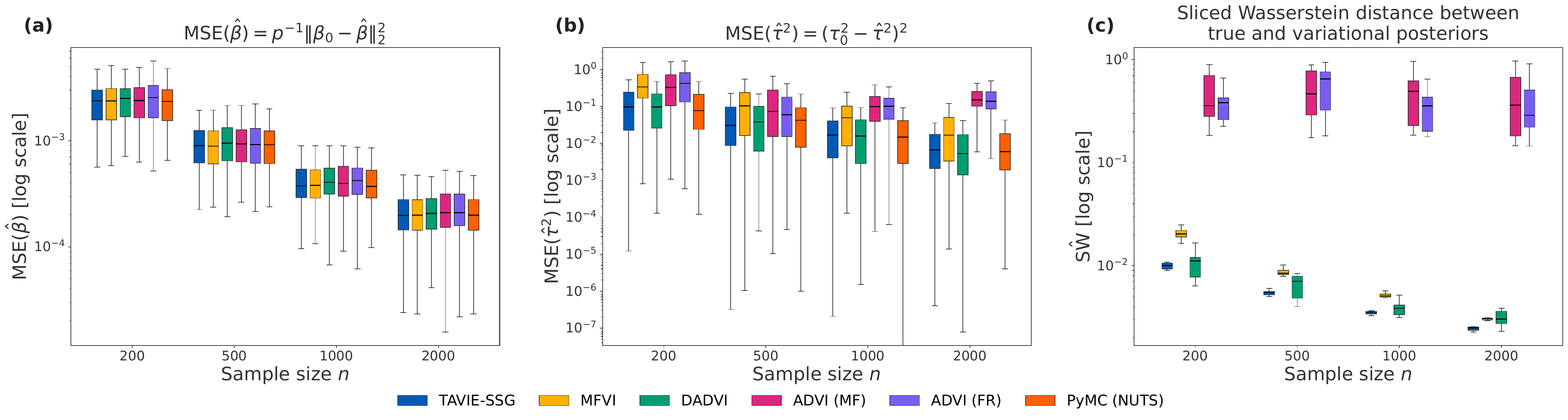}
    \caption{\footnotesize{(a)-(b) MSEs of $(\widehat{\beta}, \widehat{\tau}^2)$ and (c) estimated sliced Wasserstein (SW) distances, both on the $\log$-scale, for $\tssg$ and competitors under the Student’s-$t$ $\mathsf{SSG}$ likelihood ($\nu=5$) in experiment E1.
    }}
    \label{main-fig:student_singlep_multin}
\end{figure}
}

{
\renewcommand{\baselinestretch}{1.0}\normalsize
\begin{figure}[!t]
    \centering
    \includegraphics[width=0.8\linewidth]{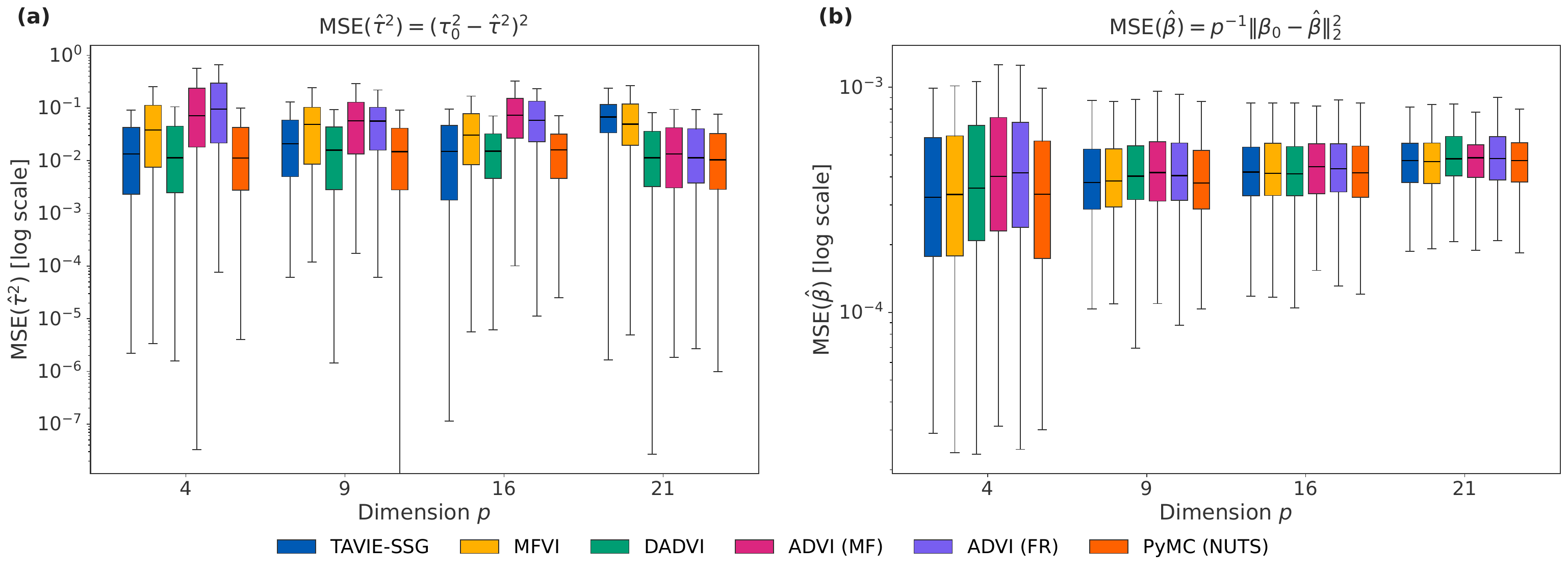}
    \caption{\footnotesize{MSEs of $(\widehat{\beta}, \widehat{\tau}^2)$ (in $\log$-scale)
    for $\tssg$ and competitors under the Student's-$t$ $\mathsf{SSG}$ likelihood ($\nu=5$) in experiment E2.}}
    \label{main-fig:student_singlen_multip}
\end{figure}
}

In experiment E1 (Figure \ref{main-fig:student_singlep_multin}), $\tssg$, DADVI, and \texttt{PyMC} (NUTS) deliver highly accurate estimates of $\tau^2$, with substantially lower MSEs than ADVI (MF/FR). For $\beta$-estimation, all methods perform competitively. Moreover, the posterior distance comparisons show that $\tssg$, DADVI, and MFVI produce closer approximations to the true posterior than ADVI variants. The MFVI algorithm proposed by~\cite{mfvi-student} uses a latent variable reformulation of the Student’s-$t$ likelihood, whereas $\tssg$ achieves better performance for $\tau^{2}$-estimation by directly banking on the intrinsic geometry of the likelihood. When $p$ increases in experiment E2 (Figure \ref{main-fig:student_singlen_multip}), $\tssg$ continues to deliver competitive accuracy. Across both experiments, Figures \hyperref[main-fig:runtime_student_BQR]{\ref{main-fig:runtime_student_BQR}a} and \hyperref[main-fig:runtime_student_BQR]{\ref{main-fig:runtime_student_BQR}b} show that $\tssg$ is also markedly more computationally efficient, with runtimes orders of magnitude smaller than those of the competitors.

For convergence diagnostics, we track the optimization objectives of $\tssg$ and ADVI (MF/FR) (Section \ref{app-convergence-ELBO-student} of Supplementary Materials). For $\tssg$, this is $\mathsf{L}(\xi)$ in~\eqref{main-eq:ELBO-general}, while for ADVI (MF/FR) it is the MC approximation of the true $\mathsf{ELBO}$ $\mathcal{L}_{\mathrm{true}}(q)$. Additional results for other $\ssg$ models, different choices of $\alpha$, and regimes with $p$ comparable to $n$ are provided in Section \ref{app:additional-simulation-studies} of Supplementary Materials.


\subsection{Bayesian Quantile Regression for U.S. Census 2000 Data}\label{main-subsec:Census-data-study}

The Census 2000 state-level dataset~\citep{yang-FastQR-pmlr} contains demographic records for a $5\%$ sample of the U.S. population. We model $\log$ annual salary on the demographic attributes (gender, age, race, marital status, and education), where $n=5\times 10^6$ and $p=11$ (including an intercept).

BQR treats the quantile loss, $\rho_u(x) = x(u - \mathds{1}(x<0))$, where $x\in \mathbb{R}$ and $u\in (0, 1)$ is the quantile level, as the negative log-likelihood yielding the asymmetric Laplace distribution (ALD)~\citep{Yu-Moeed-Bayesian-QR}, $p_{\mathrm{ALD}}(x\mid \tau, u) = 2\tau u(1-u)\exp\left\{-2\tau \rho_u(x)\right\}$, for $\tau\in \mathbb{R}^{+}$. Following~\cite{yang-FastQR-pmlr}, we consider $\tau=\tau_0$ (known) and the standard likelihood setup ($\alpha=1$). The joint likelihood for our purpose is $p(y\mid \mathbf{X}, \beta) = \prod_{i\in [n]}p_{\mathrm{ALD}}(y_i - \mathbf{x}_i^{\top}\beta\mid \tau_0, u)$.

We adapt the $\tssg$ framework to BQR using the general construction of Section \ref{main-section:methods} for the ALD likelihood. The corresponding implementation details are in Section \ref{subsec:TAVIE-BQR-setup} of Supplementary Materials. We evaluate the BQR performance of $\tssg$ in Section \ref{main-subsubsec:TAVIE-SSG-BQR}, followed by a comparative analysis against competing methods in Section \ref{main-subsubsec:tavie-ssg-BQR-competitors}.


\subsubsection{Performance of \texorpdfstring{$\tssg$}{TAVIE-SSG} on the Full Census Dataset}\label{main-subsubsec:TAVIE-SSG-BQR}

Using $\tau_0=1$ and the prior $\beta\sim\mathcal N_p(\mu,\Sigma)$ with $\mu=0$ and $\Sigma=I_p$, we implement $\tssg$ on the full Census dataset $(n=5\times10^6)$ independently across quantiles $u\in\{0.05,\ldots,0.95\}$, monitoring convergence with $\texttt{tol}=10^{-6}$.

{
\renewcommand{\baselinestretch}{1.0}\normalsize
\begin{figure}[!htbp]
    \centering
    \begin{subfigure}[t]{0.38\textwidth}
        \centering
        \includegraphics[width=\linewidth]{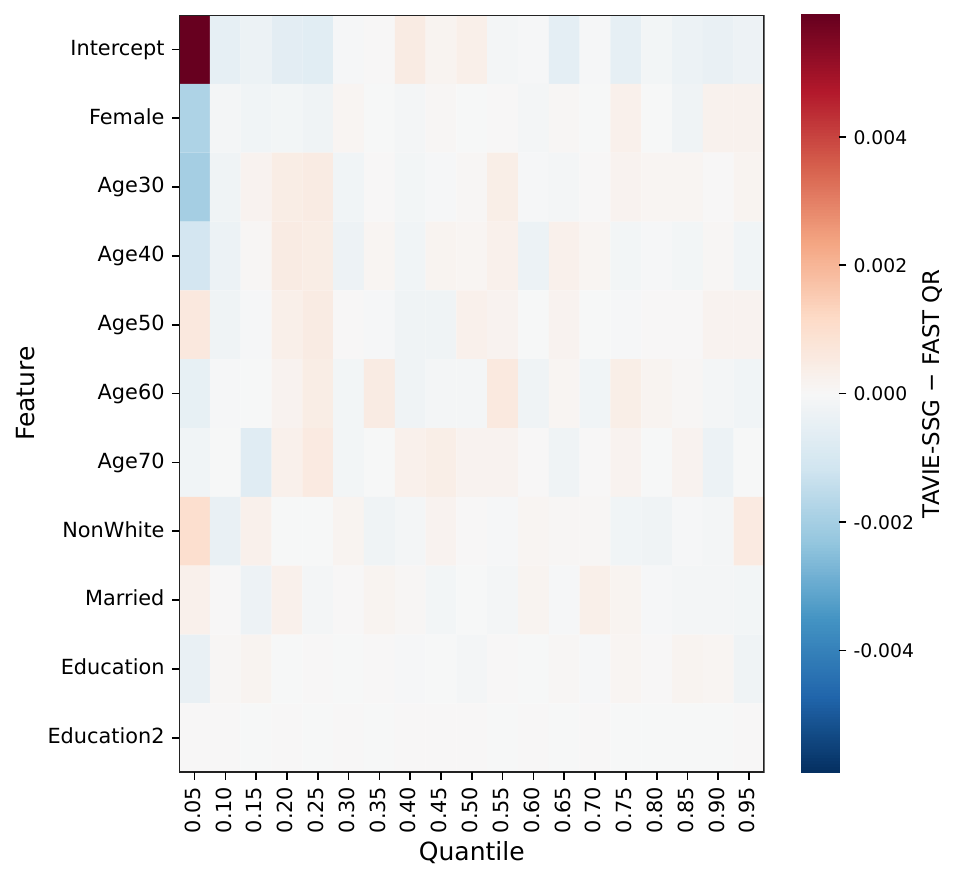}
        \caption{\footnotesize{Heatmap of differences between the $\tssg$ and FAST QR estimates on the full Census dataset.}}
        \label{main-fig:heatmap_TAVIEQR_FASTQR}
    \end{subfigure}
    \hfill
    \begin{subfigure}[t]{0.54\textwidth}
        \centering
        \includegraphics[width=\linewidth]{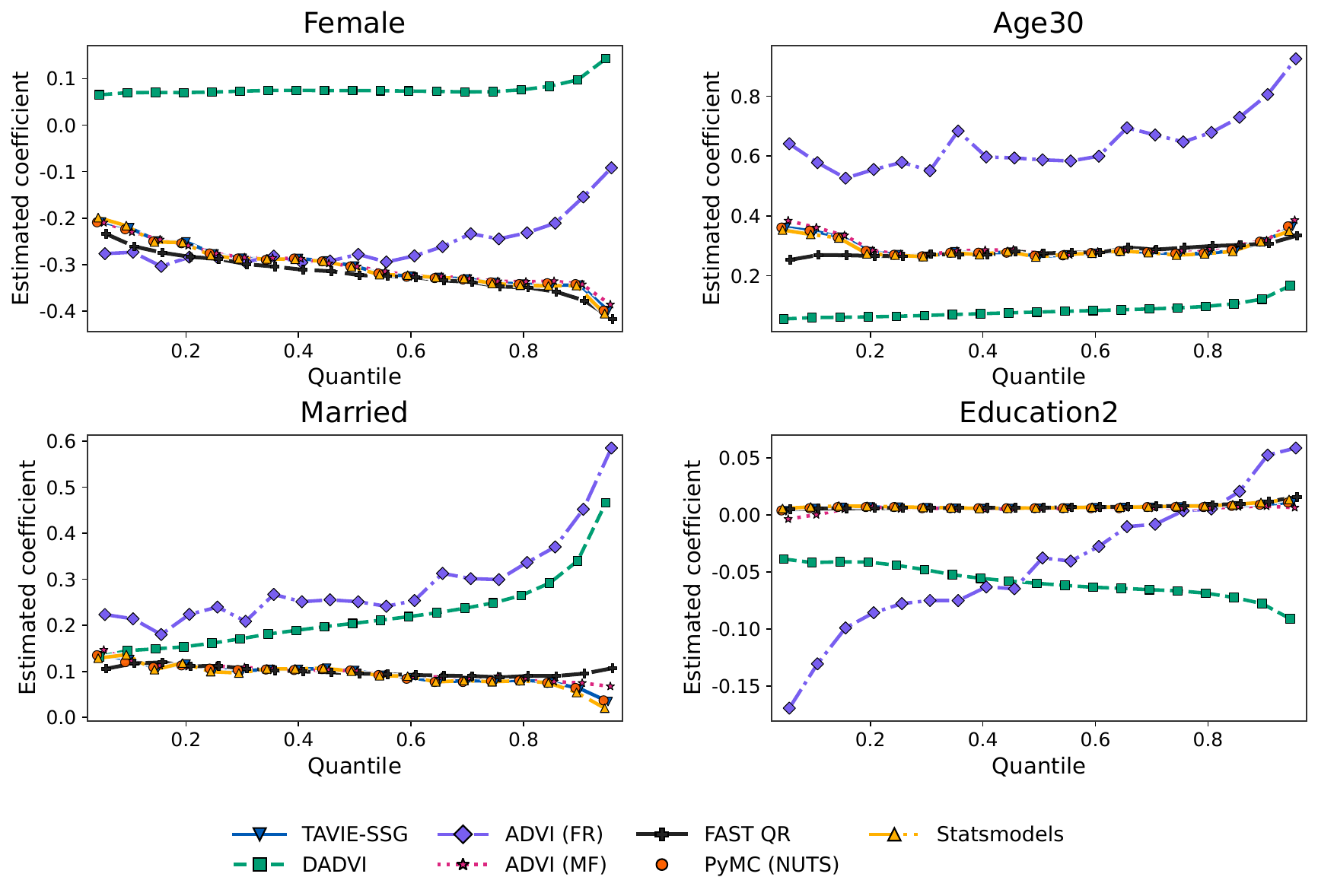}
        \caption{\footnotesize{Comparison of selected $\tssg$ estimates and competitors on the sub-sampled ($n=10^4$) dataset.}}
        \label{main-fig:census_estimates_competing}
    \end{subfigure}
    
    \caption{\footnotesize{BQR performance of $\tssg$ and competing methods on the U.S. Census 2000 dataset.}}
    \label{main-fig:TAVIE_QR_estimates_against_FAST_QR}
\end{figure}
}
For benchmarking, we compare against FAST QR~\citep{yang-FastQR-pmlr}, a large-scale approximate quantile regression method. The heatmap in Figure~\ref{main-fig:heatmap_TAVIEQR_FASTQR}, together with the $\tssg$ coefficient estimates and $95\%$ point-wise credible intervals in Figure \ref{fig:tavie_QR_estimates_census_all} of Supplementary Materials, shows that $\tssg$ attains accuracy competitive with FAST QR while efficiently performing large-scale BQR. This comparison is particularly appropriate since both methods estimate quantiles independently, in lieu of joint modeling.


\subsubsection{Comparative Analysis of \texorpdfstring{$\tssg$}{TAVIE-SSG} Against Competing Methods}\label{main-subsubsec:tavie-ssg-BQR-competitors}

We further compare $\tssg$ against DADVI, ADVI (MF/FR), \texttt{PyMC} (NUTS), and \texttt{statsmodels}~\citep{seabold2010statsmodels}, using the original FAST QR estimates as the reference. Due to the inability of competitors to scale to the full Census dataset, we sub-sample $n=10^4$ observations and run all methods except FAST QR. $\tssg$ uses the same configuration as above, while competitor settings are in Section \ref{app-competing-methods-details} of Supplementary Materials.

Figure \ref{main-fig:census_estimates_competing} reports the comparison, with additional features shown in Figure \ref{fig:census_estimates_competing_methods_n_10000_all} of Supplementary Materials. DADVI and ADVI (FR) exhibit pronounced instability in coefficient estimates across features and quantiles, diverging notably from FAST QR. The trajectory of the MC approximation of ADVI (FR)'s true $\mathsf{ELBO}$ (Section \ref{subsec:BQR-additional-plots} of Supplementary Materials) confirms that estimation instability persists despite algorithmic convergence. In contrast, $\tssg$, \texttt{PyMC} (NUTS), and \texttt{statsmodels} remain closely aligned with FAST QR estimates. However, ADVI (MF) achieves comparable accuracy only with more iterations and a larger learning rate, incurring additional computational cost.

Figure \hyperref[main-fig:runtime_student_BQR]{\ref{main-fig:runtime_student_BQR}c} reports runtimes across the considered quantiles, showing that $\tssg$ is orders of magnitude faster than the competing VI and MC methods. Among the BBVI approaches, DADVI has the largest computational overhead and, despite being approximate, it is less stable and more expensive than \texttt{PyMC} (NUTS). Altogether, these results establish $\tssg$ as an accurate, scalable, and computationally efficient approach for large-scale BQR.


\subsection{Predicting Spatial Gene Expressions in STARmap Data}\label{main-subsec:STARmap-data-study}

Spatially resolved transcript amplicon readout mapping (STARmap)~\citep{STARmap} enables high-resolution gene expression profiling across spatial locations at single-cell precision. We analyze the primary visual cortex dataset from $4$ mice in~\cite{STARmap}, focusing on the tissue sample corresponding to dark-housed biological replicate $1$, with $G=160$ genes measured over $n=941$ spatial locations. Spatial gene expression counts are modeled using a Negative-Binomial likelihood to capture complex spatial patterns.

Let $y_g = (y_{g1}, \ldots, y_{gn})^{\top} \in \mathbb{Z}_{+, 0}^{n}$ denote the gene expression counts for gene $g\in [G]$ across $n$ spatial locations $S = \{s_i: i\in [n]\}  \subset \mathbb{R}^{2}$. To model smooth spatial variation, we map each spatial location in $S$ to $27$ two-dimensional cubic $B$-spline~\citep{deBoor1978} basis functions, which, along with an intercept, yield $\mathbf{X}\in \mathbb{R}^{n\times p}$ with $p=28$. For each gene, we fit a Negative-Binomial model, $y_{gi}\mid \mathbf{x}_i, \beta_{g}\sim \mathrm{NB}(m_i, p_{gi})$ having probability mass function as in Section \ref{sec:parameterization-details} of Supplementary Materials. We set $\mathbb{E}(y_{gi}\mid \mathbf{x}_i, \beta_g)$ proportional to library size, i.e., $m_i = \sum_{g\in [G]}y_{gi}$~\citep{Robinson2010TMM}.
Success probabilities are modeled as $p_{gi} = \sigma(\mathbf{x}_i^{\top}\beta_g)$, where $\sigma$ is the sigmoid function. $\tssg$ is implemented under $\alpha=1$ with a Gaussian prior over $\beta_{g}$ having $\mu = 0$ and $\Sigma = \mathrm{diag}(10, 1, \ldots, 1)$. Convergence is monitored using $\texttt{tol}=10^{-9}$. Performance is compared with DADVI, ADVI (MF/FR), and \texttt{PyMC} (NUTS), executed under the configurations in Section \ref{app-competing-methods-details} of Supplementary Materials.

{
\renewcommand{\baselinestretch}{1}\normalsize
\begin{figure}[t]
    \centering
    \begin{subfigure}[t]{0.48\textwidth}
        \centering
        \includegraphics[width=1\linewidth]{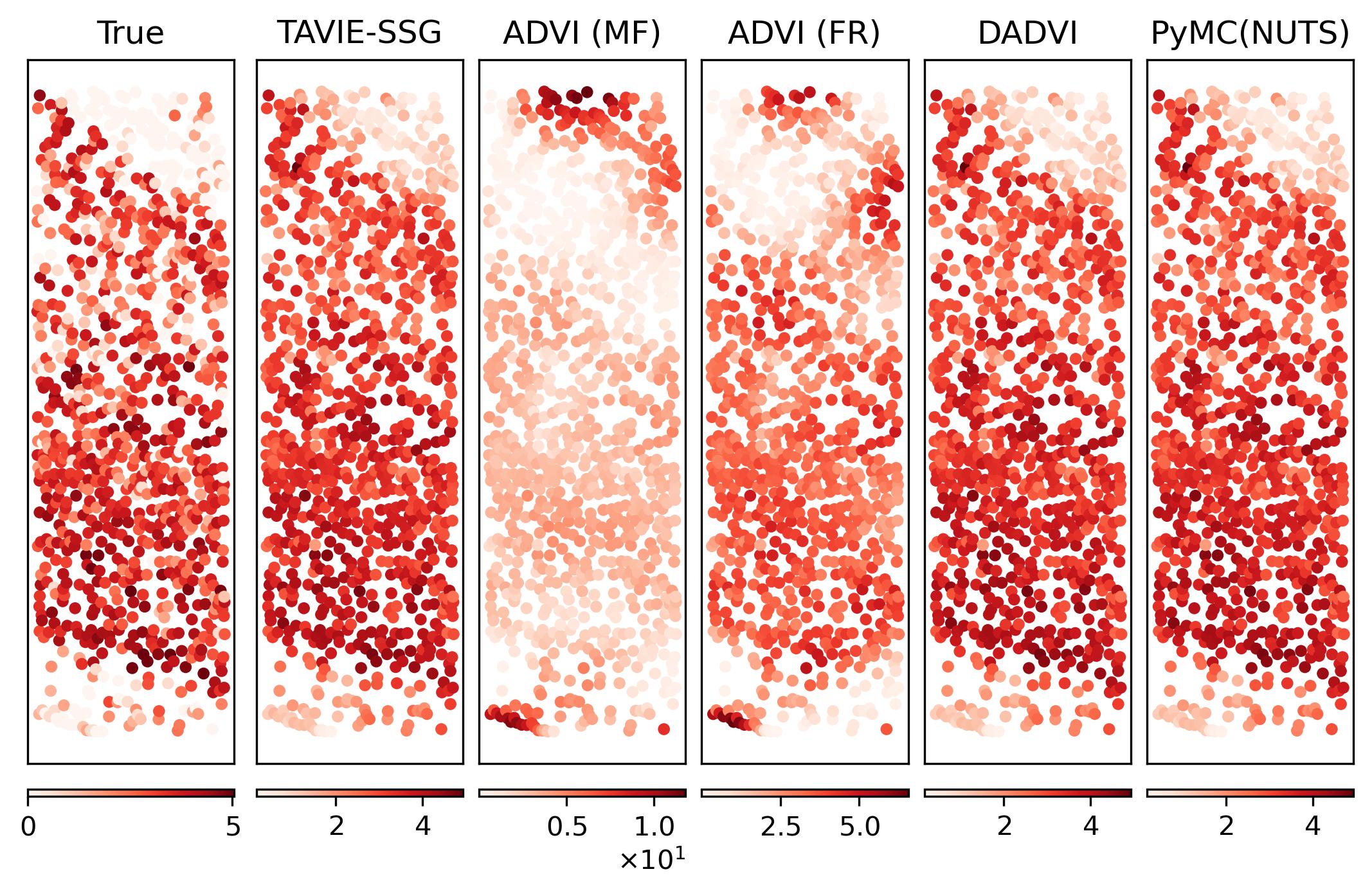}
        \caption{\footnotesize{\texttt{Slc17a7} gene expressions.
        }}
        \label{main-fig:gene1_pred}
    \end{subfigure}
    \hfill
    \begin{subfigure}[t]{0.48\textwidth}
        \centering
        \includegraphics[width=1\linewidth]{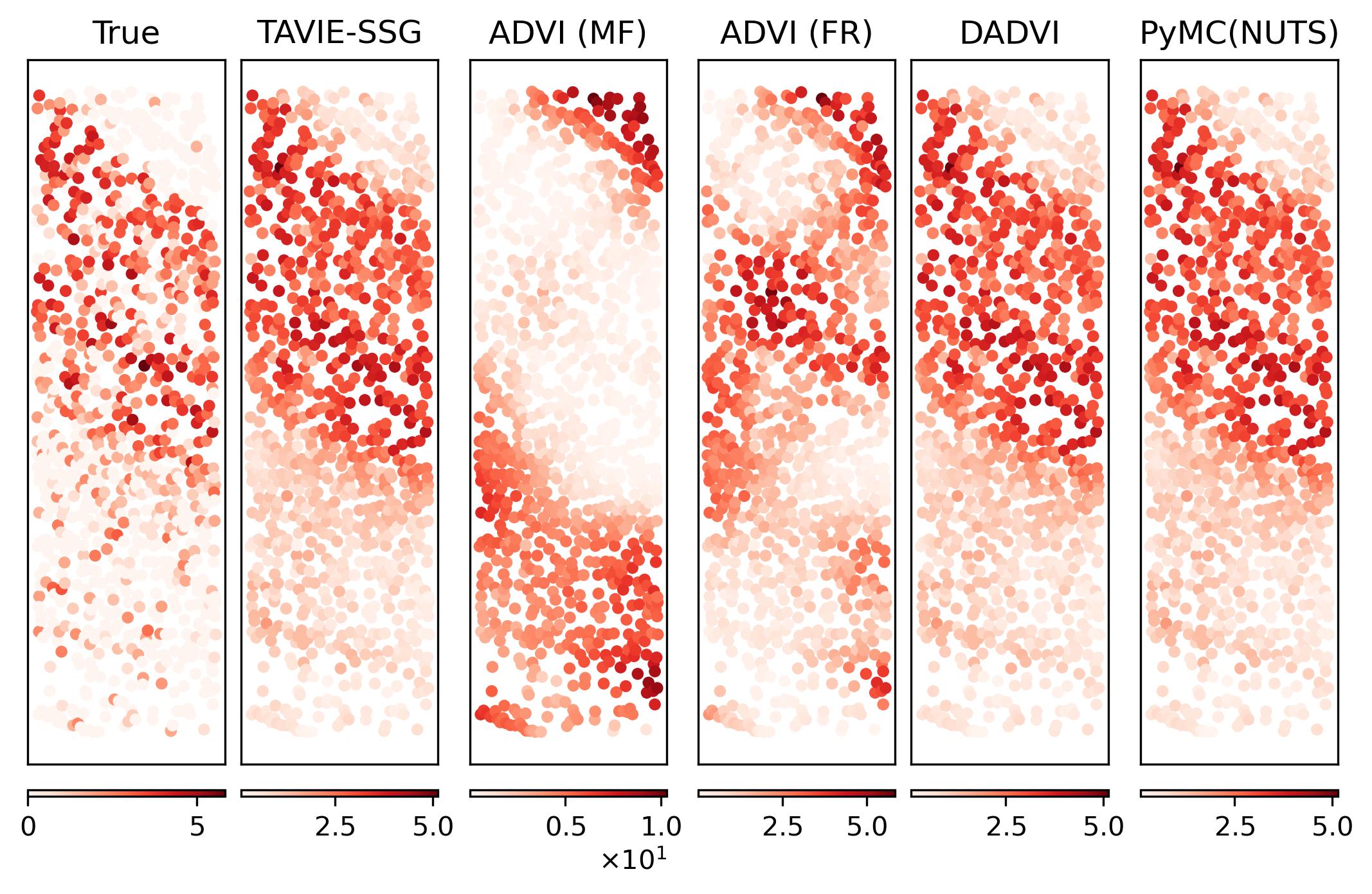}
        \caption{\footnotesize{\texttt{Pcp4} gene expressions.
        }}
        \label{main-fig:gene2_pred}
    \end{subfigure}
    
    \caption{\footnotesize{Log-normalized true and predicted gene expression counts obtained from $\tssg$, ADVI (MF/FR), DADVI, and \texttt{PyMC} (NUTS).}}
    \label{main-fig:gene_pred}
\end{figure}
}

Figure~\ref{main-fig:gene_pred} presents the true and predicted expression surfaces for genes \texttt{Slc17a7} and \texttt{Pcp4}. These genes represent interpretable biomarkers of excitatory neurons~\citep{STARmap}, which have distinct spatial expression profiles, thus providing contrasting settings for evaluating spatial reconstruction. For both genes, $\tssg$, DADVI, and \texttt{PyMC} (NUTS) achieve comparable predictive performance and accurately recover complex spatial patterns and high-expression regions, while ADVI (MF/FR) fails to capture the spatial structure, yielding comparatively poorer predictions. Additional genes, selected to encompass varying degrees of expression intensity, sparsity, and spatial heterogeneity, are presented in Section \ref{sec:STARmap-additional-results} of Supplementary Materials and exhibit similar comparative trends. These conclusions are further corroborated by Pearson residual sum of squares comparisons across all $G=160$ genes in Figure \ref{fig:heatmap} of Supplementary Materials, where $\tssg$, DADVI, and \texttt{PyMC} (NUTS) maintain uniformly low residuals, unlike the systematically larger errors of ADVI (MF/FR).


\section{Discussion}\label{main-section:discussion}

The versatility of $\tssg$ opens several compelling future research directions. First, a natural extension is its integration with modern sparsity-inducing priors, {hierarchical formulations, and mixture models}, with potential applications in high-dimensional variable selection, high-throughput genomic studies, representation learning, and risk modeling. {Such extensions may, however, lose the closed-form variational updates available in the present setting, thereby necessitating further approximations and more involved optimization.}
Second, the convergence guarantees may be sharpened by refining the \L{}ojasiewicz exponent to $1/2$, or equivalently, by establishing a PL inequality through the positive definiteness of the Hessian of $-\mathsf{L}$. This may require additional structural constraints on the design matrix to ensure local strong convexity that need not hold in general. A rigorous delineation of such conditions, along with corresponding refinement of the convergence analysis, is thus deferred to future work.
Third, motivated by the recent use of variational objectives for scalable model selection~\citep{cheriefabdellatif2019elbo,ZhangYang2024MFVBModelSelection}, a promising future direction is to develop a formal model selection regime for $\tssg$ based on the optimized objective $\mathsf{L}(\xi^{\star})$. Our empirical order analysis in Section \ref{subsec:empirical-elbo-variational-gap} of Supplementary Materials suggests that the post-convergence variational gap $\log \int_{\theta\in \Theta}p_{\alpha}(y\mid \mathbf{X}, \theta)\pi(\theta)d\theta - \mathsf{L}(\xi^{\star})$ scales at the canonical Bayesian Information Criterion (BIC) order $(p/2)\log n$, indicating that maximizing $\mathsf{L}(\xi^{\star})$ may offer a principled and computationally tractable model selection criterion. 
Fourth, statistical model misspecification is an important direction for future work. Section \ref{app:model-misspecification-empirics} of Supplementary Materials provides empirical evidence that $\tssg$ closely reproduces the predictive performance of the corresponding true $\alpha$-fractional posterior approximated using \texttt{PyMC} (NUTS) under misspecification of tail behavior, distributional shape, and conditional variance. A full theoretical treatment of the pseudo-true fractional posterior limit and the additional tangent minorization gap remains open.
Finally, another promising avenue is the construction of {tighter tangent-based lower bounds} for $\ssg$ $\log$-likelihoods following~\cite{durante}.


\subsection*{Code and Data Availability Statement}
A \texttt{Python} implementation of $\tssg$ is available at
\href{https://github.com/Roy-SR-007/TAVIE-SSG}
{\nolinkurl{github.com/Roy-SR-007/TAVIE-SSG}}.
The STARmap and U.S.\ Census 2000 datasets are available at
\href{https://lce.biohpc.swmed.edu/star/download/starmap/primary_visual_cortex_160_genes/processed_data/dark_replicate_1.zip}
{\nolinkurl{lce.biohpc.swmed.edu/star/download/starmap/primary_visual_cortex_160_genes/processed_data/dark_replicate_1.zip}}
and
\href{https://www.census.gov/data/datasets/2000/dec/microdata.html}
{\nolinkurl{census.gov/data/datasets/2000/dec/microdata.html}},
respectively.


\subsection*{Acknowledgments}
The authors sincerely thank the Editor, Associate Editor, and the referees for their constructive and valuable feedback, which greatly improved the quality of this article.


\subsection*{Funding}
Research reported in this publication was supported by the National Institutes of Health through the National Institute of General Medical Sciences under Award No. R01GM163238 and the National Institute of Dental and Craniofacial Research under Award No. 1R01DE031134-01A1; the National Science Foundation under Award Nos. 2345874 and DMS-2413715; the Texas A\&M Semiconductor Institute; and Sandia National Laboratories.


\subsection*{Disclosure Statement}
The authors have no conflicts of interest to declare.


\subsection*{Supplementary Materials}
The Supplementary Materials provide details on notation and parameterization, empirical studies on some properties of $\tssg$, proofs and supporting details for theoretical guarantees, and additional results from the simulation and real-data studies.


\spacingset{1.3}
\putbib[bibliography]
\end{bibunit}

\tavieStartSupplement
\begin{bibunit}
\def\spacingset#1{\renewcommand{\baselinestretch}%
{#1}\small\normalsize} 

\spacingset{1.3}

\renewcommand\thesection{\Alph{section}}

\renewcommand{\thefigure}{S\arabic{figure}}
\renewcommand{\thetable}{S\arabic{table}}
\begin{center}
{\LARGE\bf Supplementary Materials for\\
A Generalized Tangent Approximation based Variational Inference Framework for Strongly Super-Gaussian Likelihoods}
\end{center}

The Supplementary Materials provide parameterization details (Section~\ref{sec:parameterization-details}), analyses of approximation gaps (Section~\ref{sec:gap-discussion}) and computational complexity (Section~\ref{sec:time-complexity-analysis-TAVIE-SSG}) for $\tssg$, sensitivity studies (Section~\ref{sec:sensitivity-alpha}) and calibration guidelines (Section~\ref{sec:calibration-alpha}) for $\alpha$, discussion of assumptions (Section~\ref{sec:discussion-assumptions}), proofs of the main theoretical results (Sections~\ref{app:TAVIE-convergence-proof} and~\ref{app:variational-risk-bounds}), empirical assessment of risk bounds (Section~\ref{sec:empirical-evidence-risk-bounds}), implementation settings for competing methods (Section~\ref{app-competing-methods-details}), and additional empirical results (Sections~\ref{app:additional-student-results}-\ref{app:model-misspecification-empirics}).

\tableofcontents

\newpage

\section{Notation}
The sets of real and natural numbers are denoted by $\mathbb{R}$ and $\mathbb{N}$, respectively. $\mathbb{R}^{+}$ and $\mathbb{R}_{0}^{+}$ are the sets of positive and nonnegative reals. For integers $n, p \geq 1$, $\mathbb{R}^{n}$ is the $n$-dimensional Euclidean space of real vectors, $\mathbb{R}_{+}^{n}$ and $\mathbb{R}_{+, 0}^{n}$ are the sets of $n$-dimensional positive and nonnegative real vectors, and $\mathbb{R}^{n\times p}$ denotes the space of real matrices of dimension $n\times p$.  We use $[n] := \{1, \ldots, n\}$ for index sets. For  $x \in \mathbb{R}^{n}$, $\lVert x\rVert_{2}$ is its $\ell_2$ norm and $\mathrm{diag}(x)$ is a diagonal matrix of order $n$ with $i$th diagonal entry $x_i$. For a square matrix $A$, $|A|$ denotes its determinant and $\mathrm{tr}(A)$ its trace. The maximum and minimum eigenvalues of a matrix $A\in \mathbb{R}^{n\times n}$ are $\lambda_{\max}(A)$ and $\lambda_{\min}(A)$, respectively. The maximum Euclidean norm among the rows of $A = (a_1^\top, \ldots, a_n^\top)^\top$ is $\lVert A\rVert_{2, \infty}:= \max\{\lVert a_i\rVert_2: i\in [n]\}$. $I_p$ denotes the $p$-dimensional identity matrix. The Hadamard product $A \circ B$ between matrices $A, B \in \mathbb{R}^{n\times p}$ is defined elementwise by $(A\circ B)_{ij} := A_{ij}B_{ij}$. For a symmetric matrix $A$, the notation $A\succ 0$ indicates that $A$ is positive definite. For two sequences $a_n$ and $b_n$, $a_n = \mathcal{O}(b_n)$ and $a_{n}\lesssim b_{n}$ indicates that there exists a constant $C>0$ such that $|a_n| \leq C|b_n|$, and similarly $a_n \gtrsim b_n$ if $a_n \geq C b_n$. Also, $a_n \asymp b_n$ means that both $a_n \lesssim b_n$ and $a_n \gtrsim b_n$ hold. For a differentiable function $f:\mathbb{R}^{p}\to \mathbb{R}$, the gradient with respect to $x$ is denoted by $\nabla_{x}f(x)$, representing the vector of partial derivatives of $f$ evaluated at $x$. $\tfrac{\partial^{k}}{\partial x^{k}}f(x)$ is the $k$th derivative of $f(x)$ for $k\geq 1$. For an extended real-valued function $f:\mathbb{R}^{p} \to (-\infty, \infty]$, the domain of $f$ is $\mathrm{dom}f:=\{x\in \mathbb{R}^{p}: f(x)<\infty\}$. If $f$ is convex but not necessarily differentiable, its subdifferential at a point $x\in \mathrm{dom}f$ is the set $\partial f(x):= \{g\in \mathbb{R}^{p}: f(y) \geq f(x) + \langle g,y-x\rangle\;\text{for all}\;y\in \mathbb{R}^{p}\}$ and any $g\in \partial f(x)$ is called a subgradient of $f$ at $x$. We write the Gamma function as $\Gamma$ and the lower incomplete Gamma function as $\widehat{\Gamma}(s, x):=\int_{0}^{x}t^{s-1}\exp\{-t\}dt$. For a set $A\subseteq \mathbb{R}^{p}$, the indicator function $\mathds{1}_{A}(x)$ is $1$ if $x\in A$ and is $0$ otherwise. Further, for an event $A$, $\mathds{1}(A)$ is $1$ if $A$ occurs and is $0$ otherwise. For two probability measures $\rho$ and $\mu$ on a measurable space, the notation $\rho \ll \mu$ indicates that $\rho$ is absolutely continuous with respect to $\mu$, i.e., $\mu(A) = 0$ implies $\rho(A)=0$ for any measurable set $A$. The Kullback-Leibler (KL) divergence between $\rho$ and $\mu$ is $\mathrm{KL}(\rho \parallel \mu):=\int\log\left(\tfrac{d\rho}{d\mu}\right)d\rho$, where $\tfrac{d\rho}{d\mu}$ is the Radon-Nikodym derivative. Further, for two probability densities $p$ and $q$ on a common measurable space (with respect to a dominating measure), the KL divergence between $p$ and $q$ is $\mathrm{KL}(p\parallel q):= \int p(x)\log\left(\tfrac{p(x)}{q(x)}\right)dx$. Finally, $\mathbb{P}_{\theta_0}$ denotes the probability measure induced by the data-generating distribution under the true parameter value $\theta_0$.

\newpage

\section{Parameterization Details} 
\label{sec:parameterization-details}

As mentioned in Section~\ref{main-subsec:SSG-likelihood-models} of the main manuscript, we consider two broad classes of strongly super-Gaussian ($\mathsf{SSG}$) likelihoods: (i) Type I $\mathsf{SSG}$ likelihoods, comprising non-Gaussian linear regression families typically with heavy tails, such as the Laplace and Student's-$t$ distributions; and (ii) Type II $\mathsf{SSG}$ likelihoods, encompassing 
discrete Bernoulli-type (count) response distributions, such as Binomial (subsuming logistic regression) and Negative-Binomial families. Since several of these distributions admit multiple parameterizations, especially in regression settings, it is imperative to specify the exact forms employed in our method and theoretical analysis. This section provides the explicit parameterizations of such likelihoods and prior distributions used throughout the main manuscript for clarity and reproducibility.

\textbf{Location-scale family, Type I $\mathsf{SSG}$ likelihoods}:
We consider a {base} (or {mother}) distribution with density $f_0(t)$, such that a member of the location-scale family with parameters $(\mu, \tau)$ admits the density:
$$
p(t \mid \mu, \tau) := \tau f_0(\tau(t - \mu)),\quad t\in \mathbb{R},
$$
where $\mu\in \mathbb{R}$ is the location parameter and $\tau\in \mathbb{R}^{+}$ is the precision parameter. Taking $f_0(t) \propto \exp\left\{h(t^2)\right\}$, with $h$ being convex and decreasing, induces a Type I $\mathsf{SSG}$ likelihood. In our formulation, we specifically consider the following parameterizations for the Laplace and Student's-$t$ distributions (with fixed degrees of freedom $\nu\in \mathbb{N}$) as:
\begin{align}\label{eq:param-laplace}
\begin{split}
    \text{Laplace:}\quad &y_i\mid \mathbf{x}_i, \theta \sim \mathrm{Laplace}(\mu_i = \mathbf{x}_i^{\top}\beta, \tau),\\
    &p(y_i \mid \mathbf{x}_i, \theta) = \tfrac{\tau}{2}\exp\{-\tau|y_i - \mathbf{x}_i^{\top}\beta|\},
\end{split}
\end{align}
independently with $y_i \in \mathbb{R}$ for $i\in [n]$ and $\theta = (\beta^{\top}, \tau^{2})^{\top} \in \mathbb{R}^{p}\times \mathbb{R}^{+}$.

\begin{align}\label{eq:param-student}
    \begin{split}
        \text{Student's-}t: \quad &y_i \mid \mathbf{x}_i, \theta, \nu \sim t_{\nu}(\mu_i = \mathbf{x}_i^{\top}\beta, \tau),\\
        &p(y_i \mid \mathbf{x}_i, \theta, \nu) = 
        \frac{\Gamma\!\left(\frac{\nu + 1}{2}\right)}
        {\sqrt{\nu\pi}\,\Gamma\!\left(\frac{\nu}{2}\right)}\tau
        \left(1 + \frac{\tau^2(y_i - \mathbf{x}_i^{\top}\beta)^{2}}{\nu}\right)^{-\frac{\nu + 1}{2}},
    \end{split}
\end{align}
independently with $y_i\in \mathbb{R}$ for $i\in [n]$ along with $\theta = (\beta^{\top}, \tau^{2})^{\top} \in \mathbb{R}^{p}\times \mathbb{R}^{+}$ and fixed (known) degrees of freedom $\nu \in \mathbb{N}$.

\textbf{Bernoulli-type discrete distributions, Type II $\mathsf{SSG}$ likelihoods}: For the Type II $\mathsf{SSG}$ likelihoods, we consider 
Bernoulli-like discrete (count) distributions, such as the Binomial and Negative–Binomial families. The following parameterizations are employed throughout the analysis:
\begin{align}\label{eq:param-binom}
    \begin{split}
        \text{Binomial}: \quad &y_i \mid \mathbf{x}_i, \beta \sim \mathrm{Bin}(m_i, p_i),\\
        & p(y_i\mid \mathbf{x}_i, \beta) = \binom{m_i}{y_i} p_i^{y_i} (1 - p_i)^{\,m_i - y_i},
    \end{split}
\end{align}
independently with $y_i \in \{0,1,\ldots, m_i\}$ and $m_i \in \mathbb{N}$ for $i\in [n]$ along with $\beta\in \mathbb{R}^{p}$.

\begin{align}\label{eq:param-negbin}
    \begin{split}
        \text{Negative-Binomial}: \quad &y_i\mid \mathbf{x}_i, \beta \sim \mathrm{NB}(m_i, p_i),\\
        & p(y_i\mid \mathbf{x}_i, \beta) = \binom{y_i + m_i - 1}{y_i} p_i^{\,m_i} (1 - p_i)^{\,y_i},
    \end{split}
\end{align}
independently with $y_i \in \mathbb{N} \cup \{0\}$ and $m_i > 0$ for $i\in [n]$ along with $\beta\in \mathbb{R}^{p}$. Here, the success probability parameter $p_i$ in \eqref{eq:param-binom} and \eqref{eq:param-negbin} is modeled as:
$$
p_i = \left[1 + \exp\{-\mathbf{x}_i^{\top}\beta\}\right]^{-1}, \quad i\in [n].
$$
Note that the Binomial model in \eqref{eq:param-binom} with $m_i=1$ corresponds to the classical logistic regression formulation.

\textbf{Prior distributions}: The multivariate Gaussian (Normal) distribution, $\mathcal{N}_p(\mu, \Sigma)$, with parameters $(\mu, \Sigma)$ is specified under the standard mean-covariance parameterization, where $\mu \in \mathbb{R}^{p}$ denotes the mean vector and $\Sigma \in \mathbb{R}^{p\times p}$ the covariance matrix. For the Gamma distribution, we adopt the shape-rate parameterization, i.e., a Gamma distribution with shape $c$ and rate $d$, denoted as $\mathcal{G}(c, d)$, has density:
\begin{align}
\label{eq:param-gamma}
\begin{split}
\text{Gamma:} \quad &y\sim \mathcal{G}(c, d),\\ 
&f(y) = \frac{d^{c}}{\Gamma(c)}\, y^{c - 1} e^{-d y},
\end{split}
\end{align}
where $y \in \mathbb{R}^{+}$, $c \in \mathbb{R}^{+}$, and $d \in \mathbb{R}^{+}$.

The Normal-Gamma prior on $(\beta^{\top}, \tau^2)^{\top} \in \mathbb{R}^{p}\times \mathbb{R}^{+}$ with parameters $(\mu, \Sigma, a, b)$, denoted as $\mathcal{NG}_{p}(\mu, \Sigma,a,b)$, is defined as a hierarchical joint distribution comprising:
\begin{align}\label{eq:param-normal-gamma}
    \beta\mid \tau^{2} \sim \mathcal{N}_p\left(\mu, \tfrac{\Sigma}{\tau^2}\right), \quad \tau^{2}\sim \mathcal{G}\left(\tfrac{a}{2}, \tfrac{b}{2}\right).
\end{align}
That is, $\beta$ follows a multivariate Gaussian distribution with mean vector $\mu$ and covariance $\tau^{-2}\Sigma$, conditional on $\tau^2$; while $\tau^{2}$ itself follows a Gamma distribution with shape parameter $\tfrac{a}{2}$ and rate parameter $\tfrac{b}{2}$, respectively. We conclude this section by summarizing the notational conventions for all the distribution families discussed above in Table \ref{tab:distributions} below.
\begin{table}[H]
\centering
\begin{tabular}{ll}
\hline
\hline
\textbf{Distribution Family} & \textbf{Notation} \\
\hline
Student's-$t$ & $t_{\nu}(\mu, \tau)$ \\[4pt]
Laplace & $\mathrm{Laplace}(\mu, \tau)$ \\[4pt]
Binomial & $\mathrm{Bin}(m, p)$ \\[4pt]
Negative-Binomial & $\mathrm{NB}(m, p)$ \\[4pt]
$p$-variate Normal (Gaussian) & $\mathcal{N}_p(\mu, \Sigma)$ \\[4pt]
Gamma & $\mathcal{G}(c, d)$ \\[4pt]
$p$-variate Normal–Gamma & $\mathcal{NG}_p(\mu, \Sigma, a, b)$ \\
\hline
\hline
\end{tabular}
\caption{Distribution families and their notations.}
\label{tab:distributions}
\end{table}

\newpage


\section{Discussion on Approximation Gaps}\label{sec:gap-discussion}

\subsection{Jensen's Gap}\label{subsec:jensen's-gap}
The key conceptual feature of the $\tssg$ framework is the minorization of the true $\ssg$ likelihood contribution $p(y_i\mid \mathbf{x}_i, \theta)$ by a tangent exponential-quadratic function $\varphi(y_i\mid \mathbf{x}_i, \theta, \xi_i)$ for each $i\in [n]$, as developed in Proposition~\ref{main-prop:tangent-lower-bound} of the main manuscript. The gap introduced by this minorization is called the Jensen's gap having the form:
\begin{align}
\label{eq:jensens-gap}
\Delta_{2, i} := \log p(y_{i}\mid \mathbf{x}_i, \theta) - \log \varphi(y_i \mid \mathbf{x}_i, \theta, \xi_i), \quad i \in [n],
\end{align}
which we empirically study here for both Type I and Type II $\ssg$ likelihoods. First note that, in~\eqref{eq:jensens-gap} by substituting the forms of $p(y_i\mid \mathbf x_i, \theta)$ and $\varphi(y_i\mid \mathbf x_i, \theta, \xi_i)$ from~\eqref{main-eq:SSG-defintion} and~\eqref{main-eq:general-minorizer} of the main manuscript, we get:
\begin{align}
\label{eq:jensens-gap-1}
\Delta_{2, i} = t_i[h(\zeta_i^2) - h(\xi_i^{2}) - h'(\xi_i^2)(\zeta_i^2 - \xi_i^2)],
\end{align}
where $\zeta_i = \tau(y_i - \mathbf{x}_i^{\top}\beta), t_i=1$ and $\zeta_i = \mathbf{x}_i^{\top}\beta, t_i = b_i$ for the Type I and Type II $\ssg$ likelihoods, respectively. The explicit forms of the function $h: \mathbb{R}^{+}_{0} \to \mathbb{R}$ for the Type I (Laplace and Student's-$t$) and Type II $\ssg$ likelihoods are given in Section~\ref{main-subsec:SSG-likelihood-models} of the main manuscript. Figure~\ref{fig:jensens-gap} gives a pictorial illustration of the Jensen's gap $\Delta_{2, i}$ in~\eqref{eq:jensens-gap-1} for the Student's-$t$ Type I and Negative-Binomial Type II $\ssg$ likelihoods, respectively, where $\log p(y_{i}\mid \mathbf{x}_i, \theta)$ and $\log \varphi(y_i\mid \mathbf{x}_i, \theta, \xi_i)$ are plotted with respect to $\zeta_{i}$ for different values of $\xi_i$. Observe that, as one moves away from the point of tangency at $\xi_i = |\zeta_i|$, $\Delta_{2, i}$ increases.
\begin{figure}[H]
    \centering

    \begin{subfigure}[t]{0.48\textwidth}
        \centering
        \includegraphics[width=\linewidth]{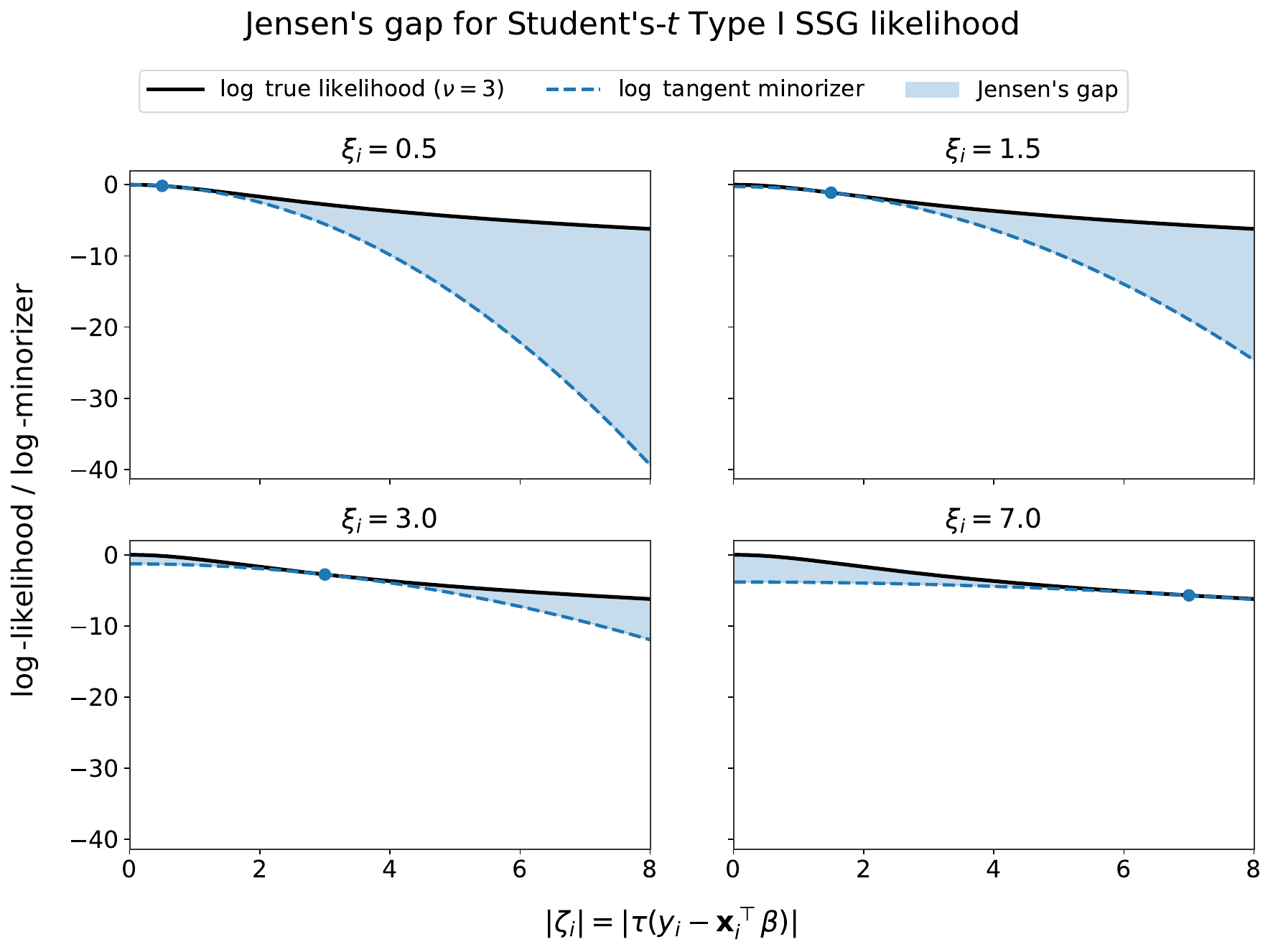}
        \label{fig:student-jensens-gap}
    \end{subfigure}
    \hfill
    \begin{subfigure}[t]{0.48\textwidth}
        \centering
        \includegraphics[width=\linewidth]{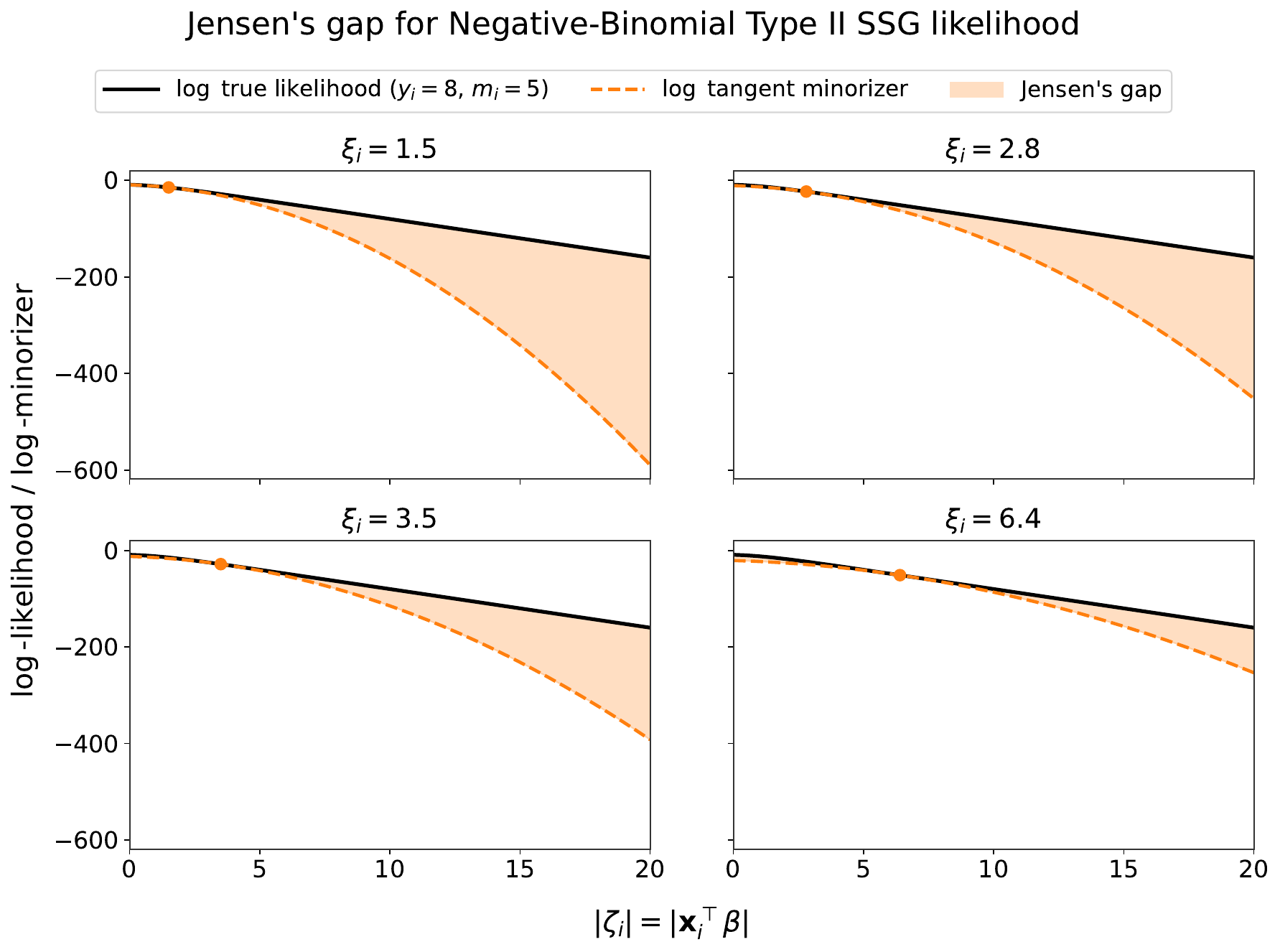}
        \label{fig:negbin-jensens-gap}
    \end{subfigure}

    \caption{\footnotesize{View of the Jensen's gap for Student's-$t$ Type I and Negative-Binomial Type II $\ssg$ likelihoods.}}
    \label{fig:jensens-gap}
\end{figure}

Further, Figure~\ref{fig:jensens-gap-variational-optima} illustrates that, for each observation $i$, $\xi_i^\star$ is the tangency point defining the local quadratic minorization of the likelihood, whereas $\zeta_i^\star$ is the corresponding fitted latent value under the variational posterior, around which the observation-specific contribution to the Jensen's gap is effectively evaluated. Observe that, $\zeta_i^{\star}$ remains close to the tangency point $\xi_i^{\star}$ indicating that the effective Jensen's gap value is small.

\begin{figure}[H]
    \centering

    \begin{subfigure}[t]{0.48\textwidth}
        \centering
        \includegraphics[width=\linewidth]{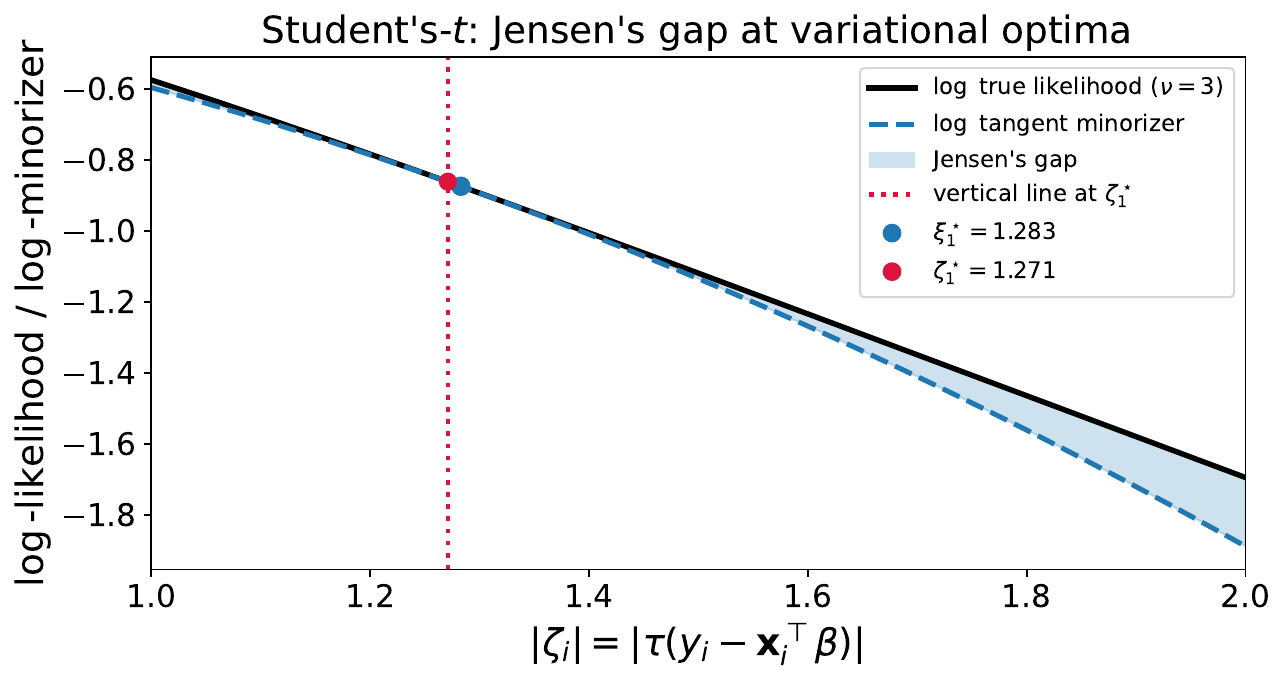}
        \label{fig:student-jensens-gap-variational-optima}
    \end{subfigure}
    \hfill
    \begin{subfigure}[t]{0.48\textwidth}
        \centering
        \includegraphics[width=\linewidth]{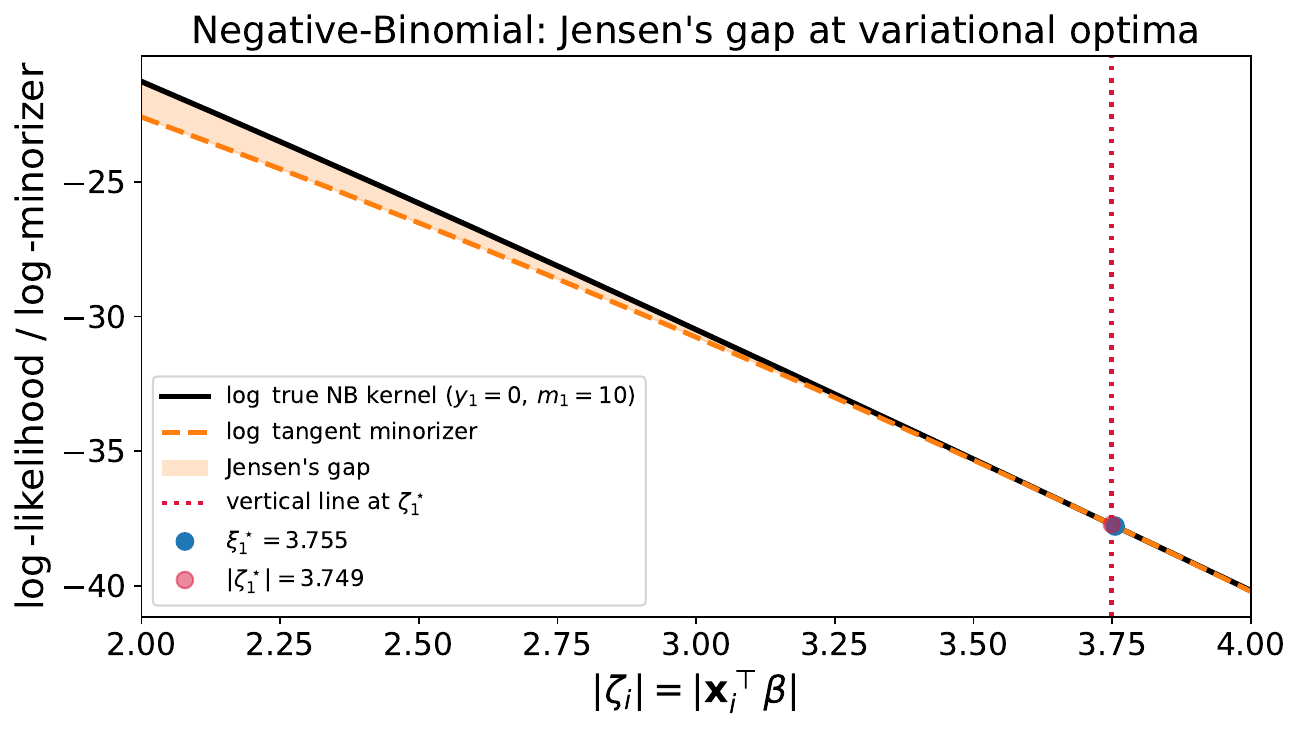}
        \label{fig:negbin-jensens-gap-variational-optima}
    \end{subfigure}

    \caption{\footnotesize{View of the Jensen's gap for Student's-$t$ Type I and Negative-Binomial Type II $\ssg$ likelihoods at the variational optima $\xi_i^{\star}$ and consequently the fitted latent value $\zeta_{i}^{\star}$ for $i=1$. For both likelihoods, $\tssg$ runs were carried out with $(n, p) = (500, 9)$, using $\nu = 3$ and $\tau^{2} = 3$ for the Student's-$t$ Type I $\ssg$ likelihood and $m_i=10$ for the Negative-Binomial Type II $\ssg$ likelihood.}}
    \label{fig:jensens-gap-variational-optima}
\end{figure}

\subsection{\texorpdfstring{$\mathsf{ELBO}$}{ELBO} Gap}\label{subsec:ELBO-gap}
\label{sec:elbo-gap}
The next component of the $\tssg$ framework is the composition of the resulting minorized (working) likelihood raised to the likelihood tempering parameter $\alpha \in (0, 1]$ viz., $\varphi_{\alpha}(y\mid \mathbf{X}, \theta, \xi) = \left\{\prod_{i\in [n]}\varphi(y_i \mid \mathbf x_i, \theta, \xi_i)\right\}^{\alpha}$ with the conjugate prior $\pi(\theta)$ over $\theta$ leading to the $\alpha$-fractional variational posterior $\pi_{\alpha}(\theta \mid \mathcal{D}_n, \xi)$ for $\xi = (\xi_1, \ldots, \xi_n)^{\top}\in \mathbb{R}^{n}_+$. The associated evidence lower bound ($\mathsf{ELBO}$) with respect to the working likelihood is:
\begin{align}
\label{eq:TAVIE-SSG-ELBO}
\mathcal{L}(q, \xi) := \int_{\theta\in \Theta}\log\frac{\varphi_{\alpha}(y\mid \mathbf X, \theta, \xi) \pi(\theta)}{q(\theta)}q(\theta)d\theta,
\end{align}
which is maximized, for fixed $\xi$, at $q_{\xi}\equiv \pi_{\alpha}(\cdot\mid \mathcal D_n, \xi)$. Hence the profile objective becomes:
\begin{align}
\label{eq:profile}
\mathsf{L}(\xi) := \mathcal{L}(q_\xi, \xi) = \log \varphi_{\alpha}(y\mid \mathbf X, \xi),
\end{align}
with the explicit closed forms as in~\eqref{main-eq:ELBO-general} of the main manuscript for Type I and Type II $\ssg$ likelihoods. The profile objective in~\eqref{eq:profile} is maximized over $\xi\in \mathbb{R}^{n}_+$ using the $\tssg$ EM routine in Algorithm~\ref{main-alg:tavie-em} of the main manuscript, which yields $\xi^\star \in \mathbb{R}^{n}_+$ as the optimizer. 

At this point, it is important to distinguish $\mathsf{L}(\xi)$ from the true variational objective based on the true $\alpha$-fractional likelihood $p_{\alpha}(y\mid \mathbf{X}, \theta) = \left\{\prod_{i\in [n]}p(y_i\mid \mathbf{x}_i, \theta)\right\}^{\alpha}$. The corresponding true $\mathsf{ELBO}$ is:
\begin{align}
\label{eq:true-ELBO}
\mathcal{L}_{\mathrm{true}}(q) := \int_{\theta\in \Theta}\log \frac{p_{\alpha}(y\mid \mathbf X, \theta)\pi(\theta)}{q(\theta)}q(\theta)d\theta, 
\end{align}
and upon substituting the same variational posterior $q_{\xi} \equiv \pi_{\alpha}(\cdot \mid \mathcal{D}_n, \xi)$, we obtain:
\begin{align}
\label{eq:true-profile-ELBO}
\mathsf{L}_{\mathrm{true}}(\xi) := \int_{\theta\in \Theta}\log \frac{p_{\alpha}(y\mid \mathbf X, \theta)\pi(\theta)}{\pi_{\alpha}(\theta \mid \mathcal D_n, \xi)}\pi_{\alpha}(\theta \mid \mathcal D_{n}, \xi)d\theta. 
\end{align}
Thus, $\mathsf{L}(\xi)$ and $\mathsf{L}_{\mathrm{true}}(\xi)$ are evaluated at the same variational posterior family, but differ in the likelihood term appearing inside the logarithm: the former is based on the tangent minorized working likelihood $\varphi_{\alpha}(y\mid \mathbf{X}, \theta, \xi)$, while the latter is based on the true $\alpha$-fractional likelihood $p_{\alpha}(y\mid \mathbf X, \theta)$. In particular, $\mathsf{L}(\xi)$ is the quantity that is naturally optimized by $\tssg$, whereas evaluating $\mathsf{L}_{\mathrm{true}}(\xi)$ requires an additional Monte Carlo (MC) layer beyond the closed-form calculations already used in the variational updates in Algorithm~\ref{main-alg:tavie-em} of the main manuscript. For this reason, $\mathsf{L}(\xi)$ serves as the primary convergence diagnostic for $\tssg$, while $\mathsf{L}_{\mathrm{true}}(\xi)$ is better viewed as a post hoc diagnostic quantity used to assess the extra approximation error induced by the tangent minorization.

Figure~\ref{fig:true_variational_elbo} compares these two quantities over the $\tssg$ iterations for representative $\ssg$ likelihood families. Consistent with the minorization property $\varphi_{\alpha}(y\mid \mathbf X, \theta, \xi) \leq p_{\alpha}(y\mid \mathbf X, \theta)$ for all $\theta\in \Theta$, we observe that $\mathsf{L}_{\mathrm{true}}(\xi)\geq \mathsf{L}(\xi)$ with both quantities displaying similar stabilization behavior near the convergence region.
\begin{figure}[H]
    \centering
    \includegraphics[width=\linewidth]{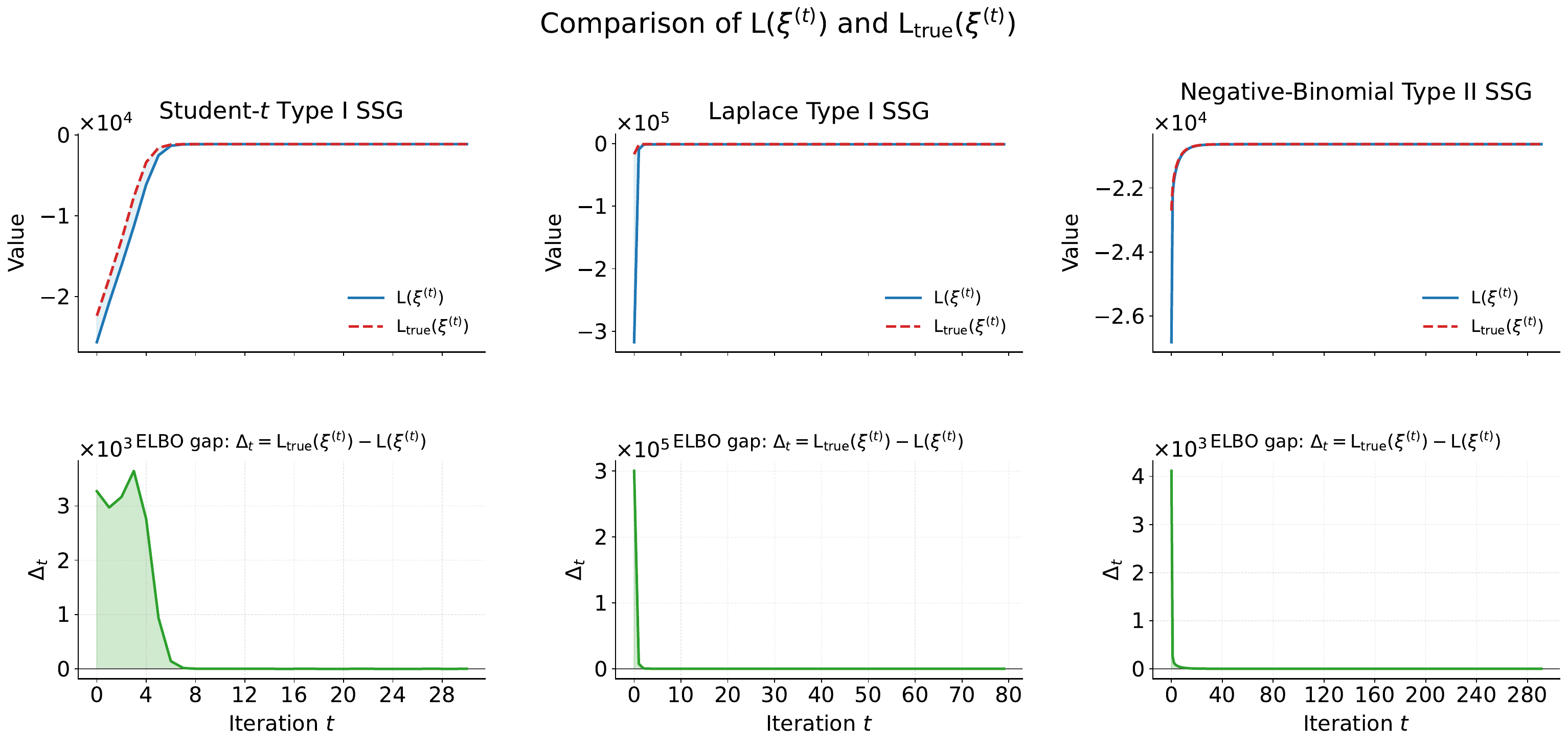}
    \caption{\footnotesize{Comparison of $\mathsf{L}(\xi)$ and $\mathsf{L}_{\mathrm{true}}(\xi)$ over $\tssg$ iterations for representative $\ssg$ likelihood families. Throughout, we fix $(n, p) = (1000, 9)$.}}
    \label{fig:true_variational_elbo}
\end{figure}

This motivates the following quantity, regarded as the $\mathsf{ELBO}$ gap which measures the discrepancy between the true profile variational objective and the profile objective actually optimized by $\tssg$:
\begin{align}
\label{eq:ELBO-gap}
\begin{split}
\mathsf{L}_{\mathrm{true}}(\xi) - \mathsf{L}(\xi) &:= \int_{\theta\in \Theta}\log \frac{p_{\alpha}(y\mid \mathbf X, \theta)}{\varphi_\alpha(y\mid \mathbf{X}, \theta, \xi)}\pi_{\alpha}(\theta\mid \mathcal D_n, \xi)d\theta \\
&= \mathbb{E}_{\xi}\left[\alpha \sum_{i\in [n]}\Delta_{2, i}\right] = \alpha \mathbb{E}_{\xi}[\Delta_2],
\end{split}
\end{align}
where $\mathbb{E}_{\xi}$ is the expectation with respect to $\pi_{\alpha}(\theta\mid \mathcal D_n, \xi)$, $\Delta_{2, i}$ is the Jensen's gap defined in~\eqref{eq:jensens-gap}, and $\Delta_2 = \sum_{i\in [n]}\Delta_{2, i}$. Hence, the $\mathsf{ELBO}$ gap is the variational posterior average of the total Jensen's gap $\Delta_2$ scaled by $\alpha$. In this sense, the $\mathsf{ELBO}$ gap provides a global summary of the approximation error induced by replacing the true likelihood with its tangent minorizer, after averaging over the posterior uncertainty encoded by $\pi_{\alpha}(\theta\mid \mathcal{D}_n, \xi)$.

The lower panel of Figure~\ref{fig:true_variational_elbo} highlights this $\mathsf{ELBO}$ gap across the $\tssg$ iterations. Empirically, the gap remains controlled and stabilizes as the algorithm converges, indicating that the profile objective $\mathsf{L}(\xi)$ optimized by $\tssg$ remains close to the corresponding true profile objective $\mathsf{L}_{\mathrm{true}}(\xi)$ in the representative $\ssg$ models considered.

Observe also that, the $\mathsf{ELBO}$ gap quantifies the approximation cost incurred by optimizing the tangent-minorized working objective $\mathsf{L}(\xi)$ in lieu of the true profile variational objective $\mathsf{L}_{\mathrm{true}}(\xi)$.
In this sense, the $\mathsf{ELBO}$ gap may be viewed as an effective complexity penalty induced by the tangent approximation: if $\mathsf{L}_{\mathrm{true}}(\xi)$ is regarded as the ideal objective and $\mathsf{L}(\xi)$ as the tractable proxy actually optimized by $\tssg$, then the $\mathsf{ELBO}$ gap can be interpreted as the loss in objective value in exchange for the corresponding gain in computational tractability.

\subsection{Variational Gap}\label{subsec:variational-gap}

The discussion above in Section~\ref{sec:elbo-gap} isolates the approximation error contributed by the tangent minorization step alone. However, the $\mathsf{ELBO}$ gap is only one component of the overall discrepancy between the exact Bayesian target and the final $\tssg$ variational approximation. The remaining component is the intrinsic discrepancy due to variational inference (VI) itself, i.e., the gap between the true $\log$-$\alpha$-fractional marginal likelihood $\log p_{\alpha}(y\mid \mathbf{X})$ and the true $\mathsf{ELBO}$ $\mathsf{L}_{\mathrm{true}}(\xi)$. By the standard $\mathsf{ELBO}$ identity (see Lemma~\ref{lemma-auxiliary}), this quantity is exactly the Kullback-Leibler (KL) divergence between the variational posterior $\pi_{\alpha}(\theta \mid \mathcal{D}_n, \xi)$ and the true $\alpha$-fractional posterior $\pi_{\alpha}(\theta\mid \mathcal D_n)$. Accordingly, the total discrepancy, which we refer to as the variational gap, admits the decomposition:
\begin{align}
\label{eq:variational-gap}
\begin{split}
    \log p_{\alpha}(y\mid \mathbf{X}) - \mathsf{L}(\xi) &:= \log p_{\alpha}(y\mid \mathbf{X}) - \mathsf{L}_{\mathrm{true}}(\xi) + \mathsf{L}_{\mathrm{true}}(\xi) - \mathsf{L}(\xi)\\
    &=\mathrm{KL}(\pi_{\alpha}(\theta \mid \mathcal{D}_n, \xi)\parallel \pi_{\alpha}(\theta \mid \mathcal{D}_n)) + \mathsf{ELBO}\text{ gap}.
\end{split}
\end{align}

Thus, the variational gap in~\eqref{eq:variational-gap} captures the total approximation error induced jointly by tangent minorization and variational approximation.
Analogously to the $\mathsf{ELBO}$ gap, the variational gap may be viewed as the total effective complexity penalty associated with the tractable $\tssg$ objective. 
This perspective naturally connects both the $\mathsf{ELBO}$ and variational gaps to classical model selection heuristics. In likelihood-based model comparison, the Bayesian Information Criterion (BIC) penalizes the increase in model dimension at the canonical scale $\tfrac{p}{2}\log n$. Therefore, if the post-convergence variational gap $\log p_{\alpha}(y\mid \mathbf{X}) - \mathsf{L}(\xi^\star)$ grows with $p$ on the same scale as the classical BIC penalty then the resulting $\tssg$ approximation error remains asymptotically comparable to the usual likelihood-complexity trade-off, rather than dominating it. More concretely, when comparing two models of dimensions $p_1$ and $p_2$ ($p_2 > p_1$), the question is whether the gain in the optimized tractable objective $\mathsf L(\xi^\star)$ is being distorted by approximation error only at the same scale as the usual BIC penalty $\tfrac{p_2-p_1}{2}\log n$, or at a substantially larger scale. Thus, empirically studying the complexity behavior of both the $\mathsf{ELBO}$ and variational gaps, and benchmarking them against the BIC penalty, provides a direct way to assess whether model selection under $\tssg$ remains compatible with standard asymptotic complexity control.

\subsection{Empirical Complexity Analysis of \texorpdfstring{$\mathsf{ELBO}$}{ELBO} and Variational Gaps}\label{subsec:empirical-elbo-variational-gap}

For illustration, we now consider the Student's-$t$ Type I $\ssg$ likelihood as in~\eqref{eq:param-student}. Data $\mathcal{D}_n:= \{(\mathbf x_i, y_i)\;:\; i\in [n]\}$ is simulated using the true model parameters: $\nu = 5$, $\tau_0^{2} = 3$, $\beta_0 \sim \mathcal{N}_{p}(0, I_{p})$, and $\mathbf{X} = (\mathbf x_1, \ldots, \mathbf x_n)^{\top}\in \mathbb{R}^{n\times p}$ with $x_{i1} = 1$ (intercept) and $x_{ij} \stackrel{\mathrm{i.i.d.}}{\sim}\mathcal{N}_{1}(0, 1)$ for $j=2, \ldots, p$. Across all experiments $\tssg$ is implemented with configurations as mentioned in Section~\ref{main-subsec:sim-exp-student} of the main manuscript. To quantify the post-convergence $\mathsf{ELBO}$ gap, we compute both $\mathsf{L}(\xi^\star)$ from $\tssg$ and a Monte Carlo (MC) approximation of $\mathsf{L}_{\mathrm{true}}(\xi^\star)$ (using $1000$ MC draws), and summarize the corresponding gap in Figure~\ref{fig:ELBO_gap_student_t_fix_n_vary_p} over $10$ replications for $n=1000$ and $p\in \{9, 16, 21, 26\}$. We estimate $\mathsf{L}_{\mathrm{true}}(\xi^\star) - \mathsf{L}(\xi^\star)$ by averaging $\mathsf{L}_{\mathrm{true}}(\xi) - \mathsf{L}(\xi)$ over the last $10\%$ of the converged $\tssg$ trajectory, rather than using only the final iterate, to reduce Monte Carlo error in $\mathsf{L}_{\mathrm{true}}(\xi^\star)$. As expected, the post-convergence $\mathsf{ELBO}$ gap increases monotonically with dimension $p$, with relatively small variability across repetitions, indicating a clear and stable growth of tangent approximation cost as model complexity increases.
\begin{figure}[!htp]
    \centering
    \includegraphics[width=0.6\linewidth]{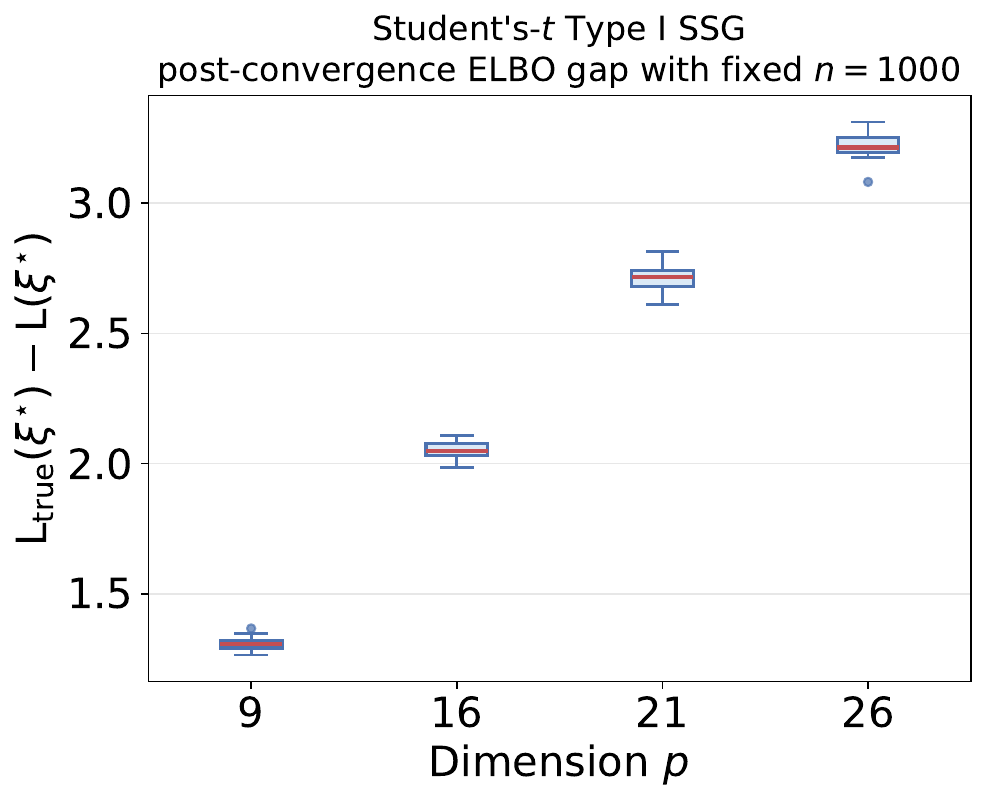}
    \caption{\footnotesize{Post-convergence $\mathsf{ELBO}$ gap $\mathsf{L}_{\mathrm{true}}(\xi^{\star}) - \mathsf{L}(\xi^{\star})$ for the Student's-$t$ Type I $\ssg$ likelihood as model dimension $p$ increases.}}
    \label{fig:ELBO_gap_student_t_fix_n_vary_p}
\end{figure}
\begin{figure}[H]
    \centering

    \begin{subfigure}[t]{0.43\textwidth}
        \centering
        \includegraphics[width=\linewidth]{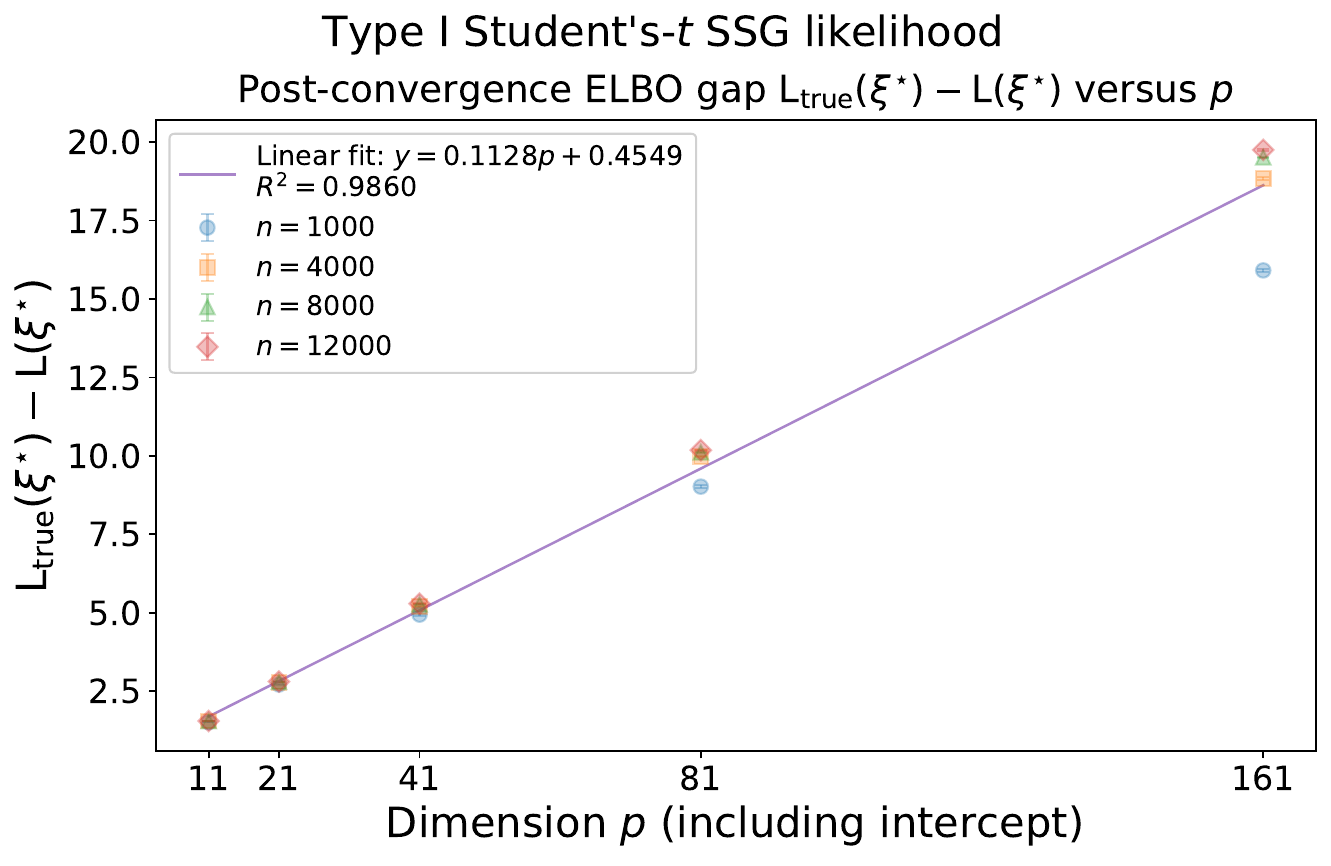}
        \caption{\footnotesize{$\mathsf{ELBO}$ gap against $p$.}}
        \label{fig:gap_p}
    \end{subfigure}
    \hfill
    \begin{subfigure}[t]{0.46\textwidth}
        \centering
        \includegraphics[width=\linewidth]{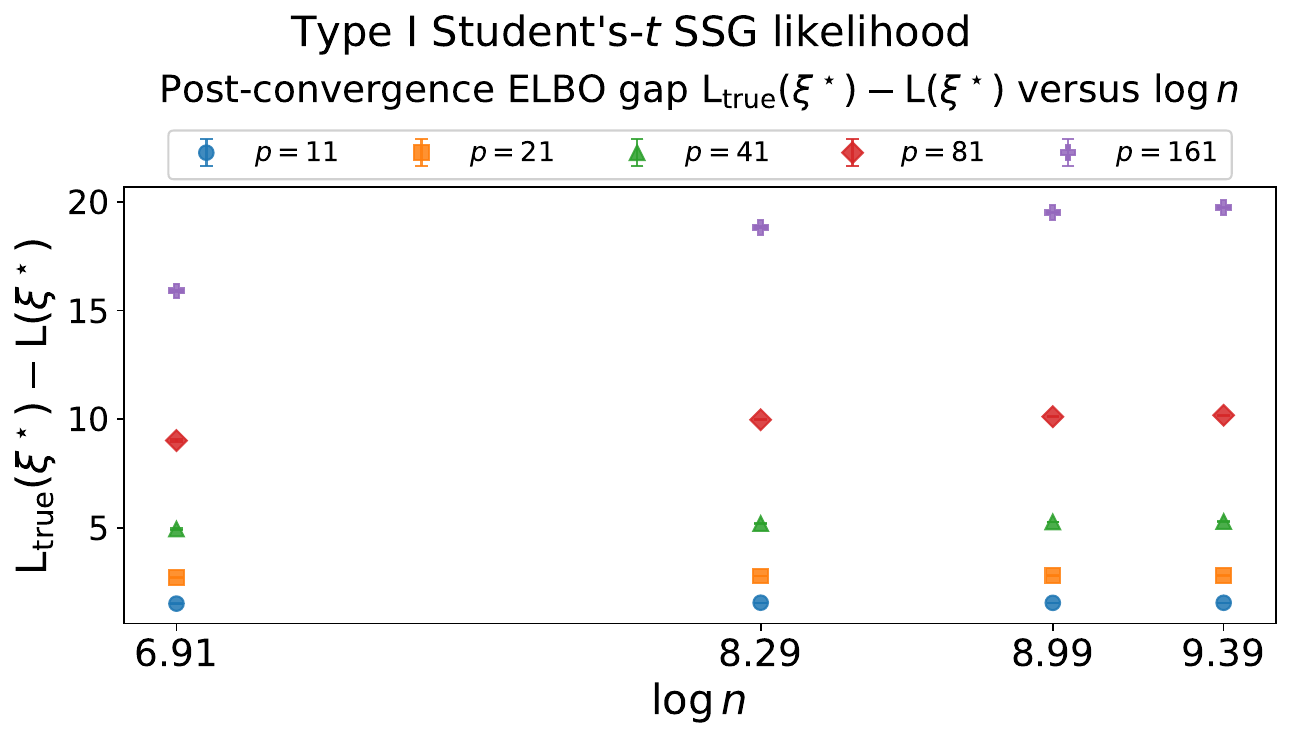}
        \caption{\footnotesize{$\mathsf{ELBO}$ gap against $\log n$.}}
        \label{fig:gap_log_n}
    \end{subfigure}

    \vspace{0.8em}

    \begin{subfigure}[t]{0.43\textwidth}
        \centering
        \includegraphics[width=\linewidth]{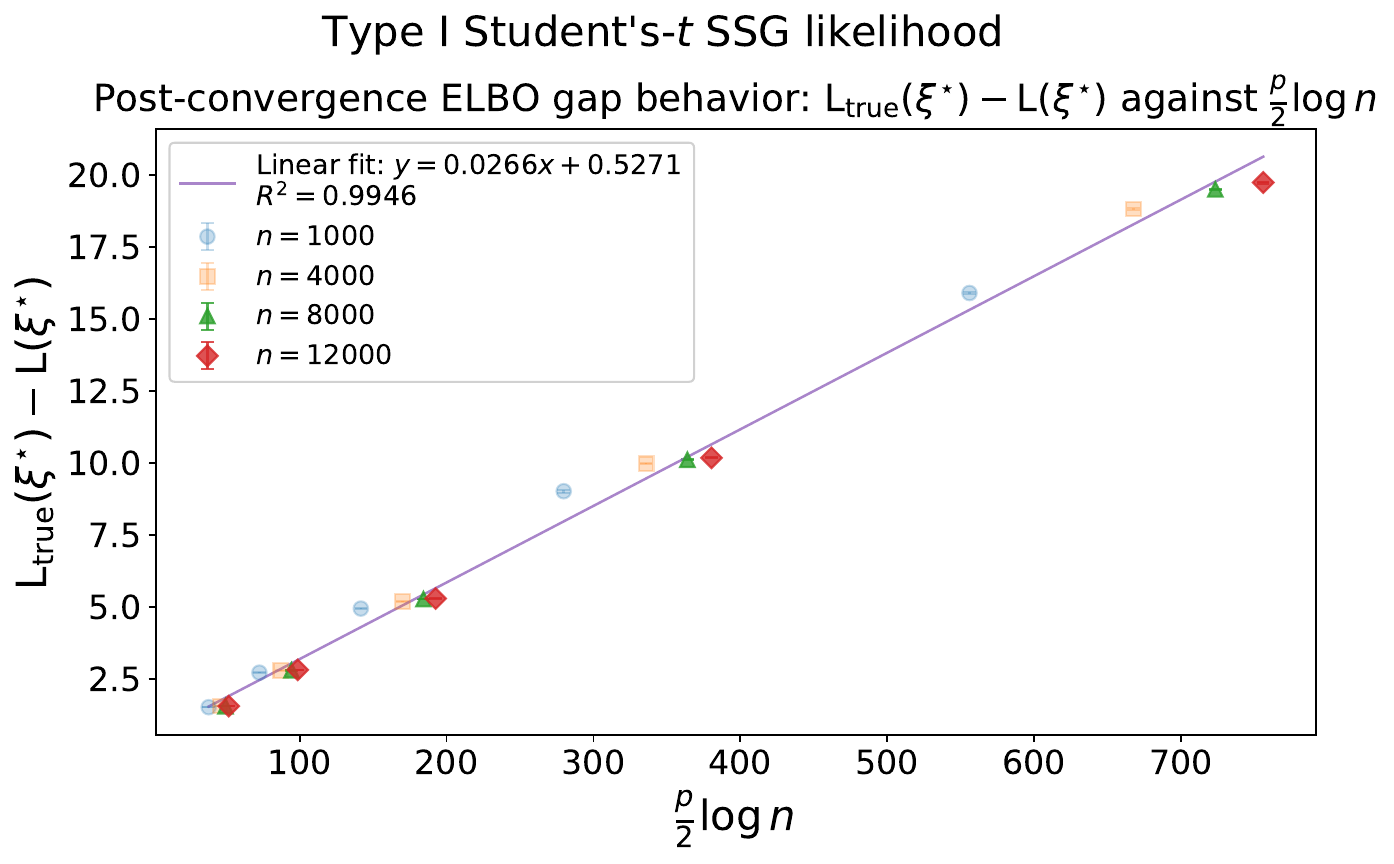}
        \caption{\footnotesize{$\mathsf{ELBO}$ gap against $\tfrac{p}{2}\log n$.}}
        \label{fig:gap_half_plog_n}
    \end{subfigure}
    \hfill
    \begin{subfigure}[t]{0.46\textwidth}
        \centering
        \includegraphics[width=\linewidth]{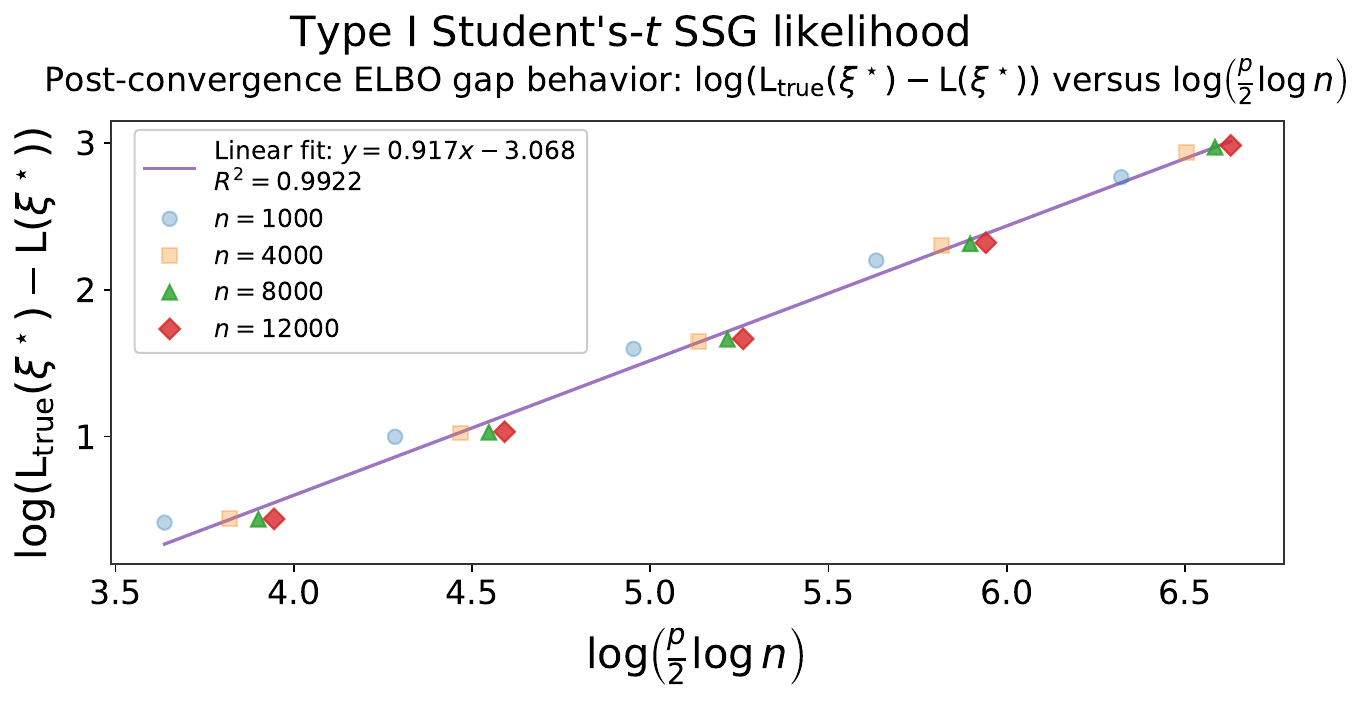}
        \caption{\footnotesize{$\log$ $\mathsf{ELBO}$ gap against $\log\left(\tfrac{p}{2}\log n\right)$.}}
        \label{fig:log_gap_log}
    \end{subfigure}

    \caption{\footnotesize{Empirical complexity analysis of the post-convergence $\mathsf{ELBO}$ gap $\mathsf{L}_{\mathrm{true}}(\xi^\star) - \mathsf{L}(\xi^\star)$ with respect to $n$ and $p$ for the Student's-$t$ Type I $\ssg$ likelihood.}}
    \label{fig:ELBO_gap_order}
\end{figure}

Now we take a closer look at the increase in the $\mathsf{ELBO}$ gap observed in Figure~\ref{fig:ELBO_gap_student_t_fix_n_vary_p} through the three diagnostics in Figure~\ref{fig:ELBO_gap_order}. Firstly, Figure~\ref{fig:gap_p} plots the post-convergence $\mathsf{ELBO}$ gap against $p\in \{11, 21, 41, 81, 161\}$ across several sample sizes $n\in \{1000, 4000, 8000, 12000\}$ over $10$ repetitions, and shows a clear linear increase with dimension $p$, indicating that the tangent approximation cost grows systematically as model complexity increases. Secondly, Figure~\ref{fig:gap_log_n} examines the corresponding gap against $\log n$ for fixed values of $p$; in comparison to the strong dependence on $p$, the variation with $n$ is relatively much milder within each dimension, suggesting that the growth of the gap is predominantly due to the dimensional complexity rather than sample size. Finally, Figures~\ref{fig:gap_half_plog_n} and~\ref{fig:log_gap_log} plot the post-convergence $\mathsf{ELBO}$ gap against $\tfrac{p}{2}\log n$ both in the original and $\log$-$\log$ scale. Particularly, the near-linear trend in the $\log$-$\log$ representation with slope close to $1$ provides empirical evidence that the post-convergence $\mathsf{ELBO}$ gap scales approximately the same as the BIC complexity order $\tfrac{p}{2}\log n$.

After analyzing the post-convergence $\mathsf{ELBO}$ gap, we now turn to the corresponding complexity analysis of the post-convergence variational gap $\log p_{\alpha}(y\mid \mathbf{X}) - \mathsf{L}({\xi^\star})$ in~\eqref{eq:variational-gap} (with the same simulation setup as in the $\mathsf{ELBO}$ gap analysis), with the goal of understanding how the overall approximation error varies with $p$ and $n$, and whether its growth remains aligned with the canonical BIC scale $\tfrac p2 \log n$. This comparison is especially informative because, unlike the $\mathsf{ELBO}$ gap, which isolates only the contribution of tangent minorization, the variational gap captures the combined effect of both the tangent approximation and the intrinsic variational approximation. It therefore provides a more complete measure of the approximation cost underlying model selection with $\tssg$.
\begin{figure}[!htp]
    \centering
    \includegraphics[width=0.8\linewidth]{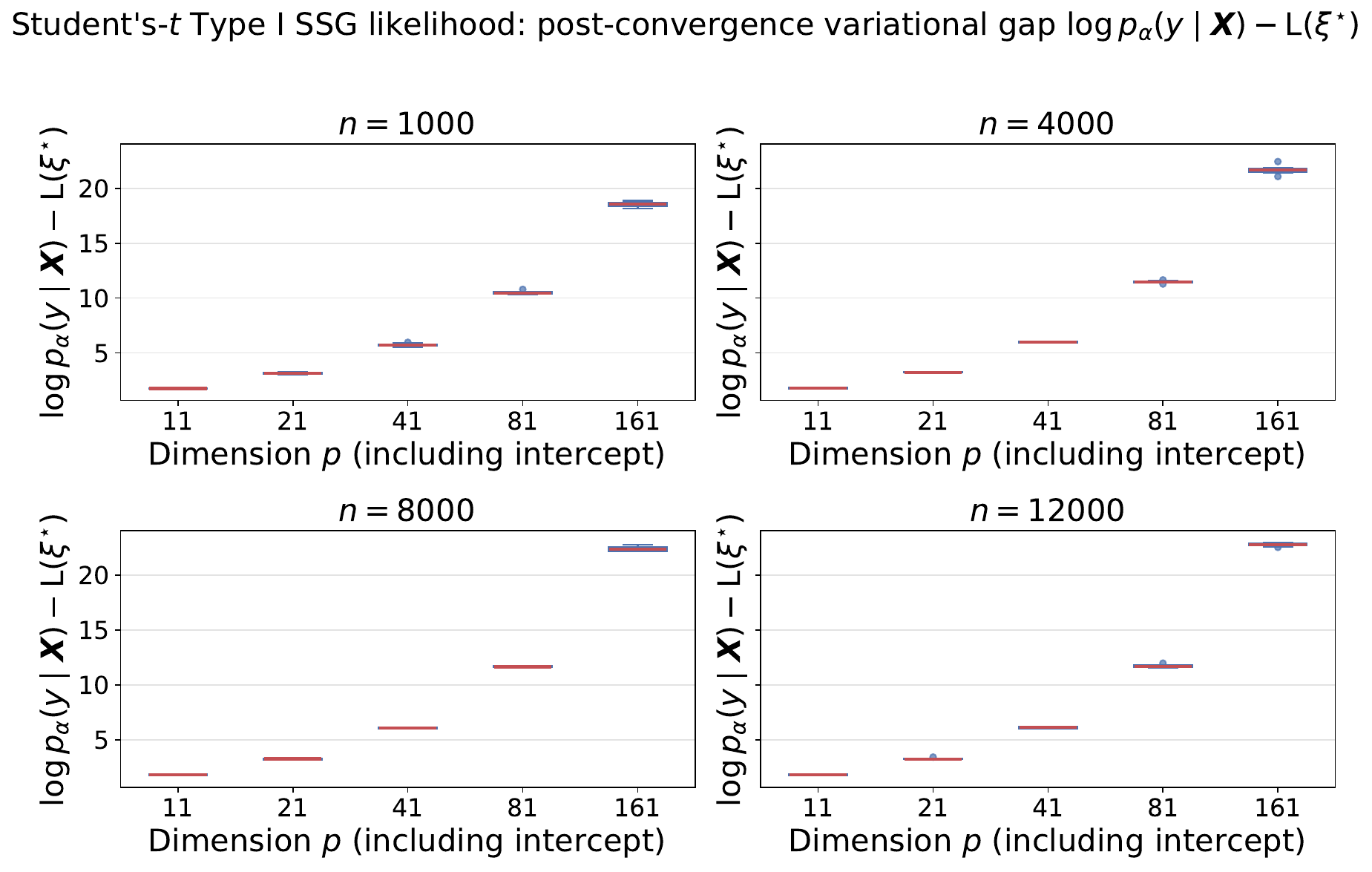}
    \caption{\footnotesize{Post-convergence variational gap $\log p_{\alpha}(y\mid \mathbf X) - \mathsf{L}(\xi^{\star})$ for the Student's-$t$ Type I $\ssg$ likelihood as model dimension $p$ increases for different sample sizes $n$.}}
    \label{fig:variational_gap_student_boxplot}
\end{figure}
Figure~\ref{fig:variational_gap_student_boxplot} first reports the post-convergence variational gap across increasing dimensions $p\in \{11, 21, 41, 81, 161\}$ for several fixed sample sizes $n \in \{1000, 4000, 8000, 12000\}$ over $10$ repetitions. Much like the $\mathsf{ELBO}$ gap, the variational gap exhibits a clear monotone increase with $p$, with only limited variability across repetitions, indicating a stable and systematic growth of total approximation cost as the model becomes more complex. At the same time, for a fixed $p$, the variation across different values of $n$ is comparatively modest, suggesting again that dimensional complexity plays the dominant role.

To examine the order of this growth more directly, Figures~\ref{fig:variational_gap_1} and~\ref{fig:variational_gap_2} plot $\log p_{\alpha}(y\mid \mathbf X) - \mathsf{L}(\xi^\star)$ against $\tfrac p2\log n$ both in the original and $\log$-$\log$ scale. The near-linear trend in the $\log$-$\log$ representation provides empirical evidence that the variational gap also scales in close accordance with the BIC-type complexity proxy $\tfrac{p}{2}\log n$. Although the variational gap is naturally larger than the $\mathsf{ELBO}$ gap, since it incorporates both sources of approximation error, its growth still appears to remain commensurate with the same effective complexity scale.
Taken together, these plots suggest that the full approximation error induced by $\tssg$ remains controlled relative to the standard likelihood-complexity trade-off, thereby supporting the view that model selection based on the tractable $\tssg$ objective continues to operate within a classical asymptotic complexity regime.

\begin{figure}[H]
    \centering

    \begin{subfigure}[t]{0.48\textwidth}
        \centering
        \includegraphics[width=\linewidth]{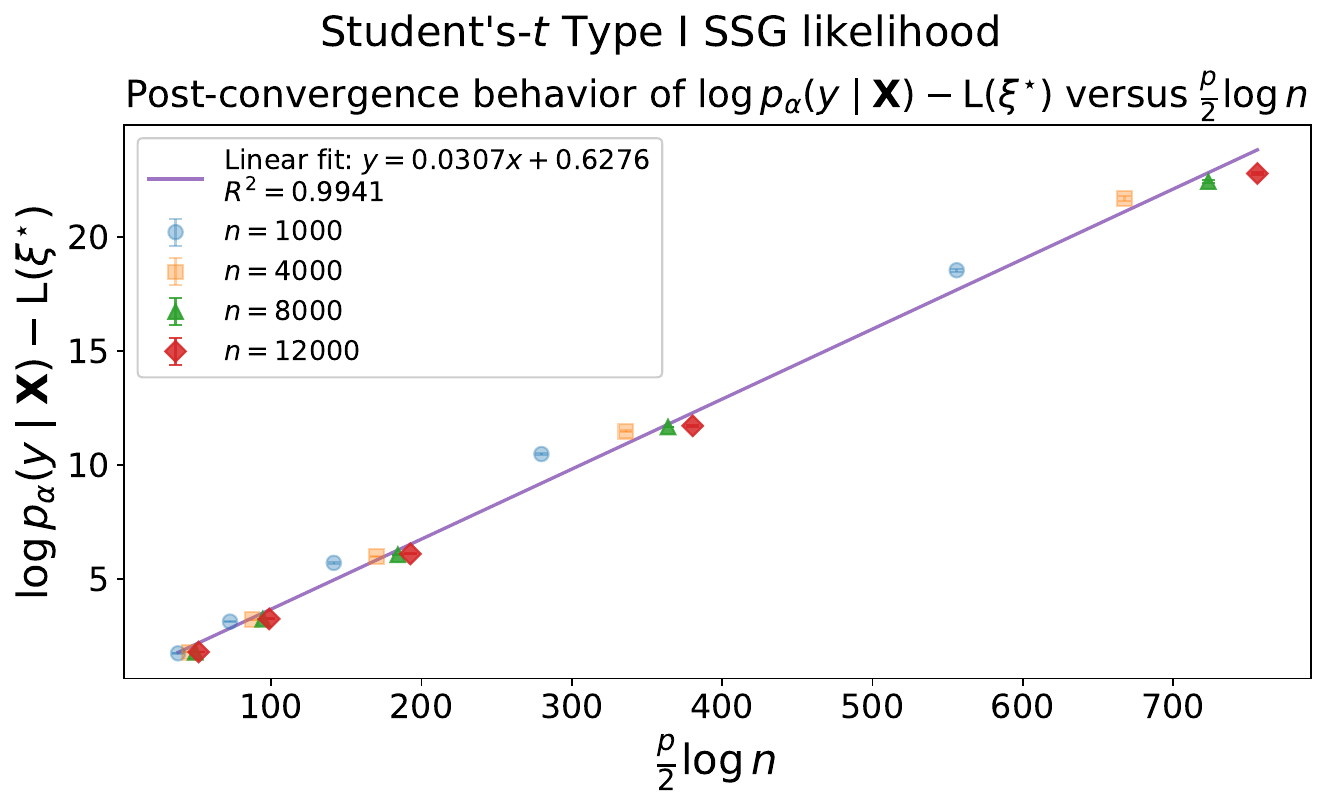}
        \caption{\footnotesize{Variational gap against $\tfrac{p}{2}\log n$.}}
        \label{fig:variational_gap_1}
    \end{subfigure}
    \hfill
    \begin{subfigure}[t]{0.48\textwidth}
        \centering
        \includegraphics[width=\linewidth]{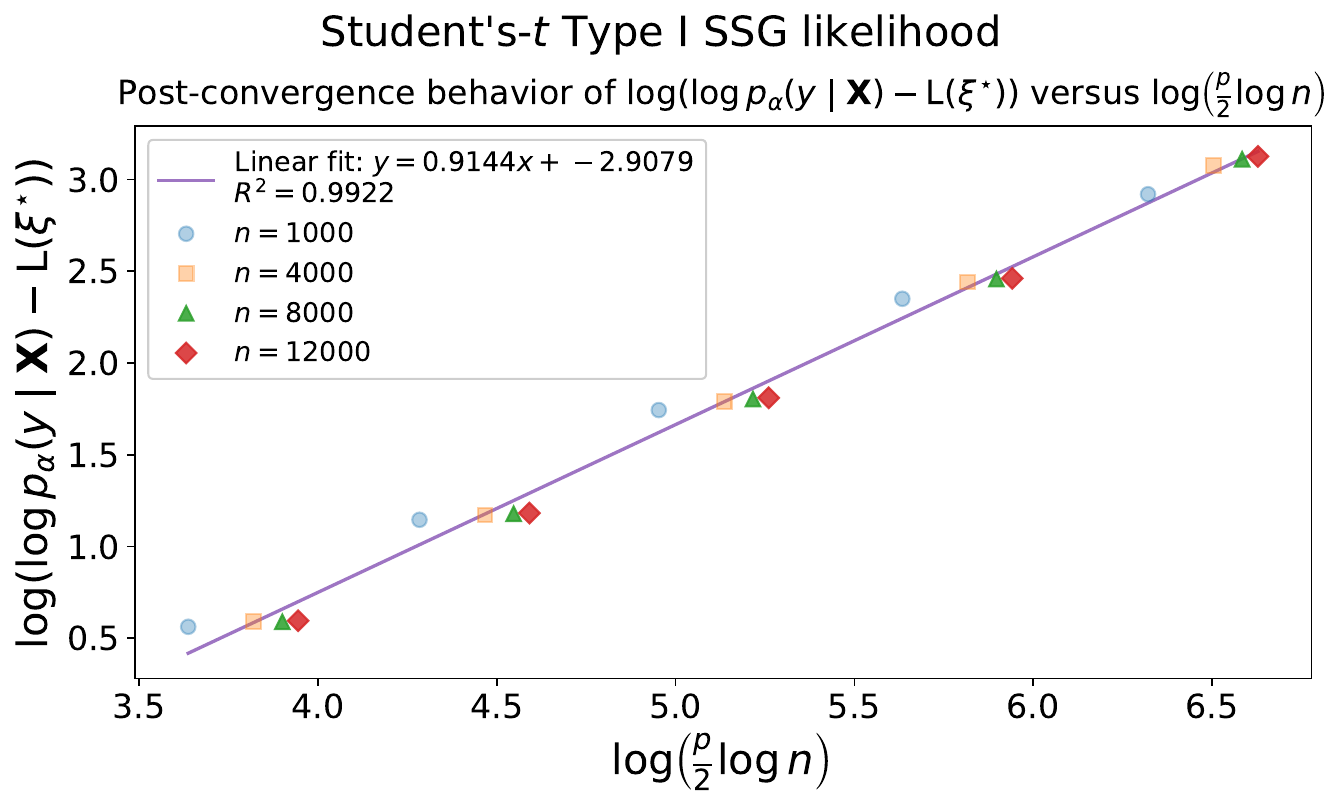}
        \caption{\footnotesize{$\log$ variational gap against $\log(\tfrac{p}{2}\log n)$.}}
        \label{fig:variational_gap_2}
    \end{subfigure}
\caption{\footnotesize{Empirical complexity analysis of the post-convergence variational gap $\log p_{\alpha}(y\mid \mathbf{X}) - \mathsf{L}(\xi^{\star})$ with respect to $\tfrac p2 \log n$ (in original and $\log$-$\log$ scale) for the Student's-$t$ Type I $\ssg$ likelihood.}}
\label{fig:student_t_log_gap_vs_log_half_p_log_n_linear_fit}
\end{figure}

In summary, these findings suggest that the post-convergence $\tssg$ $\mathsf{ELBO}$ $\mathsf{L}(\xi^\star)$ may serve as a practically meaningful criterion for model selection, in the sense that selecting the model maximizing $\mathsf{L}(\xi^\star)$ appears compatible with classical complexity control. Indeed, both the post-convergence $\mathsf{ELBO}$ and variational gaps exhibit growth on the canonical BIC scale $\tfrac{p}{2}\log n$. This interpretation is further reinforced by Figure~\ref{fig:student_t_elbo_minus_true_loglik_vs_half_p_log_n_linear_fit}, which shows that the post-convergence discrepancy $\mathsf{L}(\xi^\star)-\log p_{\alpha}(y\mid \mathbf X,\theta_0)$ displays an essentially perfect negative linear trend as a function of $\tfrac{p}{2}\log n$, where $\theta_0$ denotes the true parameter vector. Hence, up to sign, the magnitude of this discrepancy also scales on the same canonical BIC order. Overall, these empirical findings strongly support the conclusion that, at least in the present regular setting, model comparison based on the tractable $\tssg$ objective reflects the same first-order complexity regime as BIC, in close agreement with the broader theoretical conclusion that $\mathsf{ELBO}$-based model selection can asymptotically align with BIC-based model selection~\citep{cheriefabdellatif2018consistency,cheriefabdellatif2019elbo,ZhangYang2024MFVBModelSelection}.

\begin{figure}[H]
    \centering
    \includegraphics[width=0.7\linewidth]{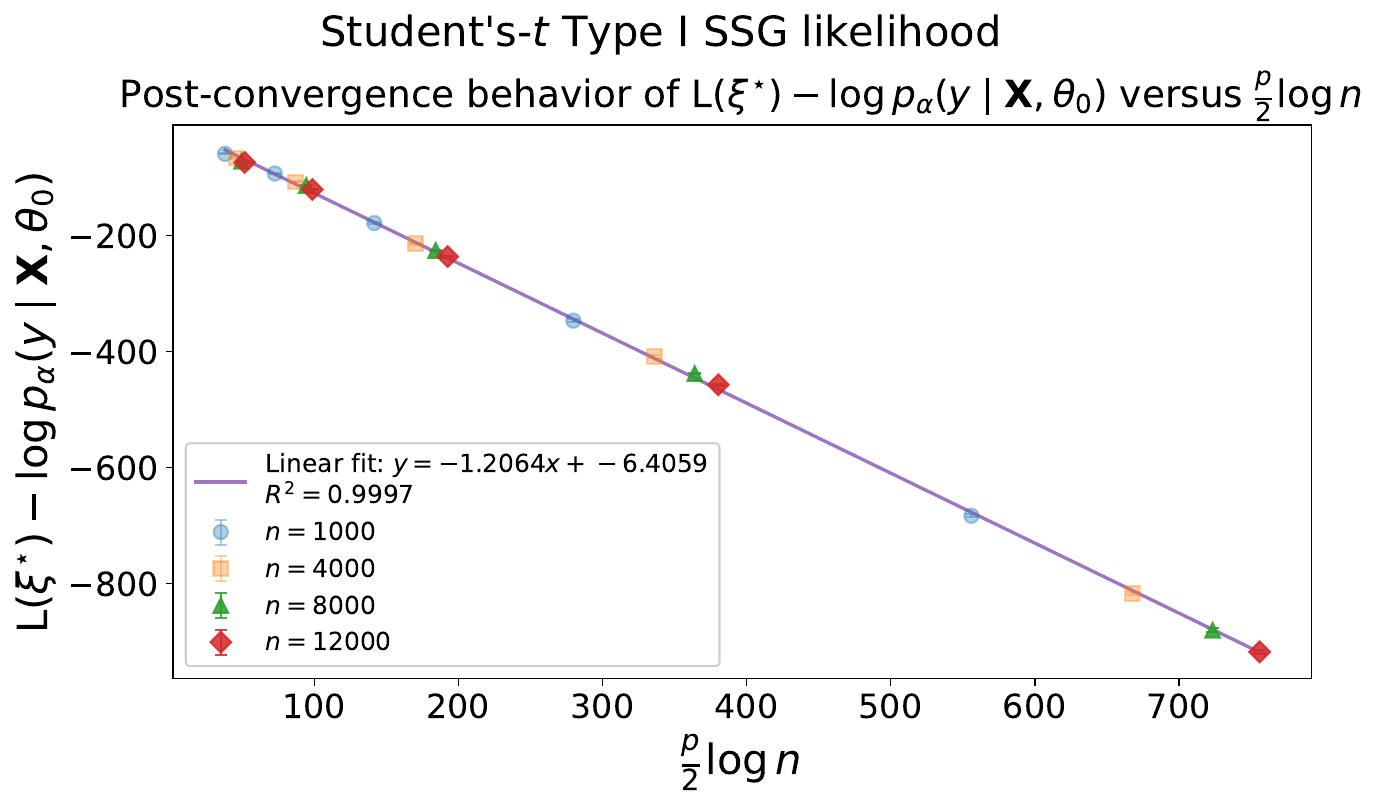}
    \caption{\footnotesize{Empirical complexity analysis of the post-convergence discrepancy $\mathsf{L}(\xi^{\star}) - \log p_{\alpha}(y\mid \mathbf{X}, \theta_0)$ with respect to $\tfrac p2\log n$ for the Student's-$t$ Type I $\ssg$ likelihood.}}
    \label{fig:student_t_elbo_minus_true_loglik_vs_half_p_log_n_linear_fit}
\end{figure}

\subsection{Empirical Comparison of True and \texorpdfstring{$\tssg$}{TAVIE-SSG} Variational Posterior Distributions}\label{subsec:empirical-true-variational-posterior}

The variational gap studied in Sections~\ref{subsec:variational-gap} and~\ref{subsec:empirical-elbo-variational-gap} above quantifies the approximation error induced by replacing the true $\alpha$-fractional posterior with its variational surrogate. While this gap is defined at the level of the marginal evidence and the corresponding optimization objective, its practical significance is ultimately posterior-geometric. This motivates a more direct empirical comparison between the true and variational posterior distributions.

For the Student's-$t$ Type I $\ssg$ likelihood with $\nu=5$, $p=9$, and $\alpha=1$, Figure~\ref{fig:contour_true_variational-1} displays the bivariate $(\beta_1, \beta_2)$ marginal contours of the variational posteriors of $\tssg$ and competitors, together with the true $\alpha$-fractional posterior contour landscape approximated using \texttt{PyMC} (NUTS)~\citep{pymc,pymc2023peerj}, for increasing $n\in \{1000, 2000\}$.  Also, Figure~\ref{fig:sliced_wasserstein_student} complements this visual comparison with boxplots of the sliced Wasserstein (SW) distance~\citep{bonneel2015sliced, kolouri2016sliced} (in $\log$-scale) between the true and variational posterior distributions across $10$ independent replications for $n\in\{200, 500, 1000, 2000\}$. 

\begin{figure}[!t]
    \centering

    \begin{subfigure}[b]{\textwidth}
        \centering
        \includegraphics[width=0.8\textwidth]{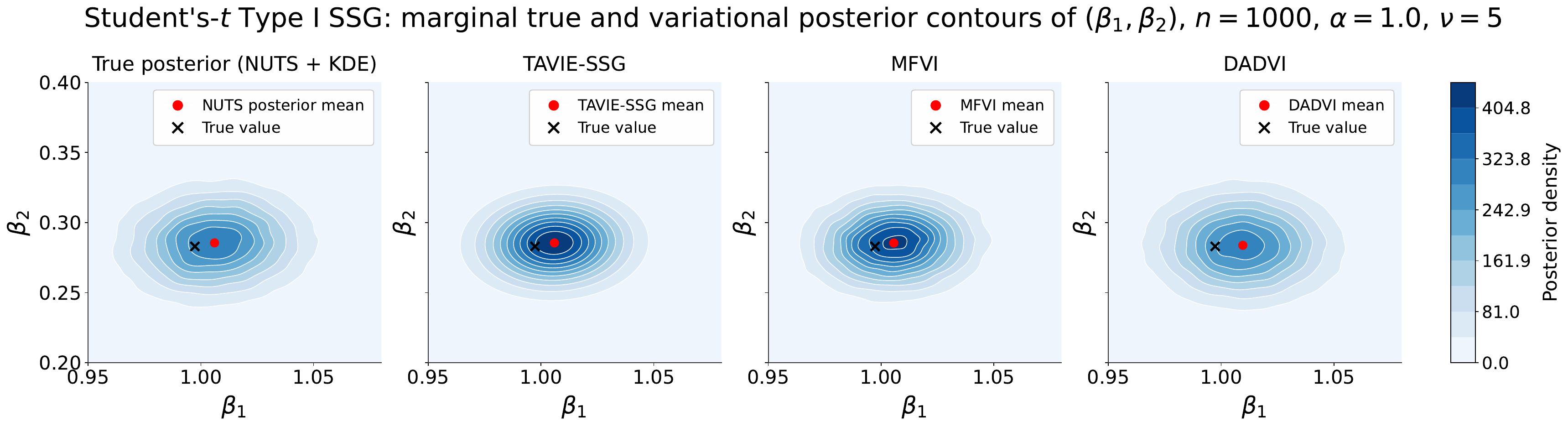}
        \label{fig:contour_1000_1}
    \end{subfigure}
    
    \vspace{0.5em}
    \begin{subfigure}[b]{\textwidth}
        \centering
        \includegraphics[width=0.8\textwidth]{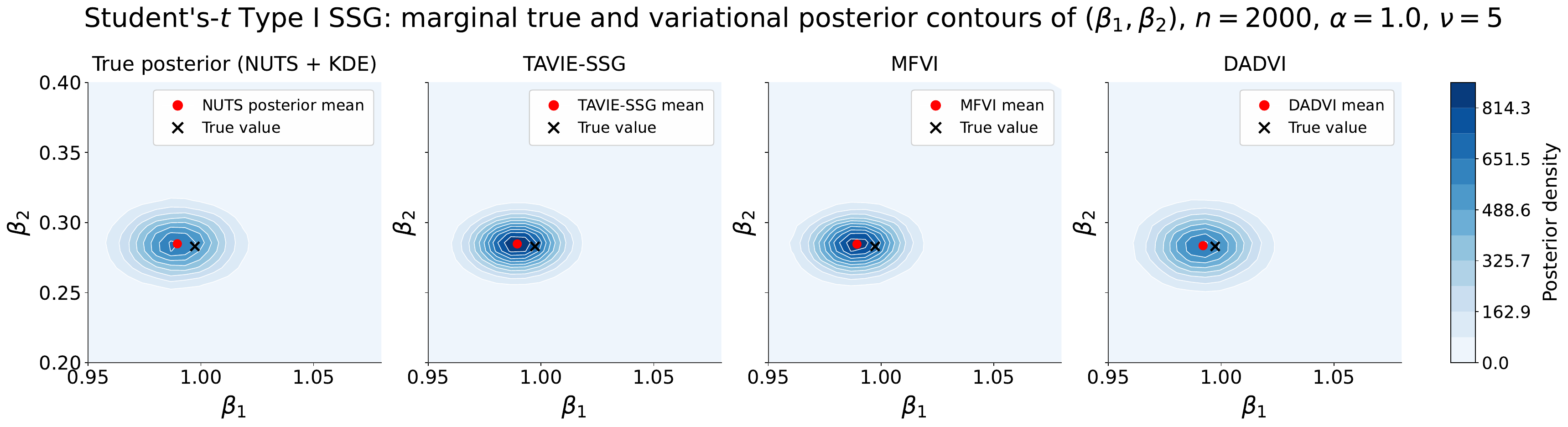}
        \label{fig:contour_2000_1}
    \end{subfigure}
    
    \caption{\footnotesize{Contour plots comparing the bivariate $(\beta_1, \beta_2)$ marginals of the variational posteriors of $\tssg$ and competitors with the corresponding fractional true posterior approximated using \texttt{PyMC} (NUTS) draws with a Gaussian kernel density estimate (KDE), under the Student's-$t$ $(\nu=5)$ Type I $\ssg$ model with $p=9$. Results are displayed for two sample sizes $n \in \{1000,2000\}$ with $\alpha=1$.}}
    \label{fig:contour_true_variational-1}
\end{figure}

\begin{figure}[H]
    \centering
    \includegraphics[width=0.6\linewidth]{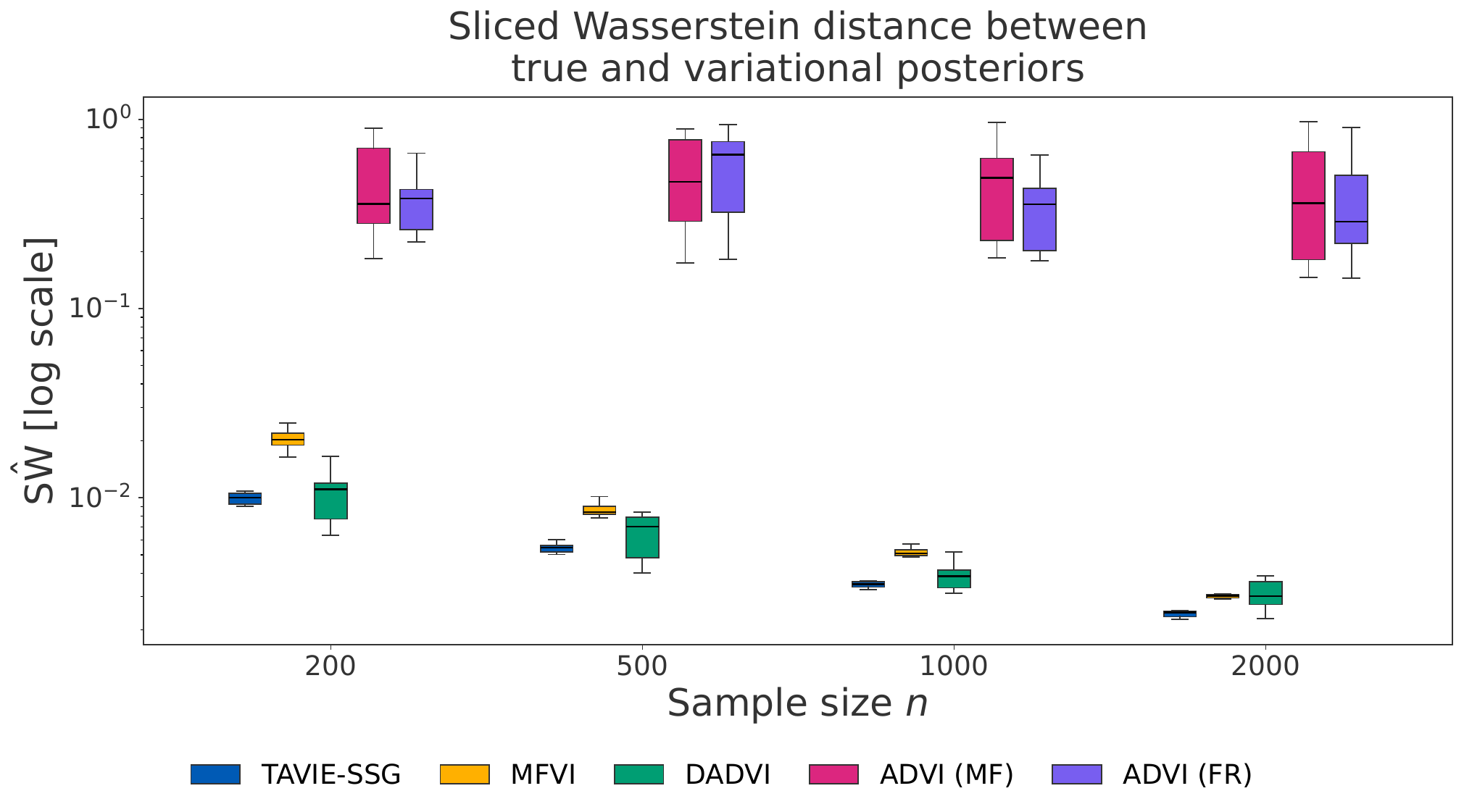}
    \caption{\footnotesize{Monte Carlo (MC) estimated sliced Wasserstein (SW) distance (in $\log$-scale) between the fractional true posterior (approximated using \texttt{PyMC} (NUTS)) and the variational posteriors of $\tssg$ and competitors under the Student's-$t$ ($\nu=5$) Type I $\ssg$ model, with $p=9$ and $\alpha=1$. For each sample size $n\in \{200, 500, 1000, 2000\}$, the boxplots summarize results over $10$ independent replications. The SW distance is computed in the joint space $(\beta,\log \tau^2)$ using random one-dimensional projections.}}
    \label{fig:sliced_wasserstein_student}
\end{figure}

Taken together, these figures show that DADVI, owing to its more flexible variational family, often yields the closest match to the true $\alpha$-fractional posterior distribution. At the same time, $\tssg$ and MFVI perform competitively both in terms of distance from the true posterior as well as the overall shape of the contours. A further clear trend is that, for $\tssg$, DADVI, and MFVI, the distance from the true posterior distribution decreases as the sample size $n$ increases, consistent with all three methods producing posteriors that become increasingly concentrated around the truth. By contrast, the variational posterior approximations obtained from the ADVI variants remain substantially farther from the true posterior across the increasing range of sample sizes $n$.

\newpage

{
\section{Time Complexity Analysis for \texorpdfstring{$\tssg$}{TAVIE-SSG}}\label{sec:time-complexity-analysis-TAVIE-SSG}

In this section, we quantify the per-iteration computational cost of $\tssg$ for Type I and Type II $\ssg$ likelihoods, as tabulated in Tables \ref{tab:type-1-time-complexity} and \ref{tab:type-2-time-complexity}, respectively. In both cases, the dominant operations arise from: (i) forming the weighted Gram matrix $\mathbf{X}^\top (\mathbf{X} \odot \mathcal{A}(\xi^{(t)}))$\footnote{Here, $\mathbf{X} \odot \mathcal{A}(\xi)$ denotes row-wise scaling. Mathematically, it is $\mathcal{A}(\xi)\mathbf{X}$, but computationally it means scaling each row of $\mathbf{X}$ by the corresponding diagonal entry of $\mathcal{A}(\xi)$. This trick avoids forming explicit matrix products and ensures that the Gram matrix computation is in $\mathcal{O}(np^2)$ time rather than $\mathcal{O}(n^2p)$.}, which costs $\mathcal{O}(np^2)$ for dense $\mathbf{X}$, and (ii) solving a $p \times p$ linear system to obtain $\Sigma_{\alpha}(\xi^{(t)})$ from its inverse, which costs $\mathcal{O}(p^3)$ via Cholesky factorization.  All remaining steps, including the coordinate-wise updates of $\xi^{(t)}$ and evaluation of $\mathcal{A}(\xi^{(t)})$, scale linearly in $n$ or quadratically in $p$ and are therefore of lower order. Consequently, each iteration has overall time complexity $\mathcal{O}(np^2 + p^3)$.

To validate this scaling empirically, we conduct a synthetic experiment using the Type I Student's-$t$ $\ssg$ model with degrees of freedom $\nu=5$, true precision $\tau_0^2=1$, $\alpha=1$, true $\beta_0 \sim \mathcal{N}_{p}(0, I_{p})$, and $\mathbf{X} = (\mathbf{x}_1, \ldots, \mathbf{x}_n)^{\top} \in \mathbb{R}^{n\times p}$, where $x_{i1}=1$ (intercept) and $x_{ij}$'s are generated independently and identically from the standard Gaussian distribution for $j=2, \ldots, p$. For each configuration, we average per-iteration runtime over $10$ independent repetitions. We study two regimes: (i) scaling in $n$ with $p$ fixed at $31$ using $n \in \{100, 200, 400, 800, 1600\}$, and (ii) scaling in $p$ with $n$ fixed at $1000$ using $p \in \{6, 11, 16, 21, 51, 81, 101\}$.

\begin{figure}[!htp]
    \centering
    \includegraphics[width=0.95\linewidth]{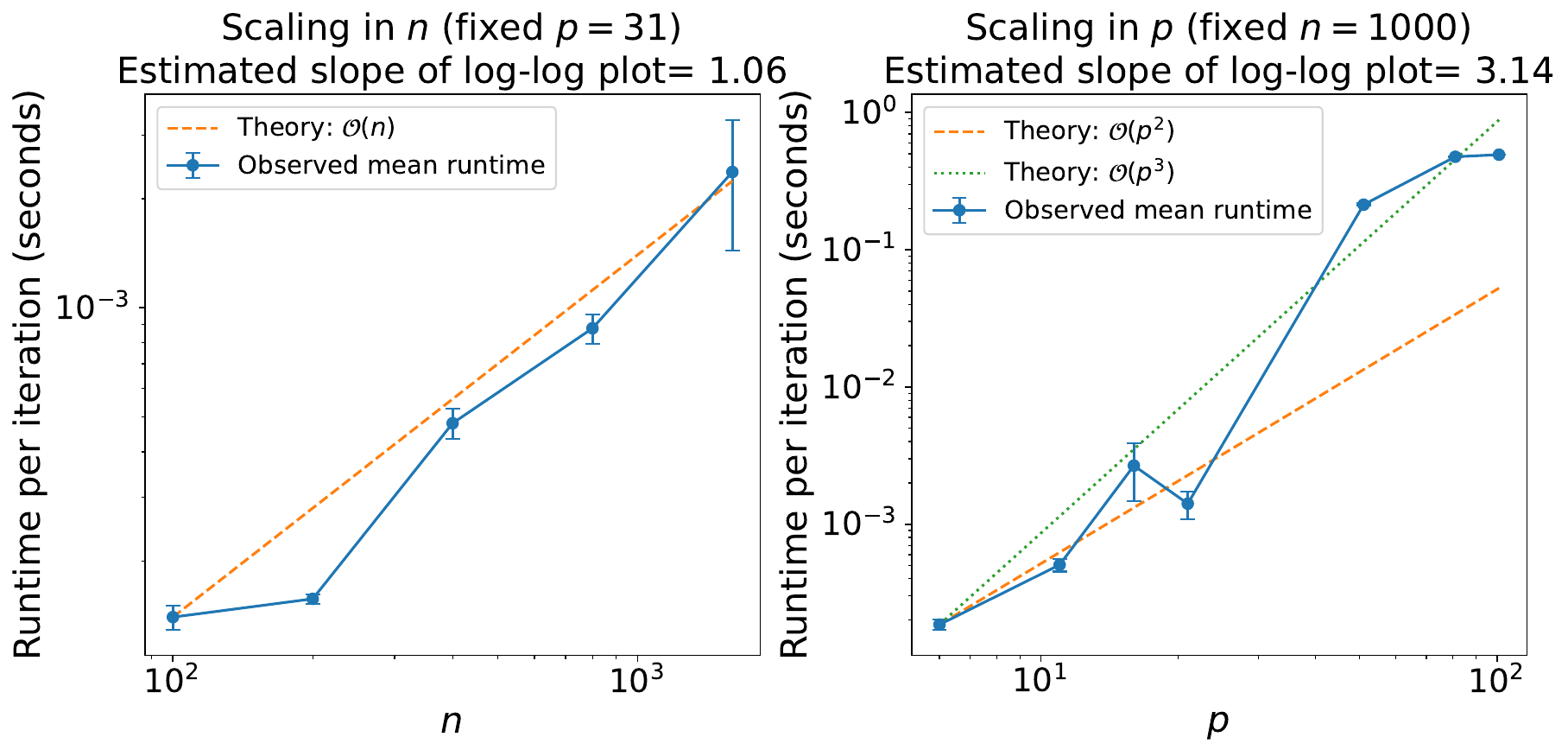}
    \caption{\footnotesize{Scaling behavior of $\tssg$ with respect to $n$ (left) and $p$ (right).}}
    \label{fig:tavie_scaling}
\end{figure}

The left panel of Figure~\ref{fig:tavie_scaling} shows that, for fixed $p$, the runtime grows approximately linearly with $n$, with an estimated $\log$-$\log$ slope of $1.06$, closely matching the theoretical $\mathcal{O}(n)$ behavior implied by $\mathcal{O}(np^2 + p^3)$ when $p$ is fixed. This reflects the fact that the $\xi$-updates decompose into $n$ independent univariate computations and are thus embarrassingly parallelizable. The right panel of Figure~\ref{fig:tavie_scaling} examines scaling in $p$ for fixed $n$. While the overall estimated slope of the $\log$-$\log$ plot is $3.14$, the empirical curve seems to interpolate between quadratic and cubic growth regimes, closely tracking these theoretical growth curves. This behavior is consistent with the theoretical $\mathcal{O}(np^2 + p^3)$ complexity.

Overall, these results confirm that $\tssg$ exhibits near-linear scaling in $n$ for fixed $p$, and polynomial scaling in $p$ consistent with the underlying linear algebra costs. This makes $\tssg$ particularly well-suited for high-sample, low-to-moderate dimensional settings, which is the primary target regime of the proposed method.

\begin{table}[H]
\centering
\caption{\footnotesize{Per-iteration time complexity of $\tssg$ for Type I $\ssg$ likelihoods.}}
\label{tab:type-1-time-complexity}
\footnotesize
\renewcommand{\arraystretch}{1.2}
\begin{tabular}{lll}
\toprule
\toprule
\textbf{Algorithm Step} & \textbf{Sub-steps} & \textbf{Time complexity} \\
\midrule

\textbf{Update $\xi,\mathcal{A}(\xi)$ } 
& Compute $\mathbf{X} \Sigma_{\alpha}(\xi^{(t-1)})$ & $\mathcal{O}(np^2)$ \\
\textbf{(parallelizable)}& Compute $\mathbf{X} \mu_{\alpha}(\xi^{(t-1)})$ & $\mathcal{O}(np)$ \\
& Compute $\operatorname{diag}(\mathbf{X} \Sigma_{\alpha}(\xi^{(t-1)}) \mathbf{X}^\top)$ & $\mathcal{O}(np)$ \\
& Update $\xi_i^{(t)} = \sqrt{\kappa_i(\xi^{(t-1)})}, i\in [n]$ & $\mathcal{O}(n)$ \\
& Evaluate $\mathcal{A}(\xi^{(t)})$ & $\mathcal{O}(n)$ \\

\midrule
\textbf{Update $\Sigma_{\alpha}(\xi)$} 
& Form $\mathbf{X} \odot \mathcal{A}(\xi^{(t)})$ & $\mathcal{O}(np)$ \\
& Compute $\mathbf{X}^\top (\mathbf{X} \odot \mathcal{A}(\xi^{(t)}))$ & $\mathcal{O}(np^2)$ \\
& Form $\Sigma_{\alpha}(\xi^{(t)})^{-1}$ & $\mathcal{O}(p^2)$ \\
& Compute $\Sigma_{\alpha}(\xi^{(t)})$ via Cholesky solve & $\mathcal{O}(p^3)$ \\

\midrule
\textbf{Update $\mu_{\alpha}(\xi)$} 
& Compute $(\mathbf{X}^\top \odot \mathcal{A}(\xi^{(t)})) y$ & $\mathcal{O}(np)$ \\
& Compute $\Sigma^{-1}\mu - 2\alpha(\mathbf{X}^\top \odot \mathcal{A}(\xi^{(t)})) y$ & $\mathcal{O}(p)$ \\
& Multiply by $\Sigma_{\alpha}(\xi^{(t)})$ & $\mathcal{O}(p^2)$ \\

\midrule
\textbf{Update $b_{\alpha}(\xi)$} 
& Compute $y^{\top}\mathcal{A}(\xi^{(t)})y$ & $\mathcal{O}(n)$ \\
& Compute $\mu_{\alpha}(\xi^{(t)})^\top \Sigma_{\alpha}(\xi^{(t)})^{-1} \mu_{\alpha}(\xi^{(t)})$ & $\mathcal{O}(p^2)$ \\
& Update $b_{\alpha}(\xi^{(t)})$ & $\mathcal{O}(1)$ \\

\midrule
\textbf{ELBO} $\mathsf{L}(\xi)$
& Compute $\log |\Sigma_{\alpha}(\xi^{(t)})|$ & $\mathcal{O}(p^3)$ \\
& Compute $\sum_{i \in [n]} \gamma(\xi_i^{(t)})$ & $\mathcal{O}(n)$ \\

\midrule
\textbf{Convergence} 
& Compute $\|\xi^{(t)} - \xi^{(t-1)}\|_2$ & $\mathcal{O}(n)$ \\

\midrule
\textbf{Total} & & $\mathbf{\mathcal{O}(np^2 + p^3)}$ \\
\bottomrule
\bottomrule
\end{tabular}
\end{table}

\begin{table}[H]
\centering
\caption{\footnotesize{Per-iteration time complexity of $\tssg$ for Type II $\ssg$ likelihoods.}}
\label{tab:type-2-time-complexity}
\footnotesize
\renewcommand{\arraystretch}{1.2}
\begin{tabular}{lll}
\toprule
\toprule
\textbf{Algorithm Step} & \textbf{Sub-steps} & \textbf{Time complexity} \\
\midrule

\textbf{Update $\xi,\mathcal{A}(\xi)$} 
& Compute $\mathbf{X} \mu_{\alpha}(\xi^{(t-1)})$ & $\mathcal{O}(np)$ \\
\textbf{(parallelizable)}& Compute $\mathbf{X} \Sigma_{\alpha}(\xi^{(t-1)})$ & $\mathcal{O}(np^2)$ \\
& Compute $\operatorname{diag}(\mathbf{X} \Sigma_{\alpha}(\xi^{(t-1)}) \mathbf{X}^\top)$ & $\mathcal{O}(np)$ \\
& Update $\xi_i^{(t)} = \sqrt{\kappa_i(\xi^{(t-1)})}, i\in [n]$ & $\mathcal{O}(n)$ \\
& Evaluate $\mathcal{A}(\xi^{(t)})\odot \mathbf{b} = -\mathbf{b} \odot \tanh(\xi^{(t)}/2)/(4\xi^{(t)})$ & $\mathcal{O}(n)$ \\

\midrule
\textbf{Update $\Sigma_{\alpha}(\xi)$} 
& Form $\mathbf{X} \odot (\mathcal{A}(\xi^{(t)}) \odot \mathbf{b})$ & $\mathcal{O}(np)$ \\
& Compute $\mathbf{X}^\top (\mathbf{X} \odot (\mathcal{A}(\xi^{(t)}) \odot \mathbf{b}))$ & $\mathcal{O}(np^2)$ \\
& Form $\Sigma_{\alpha}(\xi^{(t)})^{-1}$ & $\mathcal{O}(p^2)$ \\
& Compute $\Sigma_{\alpha}(\xi^{(t)})$ via Cholesky solve & $\mathcal{O}(p^3)$ \\

\midrule
\textbf{Update $\mu_{\alpha}(\xi)$} 
& Multiply $\Sigma_{\alpha}(\xi^{(t)})\left[\Sigma^{-1}\mu
+
\alpha\mathbf{X}^\top (\mathbf{a} - \mathbf{b}/2)\right]$ & $\mathcal{O}(p^2)$ \\

\midrule
\textbf{ELBO} $\mathsf{L}(\xi)$
& Compute $\mu_{\alpha}(\xi^{(t)})^\top \Sigma_{\alpha}(\xi^{(t)})^{-1} \mu_{\alpha}(\xi^{(t)})$ & $\mathcal{O}(p^2)$ \\
& Compute $\log |\Sigma_{\alpha}(\xi^{(t)})|$ & $\mathcal{O}(p^3)$ \\
& Compute $\sum_{i \in [n]} b_i\, \gamma(\xi_i^{(t)})$ & $\mathcal{O}(n)$ \\

\midrule
\textbf{Convergence} 
& Compute $\|\xi^{(t)} - \xi^{(t-1)}\|_2$ & $\mathcal{O}(n)$ \\

\midrule
\textbf{Total} & & $\mathbf{\mathcal{O}(np^2 + p^3)}$ \\
\bottomrule
\bottomrule
\end{tabular}
\end{table}
}

\newpage

{
\section{Sensitivity Analysis for \texorpdfstring{$\alpha$}{alpha}}\label{sec:sensitivity-alpha}
The likelihood tempering parameter $\alpha \in (0,1]$ is a key component of the $\tssg$ framework.
At an intuitive level, $\alpha$ controls the strength with which the data update the prior over the model parameters, and hence the degree of trust placed in the assumed likelihood. The choice $\alpha=1$ recovers the standard Bayesian update, whereas $\alpha\in (0, 1)$ yields a tempered update that downweights the likelihood contribution and therefore induces a more conservative learning rule. This interpretation is standard in the literature on fractional posteriors~\citep{Bayesian-fractional-posterior}, where likelihood tempering is used to moderate posterior learning and improve robustness under model misspecification; it is also closely related to generalized Bayes~\citep{BissiriHolmesWalker2016}, where the likelihood or loss contribution is explicitly scaled to calibrate the update, and to the SafeBayes perspective~\citep{GrunwaldVanOmmen2017}, where such scaling serves as a safeguard against overly aggressive learning under imperfect models. In this section, we conduct several simulation experiments to better understand the practical impact of $\alpha$ and to empirically study the sensitivity of our $\tssg$ algorithm to $\alpha$ along three complementary directions: the landscape of the variational posterior distribution, the accuracy of the resulting variational point estimates, and the frequentist coverage of credible intervals.

\textbf{Effect of $\alpha$ over the shape of variational posterior}. A natural first step is to understand how $\alpha$ influences the geometry of the variational posterior distribution. To this end, we consider fitting the $\tssg$ model with Laplace Type I $\ssg$ likelihood as in Section~\ref{app:additional-laplace-results}, with true parameters $\tau_0^{2}=3, \beta_0 \sim \mathcal{N}_{p}(0, I_{p})$, and $\mathbf{X} = (\mathbf{x}_1, \ldots, \mathbf{x}_n)^{\top} \in \mathbb{R}^{n\times p}$, where $x_{i1}=1$ (intercept) and $x_{ij}$'s are generated independently and identically from the standard Gaussian distribution for $j=2, \ldots, p$; here we fix $p=6$. For $\alpha\in \{0.20, 0.40, 0.60, 0.95, 1.00\}$ and $n\in \{1000, 2000, 4000\}$, Figure~\ref{fig:contour_beta1_beta2} displays the contours of the bivariate marginal variational posterior distribution over two coordinates of $\beta$. For each fixed $n$, the contours exhibit a systematic increase in concentration around the variational posterior mean as $\alpha$ increases. This behavior can be directly attributed to the form of the Type I $\ssg$ variational posterior precision matrix, $\Sigma_{\alpha}^{-1}(\xi^{\star}) = \Sigma^{-1} - 2\alpha \mathbf{X}^{\top}\mathcal{A}(\xi^{\star})\mathbf{X}$, where $\mathcal{A}(\xi^{\star})$ is a diagonal matrix with negative entries. Hence, increasing $\alpha$ strengthens the likelihood contribution to the posterior precision, thereby yielding tighter variational marginals and visibly sharper contour plots. Furthermore, as expected with increasing $n$, we observe both the concentration of the variational posterior around its mean as well as the variational posterior mean moves closer to the true value of the parameter $(\beta_1, \beta_2)$.
\begin{figure}[!htp]
    \centering
    \includegraphics[width=0.9\linewidth]{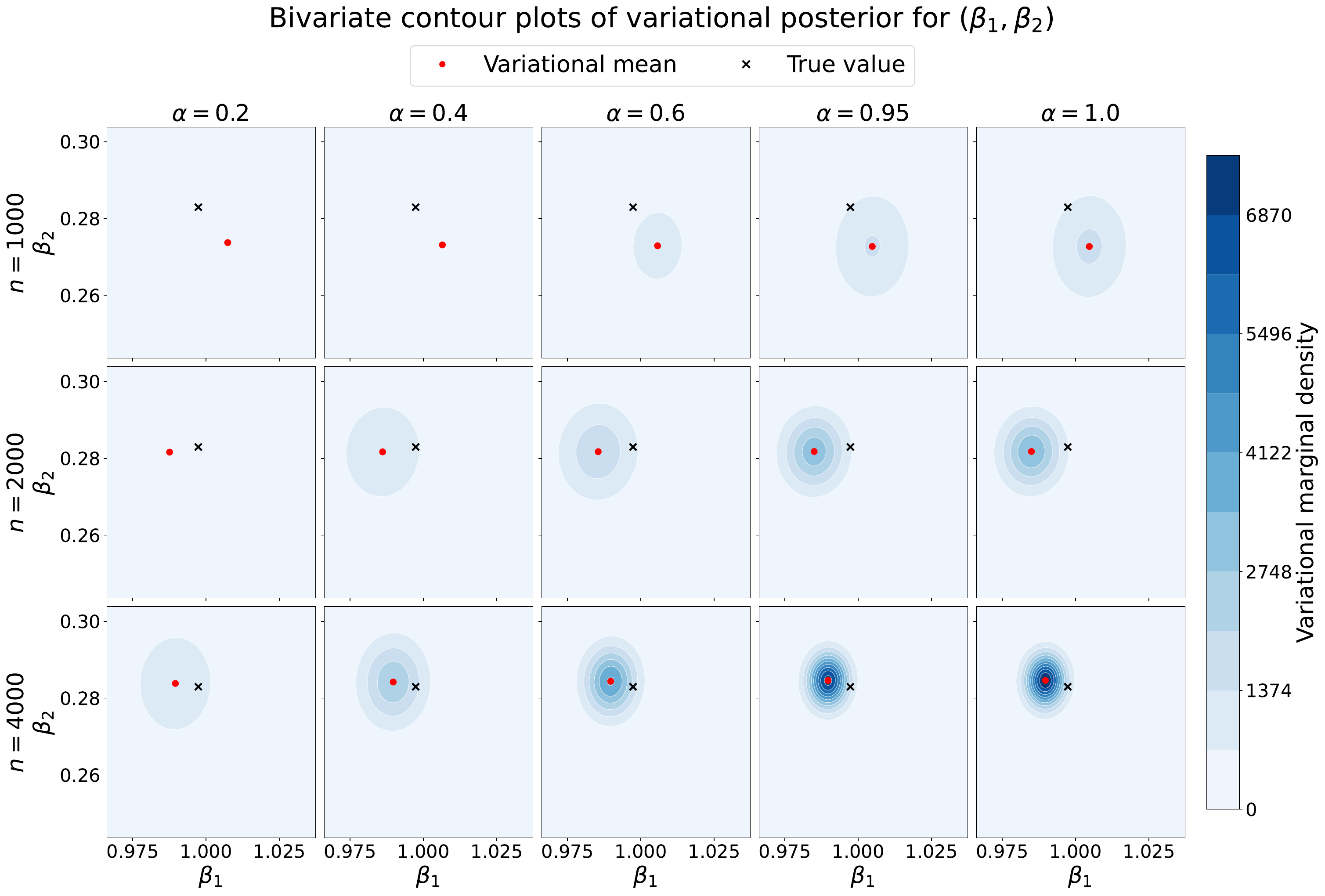}
    \caption{\footnotesize{Concentration of the marginal variational posterior distribution of $(\beta_1, \beta_2)$ around its mean $((\mu_{\alpha}(\xi^{\star}))_1, (\mu_{\alpha}(\xi^{\star}))_2)$ across different $(\alpha, n)$ for Laplace Type I $\ssg$ likelihood.}}
    \label{fig:contour_beta1_beta2}
\end{figure}

\textbf{Effect of $\alpha$ over point-estimation accuracy}. We next examine how this change in variational posterior concentration is reflected in point-estimation accuracy. Under the same simulation regime as mentioned above, for fixed $p=9$, we fit the $\tssg$ model with the Laplace Type I $\ssg$ likelihood over a range of tempering levels $\alpha \in \{0.20, 0.40, 0.60, 0.80, 0.95, 0.99\}$ and sample sizes $n\in \{200, 500, 1000, 2000\}$. We summarize the performance of $\tssg$ through the mean-squared error (MSE) of the variational posterior estimates of $\beta$ and $\tau^2$, $p^{-1}\lVert \beta_0 - \widehat{\beta}\rVert_{2}^{2}$ and $(\tau_0^{2} - \widehat{\tau}^{2})^{2}$, over $100$ data regenerations; see Figure~\ref{fig:boxplot_mse_alpha_laplace}. The effect of $\alpha$ is systematic: for each fixed $n$, the MSE decreases as $\alpha$ increases, and for each fixed $\alpha$, the MSE decreases as $n$ increases. Thus, larger values of $\alpha$ yield more accurate variational point estimates, consistent with the fact that they place greater weight on the likelihood and therefore anchor the variational approximation more strongly to the data-generating mechanism.
\begin{figure}[t!]
    \centering
    \begin{subfigure}[t]{0.6\textwidth}
        \centering
        \includegraphics[width=\textwidth]{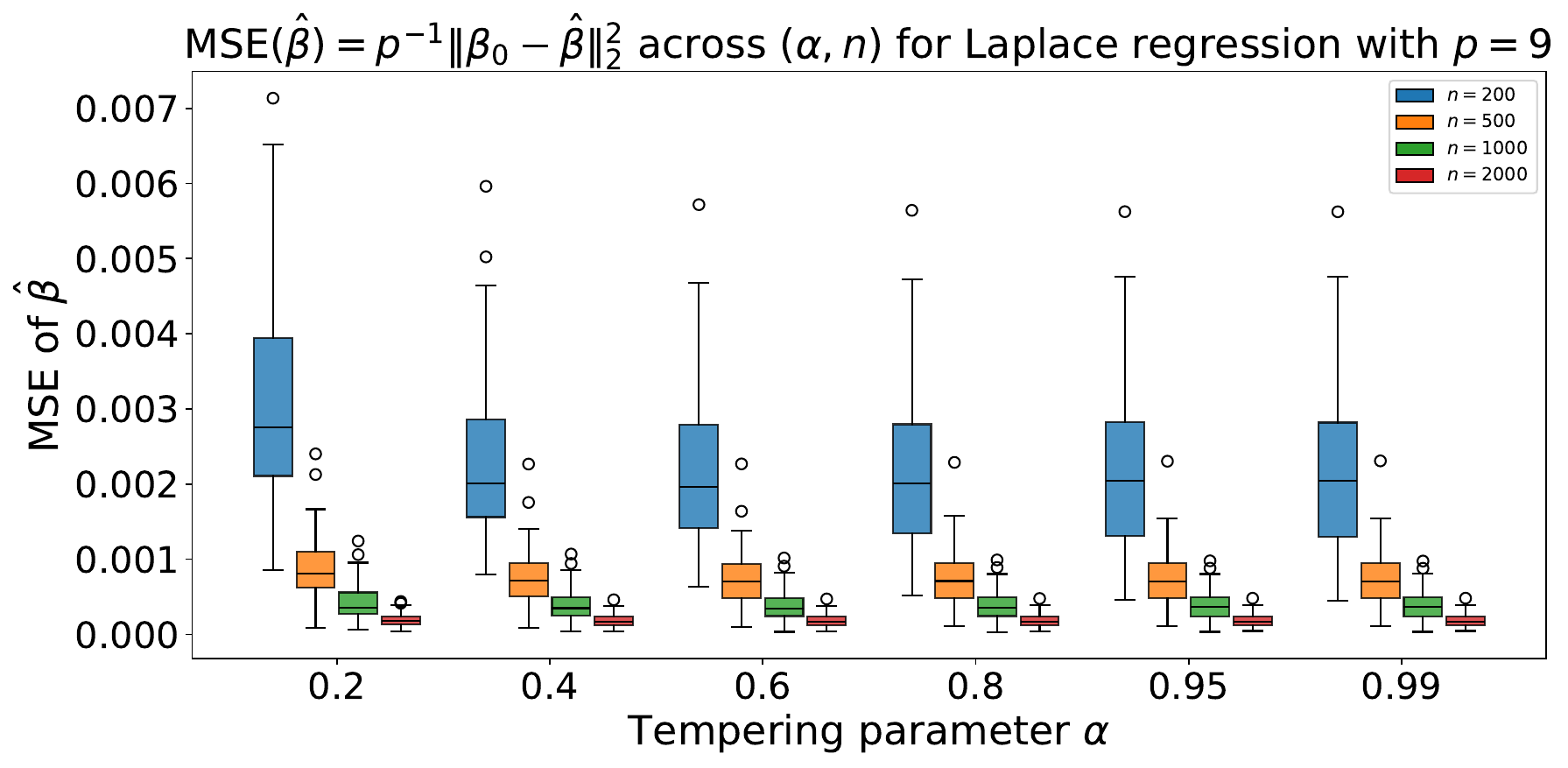}
        \label{fig:boxplot_mse_beta_alpha_laplace}
    \end{subfigure}
    \hfill
    \begin{subfigure}[t]{0.6\textwidth}
        \centering
        \includegraphics[width=\textwidth]{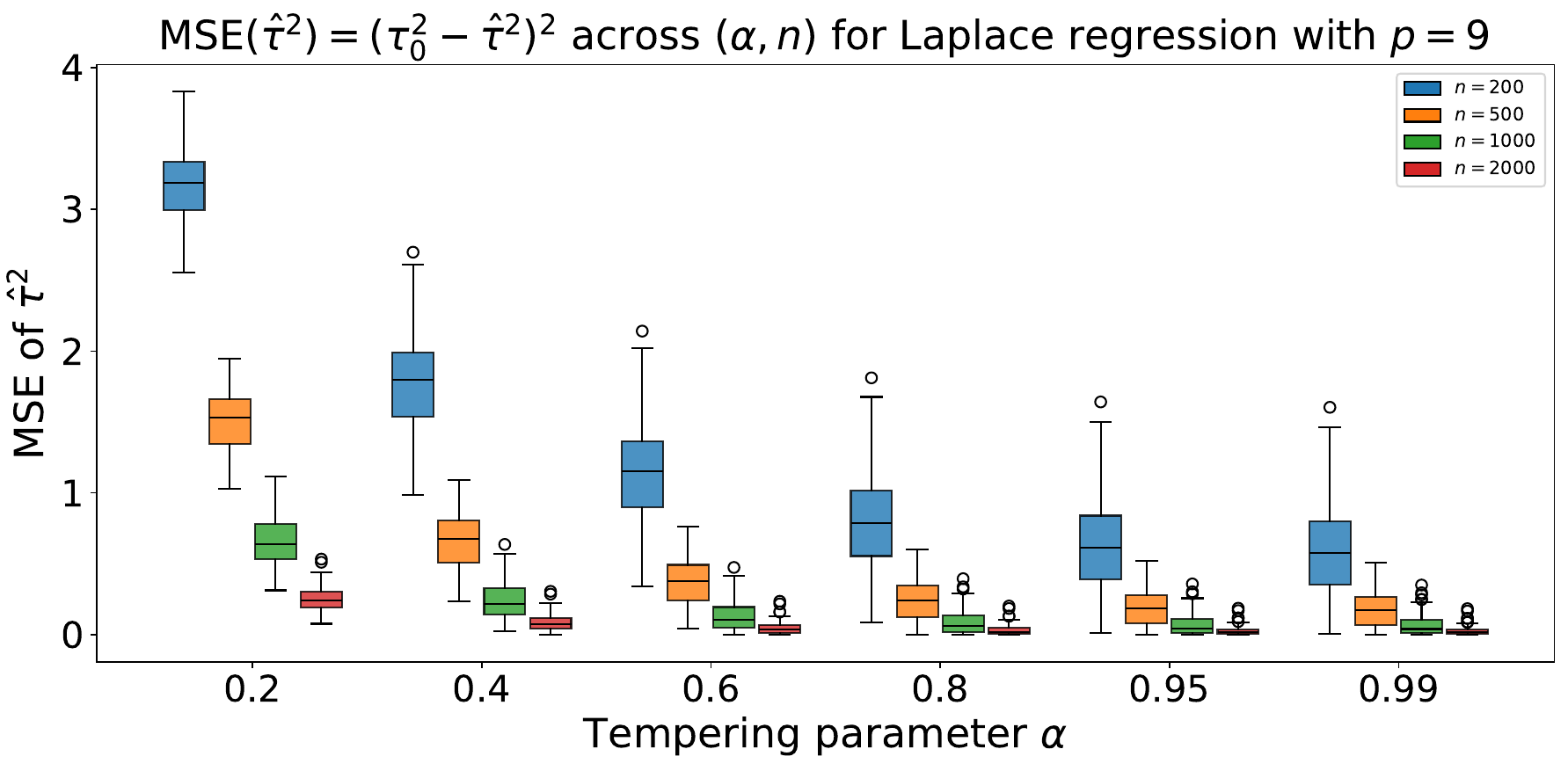}
        \label{fig:boxplot_mse_tau2_alpha_laplace}
    \end{subfigure}
    \caption{\footnotesize{Accuracy of variational posterior estimates of $\beta$ and $\tau^2$ across different $(\alpha, n)$ for Laplace Type I $\ssg$ likelihood.}}
    \label{fig:boxplot_mse_alpha_laplace}
\end{figure}

Collectively, the contour and MSE analyses show that $\alpha$ imparts a dampening effect on the rate of variational posterior concentration around the mean. Smaller values of $\alpha$ yield a more diffuse variational posterior, whereas larger values lead to sharper concentration with more accurate variational point estimates.

Additionally, we also investigate the $\alpha$-R\'{e}nyi divergence $D_{\alpha}(\widehat{\theta}, \theta_0)$ as defined in Section~\ref{main-subsec:var-risk} of the main manuscript. Figure~\ref{fig:boxplot_renyi_alpha_laplace} shows that the behavior of $D_{\alpha}(\widehat{\theta}, \theta_0)$ is more nuanced. While the divergence increases mildly with $\alpha$ for fixed $n$, it decreases markedly as $n$ increases. This pattern reflects the fact that $D_{\alpha}(\widehat{\theta}, \theta_0)$ captures not only the proximity of the variational point estimate to the true parameter, but also the distributional effect of likelihood tempering induced by $\alpha$.
\begin{figure}[H]
    \centering
    \includegraphics[width=0.7\linewidth]{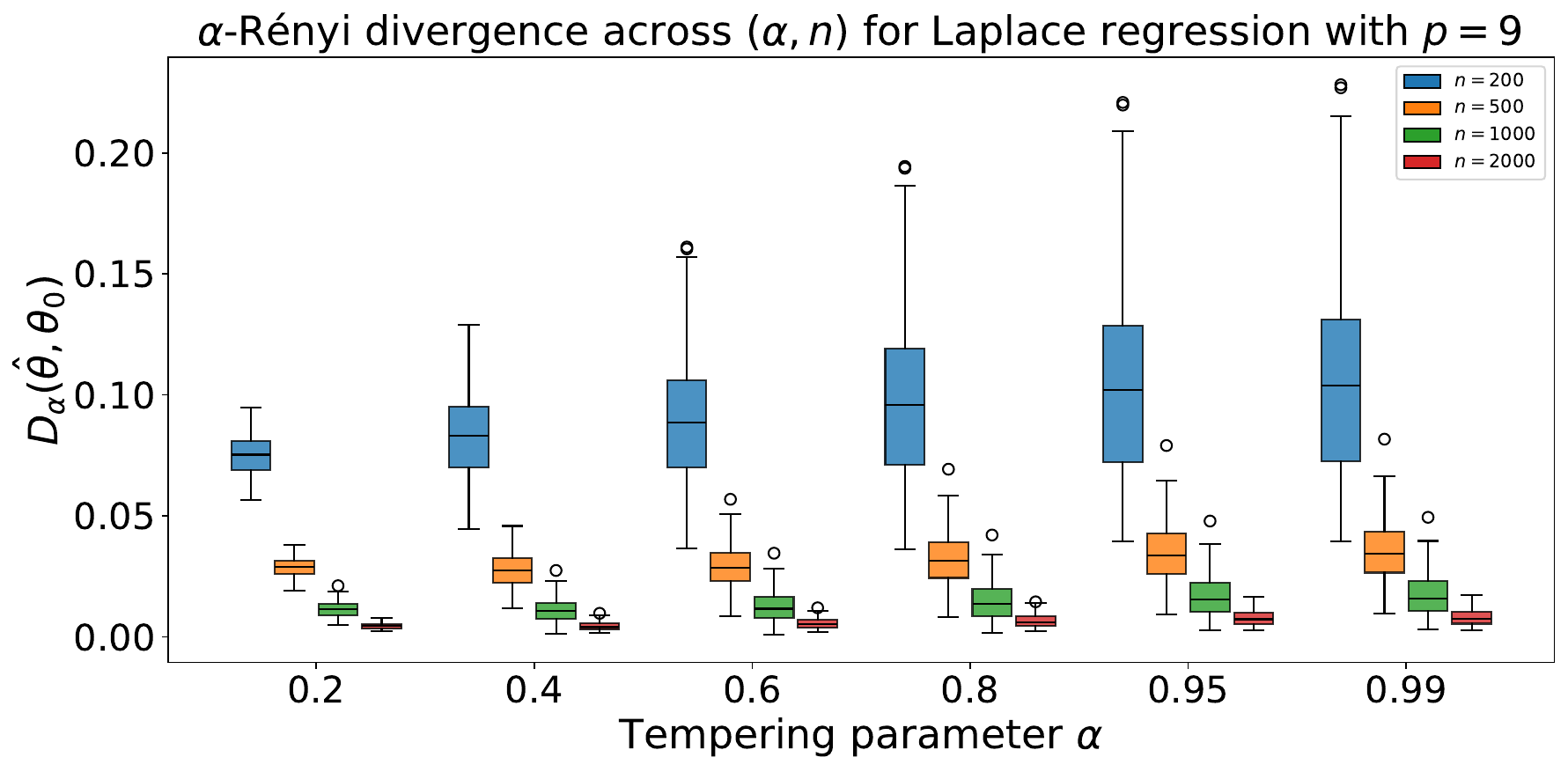}
    \caption{\footnotesize{$D_{\alpha}(\widehat{\theta}, \theta_0)$ across different $(\alpha, n)$ and $100$ data regenerations for Laplace Type I $\ssg$ likelihood.}}
    \label{fig:boxplot_renyi_alpha_laplace}
\end{figure}

\textbf{Effect of $\alpha$ over frequentist coverage}. The preceding experiments show that $\alpha$ directly regulates the spread of the variational posterior distribution. This, in turn, motivates an examination of uncertainty quantification. A growing literature has emphasized that variational methods can achieve strong point-estimation performance while still misrepresenting posterior dispersion, leading to poor calibration of posterior uncertainty, particularly under restrictive approximation families such as mean-field variational Bayes. Existing efforts to address this issue include post hoc covariance-correction methods, such as linear-response variational Bayes~\citep{GiordanoBroderickJordan2018}, resampling-based approaches, such as the variational weighted likelihood bootstrap of~\cite{han2019statisticalinferencemeanfieldvariational}, and more recent work formalizing trade-offs inherent in variational uncertainty quantification~\citep{MargossianPillaudVivienSaul2025}. Against this backdrop, we next investigate the effect of $\alpha$ on the frequentist coverage of credible intervals produced by $\tssg$.

\begin{figure}[H]
    \centering
    \begin{subfigure}[t]{0.6\textwidth}
        \centering
        \includegraphics[width=\textwidth]{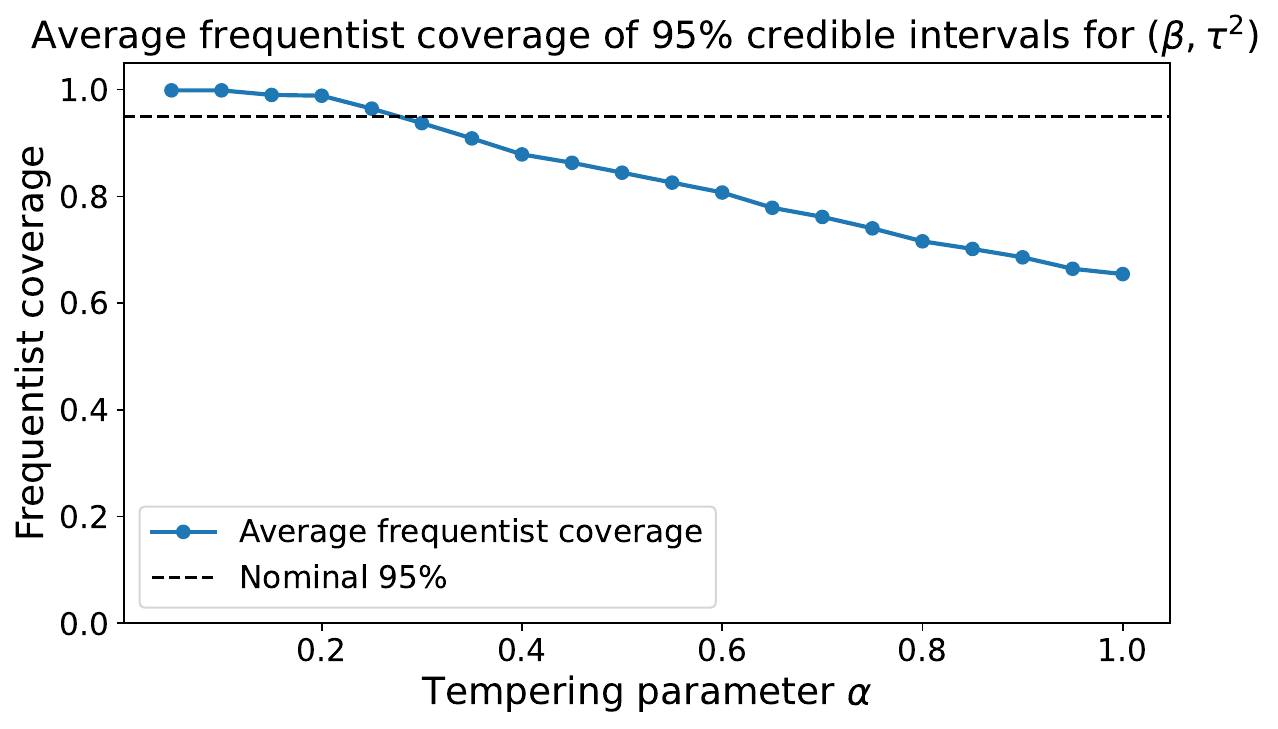}
        \label{fig:coverage_overall_vs_alpha}
    \end{subfigure}
    \hfill
    \begin{subfigure}[t]{0.6\textwidth}
        \centering
        \includegraphics[width=\textwidth]{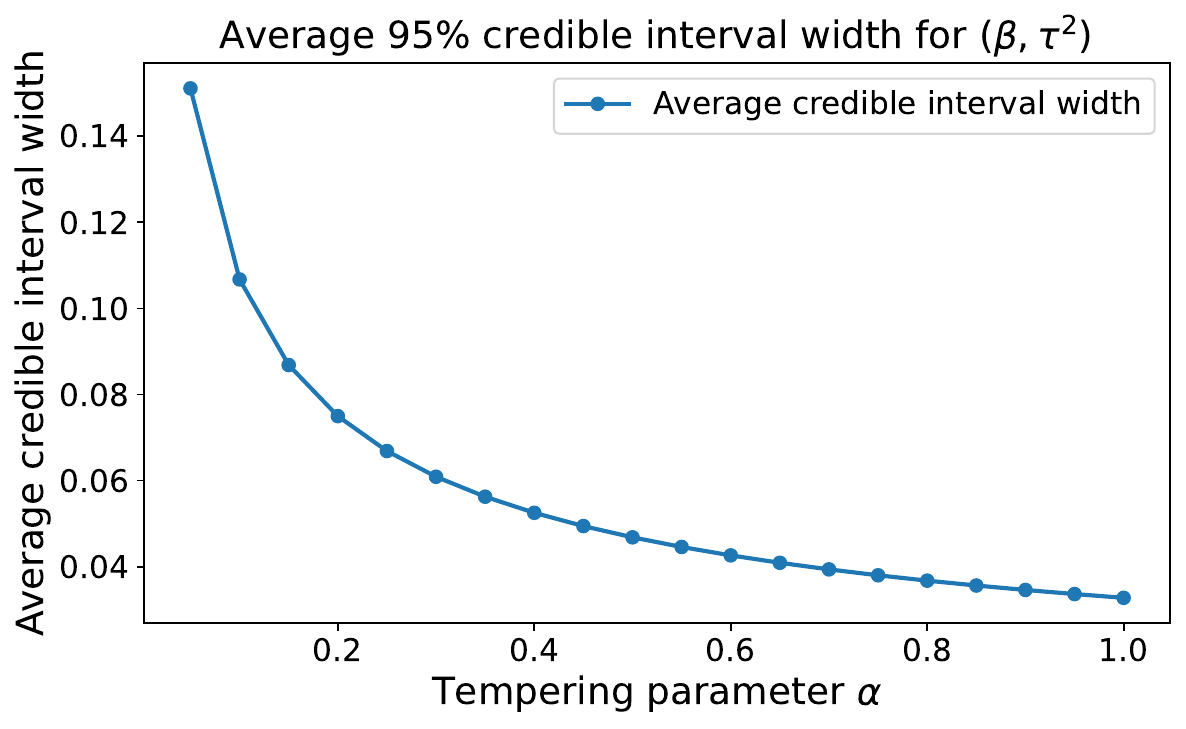}
        \label{fig:width_overall_vs_alpha}
    \end{subfigure}
    \caption{\footnotesize{Average frequentist coverage and width of $95\%$ marginal credible intervals of $(\beta, \tau^2)$ for Laplace Type I $\ssg$ likelihood.}}
    \label{fig:frequentist_coverage}
\end{figure}

Specifically, we evaluate the frequentist coverage, over $100$ independently regenerated datasets, of nominal $95\%$ marginal credible intervals for each of the $p+1$ model parameters in $(\beta^{\top}, \tau^{2})^{\top} \in \mathbb{R}^p \times \mathbb{R}^{+}$. Following a similar experimental setup for the Laplace Type I $\ssg$ likelihood as above, at fixed $n=10000$ and $p=6$, Figure~\ref{fig:frequentist_coverage} shows that the average coverage (averaged over $p+1$ model parameters) decreases as $\alpha$ increases from $0.05$ to $1.00$, while the corresponding average interval width shrinks in parallel. The same pattern persists in the varying-$n$ experiment (with $p=31$) as depicted by Figure~\ref{fig:coverage_boxplots_by_n_alpha}: for each fixed $\alpha$, the average coverage decreases as $n$ increases, while for each fixed $n$, larger $\alpha$ produces systematically lower coverage.

\begin{figure}[!htp]
    \centering
    \includegraphics[width=\linewidth]{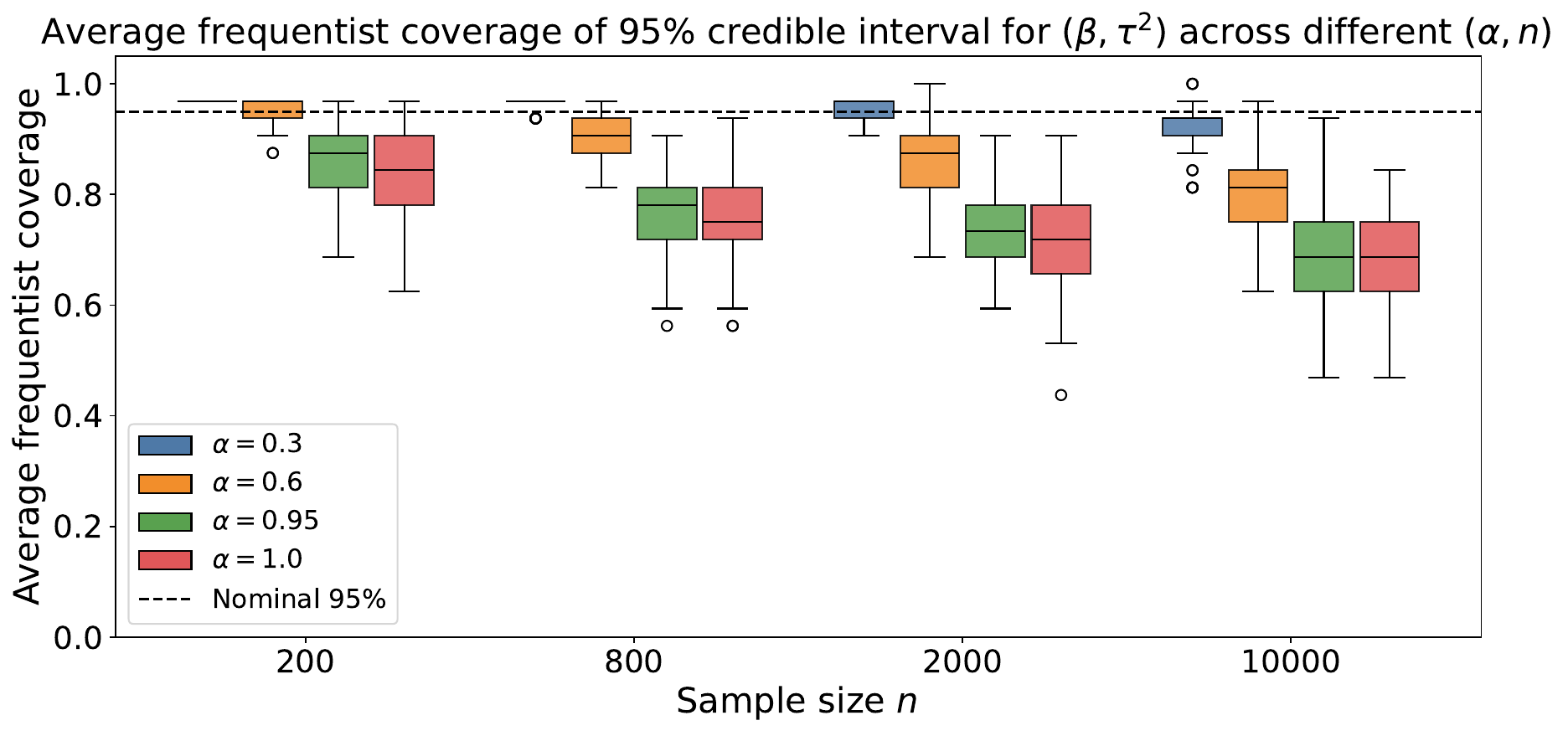}
    \caption{\footnotesize{Behavior of average frequentist coverage of $95\%$ marginal credible intervals of $(\beta, \tau^2)$ across different $(\alpha, n)$ for Laplace Type I $\ssg$ likelihood.}}
    \label{fig:coverage_boxplots_by_n_alpha}
\end{figure}

This finding is also aligned with the general behavior of variational approximations. Since larger $\alpha$ yields a more concentrated working posterior, credible intervals become tighter; in finite samples, this leads to undercoverage relative to the nominal level. The effect becomes stronger with increasing sample size because the posterior concentrates further. Therefore, in the present experiment, there is a clear accuracy-coverage trade-off induced by $\alpha$: values of $\alpha$ closer to $1$ improve point estimation but degrade frequentist coverage. It is therefore natural to ask whether the additional degree of control provided by the likelihood tempering parameter $\alpha$ can better account for uncertainty in the variational posterior.
}

\newpage

{

\section{Calibration of \texorpdfstring{$\alpha$}{alpha}}\label{sec:calibration-alpha}
The preceding sensitivity analysis shows that $\alpha$ serves two distinct inferential roles: it controls the concentration of the variational posterior distribution and, at the same time, determines the frequentist coverage of the resulting credible region. In this section, we therefore develop a practical strategy for selecting and calibrating $\alpha$ for the $\tssg$ framework so that the credible intervals attain the nominal $95\%$ coverage level, while recognizing that such calibration may come at some cost to variational point-estimation accuracy. Our proposed procedure for calibrating $\alpha$ follows from the general posterior calibration framework of~\cite{SyringMartin2019}, which advocates tuning a posterior spread parameter through a frequentist coverage criterion.

For illustration, following~\citet[Algorithm 1]{SyringMartin2019}, we present the procedure for the $\tssg$ Type I Student's-$t$ model with the model parameters $\theta = (\beta^{\top}, \tau^{2})^{\top}\in \mathbb{R}^p \times \mathbb{R}^{+}$ and $\nu\in \mathbb{N}$ along with the design matrix $\mathbf{X}\in \mathbb{R}^{n\times p}$ (including intercept) configured as in Section~\ref{main-subsec:sim-exp-student} of the main manuscript; we fix $(n, p) = (10000, 9)$. The same procedure applies more broadly to $\tssg$ models across both Type I and Type II $\ssg$ families. Fix a target credibility level $1-\gamma \in (0, 1)$, a convergence tolerance $\mathcal{E} > 0$, an initial guess $\alpha^{(0)}\in (0, 1]$, a step-size constant $s>0$, and a maximum number of iterations $T_{\max}\in\mathbb{N}$. Let $\mathcal{D}_n := \{(\mathbf{x}_i, y_i)\;:\;i\in [n]\}$ denote the observed data, and let $\widehat{\theta}(\mathbb P_n) := (\widehat{\beta}(\mathbb P_n)^{\top}, \widehat{\tau}^{2}(\mathbb{P}_n))^{\top}\in \mathbb{R}^{p}\times \mathbb{R}^{+}$ be a plug-in estimate obtained from the $\tssg$ Student's-$t$ model fit on the observed sample. Here, $\mathbb{P}_n$ is the empirical distribution of the observed data. For each $\alpha$, let:
\begin{align}
\label{eq:calibration-1}
\mathcal{C}_{\alpha, \gamma}(\mathcal D_n) := \left\{\theta\;:\; \pi_{\alpha}(\theta\mid \mathcal{D}_n, \xi^\star)\geq c_{\gamma}(\alpha)\right\},
\end{align}
denote the joint $100(1-\gamma)\%$ highest posterior density (HPD) credible region under the variational posterior distribution, where the threshold $c_{\gamma}(\alpha)$ is chosen such that the variational posterior probability assigned to $\mathcal{C}_{\alpha, \gamma}(\mathcal{D}_n)$ is $1-\gamma$. Then, define the empirical coverage function:
\begin{align}
\label{eq:calibration-2}
\widehat{c}_{\gamma}(\alpha; \mathbb{P}_n) := \frac{1}{B}\sum_{b\in [B]} \mathds{1}\left\{\widehat{\theta}(\mathbb{P}_n) \in \mathcal{C}_{\alpha, \gamma}(\mathcal{D}^{\star}_{n(b)})\right\},
\end{align}
where $\mathcal{D}^{\star}_{n(b)}$ for $b\in [B]$ are bootstrap resamples (with replacement) from the observed data. 
In Algorithm~\ref{alg:calibration-algorithm}, these resamples are redrawn at every iteration, so that each iteration uses an independent bootstrap evaluation of the coverage. The objective is to choose $\alpha$ so that $\widehat{c}_{\gamma}(\alpha; \mathbb P_n) = 1-\gamma$ (or, at least approximately), i.e., the credible regions produced by $\tssg$ achieve the desired nominal frequentist coverage. The resulting calibration algorithm is given in Algorithm~\ref{alg:calibration-algorithm}.

{
\renewcommand{\baselinestretch}{1.0}\footnotesize
\begin{algorithm}[H]
\caption{Calibrating $\alpha \in (0,1]$ for the $\tssg$ Type I Student's-$t$ model}
\label{alg:calibration-algorithm}
\DontPrintSemicolon

\KwIn{Observed data $\mathcal{D}_n$, plug-in estimate $\widehat{\theta}(\mathbb P_n)$, target credibility level $1-\gamma$, tolerance $\mathcal E>0$, initial guess $\alpha^{(0)}$, step-size constant $s>0$, maximum number of iterations $T_{\max}\in\mathbb{N}$, number of bootstrap samples $B$, number of Monte Carlo (MC) samples $M$, admissible interval $[\alpha_{\min},\alpha_{\max}]$.}
\KwOut{Calibrated likelihood tempering parameter $\alpha^{\star}$.}

\For{$t=0,1,\ldots,T_{\max}-1$}{
  \tcc{Step 1: Fit $\tssg$ on fresh bootstrap samples and construct HPD credible regions}
  Generate $B$ bootstrap resamples $\mathcal{D}_{n(b)}^{\star}=\{(\mathbf{x}_{i(b)}^{\star}, y_{i(b)}^{\star})\;:\;i\in [n]\}$ for $b\in [B]$, independently of those used in previous iterations. \;
  \For{$b\in [B]$}{
    Fit the $\tssg$ Type I Student's-$t$ model on $\mathcal{D}_{n(b)}^{\star}$ with tempering parameter $\alpha^{(t)}$, yielding the fitted Normal-Gamma variational posterior $\pi_{\alpha^{(t)}}(\theta\mid \mathcal{D}_{n(b)}^{\star},\xi^{\star})$. Next, generate $M$ MC samples:
    \begin{equation*}
        \theta_b^{(1)},\ldots,\theta_b^{(M)} \sim \pi_{\alpha^{(t)}}(\theta\mid \mathcal{D}_{n(b)}^{\star},\xi^{\star}),
    \end{equation*}
    and compute the corresponding joint variational $\log$-densities
    $
    \ell_b^{(m)}=\log \pi_{\alpha^{(t)}}(\theta_b^{(m)}\mid \mathcal{D}_{n(b)}^{\star},\xi^{\star})
    $, for $m\in [M]$. Let $\widehat{c}_{\gamma,b}^{\mathrm{MC}}$ be the empirical lower $\gamma$-quantile of $\{\ell_b^{(m)}\}_{m\in [M]}$. Define the joint $100(1-\gamma)\%$ HPD credible region by:
    \begin{equation*}
    \mathcal C_{\alpha^{(t)},\gamma}(\mathcal{D}_{n(b)}^{\star})
    :=
    \left\{
    \theta:\log \pi_{\alpha^{(t)}}(\theta\mid \mathcal{D}_{n(b)}^{\star},\xi^{\star})\ge \widehat{c}_{\gamma,b}^{\mathrm{MC}}
    \right\}.
    \end{equation*}
    \;
  }

  \tcc{Step 2: Evaluate empirical coverage probability}
  Compute:
  \begin{equation*}
    \widehat{c}_{\gamma}(\alpha^{(t)};\mathbb P_n)
  =
  \frac{1}{B}\sum_{b\in [B]}
  \mathds{1}\left\{
  \widehat{\theta}(\mathbb P_n)\in
  \mathcal C_{\alpha^{(t)},\gamma}(\mathcal{D}_{n(b)}^{\star})
  \right\}.
  \end{equation*}
  \;

  \tcc{Step 3: Check convergence at the evaluated iterate}
  \If{$\bigl|\widehat{c}_{\gamma}(\alpha^{(t)};\mathbb P_n)-(1-\gamma)\bigr|<\mathcal E$}{
    \Return{$\alpha^{\star}=\alpha^{(t)}$} \;
  }

  \tcc{Step 4: Update $\alpha$ via Robbins-Monro stochastic approximation scheme}
  As per~\cite{SyringMartin2019}, set $\kappa_{t} = (t+1)^{-0.51}$ and update~\citep{RobbinsMonro1951}:
  \begin{equation*}
      \alpha^{(t+1)}
  \gets
  \Pi_{[\alpha_{\min},\alpha_{\max}]}
  \left[
  \alpha^{(t)} + s\kappa_t
  \left\{
  \widehat{c}_{\gamma}(\alpha^{(t)};\mathbb P_n)-(1-\gamma)
  \right\}
  \right],
  \end{equation*}
  where $\Pi_{[\alpha_{\min},\alpha_{\max}]}(x) = \min\{\alpha_{\max}, \max\{\alpha_{\min}, x\}\}$ denotes projection onto $[\alpha_{\min},\alpha_{\max}]$. \;
}

\Return{$\alpha^{\star}=\alpha^{(T_{\max})}$} \;
\end{algorithm}

\newpage

\begin{figure}[!t]
    \centering
    \begin{subfigure}[t]{0.48\textwidth}
        \centering
        \includegraphics[width=\textwidth]{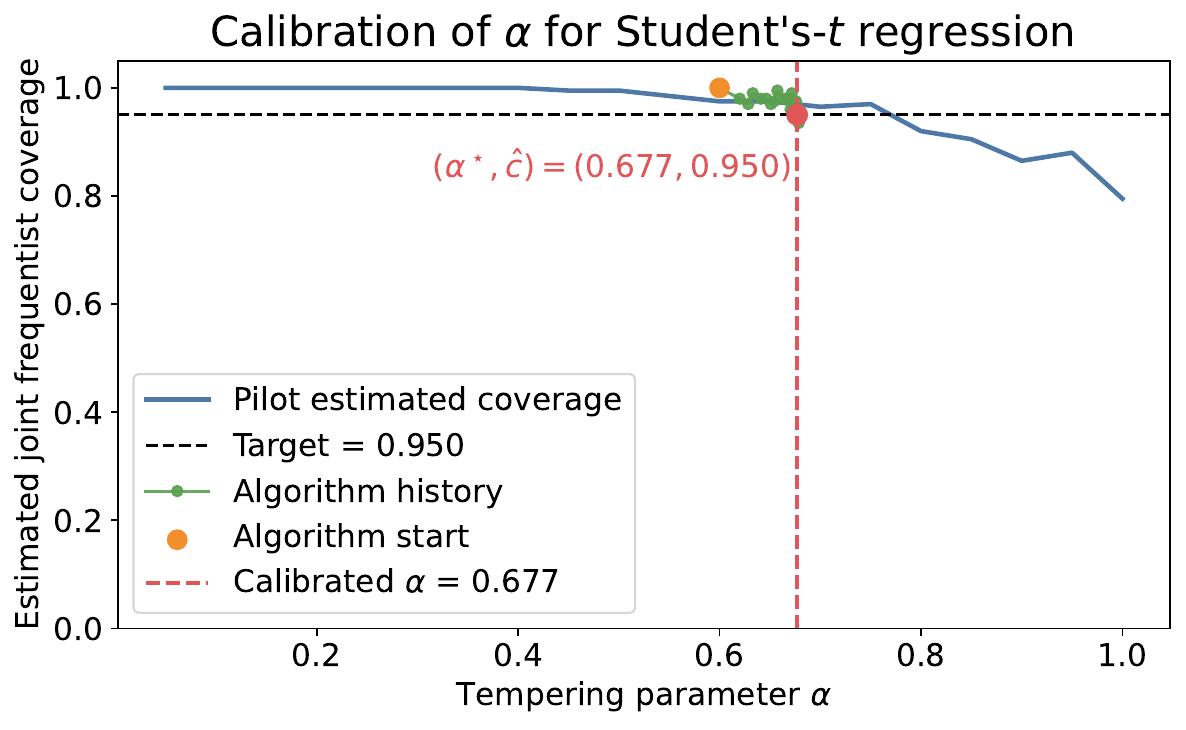}
        \label{fig:calibrated_alpha_with_history}
    \end{subfigure}
    \hfill
    \begin{subfigure}[t]{0.48\textwidth}
        \centering
        \includegraphics[width=\textwidth]{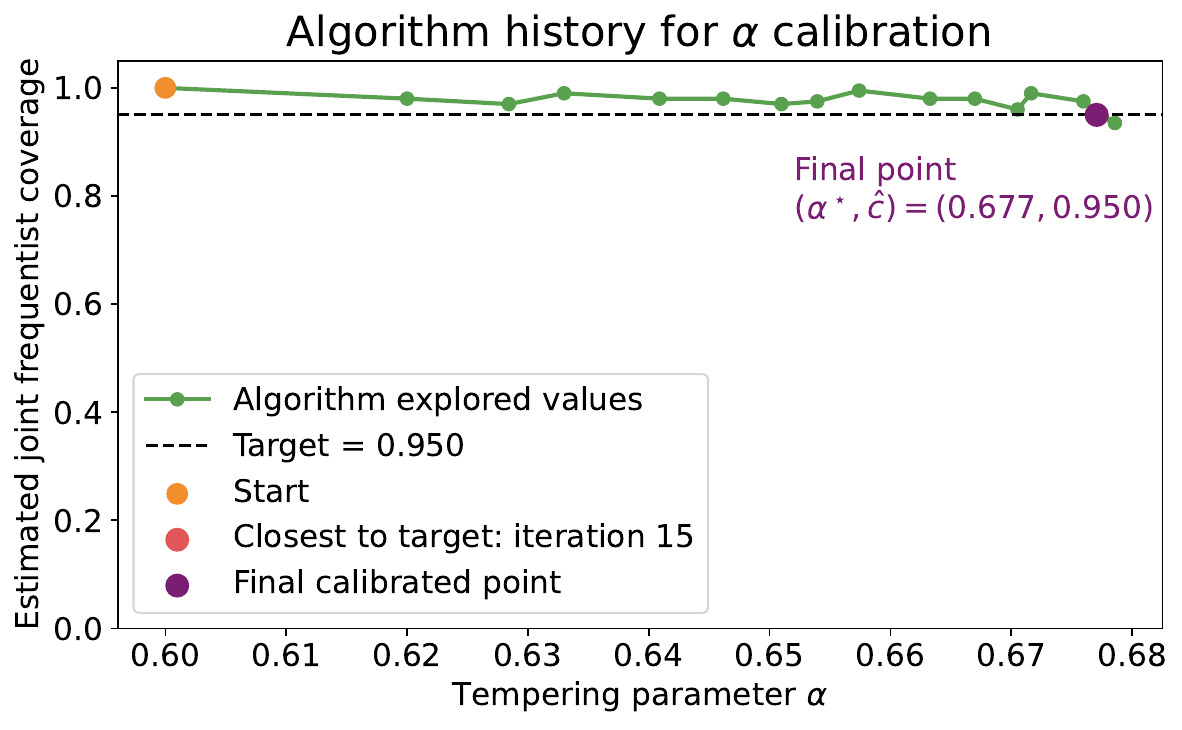}
        \label{fig:algorithm_history}
    \end{subfigure}
    \caption{\footnotesize{Calibration of $\alpha$ for the $\tssg$ Type I Student's-$t$ model with target $1-\gamma=0.95$. $\widehat{\theta}(\mathbb{P}_n)$ is obtained by fitting $\tssg$ under $\alpha=1.00$ and $(\mu, \Sigma, a, b, \texttt{tol}) = (0, I_{p}, 0.05, 0.05, 10^{-9})$. For each candidate $\alpha$, calibration is based on $B=200$ bootstrap resamples, with the joint $95\%$ HPD credible region for $(\beta^{\top}, \tau^2)^{\top}$ constructed from $M=5000$ MC draws from the fitted variational posterior. A pilot grid search over $[\alpha_{\min}, \alpha_{\max}] = [0.05, 1.00]$ (left panel) provides a reasonable initial value $\alpha^{(0)} = 0.60$, which is then refined via the Robbins-Monro scheme~\citep{RobbinsMonro1951} in Algorithm~\ref{alg:calibration-algorithm} (see right panel for algorithm history) with $(s, T_{\max}, \mathcal{E}) = (0.4, 40, 0.005)$, which terminates at iteration $t=15$. Algorithm~\ref{alg:calibration-algorithm} results into $(\alpha^{\star}, \widehat{c}_{0.05}(\alpha^{\star}; \mathbb{P}_n)) = (0.677, 0.95)$.
    }}
    \label{fig:calibration_algorithm_results}
\end{figure}
}

Taken together, the plots in Figure~\ref{fig:calibration_algorithm_results} summarize the calibration results obtained from Algorithm~\ref{alg:calibration-algorithm}. The pilot grid coverage estimates (left panel), along with the stochastic approximation path (right panel), show that the nominal $95\%$ coverage level is attained at an interior value of $\alpha$, rather than at either boundary of the admissible interval $[\alpha_{\min}, \alpha_{\max}] = [0.05, 1.00]$. The pilot search identifies $\alpha^{(0)} = 0.60$ as a suitable starting value, and the Robbins-Monro refinement in Algorithm~\ref{alg:calibration-algorithm} then adjusts this value slightly upward, converging to the calibrated choice $\alpha^{\star} = 0.677$. At this value, the estimated joint coverage is $\widehat{c}_{0.05}(0.677; \mathbb{P}_n) = 0.95$, up to MC error.

This is practically important. The calibration does not simply drive $\alpha$ toward very small values in order to inflate uncertainty, but instead selects a moderate level of tempering that balances posterior concentration against interval calibration, given the target credibility level $1-\gamma=0.95$. In this sense, Algorithm~\ref{alg:calibration-algorithm} identifies a data-driven compromise between the overconservative regime at small $\alpha$ and the under-covering regime near $\alpha=1$, thereby providing a general and practical recipe for calibrating and selecting $\alpha$ in the $\tssg$ framework.
}

\newpage


\section{Discussion on Assumptions}\label{sec:discussion-assumptions}

\subsection{\texorpdfstring{Assumption~\ref{main-ass:1}}{Assumption 1}}

\begin{assumption}[Regularity on $h$]
\label{ass:1}
The function $h: \mathbb{R}^{+}_{0}\to \mathbb{R}$ in Definition~\ref{main-def:SSG} of the main manuscript satisfies the following regularity conditions: 
\begin{enumerate}
    \item [(i)] {$h$ is thrice differentiable on $\mathbb{R}^{+}$ and all derivatives up to order $3$ are continuous on $\mathbb{R}^{+}$;}

    \item [(ii)] $h$ is strictly decreasing with $\lim_{t \rightarrow \infty} h(t) = -\infty$;

    \item [(iii)] $h''(t) > 0$; and

    \item [(iv)] there exists a constant $K\in \mathbb{R}^{+}$ such that the function $t \mapsto h(t^2)$ is $K$-Lipschitz. 
\end{enumerate}
\end{assumption}

Assumption~\ref{ass:1} imposes mild regularity on the function $h: \mathbb{R}_{0}^{+} \to \mathbb{R}$ defining the $\ssg$ likelihoods in Definition~\ref{main-def:SSG} of the main manuscript. Below, we briefly justify each component of the assumption and verify that all model families with the form of $h$ in \eqref{h-func} satisfy these conditions.
\begin{equation}
\label{h-func}
h(t)=
\begin{cases}
-\sqrt{t}, & \text{for Laplace family},\\[6pt]
-\dfrac{\nu+1}{2}\,\log\!\left(1+\dfrac{t}{\nu}\right), 
& \text{for Student's-$t$ family with fixed }\nu\in\mathbb{N},\\[10pt]
-\log\!\left[2\cosh\!\left(\dfrac{\sqrt{t}}{2}\right)\right], 
& \text{for Type II $\ssg$ likelihoods}.
\end{cases}
\end{equation}

\noindent
\emph{Smoothness}. { Assumption \ref{ass:1}(i) requires $h$ to be thrice differentiable on $\mathbb{R}^{+}$ with the corresponding  derivatives being continuous on $\mathbb{R}^{+}$}. All three forms of $h$ in \eqref{h-func} are thrice differentiable on $\mathbb{R}^{+}$ with continuity on $\mathbb{R}^{+}$. The existence of first and second order derivatives is essential for constructing tangent minorizers and for controlling the Jensen's gap:
$$
\Delta_{2, i} := \log p(y_i\mid \mathbf{x}_i,\theta) - \log \varphi(y_i\mid \mathbf{x}_i, \theta, \xi_i), \quad i\in [n], 
$$
which plays a central role in the variational risk analysis in Section~\ref{app:variational-risk-bounds}. The third derivative ensures the local $\omega$-strong concavity of the Expectation-Maximization (EM) surrogate function $\mathcal{Q}$, which underpins the convergence guarantees for the $\tssg$ iterates in Section \ref{app:TAVIE-convergence-proof}.

\noindent
\emph{Monotonicity and tail behavior}. Assumption \ref{ass:1}(ii) requires $h$ to be strictly decreasing with $\lim_{t\to \infty}h(t) = -\infty$. This is inherent to $\ssg$ likelihoods: for Type II models the claim follows directly from the functional form, while for Type I models the divergence of $h(t)$ to $-\infty$ reflects the fact that any density supported on $\mathbb{R}$ must decay to zero in the tails.

\noindent
\emph{Convexity}. Assumption \ref{ass:1}(iii) states that $h''(t) > 0$, i.e., $h$ is strictly convex on $\mathbb{R}_0^{+}$. This convexity is a defining property of the $\ssg$ class and is easily verified for the Laplace, Student's-$t$, and Type II $\ssg$ families.

\noindent
\emph{Lipschitz regularity of the composition $t\mapsto h(t^2)$}. Assumption \ref{ass:1}(iv) requires $h(t^2)$ to be $K$-Lipschitz for some $K\in \mathbb{R}^{+}$. This condition holds with $K=1$ for Type II $\ssg$ likelihoods and the Laplace (Type I $\ssg$) family, and with $K=\tfrac{\nu+1}{2\sqrt{\nu}}$ for the Student's-$t$ family with $\nu\in \mathbb{N}$ (fixed). This Lipschitz property is used to control the $\log$-likelihood ratio terms that arise in the proof of variational risk bounds in Section \ref{app:variational-risk-bounds}.


\subsection{\texorpdfstring{Assumptions~\ref{main-ass:2} and~\ref{main-ass:3}}{Assumptions 2 and 3}}

Assumptions \ref{ass:2} and \ref{ass:3} impose mild conditions on the design matrix $\mathbf{X}$ under both Type I and Type II $\ssg$ likelihoods, and on the moment structure of Type I $\ssg$ likelihoods, respectively. We summarize their roles below.

\begin{assumption}[Design regularity]
\label{ass:2}
The rows $\{\mathbf{x}_i: i \in [n]\}$ of $\mathbf{X}$ are non-zero.
\end{assumption}

Requiring each row $\mathbf{x}_i$ of the design matrix $\mathbf{X}$ to be non-zero ensures that the $\tssg$ iterates do not get stuck at boundary points where one or more coordinates of $\xi = (\xi_1, \ldots, \xi_n)^{\top}$ collapse to zero. This prevents degeneracy in the surrogate objective and guarantees that each limit point lies in a region where the EM surrogate $\mathcal{Q}$ maintains local $\omega$-strong concavity, an essential ingredient for the convergence analysis of the $\tssg$ updates in Section \ref{app:TAVIE-convergence-proof}.

\begin{assumption}[Existence of moments for Type I $\mathsf{SSG}$ likelihoods]
\label{ass:3}
The $k$th moment exists and
$\mathscr{E}_k := \mathcal{Z}_h^{-1}\int |s|^k\exp\left\{h(s^2)\right\}ds  < \infty$, for all $k\in [2]$, where $\mathcal{Z}_h = \int_{\mathbb{R}} \exp\left\{h(s^2)\right\}ds$.
\end{assumption}

Assumption \ref{ass:3} requires that the first two moments of the underlying Type I $\ssg$ density be finite. These moment bounds allow high probability control over stochastic terms arising in the tangent minorization expansion and are used in bounding the Jensen's gap that appears in the variational risk bounds developed in Section \ref{app:variational-risk-bounds}. The condition is mild and is satisfied by all the representative Type I $\ssg$ families, including Laplace and Student's-$t$ models (with fixed degrees of freedom $\nu >2$).

\newpage
\section{Convergence Guarantees for the \texorpdfstring{$\tssg$}{TAVIE-SSG} EM Algorithm}\label{app:TAVIE-convergence-proof}

\subsection{Proof of \texorpdfstring{Proposition~\ref{main-proposition:stationarity-equivalence}}{Proposition 2}}

\begin{proof}
By uniqueness and differentiability of $q_\xi$, the envelope/Danskin's theorem~\citep{Danskin1967} gives $\nabla_{\xi} \mathsf{L}(\xi) =\nabla_\xi\mathcal L(q_\xi,\xi)$.  If $(q^\star,\xi^\star)$ is stationary, then $q^\star=q_{\xi^\star}$ and $\nabla_\xi\mathcal L(q^\star,\xi^\star)=0$, hence $\nabla_{\xi} \mathsf{L}(\xi^\star)=0$.  Conversely, if $\nabla_{\xi} \mathsf{L}(\xi^\star)=0$ then $\nabla_\xi\mathcal L(q_{\xi^\star},\xi^\star)=0$, and since $q_{\xi^\star}$ maximizes $\mathcal{L}(q, \xi^{\star})$ in $q$, the functional derivative in the $q$-direction vanishes at $q_{\xi^{\star}}$; thus $(q_{\xi^\star},\xi^\star)$ is stationary for the joint maximization problem. 
\end{proof}

\subsection{Lemmata for \texorpdfstring{Theorem~\ref{main-theorem:convergence}}{Theorem 1}}\label{app:lemmata-convergence-TAVIE}

\begin{lemma}[Relative compactness of upper level sets of $\mathsf{L}$]
\label{lemma-conv:1}  
Under Assumptions \ref{ass:1}(i)-\ref{ass:1}(iv) and \ref{ass:2} of the main manuscript, $\mathsf{L}:\mathbb{R}_{+}^{n}\to \mathbb{R}$ for both Type I and Type II $\mathsf{SSG}$ likelihoods in~\eqref{main-eq:SSG-typeI-likelihood} and~\eqref{main-eq:SSG-typeII-likelihood} of the main manuscript is continuous and satisfies:
\begin{gather}
    \label{eq:L1.1}
        \lim_{\max_{i\in [n]} \xi_i \uparrow \infty}\mathsf{L}(\xi) = -\infty.
\end{gather}
Moreover, one of the following holds:
\begin{enumerate}
    \item If $\lim_{t\downarrow 0}h'(t) = -\infty$, then:
    \begin{align}
    \label{eq:L1.2}
    \lim_{\min_{i\in [n]} \xi_i \downarrow 0} \mathsf{L}(\xi) = -\infty,
    \end{align}
    and hence every upper level set $\mathsf{U}_c := \{\xi\in \mathbb{R}^{n}_{+}\;:\;\mathsf{L}(\xi) \geq c\}$ is compact in $\mathbb{R}^{n}_{+}$.

    \item If $\lim_{t\downarrow 0}h'(t) = \ell_{h}\in \mathbb{R}$, then $\mathsf{L}$ admits a continuous extension $\tilde{\mathsf{L}}: \mathbb{R}^{n}_{+, 0} \to \mathbb{R}$, and for every $c$ in the range of $\tilde{\mathsf{L}}$, the closure $\overline{\mathsf{U}}_c$ of $\mathsf{U}_c$ is compact in $\mathbb{R}^{n}_{+, 0}$.
\end{enumerate}
\end{lemma}

\begin{proof}
Continuity of $\mathsf{L}$ for both Type I and Type II $\mathsf{SSG}$ likelihoods on $\mathbb{R}_{+}^{n}$ follows from Assumption \ref{ass:1}(i) of the main manuscript. Note that, \eqref{eq:L1.1} can be immediately concluded  using the form of Type I and Type II $\mathsf{SSG}$ likelihoods in~\eqref{main-eq:SSG-typeI-likelihood} and~\eqref{main-eq:SSG-typeII-likelihood} of the main manuscript along with Assumption \ref{ass:1}(ii) of the main manuscript, where $\gamma(t) = h(t^2) - t^2 h'(t^2) \to -\infty$ as $|t| \to \infty$. 

We now analyze the behavior near the boundary $\{\xi_i=0\}$. Consider the function $h': \mathbb{R}^{+} \to (-\infty, 0)$. By Assumption~\ref{ass:1}(iii) of the main manuscript, $h'$ is increasing over $\mathbb{R}^{+}$. Consequently, $\lim_{t\downarrow 0}h'(t)$ exists and can be either equal to a fixed $\ell_h\in \mathbb{R}$, or it diverges to $-\infty$.

\textbf{Case 1}: $\lim_{t\downarrow 0}h'(t) = -\infty$. Note that, by Assumption~\ref{ass:1}(iv) of the main manuscript, $|2th'(t^2)| \leq K$ for a positive constant $K\in \mathbb{R}^{+}$. Hence, $\lim_{t\downarrow 0}\gamma(t)$ remains bounded. Recall, $A(t) = h'(t^2)$. Hence, if $\min_{i\in [n]}\xi_i \downarrow 0$:
\begin{align}
\label{eq:H.3}
\begin{split}
&\implies \min_{i\in [n]}A(\xi_i) \downarrow -\infty 
\implies \lambda_{\max}(\Sigma^{-1}_{\alpha}(\xi)) \geq \frac{\mathbf{x}_j^{\top}\Sigma^{-1}_{\alpha}(\xi)\mathbf{x}_j}{\lVert\mathbf{x}_j\rVert_2^{2}} \geq 2\alpha t_j|A(\xi_j)|\lVert\mathbf{x}_j\rVert_2^{2} \uparrow \infty\\
&
\implies \lambda_{\min}(\Sigma_{\alpha}(\xi)) \downarrow 0 
\implies \log |\Sigma_{\alpha}(\xi)| \downarrow -\infty
\implies \mathsf{L}(\xi) \downarrow -\infty,
\end{split}
\end{align}
where $j\in\arg\min_{i\in[n]}\xi_i$, $t_j=1$ for Type I and $t_j=b_j>0$ for Type II $\ssg$ likelihoods, and the divergence uses $\mathbf{x}_j\neq\mathbf{0}$ by Assumption~\ref{ass:2} of the main manuscript. Therefore, \eqref{eq:H.3} implies $\lim_{\min_{i\in [n]}\xi_i\downarrow 0}\mathsf{L}(\xi) = -\infty$, which proves~\eqref{eq:L1.2}. Observe, this scenario applies to the Laplace Type I $\ssg$ likelihood. Now fix $c\in \mathrm{range}(\mathsf{L})$. By continuity, $\mathsf{U}_c$ is closed in $\mathbb{R}^{n}_+$. Moreover,~\eqref{eq:L1.1} prevents $\mathsf{U}_c$ from escaping to infinity, while~\eqref{eq:L1.2} prevents it from approaching the boundary $\{\xi_i=0\}$. Hence $\mathsf{U}_c$ is bounded and bounded away from the boundary of $\mathbb{R}^{n}_+$. Therefore, $\mathsf{U}_c$ is closed and bounded as a subset of $\mathbb{R}^{n}$, and thus compact by the Heine-Borel theorem.

\textbf{Case 2}: $\lim_{t\downarrow 0}h'(t) = \ell_h\in \mathbb{R}$. Define the extension of $h'$, $\tilde{h}':\mathbb{R}^{+}_0 \to (-\infty, 0)$ by setting $\tilde{h}'(t) = h'(t)$ for $t>0$ and $\tilde{h}'(0) = \ell_h$. Since $h'$ is monotone, this extension is continuous at $0$. Replacing $h'$ by $\tilde{h}'$ in the closed-form expressions of $\mathsf{L}(\xi)$ for both Type I and Type II $\ssg$ likelihoods in~\eqref{main-eq:ELBO-general} of the main manuscript, $\mathsf{L}(\xi)$ can be extended continuously to $\tilde{\mathsf{L}}:\mathbb{R}^{n}_{+, 0} \to \mathbb{R}$. Now let, $\overline{\mathsf{U}}_c$ be the closure in $\mathbb{R}^{n}_{+, 0}$ of $\mathsf{U}_c$ as defined above. Further, since $\xi_i$'s are bounded below by $0$ for all $i\in [n]$, combining this observation along with \eqref{eq:L1.1}, establishes boundedness of $\overline{\mathsf{U}}_c$. Consequently, an application of the Heine-Borel theorem yields the compactness of $\overline{\mathsf{U}}_{c}$.
\end{proof}

\begin{lemma}[$\tssg$ iterates remain in $\mathbb{R}^{n}_+$]
\label{lemma-conv:bounded-away-zero}
Under Assumptions~\ref{ass:1}(ii) and~\ref{ass:2} of the main manuscript, for both Type I and Type II $\ssg$ likelihoods in~\eqref{main-eq:SSG-typeI-likelihood} and~\eqref{main-eq:SSG-typeII-likelihood} of the main manuscript, the sequence of $\tssg$ iterates $\{\xi^{(t)}:t\geq 0\}$ remains in $\mathbb{R}_{+}^{n}$, i.e., no iterate hits the boundary of $\mathbb{R}_{+}^{n}$, if $\xi^{(0)}\in \mathbb{R}^{n}_+$.
\end{lemma}

\begin{proof}
Observe that, $\Sigma_{\alpha}(\xi) = [\Sigma^{-1} - 2\alpha \mathbf{X}^{\top}\mathcal{A}(\xi)\mathbf{X}]^{-1} \succ 0$ for Type I $\ssg$ likelihoods and $\Sigma_{\alpha}(\xi) = \left[\Sigma^{-1} - 2\alpha \mathbf{X}^{\top}\left[\mathrm{diag}(\mathbf{b}) \circ \mathcal{A}(\xi)\right]\mathbf{X}\right]^{-1} \succ 0$ for Type II $\ssg$ likelihoods, by Assumptions \ref{ass:1}(ii) and \ref{ass:2} of the main manuscript. 

Further using Assumption \ref{ass:2} of the main manuscript, $\kappa_i(\xi) \geq \mathbf{x}_i^{\top}\Sigma_{\alpha}(\xi)\mathbf{x}_i > 0$ for all $i\in [n]$ and $\xi\in \mathbb{R}^{n}_{+}$. Consider $\xi^{(t)}$ to be the variational parameter vector at the $t$-th $\tssg$ iteration. For $t=0$, the initial variational parameter vector $\xi^{(0)} \in \mathbb{R}^{n}_+$, and for $t > 0$, $\xi_{i}^{(t)} = \sqrt{\kappa_i(\xi^{(t-1)})} > 0$. Hence, $\xi^{(t)}$ does not hit the zero-boundary of $\mathbb{R}^{n}_+$.
\end{proof}

\begin{lemma}[Cluster points of $\tssg$ iterates are fixed points]
\label{lemma-conv:1.1}
Under the assumptions of Lemma \ref{lemma-conv:1}, any cluster point $\xi^{\star}$ of the sequence of $\tssg$ iterates $\{\xi^{(t)}: t\geq 0\}$ is a fixed point of $H: \mathbb{R}^{n}_{+} \to \mathbb{R}^{n}_{+}$ defined as:
\begin{align}
    H(\xi):= \left(\sqrt{\kappa_1(\xi)}, \ldots, \sqrt{\kappa_n(\xi)}\right)^{\top},
\end{align}
where $\kappa_{i}(\xi)$ for Type I and Type II $\ssg$ likelihoods are as given in Section~\ref{main-subsec:tavie-algo} of the main manuscript.
\end{lemma}

\begin{proof}
Using the definition of the variational EM algorithm:
\begin{align}\label{eq:lemma-conv-1.1-1}
    \mathsf{L}(\xi^{(t+1)}) \geq \mathcal{Q}(\xi^{(t+1)} \mid \xi^{(t)}) \geq \mathcal{Q}(\xi^{(t)} \mid \xi^{(t)}) = \mathsf{L}(\xi^{(t)}).
\end{align}
Following \eqref{eq:lemma-conv-1.1-1}, we obtain the monotone ascent of the scalar sequence $\{\mathsf{L}(\xi^{(t)}): t\geq 0\}$. Hence for $t\geq 0$, $\xi^{(t)} \in \mathsf{U}_{\mathsf{L}(\xi^{(0)})}$. By compactness of the closure $\overline{\mathsf{U}}_{\mathsf{L}(\xi^{(0)})}$ of the upper level set $\mathsf{U}_{\mathsf{L}(\xi^{(0)})}$ of $\mathsf{L}$ as per Lemma \ref{lemma-conv:1}, there exists a convergent sub-sequence $\xi^{(t_j)} \to \xi^{\star}$. Define the shifted sub-sequence $\xi^{(t_j+1)}$. Using the update rule $\xi^{(t_j+1)} = H(\xi^{(t_j)})$ and by continuity of $H$, we get $\xi^{(t_j+1)} \to H(\xi^{\star})$. Call $H(\xi^{\star}) = \eta$, i.e., $\xi^{(t_j+1)} \to \eta$. Hence, it only remains to show that $\eta = \xi^{\star}$.

Further, the sequence $\{\mathsf{L}(\xi^{(t)}):t\geq 0\}$ is bounded above by the evidence $\log p_{\alpha}(y\mid \mathbf{X})$ and hence convergent. Let $\ell:= \lim_{t\to \infty}\mathsf{L}(\xi^{(t)})$. 
Consequently, the sub-sequences:
\begin{align}\label{eq:lemma-conv-1.1-2}
    \mathsf{L}(\xi^{(t_j)}) \to \ell, \quad \mathsf{L}(\xi^{(t_j+1)}) \to \ell.
\end{align}
Substituting $t=t_j$ in \eqref{eq:lemma-conv-1.1-1} and taking limit as $j\to \infty$:
\begin{align}
\label{eq:lemma-conv-1.1-2.1}
    \lim_{j\to \infty} \mathsf{L}(\xi^{(t_j+1)}) \geq \lim_{j\to\infty}\mathcal{Q}(\xi^{(t_j+1)}\mid \xi^{(t_j)}) \geq \lim_{j\to \infty}\mathcal{Q}(\xi^{(t_j)}\mid \xi^{(t_j)}) = \lim_{j\to \infty}\mathsf{L}(\xi^{(t_j)}).
\end{align}
From continuity of $\mathsf{L}$, $\mathcal{Q}$ (in both arguments), the convergences $\xi^{(t_j)} \to \xi^{\star}$ and $\xi^{(t_j+1)} \to \eta$, and \eqref{eq:lemma-conv-1.1-2}, we get from \eqref{eq:lemma-conv-1.1-2.1} that:
\begin{align}
\begin{split}
\label{eq:lemma-conv-1.1-3}
    \ell \geq \mathcal{Q}(\eta&\mid \xi^{\star}) \geq \mathcal{Q}(\xi^{\star}\mid \xi^{\star}) = \mathsf{L}(\xi^{\star}),\\
\mathsf{L}(\xi^{\star}) &= \lim_{j\to \infty}\mathsf{L}(\xi^{(t_j)}) = \ell,\\
\mathsf{L}(\eta) &= \lim_{j\to \infty}\mathsf{L}(\xi^{(t_j+1)}) = \ell.
\end{split}
\end{align}
Finally, using \eqref{eq:lemma-conv-1.1-3} we get:
\begin{align}
\label{eq:lemma-conv-1.1-4}
    \mathsf{L}(\eta) = \mathcal{Q}(\eta\mid \xi^{\star}) = \mathcal{Q}(\xi^{\star}\mid \xi^{\star}) = \mathsf{L}(\xi^{\star}) = \ell.
\end{align}
By the construction of the variational EM algorithm, $\mathsf{L}(\xi)\geq \mathcal{Q}(\xi\mid v)$ for all $\xi$ with equality exactly when $\xi = v$. Taking $v = \xi^{\star}$ and $\xi = \eta$, we have $\mathsf{L}(\eta) \geq  \mathcal{Q}(\eta\mid \xi^{\star})$ and since equality holds by \eqref{eq:lemma-conv-1.1-4}, we conclude $\eta = \xi^{\star}$.
\end{proof}

\begin{lemma}[Critical points of $\mathsf{L}$ for Type I and Type II $\ssg$ likelihoods]
\label{lemma:critical-point-L}
For Type I and Type II $\ssg$ likelihoods, any point $\xi^{\star}$ is a fixed point of Algorithm~\ref{main-alg:tavie-em} in the main manuscript, if and only if it is a critical point of $\mathsf{L}(\xi)$, i.e., $\nabla_{\xi} \mathsf{L}(\xi^{\star}) = 0$. 
\end{lemma}

\begin{remark}
Although Lemma \ref{lemma:critical-point-L} has the same conclusion as Proposition~\ref{main-proposition:stationarity-equivalence} of the main manuscript, the independent importance of the lemma stems from the fact that while Proposition~\ref{main-proposition:stationarity-equivalence} is true under the general definition of $\ssg$ likelihood functions, it requires the uniqueness/differentiability conditions for Danskin's theorem~\citep{Danskin1967} to be verified. However, Lemma \ref{lemma:critical-point-L} demonstrates that for the Type I and Type II $\ssg$ likelihoods, the same conclusion as that of Proposition~\ref{main-proposition:stationarity-equivalence} holds with no additional assumptions on the respective likelihoods.
\end{remark}

\begin{proof}[Proof of Lemma \ref{lemma:critical-point-L}]

\noindent

\textbf{Type I $\mathsf{SSG}$ likelihoods}: The function $\mathsf{L}:\mathbb{R}_{+}^{n}\to \mathbb{R}$ for Type I $\mathsf{SSG}$ likelihood up to additive constants is:
\begin{align}\label{eq:f-Type-I}
    \mathsf{L}(\xi) := -\frac{a+n\alpha}{2}\log \frac{b_{\alpha}(\xi)}{2} + \alpha\sum_{i=1}^{n}\gamma(\xi_i) + \frac{1}{2}\log |\Sigma_{\alpha}(\xi)|.
\end{align}
The $\alpha$-posterior covariance matrix $\Sigma_{\alpha}(\xi)$ is:
\begin{align}\label{eq:L5.1}
    \Sigma_{\alpha}(\xi) = \left[\Sigma^{-1} - 2\alpha \mathbf{X}^{\top}\mathcal{A}(\xi)\mathbf{X}\right]^{-1},\quad M(\xi) := \Sigma_{\alpha}^{-1}(\xi).
\end{align}
Note that, for any invertible matrix $M(\psi)$, where $\psi\in \mathbb{R}$:
\begin{align}\label{eq:L5.2}
    \frac{\partial}{\partial \psi}\log |M(\psi)| = \mathrm{tr}\left(M^{-1}(\psi)\frac{\partial}{\partial \psi}M(\psi)\right).
\end{align}
Using \eqref{eq:L5.2}:
\begin{align}\label{eq:L5.3}
\begin{split}
    \frac{\partial}{\partial \xi_j}\left[\frac{1}{2}\log |\Sigma_{\alpha}(\xi)|\right] &=  -\frac{1}{2}\frac{\partial}{\partial \xi_j}\log|M(\xi)|\\
    &= -\frac{1}{2}\mathrm{tr}\left(\Sigma_{\alpha}(\xi)\frac{\partial}{\partial \xi_j}M(\xi)\right)\\
    &= -\frac{1}{2}\mathrm{tr}\left(\Sigma_{\alpha}(\xi)\left[-2\alpha A'(\xi_j)\mathbf{x}_j\mathbf{x}_{j}^{\top}\right]\right)\\
    &= \alpha A'(\xi_j)\mathbf{x}_j^{\top}\Sigma_{\alpha}(\xi)\mathbf{x}_j.
\end{split}
\end{align}
Now:
\begin{align}\label{eq:L5.4}
    \frac{\partial}{\partial \xi_j}b_{\alpha}(\xi) = \frac{\partial}{\partial \xi_j}\left(-2\alpha y^{\top}\mathcal{A}(\xi) y - \mu_{\alpha}(\xi)^{\top}\Sigma_{\alpha}^{-1}(\xi)\mu_{\alpha}(\xi)\right)
\end{align}
Note:
\begin{align}\label{eq:L5.5}
\begin{split}
    \frac{\partial}{\partial \xi_j}\Sigma_{\alpha}(\xi) &= -\Sigma_{\alpha}(\xi)\left(\frac{\partial}{\partial \xi_j}\Sigma_{\alpha}^{-1}(\xi)\right)\Sigma_{\alpha}(\xi)\\
    &= -\Sigma_{\alpha}(\xi)\left(\frac{\partial}{\partial \xi_j}M(\xi)\right)\Sigma_{\alpha}(\xi)\\
    &= -\Sigma_{\alpha}(\xi)\left(-2\alpha A'(\xi_j)\mathbf{x}_j\mathbf{x}_j^{\top}\right)\Sigma_{\alpha}(\xi),\quad \text{from \eqref{eq:L5.3}}\\
    &= 2\alpha A'(\xi_j)\Sigma_{\alpha}(\xi)\mathbf{x}_j\mathbf{x}_j^{\top}\Sigma_{\alpha}(\xi)
\end{split}
\end{align}
Let $c(\xi):= \Sigma^{-1}_{\alpha}(\xi)\mu_{\alpha}(\xi) = \Sigma^{-1}\mu - 2\alpha \mathbf{X}^{\top}\mathcal A(\xi) y$. Then:
\begin{align}\label{eq:L5.6}
    \begin{split}
        \frac{\partial}{\partial \xi_j}c(\xi) &= -2\alpha \mathbf{X}^{\top}A'(\xi_j) y_j e_j\\
        &= -2\alpha A'(\xi_j) y_j \mathbf{x}_j,
    \end{split}
\end{align}
where $e_j = (0, \ldots, 1, \ldots, 0)^{\top}\in \mathbb{R}^n$ is the $j$th canonical basis vector. From \eqref{eq:L5.5} and \eqref{eq:L5.6}:
\begin{align}\label{eq:L5.7}
    \begin{split}
        \frac{\partial}{\partial \xi_j}\mu_{\alpha}(\xi) &= \left[2\alpha A'(\xi_j)\Sigma_{\alpha}(\xi)\mathbf{x}_j\mathbf{x}_j^{\top}\Sigma_{\alpha}(\xi)\right]c(\xi) + \Sigma_{\alpha}(\xi)\left[-2\alpha A'(\xi_j) y_j \mathbf{x}_j\right]\\
        &= 2\alpha A'(\xi_j)\Sigma_{\alpha}(\xi) \mathbf{x}_j\left[\mathbf{x}_{j}^{\top}\mu_{\alpha}(\xi) - y_j\right]\\
        &= -2\alpha A'(\xi_j)\Sigma_{\alpha}(\xi)\mathbf{x}_j\mathfrak{R}_j(\xi),
    \end{split}
\end{align}
where $\mathfrak{R}_j(\xi) = y_j - \mathbf{x}_j^{\top}\mu_{\alpha}(\xi)$ for $j\in [n]$ and the penultimate equality in \eqref{eq:L5.7} above follows from the fact that $\mathbf{x}_{j}^{\top}\mu_{\alpha}(\xi) = \mathbf{x}_{j}^{\top}\Sigma_{\alpha}(\xi)c(\xi)$. Using \eqref{eq:L5.6} and \eqref{eq:L5.7} in \eqref{eq:L5.4}:
\begin{align}\label{eq:L5.8}
    \begin{split}
        \frac{\partial}{\partial \xi_j}b_{\alpha}(\xi) &= -2\alpha y^{\top}\left(\frac{\partial}{\partial \xi_j}\mathcal{A}(\xi)\right) y - \frac{\partial}{\partial \xi_j}\left(\mu_{\alpha}(\xi)^{\top}c(\xi)\right)\\
        &= -2\alpha A'(\xi_j) y_j^2 - \left[\frac{\partial}{\partial\xi_j}\mu_{\alpha}(\xi)^{\top}\right]c(\xi) - \mu_{\alpha}(\xi)^{\top}\left[\frac{\partial}{\partial\xi_j}c(\xi)\right]\\
        &= -2\alpha A'(\xi_j)\mathfrak{R}^{2}_j(\xi).
    \end{split}
\end{align}
Observe that, $\gamma(\xi_j) = h(\xi_j^2) - \xi_j^2h'(\xi_j^2)$, which implies $\gamma'(\xi_j) = -\xi_j^2 A'(\xi_j)$. Using this observation along with \eqref{eq:L5.3} and \eqref{eq:L5.8} in \eqref{eq:f-Type-I}:
\begin{align}\label{eq:L5.9}
    \begin{split}
        \frac{\partial}{\partial \xi_j}\mathsf{L}(\xi) &= \alpha A'(\xi_j)\left[\frac{a+n\alpha}{b_{\alpha}(\xi)}\mathfrak{R}^{2}_j(\xi) + \mathbf{x}_j^{\top}\Sigma_{\alpha}(\xi)\mathbf{x}_j - \xi_j^{2}\right]\\
        &= \alpha A'(\xi_j)\left[\kappa_j(\xi) - \xi_j^2\right],
    \end{split}
\end{align}
where $\kappa_j(\xi) = \frac{a+n\alpha}{b_{\alpha}(\xi)}\mathfrak{R}^{2}_j(\xi) + \mathbf{x}_j^{\top}\Sigma_{\alpha}(\xi)\mathbf{x}_j$ as defined in Section~\ref{main-subsec:tavie-algo} of the main manuscript. Defining $\kappa(\xi) := (\kappa_1(\xi),\ldots,\kappa_n(\xi))^{\top}$, we get from \eqref{eq:L5.9}:
\begin{align}\label{eq:L5.10}
    \nabla_{\xi}\mathsf{L}(\xi) = \alpha A'(\xi)\circ \left[ \kappa(\xi) - \xi \circ \xi\right],
\end{align}
where $A'(\xi) = (A'(\xi_1), \ldots, A'(\xi_n))^{\top}$. For any $\xi^{\star}$,  $\kappa(\xi^{\star}) = \xi^{\star} \circ \xi^{\star} \iff \nabla_{\xi}\mathsf{L}(\xi^{\star}) = 0$ from \eqref{eq:L5.10}. Therefore, we can conclude that, $\xi^{\star}$ is a critical point of $\mathsf{L}$ if and only if it is a fixed point of Algorithm~\ref{main-alg:tavie-em} of the main manuscript.

\noindent

\textbf{Type II $\mathsf{SSG}$ likelihoods}: The function $\mathsf{L}: \mathbb{R}_{+}^{n} \to \mathbb{R}$ for Type II $\mathsf{SSG}$ likelihood up to additive constants is:
\begin{align}\label{eq:f-Type-II}
\mathsf{L}(\xi) := \frac{1}{2}w^{\top}\Sigma_{\alpha}(\xi) w  +\alpha\sum_{i=1}^{n}b_i\gamma(\xi_i) + \frac{1}{2}\log |\Sigma_{\alpha}(\xi)|,
\end{align}
where $w = \Sigma^{-1}\mu + \alpha \mathbf{X}^{\top}\left(\mathbf{a} - \tfrac{\mathbf{b}}{2}\right)$ with $\mathbf{a} = (a_1, \ldots, a_n)^{\top}$ and $\mathbf{b} = (b_1, \ldots, b_n)^{\top}$. The $\alpha$-posterior covariance matrix $\Sigma_{\alpha}(\xi)$ is:
\begin{align}\label{eq:L6.1}
    \Sigma_{\alpha}(\xi) = [\Sigma^{-1} - 2\alpha \mathbf{X}^{\top}(\mathcal{A}(\xi)\circ \mathrm{diag}(\mathbf{b}))\mathbf{X}]^{-1}. 
\end{align}
Therefore:
\begin{align}\label{eq:L6.2}
    \begin{split}
        \frac{\partial}{\partial \xi_j}\Sigma_{\alpha}(\xi) &= -\Sigma_{\alpha}(\xi)\left(\frac{\partial}{\partial \xi_j}(\Sigma^{-1} - 2\alpha \mathbf{X}^{\top}(\mathcal{A}(\xi)\circ \mathrm{diag}(\mathbf{b}))\mathbf{X})\right)\Sigma_{\alpha}(\xi)\\
        &= 2\alpha A'(\xi_j)b_j\Sigma_{\alpha}(\xi)\mathbf{x}_j\mathbf{x}_j^{\top}\Sigma_{\alpha}(\xi),
    \end{split}
\end{align}
which follows from the first equality in \eqref{eq:L5.5}. Using \eqref{eq:L6.2}, we have:
\begin{align}\label{eq:L6.3}
\begin{split}
        \frac{1}{2}\frac{\partial}{\partial \xi_j}w^{\top}\Sigma_{\alpha}(\xi)w &= \alpha A'(\xi_j)b_jw^{\top}\Sigma_{\alpha}(\xi)\mathbf{x}_j\mathbf{x}_j^{\top}\Sigma_{\alpha}(\xi)w\\
        &= \alpha A'(\xi_j)b_j\left(\mathbf{x}_j^{\top}\mu_{\alpha}(\xi)\right)^2.
\end{split}
\end{align}
From \eqref{eq:L5.2}:
\begin{align}\label{eq:L6.4}
    \begin{split}
        \frac{1}{2}\frac{\partial}{\partial \xi_j}\log|\Sigma_{\alpha}(\xi)| &= -\frac{1}{2}\frac{\partial}{\partial \xi_j}\log|\Sigma^{-1} - 2\alpha \mathbf{X}^{\top}\left(\mathcal{A}(\xi)\circ \mathrm{diag}(\mathbf{b})\right)\mathbf{X}|\\
        &= -\frac{1}{2}\mathrm{tr}\left(\Sigma_{\alpha}(\xi)\left(-2\alpha b_j A'(\xi_j)\mathbf{x}_j\mathbf{x}_j^{\top}\right)\right)\\
        &= \alpha b_j A'(\xi_j)\mathbf{x}_j^{\top}\Sigma_{\alpha}(\xi)\mathbf{x}_j.
    \end{split}
\end{align}
Using \eqref{eq:L6.3}, \eqref{eq:L6.4}, and $\gamma'(\xi_j) = -\xi_j^2A'(\xi_j)$ in \eqref{eq:f-Type-II}:
\begin{align}\label{eq:L6.5}
    \begin{split}
        \frac{\partial}{\partial \xi_j}\mathsf{L}(\xi) &= -\alpha b_j A'(\xi_j)\left[\xi_j^2 - \left(\mathbf{x}_j^{\top}\mu_{\alpha}(\xi)\right)^2 - \mathbf{x}_j^{\top}\Sigma_{\alpha}(\xi)\mathbf{x}_j\right]\\
        &= \alpha b_j A'(\xi_j)\left[\kappa_j(\xi) - \xi_j^2\right],
    \end{split}
\end{align}
where $\kappa_j(\xi) = \left(\mathbf{x}_j^{\top}\mu_{\alpha}(\xi)\right)^2 + \mathbf{x}_j^{\top}\Sigma_{\alpha}(\xi)\mathbf{x}_j$ as defined in Section~\ref{main-subsec:tavie-algo} of the main manuscript. Defining $\kappa(\xi) := (\kappa_1(\xi),\ldots,\kappa_n(\xi))^{\top}$, we get from \eqref{eq:L6.5}:
\begin{align}\label{eq:L6.6}
    \nabla_{\xi} \mathsf{L}(\xi) = \alpha \mathbf{b}\circ A'(\xi)\circ \left[\kappa(\xi) - \xi \circ \xi\right],
\end{align}
where $A'(\xi) = (A'(\xi_1), \ldots, A'(\xi_n))^{\top}$.
Following the similar analogy as in the case of Type I $\mathsf{SSG}$ likelihoods, using \eqref{eq:L6.6}, we conclude that, for any $\xi^{\star}$,  $\kappa(\xi^{\star}) = \xi^{\star} \circ \xi^{\star} \iff \nabla_{\xi}\mathsf{L}(\xi^{\star}) = 0$. Therefore, we can conclude that, $\xi^{\star}$ is a critical point of $\mathsf{L}$ if and only if it is a fixed point of Algorithm~\ref{main-alg:tavie-em} of the main manuscript.
\end{proof}

\subsection{Proof of Convergence of the \texorpdfstring{$\tssg$}{TAVIE-SSG} EM Algorithm \texorpdfstring{(Theorem~\ref{main-theorem:convergence})}{(Theorem 1)}}\label{app:TAVIE-convergence}

\begin{definition}[Kurdyka-\L{}ojasiewicz (K\L) property and exponent~\citep{bolte2014proximal}]\label{def:KL-property} A proper closed function $\mathfrak{h}:\mathbb{R}^{n} \to (-\infty, \infty]$ satisfies the Kurdyka-\L{}ojasiewicz (K\L) property at $\hat{x} \in \mathrm{dom}\;\partial \mathfrak{h}$ if there exists a $\rho \in \mathbb{R}^{+}$, a neighborhood $\mathcal{U}$ of $\hat{x}$, and a continuous concave function $\phi: [0, \rho) \to \mathbb{R}^{+}_{0}$ with $\phi(0) = 0$ such that:
\begin{enumerate}
    \item [i.] $\phi$ is continuously differentiable on $(0, \rho)$ with $\phi' > 0$ on $(0, \rho)$; and

    \item [ii.] for every $x\in \mathcal{U}$ with $\mathfrak{h}(\hat{x}) < \mathfrak{h}(x) < \mathfrak{h}(\hat{x}) + \rho$, it holds that:
    \begin{align}\label{eq:KL}
        \phi'(\mathfrak{h}(x) - \mathfrak{h}(\hat{x}))\lVert \partial \mathfrak{h}(x)\rVert_2 \geq 1.
    \end{align}
\end{enumerate}
Further, if $\mathfrak{h}$ satisfies the K\L{} property at $\hat{x}\in \mathrm{dom}\;\partial \mathfrak{h}$ and $\phi$ in \eqref{eq:KL} can be chosen as $\phi(x) = \rho_0 x^{1-\Omega}$ for some $\rho_0>0$ and $\Omega\in (0, 1)$, then $\mathfrak{h}$ satisfies the K\L{} property at $\hat{x}$ with exponent $\Omega$, i.e., there exists a constant $c_{\mathrm{K\text{\L}}} > 0$ such that:
\begin{align}\label{eq:KL-1}
    \lVert \partial \mathfrak{h}({x})\rVert_{2} \geq c_{\mathrm{K\text{\L}}}\left(\mathfrak{h}(x) - \mathfrak{h}(\hat{x})\right)^{\Omega}.
\end{align}
\end{definition}

\begin{proof}[Proof of Theorem~\ref{main-theorem:convergence}]
Define $H: \mathbb{R}_{+}^{n} \to \mathbb{R}^{n}_{+}$ as $H(z) := \left(\sqrt{\kappa_1(z)}, \ldots, \sqrt{\kappa_n(z)}\right)^{\top}$, where:
\begin{align}\label{eq:kappa-Type-I-II}
    \kappa_i(z) = \begin{cases}
        \mathbf{x}_i^{\top}\Sigma_{\alpha}(z) \mathbf{x}_i + \frac{a_\alpha}{b_{\alpha}(z)}\left(y_i - \mathbf{x}_{i}^{\top}\mu_{\alpha}(z)\right)^2, &\text{Type I } \mathsf{SSG},\\
        \mathbf{x}_i^{\top}\Sigma_{\alpha}(z)\mathbf{x}_i + \left(\mathbf{x}_i^{\top}\mu_{\alpha}(z)\right)^{2}, & \text{Type II } \mathsf{SSG},
    \end{cases}
\end{align}
for $i\in [n]$. We prove Theorem~\ref{main-theorem:convergence} of the main manuscript using the following steps.

\noindent

\emph{Monotonic ascent and boundedness of $\mathsf{L}(\xi^{(t)})$}.
Since $\xi^{(t+1)} = \arg\max_{\xi \in \mathbb{R}_{+}^{n}}\mathcal{Q}(\xi\mid \xi^{(t)})$:
\begin{align}\label{eq:P.1}
    \mathsf{L}(\xi^{(t+1)}) \geq \mathcal{Q}(\xi^{(t+1)}\mid \xi^{(t)}) \geq \mathcal{Q}(\xi^{(t)}\mid \xi^{(t)}) = \mathsf{L}(\xi^{(t)}),
\end{align}
Hence, $\mathsf{L}(\xi^{(t)})$ is nondecreasing and bounded above by the evidence $\log p_{\alpha}(y\mid \mathbf{X})$, as a result of which it is convergent.

\noindent
\emph{Cluster points are fixed points}. 
By compactness of $\overline{\mathsf{U}}_{\mathsf{L}(\xi^{(0)})}$ from Lemma \ref{lemma-conv:1}, the sequence of $\tssg$ iterates $\{\xi^{(t)}: t\geq 0\}$ admit convergent sub-sequences and hence has cluster points.
By invoking Lemma \ref{lemma-conv:1.1}, any cluster point $\xi^{\star}$ of the sequence of $\tssg$ iterates $\{\xi^{(t)}: t\geq 0\}$ satisfies $\xi^{\star} = H(\xi^{\star})$.

\noindent
\emph{Cluster points are bounded away from zero}. Using Lemma~\ref{lemma-conv:bounded-away-zero} and $\xi^{\star} = H(\xi^{\star})$, we have $\xi_i^{\star} = \sqrt{\kappa_i(\xi^{\star})} > 0$ for all $i\in [n]$. Let $\delta(\xi^{\star}) := \min_i \xi_i^{\star} > 0$. Consider any sub-sequence $\xi^{(t_j)} \to \xi^{\star}$. Then there exists $J \in \mathbb{N}$ such that for $j \geq J$, $\xi^{(t_j)} \in \left[\tfrac{\delta(\xi^{\star})}{2}, \infty\right)^{n}$. By Lemma \ref{lemma-conv:1}, the compactness of $\overline{\mathsf{U}}_{\mathsf{L}(\xi^{(t_J)})}$ implies the existence of $R>0$ such that for all $j\geq J$, $\xi^{(t_j)} \in \left[\tfrac{\delta(\xi^{\star})}{2}, R\right]^{n}$.

\noindent
\emph{Local $\omega$-strong concavity of EM surrogate $\mathcal{Q}(\cdot \mid v)$}. Since $h''$ is continuous and strictly positive by Assumptions \ref{ass:1}(i) and \ref{ass:1}(iii) of the main manuscript:
\begin{align}\label{eq:P.1.1}
m := \inf_{t\in \left[\tfrac{\delta(\xi^{\star})}{2}, R\right]}\;th''(t^2) > 0.
\end{align}
We work in $\left[\tfrac{\delta(\xi^{\star})}{2}, R\right]^{n}$. For a single coordinate $\xi_i$ under the Type I $\mathsf{SSG}$ likelihood:
\begin{align}\label{eq:Q-TypeI}
\begin{split}
    \frac{\partial^2}{\partial \xi_i^2}\mathcal{Q}(\xi\mid v)\; &= \; \alpha \left[4\xi^2_ih'''(\xi_i^2)\kappa_i(v) + 2h''(\xi_i^2)\kappa_i(v) - 6\xi_i^2h''(\xi_i^2) - 4\xi_i^4h'''(\xi_i^2)\right]\\
    &=\; \alpha \left[\kappa_i(v)\left(4\xi_i^2h'''(\xi_i^2) + 2h''(\xi_i^2)\right) - \left(6\xi_i^2h''(\xi_i^2) + 4\xi_i^4 h'''(\xi_i^2)\right)\right]\\
    &=\; -4\alpha \xi_i^2h''(\xi_i^2),
\end{split}
\end{align}
evaluated at $\kappa_i(v) = \xi_i^2$. Using \eqref{eq:P.1.1} in \eqref{eq:Q-TypeI}, we get:
\begin{align}\label{eq:TypeI-concavity}
    \frac{\partial^2}{\partial \xi_i^2}\mathcal{Q}(\xi\mid v)\Big|_{\xi_i = \sqrt{\kappa_i(v)}}\;\leq\; -2\alpha m\delta(\xi^{\star}).
\end{align}
Similarly, for Type II $\mathsf{SSG}$ likelihood:
\begin{align}\label{eq:Q-TypeII}
\begin{split}
    \frac{\partial^2}{\partial \xi_i^2}\mathcal{Q}(\xi\mid v)\; &= \; \alpha b_i \left[4\xi^2_ih'''(\xi_i^2)\kappa_i(v) + 2h''(\xi_i^2)\kappa_i(v) - 6\xi_i^2h''(\xi_i^2) - 4\xi_i^4h'''(\xi_i^2)\right]\\
    &=\; \alpha b_i \left[\kappa_i(v)\left(4\xi_i^2h'''(\xi_i^2) + 2h''(\xi_i^2)\right) - \left(6\xi_i^2h''(\xi_i^2) + 4\xi_i^4 h'''(\xi_i^2)\right)\right]\\
    &=\; -4\alpha b_i \xi_i^2h''(\xi_i^2),
\end{split}
\end{align}
evaluated at $\kappa_i(v) = \xi_i^2$.  From \eqref{eq:P.1.1} and \eqref{eq:Q-TypeII}, we have:
\begin{align}\label{eq:TypeII-concavity}
\frac{\partial^2}{\partial \xi_i^2}\mathcal{Q}(\xi\mid v)\Big|_{\xi_i = \sqrt{\kappa_i(v)}}\;\leq\; -2\alpha m\delta(\xi^{\star})\cdot \min_{i}b_i.
\end{align}
Hence, from \eqref{eq:TypeI-concavity} and \eqref{eq:TypeII-concavity}, the local $\omega$-strong concavity of the EM surrogate function $\mathcal{Q}(\xi \mid v)$ holds for $\omega = 2\alpha m \delta(\xi^{\star})$ and $\omega = 2\alpha m \delta(\xi^{\star}) \cdot \min_i b_i$ under Type I and Type II $\mathsf{SSG}$ likelihoods, respectively, in $\xi, v \in \left[\tfrac{\delta(\xi^{\star})}{2}, R\right]^{n}$.

\noindent
\emph{Sufficient ascent and relative-error bound}. Let $v \in \left[\tfrac{\delta(\xi^{\star})}{2}, R\right]^{n}$ and $\xi = H(v)$, the local $\omega$-strong concavity of $\mathcal{Q}(\cdot \mid v)$ in $\left[\tfrac{\delta(\xi^{\star})}{2}, R\right]^{n}$ holds with modulus $\omega > 0$ independent of $v$. Therefore:
\begin{align}\label{eq:P.2}
    \mathcal{Q}(\xi\mid v) - \mathcal{Q}(v\mid v) \geq \frac{\omega}{2}\lVert \xi - v\rVert^{2}_2.
\end{align}
Using $\mathsf{L}(\xi) \geq \mathcal{Q}(\xi\mid v)$ with equality when $v=\xi$, we get:
\begin{align}\label{eq:P.3}
    \mathsf{L}(\xi) - \mathsf{L}(v) \geq \frac{\omega}{2}\lVert \xi - v\rVert^{2}_2,\quad (\text{\textit{sufficient ascent}}).
\end{align}
Since $\nabla_{\tilde{\xi}} \mathcal{Q}(\tilde{\xi}=\xi\mid v) = 0$ and $\nabla_{\tilde{\xi}} \mathsf{L}(\xi) = \nabla_{\tilde{\xi}}\mathcal{Q}(\tilde{\xi}=\xi\mid\xi)$, we obtain:
\begin{align}\label{eq:P.4}
\begin{split}
    \lVert \nabla_{\tilde{\xi}} \mathsf{L}(\xi)\rVert_2 &= \lVert \nabla_{\tilde{\xi}} \mathcal{Q}(\tilde{\xi}=\xi\mid \xi) - \nabla_{\tilde{\xi}}\mathcal{Q}(\tilde{\xi}=\xi\mid v)\rVert_2\\
    &\leq C\lVert \xi - v\rVert_2,\quad (\text{\textit{relative-error bound}}),
\end{split}
\end{align}
for some $C>0$ independent of $v$, because the map $(\xi, v) \mapsto \nabla_{\tilde{\xi}}\mathcal{Q}(\xi \mid v)$ is continuously differentiable on $ \left[\tfrac{\delta(\xi^{\star})}{2}, R\right]^{n}\times  \left[\tfrac{\delta(\xi^{\star})}{2}, R\right]^{n}$ and hence Lipschitz in the second argument on  this compact set.

\noindent
\emph{Convergence to a fixed point}. Let $\xi^{\star}$ be a fixed point of $H$. Also, consider $\mathcal{V}$ to be an open neighborhood of $\xi^{\star}$ contained in $\left[\tfrac{\delta(\xi^{\star})}{2}, R\right]^{n}$.
Note that, the function $\mathsf{L}:\mathbb{R}_{+}^{n}\to \mathbb{R}$, under both Type I and Type II $\mathsf{SSG}$ likelihoods, is an analytic map. Therefore, it satisfies the K\L{} property in Definition \ref{def:KL-property} on compact sets (in particular at the fixed point $\xi^{\star}$)~\citep{KL-SIAM-1,bolte2014proximal}. Hence, there exists a neighborhood $\mathcal{U}$ of $\xi^{\star}$, $\eta>0$, and a concave function $\phi: [0, \eta) \to \mathbb{R}^{+}_0$, which is continuously differentiable on $(0, \eta)$ with $\phi'> 0$ such that:
\begin{align}\label{eq:P.7}
\phi'(f(\xi) - f(\xi^{\star}))\lVert \nabla_\xi f(\xi)\rVert_2 \geq 1,\text{ for all }\xi\in \mathcal{U}\text{ with }0<f(\xi) - f(\xi^{\star})< \eta, 
\end{align}
where $f(\xi) = -\mathsf{L}(\xi)$. As $f(\xi^{(t)}) \downarrow f(\xi^{\star})$ and $\xi^{(t)}$ eventually stays in $\tilde{\mathcal{U}} = \mathcal{U} \cap \mathcal{V}$, there exists $T\in \mathbb{N}$ with $\xi^{(t)} \in \tilde{\mathcal{U}}$ and $0<f(\xi^{(t)}) - f(\xi^{\star}) <\eta$, for all $t\geq T$. For $t\geq T$, define:
\begin{align}\label{eq:P.8}
\begin{split}
    S_t:= f(\xi^{(t)}) - f(\xi^{\star}) \in (0, \eta)
    \implies S_{t+1} - S_{t} = f(\xi^{(t+1)}) - f(\xi^{(t)}). 
\end{split}
\end{align}
From \eqref{eq:P.7} and the relative-error bound in \eqref{eq:P.4}:
\begin{align}\label{eq:P.9}
    \phi'(S_t) \geq \frac{1}{\lVert \nabla_{\xi} f(\xi^{(t)})\rVert_2} \geq \frac{1}{C\lVert \Delta^{(t)}\rVert_2},
\end{align}
where $\Delta^{(t)} = \xi^{(t)} - \xi^{(t-1)}$. Concavity of $\phi$ gives:
\begin{align}\label{eq:P.10}
    \begin{split}
        \phi(S_{t+1}) \leq \phi(S_t) + \phi'(S_t)(S_{t+1} - S_t)
        \implies \phi(S_t) - \phi(S_{t+1}) \geq \phi'(S_t)(S_t-S_{t+1}).
    \end{split}
\end{align}
Using $S_t - S_{t+1} = f(\xi^{(t)}) - f(\xi^{(t+1)})$, \eqref{eq:P.9}, and sufficient ascent in \eqref{eq:P.3}, we get:
\begin{align}\label{eq:P.11}
    \begin{split}
        \phi(S_t) - \phi(S_{t+1}) &\geq \phi'(S_t)(f(\xi^{(t)}) - f(\xi^{(t+1)}))\\
        &= \phi'(S_t)(\mathsf{L}(\xi^{(t+1)}) - \mathsf{L}(\xi^{(t)}))\\
        &\geq \frac{1}{C\lVert \Delta^{(t)}\rVert_2}\cdot \frac{\omega}{2}\lVert \Delta^{(t+1)}\rVert^2_2\\
        &= \frac{\omega}{2C}\cdot \frac{\lVert \Delta^{(t+1)}\rVert^{2}_2}{\lVert \Delta^{(t)}\rVert_2}.
    \end{split}
\end{align}
Observe:
\begin{align}\label{eq:P.12}
 \phi(S_t) - \phi(S_{t+1}) \geq \frac{\omega}{2C}\left(2\lVert \Delta^{(t+1)}\rVert_2 - \lVert \Delta^{(t)}\rVert_2\right),   
\end{align}
using $b^{-1}a^2 \geq 2a-b$ for $a, b>0$. From \eqref{eq:P.12}, we get:
\begin{align}\label{eq:P.13}
    \lVert \Delta^{(t+1)}\rVert_2 \leq \frac{1}{2}\lVert \Delta^{(t)}\rVert_2+ \frac{C}{\omega}(\phi(S_t) - \phi(S_{t+1})).
\end{align}
Defining $\mathcal{S}_N:= \sum_{t=T}^{N}\lVert \Delta^{(t)}\rVert_2$ and summing \eqref{eq:P.13} over $T\leq t\leq N$, we have:
\begin{align}\label{eq:P.14}
\begin{split}
    \mathcal{S}_{N+1} - \lVert \Delta^{(T)}\rVert_2 &\leq \frac{\mathcal{S}_N}{2} + \frac{C}{\omega}(\phi(S_T) - \phi(S_{N+1}))\\
    &\leq \frac{\mathcal{S}_N}{2} + \frac{C}{\omega}\phi(S_T).
\end{split}
\end{align}
Therefore:
\begin{align}\label{eq:P.15}
    \mathcal{S}_{N+1} \leq \frac{\mathcal{S}_N}{2} + \lVert \Delta^{(T)}\rVert_2 + \frac{C}{\omega}\phi(S_T).
\end{align}
By recurrence on \eqref{eq:P.15}:
\begin{align}\label{eq:P.16}
    \mathcal{S}_N \leq \left(\frac{1}{2}\right)^{N-T}\mathcal{S}_T + 2\left[\lVert \Delta^{(T)}\rVert_2 +\frac{C}{\omega}\phi(S_T)\right],
\end{align}
which implies $\sup_N \mathcal{S}_N <\infty$. Hence, $\sum_{t=1}^{\infty}\lVert \Delta^{(t)}\rVert_2 <\infty$, implying that the sequence $\{\xi^{(t)}: t\geq 0\}$ is Cauchy and converges to a fixed point $\xi^{\star}$. Hence, Lemma \ref{lemma:critical-point-L} implies $\nabla_{\xi}\mathsf{L}(\xi^{\star}) = 0$, thereby completing the proof of Theorem~\ref{main-theorem:convergence} in the main manuscript.

\end{proof}

\subsection{Proof of Convergence Rate of the \texorpdfstring{$\tssg$}{TAVIE-SSG} EM Algorithm \texorpdfstring{(Theorem~\ref{main-theorem:convergence-rate})}{(Theorem 2)}}\label{app:TAVIE-convergence-rate}

\begin{proof}[Proof of Theorem~\ref{main-theorem:convergence-rate}]
Let $\Psi$ be the nonempty compact set of the cluster points of the sequence of $\tssg$ iterates $\{\xi^{(t)}: t\geq 0\}$ and $f := -\mathsf{L}$. In extension to the proof of Theorem~\ref{main-theorem:convergence} in Section \ref{app:TAVIE-convergence} (particularly the step corresponding to the convergence of $\{\xi^{(t)}: t\geq 0\}$ to a fixed point), using the K\L{} exponent property in Definition \ref{def:KL-property}, we assert the following Claim \ref{claim:uniformized-KL}.

\begin{claim}[Uniformized Kurdyka-\L{}ojasiewicz (K\L{}) neighborhood]
\label{claim:uniformized-KL}
There exists an open neighborhood $\mathcal{U}$ of $\Psi$, a $\rho \in (0, 1)$, a constant $c_{\mathrm{K\text{\L}}} > 0$, and an exponent $\Omega \in (0, 1)$ such that:
\begin{align}\label{eq:UKL-1}
    \lVert \nabla_{\psi} f(\psi)\rVert_{2} \geq c_{\mathrm{K\text{\L}}}\left(f(\psi) - \inf_{\Psi} f\right)^{\Omega},\text{ for all } \psi\in \mathcal{U}\text{ with }0<f(\psi) - \inf_{\Psi}f < \rho.
\end{align}
Moreover, there exists $T\in \mathbb{N}$ such that $\xi^{(t)} \in \mathcal{U}$ and $0 < f(\xi^{(t)}) - \inf_{\Psi} f < \rho$ for all $t\geq T$.
\end{claim}
\begin{proof}[Proof of Claim \ref{claim:uniformized-KL}]

Since the $\tssg$ iterates $\xi^{(t)}$ lie in the compact set $\overline{\mathsf{U}}_{\mathsf{L}(\xi^{(0)})}$ (see Lemma \ref{lemma-conv:1}), the sequence $\{\xi^{(t)}: t\geq 0\}$ is bounded, so $\Psi$ containing its cluster points is nonempty and compact. Because, $f = -\mathsf{L}$ is real-analytic in a neighborhood of $\Psi$, it is, in particular, continuous.

Let $\overline{f}:= \lim_{t\to \infty} f(\xi^{(t)})$, which exists as $f(\xi^{(t)})$ is nonincreasing and bounded below. If $\psi \in \Psi$, there exists a sub-sequence $\xi^{(t_j)} \to \psi$, hence by continuity $f(\psi) = \lim_{j\to \infty} f(\xi^{(t_j)}) = \overline{f}$. Therefore:
\begin{align}\label{eq:UKL-2}
    f\text{ is constant on $\Psi$}\quad\text{and}\quad \inf_{\Psi} f = \overline{f}.
\end{align}

\noindent
\emph{Pointwise K\L{} property at each $\overline{\psi} \in \Psi$}. Since $f$ is real-analytic, it satisfies the (pointwise) K\L{} property in Definition \ref{def:KL-property} at every $\overline{\psi}\in \Psi$~\citep{KL-SIAM-1}, with an open neighborhood $\mathcal{U}_{\overline{\psi}}$ of $\overline{\psi}$, constants $c_{\overline{\psi}} > 0$, $\Omega_{\overline{\psi}} \in (0, 1)$, and $\rho_{\overline{\psi}} \in (0, 1)$, such that for every $\psi \in \mathcal{U}_{\overline{\psi}}$ and $0< f(\psi) - f(\overline{\psi}) < \rho_{\overline{\psi}}$, \eqref{eq:KL-1} in Definition \ref{def:KL-property} holds. By \eqref{eq:UKL-2}, $f(\overline{\psi}) = \overline{f}$ for every $\overline{\psi} \in \Psi$, so:
\begin{align}\label{eq:UKL-3}
    \lVert \nabla_{\psi} f(\psi)\rVert_{2} \geq c_{\overline{\psi}}\left(f(\psi) - \overline{f}\right)^{\Omega_{\overline{\psi}}},\text{ for $\psi \in \mathcal{U}_{\overline{\psi}}, 0<f(\psi) - \overline{f}<\rho_{\overline{\psi}}$}.
\end{align}

\noindent
\emph{Finite sub-cover and uniform constants}. The family $\{\mathcal{U}_{\overline{\psi}}\}_{\overline{\psi}\in \Psi}$ is an open cover of the compact set $\Psi$, hence it admits a finite sub-cover:
\begin{align}\label{eq:UKL-4}
\Psi \subset \cup_{k=1}^{K}\mathcal{U}_{k},\text{ where }\mathcal{U}_{k}:= \mathcal{U}_{\overline{\psi}_{k}},\; c_{k}:= c_{\overline{\psi}_{k}}>0,\; \Omega_{k}:= \Omega_{\overline{\psi}_{k}}\in (0, 1),\; \rho_{k}:= \rho_{\overline{\psi}_{k}}\in (0, 1).
\end{align}
Now, set:
\begin{align}\label{eq:UKL-5}
c_{\mathrm{K\text{\L}}}:=\min_{k\in[K]} c_{k} > 0,\; \Omega:= \max_{k\in [K]}\Omega_{k} < 1,\; \rho:= \min\{1, \rho_1, \ldots, \rho_K\}\in (0, 1).
\end{align}
Let $\mathcal{U}:= \cup_{k=1}^{K}\mathcal{U}_{k}$, which is an open neighborhood of $\Psi$. Fix any $\psi \in \mathcal{U}$ with $0< f(\psi) - \overline{f} < \rho$. Then, $\psi \in \mathcal{U}_{k^{\star}}$ for some $k^{\star}$. Since, $0<f(\psi) - \overline{f}<\rho\leq \rho_{k^{\star}} \leq 1$ and $\Omega \geq \Omega_{k^{\star}}$, we have:
\begin{align}\label{eq:UKL-6}
    \left(f(\psi) - \overline{f}\right)^{\Omega_{k^{\star}}} \geq \left(f(\psi) - \overline{f}\right)^{\Omega},
\end{align}
and hence, by local inequality on $\mathcal{U}_{k^{\star}}$:
\begin{align}\label{eq:UKL-7}
    \lVert \nabla_{\psi}f(\psi)\rVert_{2} \geq c_{k^{\star}}\left(f(\psi) - \overline{f}\right)^{\Omega_{k^{\star}}} \geq c_{\mathrm{K\text{\L}}}\left(f(\psi) - \overline{f}\right)^{\Omega}.
\end{align}
This proves \eqref{eq:UKL-1} on $\mathcal{U}$ with the single pair $(c_{\mathrm{K\text{\L}}}, \Omega)$ and parameter $\rho$.

\noindent
\emph{Eventual containment in $\mathcal{U}$ and small gaps}. Let $\mathrm{dist}(x, A) := \inf_{y \in A} \lVert x - y \rVert_2$ where $x \in \mathbb{R}^n$ and $A \subset \mathbb{R}^n$ is compact. We show that $\mathrm{dist}(\xi^{(t)}, \Psi) \to 0$. Suppose not. Then, there exists $\epsilon > 0$ and a sub-sequence $\{\xi^{(t_j)}: j\geq 1\}$ with $\mathrm{dist}(\xi^{(t_j)}, \Psi) \geq \epsilon$, for all $j \geq 1$. By compactness of $\overline{\mathsf{U}}_{\mathsf{L}(\xi^{(0)})}$, $\{\xi^{(t_j)}: j\geq 1\}$ admits a further convergent sub-sequence $\xi^{(t_{j_l})} \to \psi$; necessarily $\psi\in \Psi$. But then, $\mathrm{dist}(\xi^{(t_{j_l})}, \Psi) \leq \lVert \xi^{(t_{j_l})} - \psi\rVert_2 \to 0$, which yields a contradiction. Hence:
\begin{align}\label{eq:UKL-8}
    \mathrm{dist}(\xi^{(t)}, \Psi) \to 0.
\end{align}
Since, $\mathcal{U}$ is an open neighborhood of $\Psi$, there exist points $\psi_1, \ldots, \psi_K \in \Psi$ and radii $\mathsf{r}_1, \ldots, \mathsf{r}_K > 0$ such that, the balls for $k\in [K]$, $B(\psi_k, \mathsf{r}_k) \subset \mathcal{U}$ cover $\Psi$. Replacing $\mathsf{r}_k$ by $\tfrac{\mathsf{r}_{k}}{2}$, define:
\begin{align}\label{eq:UKL-9}
\epsilon:= \frac{1}{2}\min_{k\in [K]}\mathsf{r}_k > 0.
\end{align}
Note that, the tubular neighborhood $\{\psi: \mathrm{dist}(\psi, \Psi) < \epsilon\}$ is contained in $\mathcal{U}$ as, if $\mathrm{dist}(\psi, \Psi) < \epsilon$, choose  $\tilde{\psi} \in \Psi$ with $\lVert \psi - \tilde{\psi}\rVert_2 < \epsilon$; then $\tilde{\psi} \in B\left(\psi_k, \tfrac{\mathsf{r}_k}{2}\right)$ for some $k$, and $\lVert \psi - \psi_k\rVert_2 \leq \lVert \psi - \tilde{\psi}\rVert_2 + \lVert \tilde{\psi} - \psi_k\rVert_2 < \epsilon + \tfrac{\mathsf{r}_k}{2} \leq \mathsf{r}_k$. Hence, $\psi \in B(\psi_k, \mathsf{r}_k)\subset \mathcal{U}$.

By \eqref{eq:UKL-8}, there exists $T_1\in \mathbb{N}$ such that $\mathrm{dist}(\xi^{(t)}, \Psi) < \epsilon$ for all $t\geq T_1$, hence $\xi^{(t)} \in \mathcal{U}$ for all $t\geq T_1$. Moreover, since $f(\xi^{(t)}) \downarrow \overline{f}$ and $\rho>0$, there exists $T_2\in \mathbb{N}$ such that $0\leq f(\xi^{(t)}) - \overline{f}< \rho$ for all $t\geq T_2$. Taking $T:= \max\{T_1, T_2\}$, yields:
\begin{align}\label{eq:UKL-10}
\xi^{(t)} \in \mathcal{U}\quad \text{and}\quad 0\leq f(\xi^{(t)}) - \overline{f}<\rho,\text{ for all $t\geq T$}.
\end{align}
This completes the proof of Claim \ref{claim:uniformized-KL}.
\end{proof}

\noindent
\emph{One-step recursion}.
Define the value gap $s_t:= f(\xi^{(t)}) - \overline{f}$, where $\overline{f}$ is as defined in the proof of Claim \ref{claim:uniformized-KL} above.
Following the proof of Theorem~\ref{main-theorem:convergence} of the main manuscript in Section \ref{app:TAVIE-convergence}, any $\psi^{\star}\in \Psi$ has strictly positive entries. Therefore, any sub-sequence $\xi^{(t_j)} \to \xi^{\star}$, eventually lies in $\left[\tfrac{\delta(\xi^{\star})}{2}, R\right]^{n}$, where $\delta(\xi^{\star}) = \min_{i}\xi_i^{\star}$ and constant $R > 0$. Hence, in line of the arguments presented in the proof of Theorem~\ref{main-theorem:convergence}, local $\omega$-strong concavity of the EM surrogate function $\mathcal{Q}$ holds inside $\left[\tfrac{\delta(\xi^{\star})}{2}, R\right]^{n}$, which implies the sufficient ascent and relative-error bounds in \eqref{eq:P.3} and \eqref{eq:P.4}, respectively.
Let $\mathcal{V}$ be an open neighborhood of $\xi^{\star}$ contained in $\left[\tfrac{\delta(\xi^{\star})}{2}, R\right]^{n}$ and define $\tilde{\mathcal{U}} = \mathcal{U} \cap \mathcal{V}$, on which the uniformized K\L{} inequality in Claim \ref{claim:uniformized-KL}, relative-error bound, and sufficient ascent property all hold.

Working in $\tilde{\mathcal{U}}$, the uniformized K\L{} inequality in Claim \ref{claim:uniformized-KL} at $\xi^{(t)}$ and the relative-error bound in \eqref{eq:P.4} yields:
\begin{align}\label{eq:UKL-13}
    \lVert \Delta^{(t)}\rVert_2:= \lVert \xi^{(t)} - \xi^{(t-1)}\rVert_2 \geq \frac{1}{C}\lVert \nabla_{\xi} f(\xi^{(t)})\rVert_2 \geq \frac{c_{\mathrm{K\text{\L}}}}{C}s_t^{\Omega}.
\end{align}
The sufficient ascent inequality in \eqref{eq:P.3} yields:
\begin{align}\label{eq:UKL-14}
    s_{t-1}-s_t = f(\xi^{(t-1)}) - f(\xi^{(t)}) \geq \frac{\omega}{2}\lVert\Delta^{(t)}\rVert_{2}^{2} \geq \frac{\omega}{2}\frac{c_{\mathrm{K\text{\L}}}^{2}}{C^2}s_t^{2\Omega}.
\end{align}
Therefore:
\begin{align}\label{eq:UKL-15}
    s_t \leq s_{t-1} - \frac{\omega}{2}\frac{c_{\mathrm{K\text{\L}}}^{2}}{C^2}s_t^{2\Omega},\quad t\geq T+1.
\end{align}

For brevity, write $\Lambda:=\frac{\omega}{2}\frac{c_{\text{K\L}}^2}{C^2}$, so that \eqref{eq:UKL-15} reads $s_t\le s_{t-1}-\Lambda s_t^{2\Omega}$ for $t\ge T+1$, and \eqref{eq:UKL-14} gives, for all $k\ge T$:
\begin{equation*}
\|\Delta^{(k+1)}\|_2\le\sqrt{\frac{2}{\omega}}\,(s_k-s_{k+1})^{1/2}\le\sqrt{\frac{2}{\omega}}\,s_k^{1/2}.
\end{equation*}
If $s_{t_0}=0$ for some $t_0\ge T$, then, since $\{s_t\}$ is non-increasing and nonnegative, $s_t=0$ for all $t\ge t_0$; by \eqref{eq:UKL-14}, $\Delta^{(t+1)}=0$ for all $t\ge t_0$, so $\xi^{(t)}=\xi^\star$ for all $t\ge t_0$ and every rate below holds trivially. Hence, assume $s_t>0$ for all $t\ge T$. Then \eqref{eq:UKL-15} gives $s_t<s_{t-1}$, and $\Delta^{(t)}\neq0$ for all $t\ge T+1$; otherwise $\xi^{(t-1)}$ would be a fixed point of $H$, so $\xi^{(k)}=\xi^{(t-1)}=\xi^\star$ for all $k\ge t-1$ and $s_{t-1}=0$.

Now, we derive the explicit convergence rates for $\Omega\in(0,1)$ by treating the cases $\Omega\in(0,\frac{1}{2})$, $\Omega=\frac{1}{2}$, and $\Omega\in(\frac{1}{2},1)$ separately.

\textit{Rate for $\Omega=\frac{1}{2}$.} From \eqref{eq:UKL-15}, $s_t\le s_{t-1}-\Lambda s_t$, i.e., $(1+\Lambda)s_t\le s_{t-1}$. Hence, $s_t\le(1+\Lambda)^{-(t-T)}s_T$ is geometric.
For the distance rate, $s_k\le(1+\Lambda)^{-(k-t)}s_t$ for all $k\ge t\ge T$, so that:
\begin{align*}
\|\xi^{(t)}-\xi^\star\|_2\le\sum_{k\ge t}\|\Delta^{(k+1)}\|_2
&\le\sqrt{\frac{2}{\omega}}\sum_{k\ge t}(s_k-s_{k+1})^{1/2}
\le\sqrt{\frac{2}{\omega}}\,s_t^{1/2}\sum_{j\ge0}(1+\Lambda)^{-j/2}\\
&=\frac{\sqrt{2/\omega}}{1-(1+\Lambda)^{-1/2}}\,s_t^{1/2}\lesssim\sqrt{s_t}.
\end{align*}

\textit{Rate for $\Omega\in(0,\frac{1}{2})$.} Since $2\Omega-1<0$ and $s_t\le s_T$ for all $t\ge T$, we have $s_t^{2\Omega-1}\ge M:=s_T^{2\Omega-1}>0$ for all $t\ge T$. From \eqref{eq:UKL-15}:
\begin{equation}
s_t\le s_{t-1}-\Lambda s_t^{2\Omega}\le s_{t-1}-\Lambda M s_t\implies s_t\le\frac{1}{1+\Lambda M}\,s_{t-1},
\label{eq:UKL-16}
\end{equation}
which is geometric. The distance rate $\|\xi^{(t)}-\xi^\star\|_2\lesssim\sqrt{s_t}$ follows exactly as for $\Omega=\frac{1}{2}$, with $1+\Lambda$ replaced by $1+\Lambda M$.

\textit{Rate for $\Omega\in(\frac{1}{2},1)$.} Let $\phi(u):=u^{1-2\Omega}$ (note $1-2\Omega<0$, so $\phi$ is decreasing).
By the mean-value theorem, for some $s_t^\star\in(s_t,s_{t-1})$, and since $u\mapsto u^{-2\Omega}$ is decreasing, so that $(s_t^\star)^{-2\Omega}>s_{t-1}^{-2\Omega}$:
\begin{equation}
\phi(s_t)-\phi(s_{t-1})=(2\Omega-1)(s_t^\star)^{-2\Omega}(s_{t-1}-s_t)\ge(2\Omega-1)\,s_{t-1}^{-2\Omega}(s_{t-1}-s_t).
\label{eq:UKL-17}
\end{equation}
Using \eqref{eq:UKL-15} in \eqref{eq:UKL-17}:
\begin{equation*}
\phi(s_t)-\phi(s_{t-1})\ge(2\Omega-1)\Lambda\left(\frac{s_t}{s_{t-1}}\right)^{2\Omega}.
\end{equation*}
Fix any $r\in(0,1)$. For each $t\ge T+1$, exactly one of the following holds:
\begin{itemize}
\item[(a)] $s_t\ge r\,s_{t-1}$. Then $(s_t/s_{t-1})^{2\Omega}\ge r^{2\Omega}$, and hence $\phi(s_t)-\phi(s_{t-1})\ge(2\Omega-1)\Lambda r^{2\Omega}$.
\item[(b)] $s_t<r\,s_{t-1}$. Since $\phi$ is decreasing, $\phi(s_t)>\phi(r s_{t-1})=r^{1-2\Omega}\phi(s_{t-1})$, and hence
$\phi(s_t)-\phi(s_{t-1})>(r^{1-2\Omega}-1)\,\phi(s_{t-1})\ge(r^{1-2\Omega}-1)\,\phi(s_T)$.
Here $r^{1-2\Omega}>1$, and $\phi(s_{t-1})\ge\phi(s_T)$ because $s_{t-1}\le s_T$.
\end{itemize}
Combining (a) and (b), for all $t\ge T+1$:
\begin{equation}
\phi(s_t)-\phi(s_{t-1})\ge\Upsilon:=\min\left\{(2\Omega-1)\Lambda r^{2\Omega},\;(r^{1-2\Omega}-1)\,s_T^{1-2\Omega}\right\}>0.
\label{eq:UKL-18}
\end{equation}
Summing \eqref{eq:UKL-18} from $T+1$ to $t$ (a telescoping sum of $t-T$ terms) gives:
\begin{equation}
s_t^{1-2\Omega}\ge s_T^{1-2\Omega}+\Upsilon\,(t-T),
\label{eq:UKL-19}
\end{equation}
and hence, since $-\frac{1}{2\Omega-1}<0$:
\begin{equation}
s_t\le\Big(s_T^{1-2\Omega}+\Upsilon(t-T)\Big)^{-\frac{1}{2\Omega-1}}\le\big(\Upsilon(t-T)\big)^{-\frac{1}{2\Omega-1}}=\mathcal{O}\Big((t-T)^{-\frac{1}{2\Omega-1}}\Big).
\label{eq:UKL-20}
\end{equation}
For the distance rate, let $\chi(s):=\frac{s^{1-\Omega}}{1-\Omega}$ for $s\ge0$, which is concave and increasing, with $\chi'(s)=s^{-\Omega}$ for $s>0$. Fix $t\ge T+1$.
Using, in turn, the concavity of $\chi$, \eqref{eq:UKL-13} in the form $s_t^{-\Omega}\ge c_{\text{K\L}}/(C\|\Delta^{(t)}\|_2)$, and \eqref{eq:UKL-14} in the form $s_t-s_{t+1}\ge\frac{\omega}{2}\|\Delta^{(t+1)}\|_2^2$:
\begin{equation*}
\chi(s_t)-\chi(s_{t+1})\ge s_t^{-\Omega}(s_t-s_{t+1})\ge\frac{\omega\,c_{\text{K\L}}}{2C}\,\frac{\|\Delta^{(t+1)}\|_2^2}{\|\Delta^{(t)}\|_2}.
\end{equation*}
Rearranging and using $2\sqrt{ab}\le a+b$ for $a,b\ge0$:
\begin{equation*}
2\|\Delta^{(t+1)}\|_2\le\|\Delta^{(t)}\|_2+\frac{2C}{\omega\,c_{\text{K\L}}}\big[\chi(s_t)-\chi(s_{t+1})\big].
\end{equation*}
Summing from $k=t$ to $N$, the bracketed terms telescope. Moreover,
$\sum_{k=t}^{N}\|\Delta^{(k)}\|_2\le\|\Delta^{(t)}\|_2+\sum_{k=t}^{N}\|\Delta^{(k+1)}\|_2$, so that:
\begin{equation*}
\sum_{k=t}^{N}\|\Delta^{(k+1)}\|_2\le\|\Delta^{(t)}\|_2+\frac{2C}{\omega\,c_{\text{K\L}}}\,\chi(s_t).
\end{equation*}
Letting $N\to\infty$, and using $\|\Delta^{(t)}\|_2\le\sqrt{2/\omega}\,s_{t-1}^{1/2}$:
\begin{equation*}
\|\xi^{(t)}-\xi^\star\|_2\le\sum_{k\ge t}\|\Delta^{(k+1)}\|_2\le\sqrt{\frac{2}{\omega}}\,s_{t-1}^{1/2}+\frac{2C}{\omega\,c_{\text{K\L}}(1-\Omega)}\,s_t^{1-\Omega}.
\end{equation*}
For $t\ge T+2$, we have $t-1-T\ge\frac{t-T}{2}$, so \eqref{eq:UKL-20} gives $s_{t-1}^{1/2}=\mathcal{O}\big((t-T)^{-\frac{1}{2(2\Omega-1)}}\big)$ and $s_t^{1-\Omega}=\mathcal{O}\big((t-T)^{-\frac{1-\Omega}{2\Omega-1}}\big)$.
Since $\Omega>\frac{1}{2}$ implies $\frac{1}{2}>1-\Omega$, the first term decays at least as fast as the second. Therefore, enlarging the constant to cover $t=T+1$:
\begin{equation}
\|\xi^{(t)}-\xi^\star\|_2\le\sum_{k\ge t}\|\Delta^{(k+1)}\|_2=\mathcal{O}\Big((t-T)^{-\frac{1-\Omega}{2\Omega-1}}\Big).
\label{eq:UKL-21}
\end{equation}
This completes the proof of Theorem \ref{main-theorem:convergence-rate} in the main manuscript.
\end{proof}

\newpage

\section{Variational Risk Bounds}\label{app:variational-risk-bounds}

\subsection{Lemmata for Variational Risk Bounds}\label{app:lemmata-variatonal-risk-bounds-general}

\begin{lemma}[Donsker and Varadhan's variational inequality]\label{lemma-variational-inequality}
Let $\mu$ be a probability measure and $h$ be a measurable function such that $e^{h}$ is integrable. Then:
\begin{equation}
    \log\int e^{h}d\mu = \sup_{\rho\ll\mu}\left[\int h d\rho - \mathrm{KL}(\rho\parallel \mu)\right].
\end{equation}
\end{lemma}

\begin{proof}
See~\cite{DonskerVaradhan1983IV}.
\end{proof}

\begin{lemma}
\label{lemma-auxiliary}
Let $x$ and $y$ be two continuous random vectors with joint density function $f(x, y)$. The maximum value of:
\begin{align}
    \int q(x)\log\left(\frac{f(x,y)}{q(x)}\right)dx,
\end{align}
over all density functions $q$ is attained at $q^{\star}(x) = f(x\mid y)$.
\end{lemma}

\begin{proof}
Observe that:
\begin{align}\label{eq:lemma2.1}
    \begin{split}
        \log f(y) = \int q(x)\log \left(\frac{f(x, y)}{q(x)}\right)dx + \mathrm{KL}(q\parallel f(\cdot\mid y)) \geq \int q(x)\log \left(\frac{f(x, y)}{q(x)}\right)dx,
    \end{split}
\end{align}
as $\mathrm{KL}(q\parallel f(\cdot\mid y))\geq 0$. Therefore, equality in \eqref{eq:lemma2.1} holds if and only if $q(x) = f(x\mid y)$.
\end{proof}

\begin{lemma}[Optimal $\tssg$ variational solution]\label{lemma-variational-optimizer}
Let $\mathcal{P}_{\Theta}$ be the set of densities supported on $\Theta$ and $\varphi_{\alpha}(y, \theta\mid \mathbf{X}, \xi) = \varphi_{\alpha}(y\mid \mathbf{X}, \theta, \xi)\pi(\theta) = \left\{\prod_{i\in [n]}\varphi(y_i\mid \mathbf{x}_i, \theta, \xi_i)\right\}^{\alpha}\pi(\theta)$, for any $\alpha \in (0,1]$. Then any maximizer $(q^\star, \xi^\star)$ of the objective function $\mathcal{L}(q, \xi):\mathcal{P}_{\Theta}\times \mathbb{R}^n_{+}\rightarrow \mathbb{R}$ defined as:
\begin{align}\label{eq:objective-function}
   \mathcal{L}(q, \xi) := \int_{\theta\in \Theta}\log\frac{\varphi_{\alpha}(y, \theta\mid \mathbf{X}, \xi)}{q(\theta)}q(\theta)d\theta,
\end{align}
satisfies:
\begin{align}\label{eq:lemma-optimal-variational-solution}
    q^\star(\theta) \equiv \pi_{\alpha}(\theta\mid \mathcal{D}_n, \xi^{\star}), \quad 
    \xi_i^{\star} = \sqrt{\kappa_i(\xi^\star)}, \quad i\in [n],
\end{align}
where $\pi_{\alpha}(\theta\mid \mathcal{D}_n, \xi)$ and $\kappa_{i}(\xi)$ are as given in Sections~\ref{main-subsec:prior-posterior} and~\ref{main-subsec:tavie-algo} of the main manuscript.
\end{lemma}
\begin{proof}
From \eqref{eq:objective-function}, we have:
\begin{align}\label{eq:optimal-proof-1}
    \mathcal{L}(q, \xi) := \int_{\theta\in \Theta}q(\theta)\log \varphi_{\alpha}(y, \theta\mid \mathbf{X}, \xi)d\theta - \int_{\theta\in \Theta}q(\theta)\log q(\theta)d\theta.
\end{align}
We want to maximize \eqref{eq:optimal-proof-1} jointly with respect to $(q, \xi) \in \mathcal{P}_{\Theta}\times \mathbb{R}_{+}^n$. Therefore, considering $q$ fixed, we set $\tfrac{\partial \mathcal{L}(q, \xi)}{\partial \xi}$ to $0$. Since, the second term on the right hand side of \eqref{eq:optimal-proof-1} is independent of $\xi$, our maximization problem equivalently amounts to:
\begin{align}\label{eq:optimal-proof-2}
    \frac{\partial}{\partial \xi}\mathbb{E}_{q}\left[\log \varphi_{\alpha}(y, \theta\mid \mathbf{X}, \xi)\right] = 0.
\end{align}
By using differentiation under the integral, from \eqref{eq:optimal-proof-2}, we have:
\begin{align}
    \mathbb{E}_{q}\left[\frac{\partial}{\partial\xi}\log \varphi_{\alpha}(y, \theta\mid \mathbf{X},\xi)\right] = 0.
\end{align}
Using Lemma \ref{lemma-auxiliary} above, $\mathcal{L}(q, \xi)$  in \eqref{eq:objective-function} can be maximized for a fixed $\xi$, which leads to the optimal variational family $q$ being the conditional distribution $\pi_{\alpha}(\theta\mid \mathcal{D}_n, \xi)$. Taking expectation in \eqref{eq:optimal-proof-2} with respect to this optimum results into:
\begin{align}\label{eq:optimal-proof-3}
    \mathbb{E}_{\pi_{\alpha}(\theta\mid \mathcal{D}_n, \xi)}\left[\frac{\partial}{\partial\xi}\log \varphi_{\alpha}(y, \theta \mid \mathbf{X}, \xi)\right] = 0.
\end{align}
In order to show that, the solution of \eqref{eq:optimal-proof-3} satisfies the fixed-point update in Algorithm~\ref{main-alg:tavie-em} of the main manuscript, we use the first-order stationarity condition for maximizing the EM surrogate function $\mathcal{Q}(\xi^{(t+1)}\mid \xi^{(t)})$ in~\eqref{main-eq:EM-surrogate} of the main manuscript with respect to $\xi^{(t+1)}$, given by:
\begin{align}\label{eq:optimal-proof-4}
\frac{\partial}{\partial\xi^{(t+1)}}\mathcal{Q}(\xi^{(t+1)}\mid \xi^{(t)}) = \mathbb{E}_{\pi_{\alpha}(\theta\mid \mathcal{D}_n, \xi^{(t)})}\left[\frac{\partial}{\partial\xi^{(t+1)}}\log\varphi_{\alpha}(y, \theta\mid \mathbf{X},\xi^{(t+1)})\right] = 0.
\end{align}

Hence, we show that the solution of \eqref{eq:optimal-proof-3} satisfies the fixed-point update of Algorithm~\ref{main-alg:tavie-em} in the main manuscript, thus completing the proof.
\end{proof}

\begin{lemma}[Gaussian translation bounds \citep{anderson, Ball1993-vo}]
\label{lem:gaussian-translation}
Let $\Phi_p$ denote the standard Gaussian measure on $\mathbb{R}^p$:
\begin{align}
\Phi_p(E) \;=\; (2\pi)^{-p/2}\int_E \exp\left\{-\frac{\|x\|_2^2}{2}\right\}\,dx,
\end{align}
and let $A\subseteq\mathbb{R}^p$ be a convex set that is symmetric about the origin (i.e., $A=-A$).
Then for every $m\in\mathbb{R}^p$:
\begin{align}
\exp\left\{-\frac{\|m\|_2^2}{2}\right\}\,\Phi_p(A) \;\le\; \Phi_p(A+m)\;\le\; \Phi_p(A),
\end{align}
where $A+m:=\{x+m:x\in A\}$. 
\end{lemma}

\begin{proof}
This is a combination of Anderson’s inequality~\citep{anderson} for the upper bound and a quantitative refinement for the lower bound obtained via $\log$-concavity~\citep{Ball1993-vo}.
\end{proof}

\begin{lemma}[Prior concentration bound for Type I $\mathsf{SSG}$ likelihoods]\label{lemma:prior-inequality}
Let $(\beta, \tau^2) \sim \mathcal{NG}_p(\mu, \Sigma, a, b)$, $\beta_0 \in \mathbb{R}^p$, and $\tau_0^2 > 0$. Then for $c_1, c_2 \in \mathbb{R}^{+}$:
\begin{align}\label{eq:lemma-anderson-normal-gamma}
\begin{split}
&\pi(\lVert\beta - \beta_0\rVert_2 \leq c_1, |\tau^2 - \tau_0^2| \leq c_2)\\
&\geq C(a,b,\tau_0)c_2
\frac{2^{-\frac{p}{2}}}{\Gamma\left(\frac{p}{2}+1\right)}\left(\sqrt{\frac{t_{\min}}{\lambda_{\mathrm{max}}(\Sigma)}}c_1\right)^{p}
\exp\left\{-\frac{t_{\max}}{2}\left[\Delta^2 + \frac{c_1^2}{\lambda_{\mathrm{max}}(\Sigma)}\right]\right\},
\end{split}
\end{align}
where $t_{\max} = \tau_0^2 + c_2$, $t_{\min} = \max\{\tau_0^2 - c_2, 0\}$, $\Delta^{2} = \lVert \Sigma^{-\frac{1}{2}}(\beta_0-\mu)\rVert_{2}^{2}$, and $C(a,b,\tau_0) = \frac{b}{2\Gamma\left(\frac{a}{2}\right)}\left(\frac{b \tau_0^2}{2}\right)^{\frac{a}{2}-1}\exp\left\{-\frac{b\tau_0^2}{2}\right\}$.
\end{lemma}

\begin{proof}
We have:
\begin{align}\label{eq:proof-lemma-anderson-ng-1}
    \begin{split}
        &\pi\left(\lVert \beta-\beta_0\rVert_2 \leq c_1, |\tau^2-\tau^2_0| \leq c_2\right)\\
        &\geq \inf_{t\in [t_{\min}, t_{\max}]}\pi\left(\lVert \beta-\beta_0\rVert_{2}\leq c_1 \Big|\tau^{2}=t\right)
        \pi\left(|\tau^2-\tau_0^2| \leq c_2\right).
    \end{split}
\end{align}
We now bound each term on the right hand side of \eqref{eq:proof-lemma-anderson-ng-1}. First, using Lemma \ref{lem:gaussian-translation} with $A := \{\beta: \lVert \beta-\mu\rVert_2 \leq c_1\}$ and $m := \beta_0-\mu$, for any $t > 0$ we get:
\begin{align}\label{eq:Ball-application}
    \pi\left(\lVert \beta-\beta_0\rVert_2 \leq c_1\mid \tau^2=t\right) \geq \exp\left\{-\frac{t}{2}\Delta^2\right\}\pi\left(\lVert \beta-\mu\rVert_2\leq c_1\mid \tau^2=t\right).
\end{align}
In \eqref{eq:Ball-application} note that:
\begin{align}\label{eq:Ball-application-1}
\begin{split}
    \pi\left(\lVert \beta-\mu\rVert_2\leq c_1\mid \tau^2=t\right) &\geq \pi\left(\lVert Z\rVert_2 \leq \sqrt{\frac{t}{\lambda_{\mathrm{max}}(\Sigma)}} c_1\right)
    = \frac{\widehat{\Gamma}\left(\frac{p}{2}, \frac{t c_1^2}{2\lambda_{\mathrm{max}}(\Sigma)}\right)}{\Gamma\left(\frac{p}{2}\right)}\\
    &\geq \frac{2^{-\frac{p}{2}}}{\Gamma\left(\frac{p}{2}+1\right)} \left(\sqrt{\frac{t}{\lambda_{\mathrm{max}}(\Sigma)}}c_1\right)^{p}\exp\left\{-\frac{1}{2}\frac{t c_1^2}{\lambda_{\mathrm{max}}(\Sigma)}\right\},
\end{split}
\end{align}
where $Z = t^{\frac{1}{2}}\Sigma^{-\frac{1}{2}}(\beta-\mu) \sim \mathcal{N}_p(0, I_p)$ and $\widehat{\Gamma}(a,s) = \int_{0}^{s} e^{-x}x^{a-1}dx$ is the lower incomplete gamma function. Combining \eqref{eq:Ball-application} and \eqref{eq:Ball-application-1}, we finally obtain:
\begin{align}\label{eq:Ball-application-2}
    \begin{split}
        &\inf_{t \in [t_{\min}, t_{\max}]}\pi\left(\lVert \beta-\beta_0\rVert_2\leq c_1\mid \tau^2=t\right) \\
        &\geq \frac{2^{-\frac{p}{2}}}{\Gamma\left(\frac{p}{2}+1\right)}\left(\sqrt{\frac{t_{\min}}{\lambda_{\mathrm{max}}(\Sigma)}}c_1\right)^{p}\exp\left\{-\frac{t_{\max}}{2}\left[\Delta^2 + \frac{c_1^2}{\lambda_{\mathrm{max}}(\Sigma)}\right]\right\}.
    \end{split}
\end{align}

\textbf{Note:} Taking $t = t_{\max} = t_{\min} = 1$ in \eqref{eq:Ball-application-2}, yields a lower bound of $\pi\left(\lVert \beta-\beta_0\rVert_2\leq c_1\right)$ for $\beta \sim \mathcal{N}_p(\mu, \Sigma)$, which is used for prior concentration bound under Type II $\mathsf{SSG}$ likelihoods.

Now, we provide a bound for $\pi\left(|\tau^2-\tau_0^2| \leq c_2\right)$. Observe that:
\begin{align}\label{eq:proof-lemma-anderson-ng-2}
\begin{split}
\pi\left(|\tau^2-\tau_0^2| \leq c_2\right) 
&= \mathbb{P}(t_{\min} < \tau^2 < t_{\max}), \quad \tau^2\sim \mathcal{G}\left(\frac{a}{2}, \frac{b}{2}\right).
\end{split}
\end{align}
By mean-value theorem, for some $t^{\star} \in [t_{\min}, t_{\max}]$:
\begin{align}\label{eq:integrand-bound}
    \mathbb{P}(t_{\min} < \tau^2 < t_{\max}) &\geq \frac{b}{2\Gamma\left(\frac{a}{2}\right)} (t_{\max} - t_{\min}) \left(\frac{b t^{\star}}{2}\right)^{\frac{a}{2}-1}\exp\left\{-\frac{b t^{\star}}{2}\right\}.
\end{align}
Hence, a lower bound follows from bounding below the right hand side of the display above in \eqref{eq:integrand-bound} over $[t_{\min}, t_{\max}]$:
\begin{align}
    \mathbb{P}(t_{\min} < \tau^2 < t_{\max}) \geq \frac{b}{2\Gamma\left(\frac{a}{2}\right)}(t_{\max} - t_{\min}) \min_{t\in [t_{\min}, t_{\max}]}\left\{\left(\frac{bt}{2}\right)^{\frac{a}{2}-1}\exp\left\{-\frac{bt}{2}\right\}\right\}.
\end{align}
We have:
\begin{align}
    t_{\max} - t_{\min} := \begin{cases}
        2 c_2, & \text{if }\tau_0^2 \geq c_2\\
        \tau_0^2 + c_2, & \text{otherwise},
    \end{cases}
\end{align}
where $c_2\in \mathbb{R}^{+}$ and since the Gamma probability density function is unimodal at $t^{\dagger}:= \frac{a-2}{b}$, for $a>2$:
\begin{align}
    \mathbb{P}(t_{\min} < \tau^2 < t_{\max}) \geq \frac{b \Delta_t}{2\Gamma\left(\frac{a}{2}\right)}\min_{s\in [t_{\min}, t_{\max}]}\left\{\left(\frac{bs}{2}\right)^{\frac{a}{2}-1}\exp\left\{-\frac{bs}{2}\right\}\right\},\quad \Delta_t:= t_{\max} - t_{\min}.
\end{align}
For small $c_2$:
\begin{align}
    \Delta_t \asymp 2c_2,
\end{align}
and since $t_{\min}, t_{\max} \approx \tau_0^2$:
\begin{align}\label{eq:final-integral-bound}
    \mathbb{P}(t_{\min} < \tau^2 < t_{\max}) \gtrsim C(a, b, \tau_0)c_2,
\end{align}
where $C(a, b, \tau_0):= \frac{b}{2\Gamma\left(\frac{a}{2}\right)}\left(\frac{b \tau_0^2}{2}\right)^{\frac{a}{2}-1}\exp\left\{-\frac{b\tau_0^2}{2}\right\}$.
Using \eqref{eq:Ball-application-2} and \eqref{eq:final-integral-bound} in \eqref{eq:proof-lemma-anderson-ng-1}, we prove the assertion in Lemma \ref{lemma:prior-inequality}.
\end{proof}

\begin{lemma}[Majorization of the variational risk under $\alpha$-R\'{e}nyi divergence]\label{lemma:majorization-Renyi}
Fix $\varepsilon \in (0, 1)$ and let $D>1$ be an arbitrary constant. The variational risk under $\alpha$-R\'{e}nyi divergence for any $\mathsf{SSG}$ likelihood satisfies the following bound with $\mathbb{P}_{\theta_0}$-probability at least $1-\varepsilon - [(D-1)^2n\varepsilon^2]^{-1}$:
\begin{align}\label{eq:lemma-majorization-renyi-1}
        n(1-\alpha)\int_{\theta\in \Theta}D_{\alpha}(\theta, \theta_{0})\pi_{\alpha}(\theta \mid \mathcal{D}_n, \xi^{\star})d\theta
        \leq Dn\alpha\varepsilon^2 - \log \pi(\mathcal{B}_n(\theta_0, \varepsilon)) + \log\left(\frac{1}{\varepsilon}\right),
\end{align}
where for arbitrary $\tilde{\xi} \in \mathbb{R}_{+}^{n}$:
\begin{align}\label{eq:lemma-majorization-renyi-2}
    \begin{split}
        \mathcal{B}_{n}(\theta_0, \varepsilon) &:= \left\{\theta \in \Theta: \tilde{D}\left(p(\cdot\mid \mathbf{X}, \theta_{0}) \parallel \varphi(\cdot\mid \mathbf{X}, \theta, \tilde{\xi})\right)\leq n\varepsilon^2,
        \right.\\&\left.\qquad 
        {V}\left(p(\cdot\mid \mathbf{X}, \theta_0) \parallel \varphi(\cdot\mid \mathbf{X}, \theta, \tilde{\xi})\right)\leq n\varepsilon^2\right\},
    \end{split}
\end{align}
with $\tilde{D}(f\parallel g) := \int f|\log(f/g)|$ and ${V}(f\parallel g) := \int f\log^{2}(f/g) - \tilde{D}^{2}(f\parallel g)$, for positive functions $f$ and $g$.
\end{lemma}

\begin{proof}
From the definition of $\alpha$-R\'{e}nyi divergence in~\eqref{main-eq:alpha-renyi-divergence} of the main manuscript and using the fact that $\varphi(y\mid \mathbf{X}, \theta, \xi)$ lower bounds $p(y\mid \mathbf{X}, \theta)$ along with $\theta_0$ being the true value of the parameter $\theta$, we get:
\begin{align}\label{eq:proof-lemma-majorization-renyi-1}
    \begin{split}
        \mathbb{E}_{\theta_0}\left[\exp\left\{\alpha\log\frac{\varphi(y\mid \mathbf{X},\theta, \xi)}{p(y\mid \mathbf{X}, \theta_0)}\right\}\right] &\leq \mathbb{E}_{\theta_0}\left[\exp\left\{\alpha\log\frac{p(y\mid \mathbf{X}, \theta)}{p(y\mid \mathbf{X}, \theta_0)}\right\}\right]\\
        &= \exp\left\{-n(1-\alpha)D_{\alpha}(\theta, \theta_0)\right\},
    \end{split}
\end{align}
where $\mathbb{E}_{\theta_0}$ is the expectation under $p(y\mid \mathbf{X}, \theta_0)$. Thus, for any $\varepsilon\in (0,1)$:
\begin{align}\label{eq:proof-lemma-majorization-renyi-2}
    \begin{split}
        \mathbb{E}_{\theta_0}\left[\exp\left\{\alpha\log\frac{\varphi(y\mid \mathbf{X}, \theta, \xi)}{p(y\mid \mathbf{X}, \theta_0)} + n(1-\alpha)D_{\alpha}(\theta, \theta_0) - \log\left(\frac{1}{\varepsilon}\right)\right\}\right] \leq \varepsilon.
    \end{split}
\end{align}
Integrating both sides of \eqref{eq:proof-lemma-majorization-renyi-2} above with respect to the prior $\pi(\theta)$ and a consequent application of Fubini's theorem yields:
\begin{align}\label{eq:proof-lemma-majorization-renyi-3}
    \mathbb{E}_{\theta_0}\left[\int \exp\left\{\alpha \log\frac{\varphi(y\mid \mathbf{X}, \theta, \xi)}{p(y\mid \mathbf{X}, \theta_0)} + n(1-\alpha)D_{\alpha}(\theta, \theta_0) - \log\left(\frac{1}{\varepsilon}\right)\right\}\pi(\theta)d\theta\right] \leq \varepsilon.
\end{align}
Using Donsker and Varadhan's variational inequality in Lemma \ref{lemma-variational-inequality} above, we have:
\begin{align}\label{eq:proof-lemma-majorization-renyi-4}
    \begin{split}
        &\mathbb{E}_{\theta_0}\Big[\sup_{q\ll \pi}\Big(\int_{\theta\in \Theta}\left\{\alpha\log\frac{\varphi(y\mid\mathbf{X}, \theta, \xi)}{p(y\mid \mathbf{X}, \theta_0)} + n(1-\alpha)D_{\alpha}(\theta, \theta_0) - \log\left(\frac{1}{\varepsilon}\right)\right\}q(\theta)d\theta\\
        &\qquad\qquad\qquad\qquad\qquad\qquad\qquad\qquad\qquad\qquad\qquad\qquad\qquad-\mathrm{KL}(q\parallel \pi)\Big)\Big]\leq \varepsilon,
    \end{split}
\end{align}
where $\pi$ represents the prior distribution over the parameters in $\theta$. Choosing $\rho$ as the optimal variational solution, i.e., $\rho = q^\star \equiv \pi_{\alpha}(\theta \mid \mathcal{D}_n, \xi^\star)$ and setting $\xi = \xi^\star$, we obtain:
\begin{align}\label{eq:proof-lemma-majorization-renyi-5}
    \begin{split}
        &\mathbb{E}_{\theta_{0}}\Big[\int_{\theta\in \Theta }\left\{\alpha\log\frac{\varphi(y\mid \mathbf{X}, \theta, \xi^\star)}{p(y\mid \mathbf{X}, \theta_{0})} + n(1-\alpha)D_{\alpha}(\theta, \theta_{0}) - \log\left(\frac{1}{\varepsilon}\right)\right\}q^{\star}(\theta)d\theta\\
        &\qquad\qquad\qquad\qquad\qquad\qquad\qquad\qquad\qquad\qquad\qquad\qquad\qquad- \mathrm{KL}(q^\star\parallel \pi)\Big] \leq \varepsilon.
    \end{split}
\end{align}
With the application of Markov's inequality followed by Lemma \ref{lemma-variational-optimizer}, we further obtain with $\mathbb{P}_{\theta_{0}}$-probability at least $(1-\varepsilon)$:
\begin{align}\label{eq:proof-lemma-majorization-renyi-6}
    \begin{split}
        &n(1-\alpha)\int_{\theta\in \Theta}D_{\alpha}(\theta, \theta_{0})q^{\star}(\theta)d\theta\\
        &\leq -\alpha \int_{\theta\in \Theta}\log \frac{\varphi(y\mid \mathbf{X}, \theta, \xi^{\star})}{p(y\mid \mathbf{X}, \theta_{0})}q^{\star}(\theta)d\theta + \mathrm{KL}(q^{\star}\parallel \pi)
        - \log\left(\varepsilon\right)\\
        &= \inf_{q, \xi}\left\{-\alpha\int_{\theta\in \Theta }\log\frac{\varphi(y\mid \mathbf{X}, \theta, \xi)}{p(y\mid \mathbf{X}, \theta_{0})}q(\theta)d\theta + \mathrm{KL}(q\parallel \pi)\right\}
        - \log\left(\varepsilon\right).
    \end{split}
\end{align}
Finally, defining $\mathcal{I}(\tilde{q},\tilde{\xi}) = -\int_{\theta\in \Theta }\log\frac{\varphi(y\mid \mathbf{X}, \theta, \tilde{\xi})}{p(y\mid \mathbf{X}, \theta_{0})}\tilde{q}(\theta)d\theta$, we conclude that the variational risk under $\alpha$-R\'{e}nyi divergence satisfies the following bound for any choices of $\tilde{q}$ and $\tilde{\xi}$:
\begin{align}\label{eq:proof-lemma-majorization-renyi-7}
n(1-\alpha)\int_{\theta\in \Theta}D_{\alpha}(\theta, \theta_{0})q^{\star}(\theta)d\theta
\leq \alpha \mathcal{I}(\tilde{q}, \tilde{\xi})+ \mathrm{KL}(\tilde{q} \parallel \pi) - \log\left(\varepsilon\right).
\end{align}
%
We now choose $\tilde{q}$ as the following (with $\tilde{\xi}$ being arbitrary but fixed):
%
\begin{align}\label{eq:proof-lemma-majorization-renyi-8}
    \begin{split}
        \tilde{q}(\theta) := \frac{\pi(\theta)}{\pi\left(\mathcal{B}_n(\theta_0, \varepsilon)\right)}\mathds{1}_{\mathcal{B}_n(\theta_0, \varepsilon)}(\theta),\quad \text{for all}\hspace{1mm}\theta\in \Theta.
    \end{split}
\end{align}
The choice of $\tilde{q}$ above in \eqref{eq:proof-lemma-majorization-renyi-8} is essentially the restriction of the prior $\pi$ into the KL neighborhood $\mathcal{B}_n(\theta_0, \varepsilon)$ around $\theta_0$ with radius $\varepsilon$, which is defined as in \eqref{eq:lemma-majorization-renyi-2}. 
Substituting $\tilde{q}(\theta)$ in \eqref{eq:proof-lemma-majorization-renyi-7} makes the second term the negative $\log$-prior mass:
\begin{align}\label{eq:proof-lemma-majorization-renyi-9}
    \mathrm{KL}(\tilde{q} \parallel \pi) = -\log(\pi(\mathcal{B}_n(\theta_0, \varepsilon))).
\end{align}
%

We now obtain a high-probability upper bound for $\mathcal{I}(\tilde{q}, \tilde{\xi})$ in \eqref{eq:proof-lemma-majorization-renyi-7}. Using Fubini's theorem, followed by the definition of $\mathcal{B}_n(\theta_0, \varepsilon)$ in \eqref{eq:lemma-majorization-renyi-2} above, we get:
\begin{align}\label{eq:proof-lemma-majorization-renyi-10}
\begin{split}
    \mathbb{E}_{\theta_0}\left[\int_{\theta\in \Theta}\tilde{q}(\theta)\log\frac{\varphi(y\mid \mathbf{X}, \theta, \tilde{\xi})}{p(y\mid \mathbf{X}, \theta_0)}d\theta\right] &= 
    \int_{\theta\in \Theta} \mathbb{E}_{\theta_{0}}\left[\log\frac{\varphi(y \mid \mathbf{X}, \theta, \tilde{\xi})}{p(y\mid \mathbf{X}, \theta_{0})}\right]\tilde{q}(\theta)d\theta \\
    &\geq -\int_{\mathcal{B}_{n}(\theta_{0}, \varepsilon)}\tilde{D}\left(p(\cdot\mid \mathbf{X},\theta_0)\parallel \varphi(\cdot\mid \mathbf{X}, \theta, \tilde{\xi})\right)\tilde{q}(\theta)d\theta\\
    &\geq -n\varepsilon^2,
\end{split}
\end{align}
where the first inequality uses $\mathbb{E}_{\theta_0}\left[\log\frac{\varphi(y\mid \mathbf{X}, \theta, \tilde{\xi})}{p(y\mid \mathbf{X}, \theta_0)}\right] = -\int p(y\mid \mathbf{X}, \theta_0)\log\frac{p(y\mid \mathbf{X}, \theta_0)}{\varphi(y\mid \mathbf{X}, \theta, \tilde{\xi})}dy \geq -\tilde{D}\left(p(\cdot\mid \mathbf{X},\theta_0)\parallel \varphi(\cdot\mid \mathbf{X}, \theta, \tilde{\xi})\right)$, and the second uses that $\tilde{q}$ is supported on $\mathcal{B}_n(\theta_0,\varepsilon)$. Now, using Cauchy-Schwarz inequality, we bound the second moment as:
\begin{align}\label{eq:proof-lemma-majorization-renyi-11}
    \begin{split}
        &\text{Var}_{\theta_{0}}\left[\int_{\theta\in \Theta} \tilde{q}(\theta)\log\frac{\varphi(y\mid \mathbf{X}, \theta, \tilde{\xi})}{p(y\mid \mathbf{X}, \theta_0)}d\theta\right]\\
        &\leq \int_{\mathcal{B}_{n}(\theta_{0}, \varepsilon)}V\left(p(.\mid \mathbf{X}, \theta_{0})\parallel \varphi(.\mid \mathbf{X}, \theta, \tilde{\xi})\right)\tilde{q}(\theta)d\theta \\&\leq n\varepsilon^2.
    \end{split}
\end{align}
For some arbitrary constant $D>1$ and using \eqref{eq:proof-lemma-majorization-renyi-10} and \eqref{eq:proof-lemma-majorization-renyi-11}, along with the application of Chebyshev's inequality, we have:
\begin{align}\label{eq:proof-lemma-majorization-renyi-12}
    \begin{split}
        &\mathbb{P}_{\theta_{0}}\left\{\int_{\theta\in \Theta} \tilde{q}(\theta)\log\frac{\varphi(y\mid \mathbf{X}, \theta, \tilde{\xi})}{p(y\mid \mathbf{X}, \theta_{0})}d\theta \leq -Dn\varepsilon^2\right\}\\
        &\leq \mathbb{P}_{\theta_{0}}\left\{\int_{\theta\in \Theta} \tilde{q}(\theta)\log\frac{\varphi(y \mid \mathbf{X}, \theta, \tilde{\xi})}{p(y \mid \mathbf{X}, \theta_{0})}d\theta\right.\\
        &\left.\qquad \qquad - \mathbb{E}_{\theta_{0}}\left[\int_{\theta\in \Theta} \tilde{q}(\theta)\log\frac{\varphi(y \mid \mathbf{X}, \theta, \tilde{\xi})}{p(y\mid \mathbf{X}, \theta_{0})}d\theta\right] \leq -(D-1)n\varepsilon^2\right\}\\
        &\leq \frac{\text{Var}_{\theta_{0}} \left[\int_{\theta\in \Theta} \tilde{q}(\theta)\log\frac{\varphi(y\mid \mathbf{X}, \theta, \tilde{\xi})}{p(y\mid \mathbf{X}, \theta_{0})}d\theta\right]}{(D-1)^{2}n^2\varepsilon^4}\\
        &\leq \frac{1}{(D-1)^{2}n\varepsilon^2}.
    \end{split}
\end{align}
From \eqref{eq:proof-lemma-majorization-renyi-12}, with $\mathbb{P}_{\theta_0}$-probability at least $1 - [(D-1)^2n\varepsilon^2]^{-1}$, $\mathcal{I}(\tilde{q} , \tilde{\xi})$ in \eqref{eq:proof-lemma-majorization-renyi-7} 
satisfies the following inequality:
\begin{align}\label{eq:proof-lemma-majorization-renyi-13}
    \mathcal{I}(\tilde{q} , \tilde{\xi}) = -\int_{\theta\in \Theta}\tilde{q}(\theta)\log\frac{\varphi(y\mid \mathbf{X}, \theta, \tilde{\xi})}{p(y\mid \mathbf{X}, \theta_0)}d\theta \leq Dn\varepsilon^2.
\end{align}
Using \eqref{eq:proof-lemma-majorization-renyi-9} and \eqref{eq:proof-lemma-majorization-renyi-13} in \eqref{eq:proof-lemma-majorization-renyi-7}, we obtain the desired bound.
\end{proof}

\begin{lemma}[Statistical identifiability;~\cite{ghosal2007convergence}]
\label{lemma:hell-risk-assumption}
For some $\varepsilon_n > 0$, any $\varepsilon \geq \varepsilon_n$, and $\pi$ being the prior measure, there exists a sieve set $\mathcal{F}_{n, \varepsilon} \subset \Theta$, where $\Theta$ is compact, and a test function $\phi_{n, \varepsilon}:\mathbb{R}^{n} \to [0, 1]$ such that, for some $c\in \mathbb{R}^{+}$:
\begin{gather}\label{eq:hell-risk-assumption-1}
\pi(\mathcal{F}_{n, \varepsilon}^{c}) \leq \exp\left\{-cn\varepsilon^2\right\},\\
\label{eq:hell-risk-assumption-2}
\mathbb{E}_{\theta_0}\left[\phi_{n, \varepsilon}\right] \leq \exp\left\{-cn \varepsilon_n^2\right\},\\
\label{eq:hell-risk-assumption-3}
\mathbb{E}_{\theta}\left[1 - \phi_{n, \varepsilon}\right] \leq \exp\left\{-c n \mathcal{H}^{2}(\theta \parallel \theta_0)\right\},
\end{gather}
for all $\theta\in \mathcal{F}_{n, \varepsilon}$ satisfying $\mathcal{H}^{2}(\theta\parallel\theta_0) \geq \varepsilon^2$, where $\mathcal{H}^{2}(\theta\parallel\theta_0)$ is the squared Hellinger distance between $p(y\mid \mathbf{X}, \theta)$ and $p(y\mid \mathbf{X}, \theta_0)$ as in~\eqref{main-eq:Hellinger-divergence} of the main manuscript.
\end{lemma}

\subsection{Proof of Variational Risk Bound under \texorpdfstring{$\alpha$-R\'{e}nyi}{alpha-Renyi} Divergence for Type I \texorpdfstring{$\mathsf{SSG}$}{SSG} Likelihoods (\texorpdfstring{Theorem~\ref{main-theorem-alpha-less-1-variational-risk-bound-Type-I}}{Theorem 3})}\label{app:proof-theorem-vb-alpha-less-1-type-1}

\begin{proof}[Proof of Theorem~\ref{main-theorem-alpha-less-1-variational-risk-bound-Type-I}]
We present the proof of Theorem~\ref{main-theorem-alpha-less-1-variational-risk-bound-Type-I} of the main manuscript using the following steps. Throughout, we consider our working likelihood under the Type I $\mathsf{SSG}$ model as:
$$\varphi_{\alpha}(y\mid \mathbf{X}, \theta, \xi) \propto \tau^{n\alpha}\exp\left\{\alpha\sum_{i \in [n]}\left[A(\xi_i)[\tau^2(y_i - \mathbf{x}_i^{\top}\beta)^{2}] + \gamma(\xi_i)\right]\right\}.$$

By Lemma \ref{lemma:majorization-Renyi}, the variational risk under $\alpha$-R\'{e}nyi divergence satisfies:
\begin{align}\label{eq:proof-theorem1-7}
        n(1-\alpha)\int_{\theta\in \Theta}D_{\alpha}(\theta, \theta_{0})\pi_{\alpha}(\theta \mid \mathcal{D}_n, \xi^{\star})d\theta
        \leq Dn\alpha\varepsilon^2 - \log \pi(\mathcal{B}_n(\theta_0, \varepsilon)) + \log\left(\frac{1}{\varepsilon}\right),
\end{align}
with $\mathbb{P}_{\theta_0}$-probability at least $1-\varepsilon-[(D-1)^2n\varepsilon^2]^{-1}$. Note that, $-\log \pi(\mathcal{B}_n(\theta_0, \varepsilon))$ in \eqref{eq:proof-theorem1-7} above, is the local Bayesian complexity (see~\cite{Bayesian-fractional-posterior}).

We now obtain a high-probability upper bound for the local Bayesian complexity in \eqref{eq:proof-theorem1-7}. To start, we obtain an upper bound for the $\log$-pseudo-likelihood ratio denoted by $\Delta(\theta, \theta_0)$:
\begin{align}\label{eq:proof-theorem1-16}
    \begin{split}
        \Delta(\theta, \theta_0) &:= \log \frac{\varphi(y\mid \mathbf{X}, \theta, \tilde{\xi})}{p(y\mid \mathbf{X}, \theta_0)}\\
        &= \left(\log p(y\mid \mathbf{X}, \theta) - \log p(y\mid \mathbf{X}, \theta_0)\right) + \left(\log \varphi(y\mid \mathbf{X}, \theta, \tilde{\xi}) - \log p(y\mid \mathbf{X}, \theta)\right)\\
        &= \Delta_1 +\widetilde \Delta_2.
    \end{split}
\end{align}
Consider $\Delta_1$ in \eqref{eq:proof-theorem1-16}. For any $\theta$ with $|\tau^2-\tau_0^2|\leq \tau_0^2/2$:
\begin{align}\label{eq:proof-theorem1-17}
    \begin{split}
        \Delta_1 &:= \log p(y\mid \mathbf{X}, \theta) - \log p(y\mid \mathbf{X}, \theta_0)\\
        &= \sum_{i\in [n]}\left[h\left(\tau^2(y_i - \mathbf{x}_i^{\top}\beta)^2\right) - h\left(\tau_0^2(y_i-\mathbf{x}_i^{\top}\beta_0)^2\right)\right] + n \log(\tau/\tau_0),\\
        |\Delta_1| &\stackrel{\star}{\leq} K\sum_{i\in [n]}\Big|\tau(y_i - \mathbf{x}_i^{\top}\beta) - \tau_0(y_i-\mathbf{x}_i^{\top}\beta_0)\Big| + \frac{n}{\tau_0^2}|\tau^2-\tau_0^2|\\
        &\leq K \sum_{i\in [n]} \left[|\tau-\tau_0||y_i-\mathbf{x}_i^{\top}\beta_0| + \left(|\tau-\tau_0| + \tau_0\right)|\mathbf{x}_i^{\top}(\beta-\beta_0)|\right] + \frac{n}{\tau_0^2}|\tau^2-\tau_0^2|\\
        &\leq K n |\tau-\tau_0| \mathcal{E}_1(y, \mathbf{X}, \beta_0) + Kn\left(|\tau-\tau_0| + \tau_0\right)\lVert \mathbf{X}\rVert_{2, \infty}\lVert \beta-\beta_0\rVert_2 + \frac{n}{\tau_0^2}|\tau^2-\tau_0^2|\\
        &\leq \frac{Kn}{\tau_0}|\tau^2-\tau_0^2|\mathcal{E}_1(y, \mathbf{X}, \beta_0) + Kn\left(\frac{|\tau^2-\tau_0^2|}{\tau_0} + \tau_0\right) \lVert \mathbf{X}\rVert_{2, \infty}\lVert \beta-\beta_0\rVert_2 + \frac{n}{\tau_0^2}|\tau^2-\tau_0^2|,
    \end{split}
\end{align}
where the inequality denoted by $\star$ in \eqref{eq:proof-theorem1-17} above is due to Assumption \ref{ass:1}(iv) in the main manuscript, which gives $|h(s^2)-h(s_0^2)|\leq K\big||s|-|s_0|\big|\leq K|s-s_0|$ for all $s, s_0\in\mathbb{R}$, together with the mean-value bound $|n\log(\tau/\tau_0)| = \frac{n}{2}|\log(\tau^2/\tau_0^2)| \leq \frac{n|\tau^2-\tau_0^2|}{2\min\{\tau^2,\tau_0^2\}}\leq \frac{n}{\tau_0^2}|\tau^2-\tau_0^2|$, where the last step uses $\min\{\tau^2,\tau_0^2\}\geq \tau_0^2/2$ under $|\tau^2-\tau_0^2|\leq\tau_0^2/2$. The last inequality in \eqref{eq:proof-theorem1-17} above follows from $|\tau-\tau_0| \leq \tau_0^{-1}|\tau^2-\tau_0^2|$, and $\mathcal{E}_k(y, \mathbf{X}, \beta) := n^{-1}\sum_{i \in [n]}|y_i - \mathbf{x}_i^{\top}\beta|^k$, for $k\in \mathbb{N}$. 

Now consider $\widetilde \Delta_2$ in \eqref{eq:proof-theorem1-16}, which is the negative Jensen's gap, i.e., $\widetilde{\Delta}_{2} = -\Delta_2$. Additionally define $s_i := |\tau(y_i - \mathbf{x}_i^{\top}\beta)|$ and $s_{i0} := |\tau_0(y_i - \mathbf{x}_i^{\top}\beta_0)|$:
\begin{align}\label{eq:proof-theorem1-21a}
    \begin{split}
        |\widetilde \Delta_2|=\Delta_2 &= \Big|\log \varphi (y\mid \mathbf{X}, \theta, \tilde{\xi}) - \log p(y\mid \mathbf{X}, \theta)\Big|=\Big|\sum_{i\in [n]}\left(h'(\tilde{\xi}_i^2)s_i^2 + \gamma(\tilde{\xi}_i) - h\left(s_i^2\right)\right)\Big|\\
        &= \sum_{i\in [n]}\Big|h'(\tilde{\xi}_i^2)\left(s_i^2-\tilde{\xi}_i^2\right) + h(\tilde{\xi}_i^2) - h\left(s_i^2\right)\Big| = \sum_{i \in [n]} \Delta_{2,i},
    \end{split}
\end{align}
where $\Delta_{2,i} = \Big|h\left(s_i^2\right) - h(\tilde{\xi}_i^2) - h'(\tilde{\xi}_i^2)\left(s_i^2-\tilde{\xi}_i^2\right)\Big|$. Since the function $t \mapsto h(t^2)$ is $K$-Lipschitz by Assumption \ref{ass:1}(iv) and has derivative $2th'(t^2)$:
\begin{align}\label{eq:proof-theorem1-21b}
    \begin{split}
        \Delta_{2,i} &:= \Big|h(s_i^2) - h(\tilde{\xi}_i^2) - h'(\tilde{\xi}_i^2)\left(s_i^2-\tilde{\xi}_i^2\right)\Big|\\
        &\leq |h(s_i^2) - h(\tilde{\xi}_i^2)| + \frac{|2\tilde{\xi}_i h'(\tilde{\xi}_i^2)|}{2\tilde{\xi}_i}|s_i^2-\tilde{\xi}_i^2|\\
        &\leq K |s_i - \tilde{\xi}_i| + \frac{K}{2\tilde{\xi}_i}| s_i - \tilde{\xi}_i|(s_i -  \tilde{\xi}_i + 2\tilde{\xi}_i)\\
        &\leq 2K|s_i - \tilde{\xi}_i| + \frac{K}{2\tilde{\xi}_i}|s_i - \tilde{\xi}_i|^2.
    \end{split}
\end{align}
In~\eqref{eq:proof-theorem1-21b}, we take $\tilde{\xi}_i = \max\{s_{i0}, C_{\star}\varepsilon^2\}$ for some $C_{\star}>0$ to be fixed later. Hence:
\begin{align}\label{eq:proof-theorem1-21c}
    \begin{split}
        &s_{i0} \leq \tilde{\xi}_i \leq s_{i0} + C_{\star}\varepsilon^2 \\ \iff
        &s_{i0} - s_i \leq \tilde{\xi}_i - s_i \leq s_{i0} - s_i + C_{\star}\varepsilon^2 \\ 
        \implies &|s_i - \tilde{\xi}_i| \leq |s_i - s_{i0}| + C_{\star}\varepsilon^2 \\
        \implies & |s_i - \tilde{\xi}_i|^2 \leq |s_i - s_{i0}|^2 + 2C_{\star}|s_i - s_{i0}| \varepsilon^2 + C_{\star}^2\varepsilon^4,
    \end{split}
\end{align}
which leads to:
\begin{align}\label{eq:proof-theorem1-21d}
    \begin{split}
        \Delta_{2,i} & \leq 2K|s_i - \tilde{\xi}_i| + \frac{K}{2\tilde{\xi}_i}|s_i - \tilde{\xi}_i|^2 \\
        & \leq 2K|s_i - \tilde{\xi}_i| + \frac{K}{2C_{\star}\varepsilon^2} |s_i - \tilde{\xi}_i|^2 \\
        & \leq 2K |s_i - s_{i0}| + 2KC_{\star}\varepsilon^2 + \frac{K}{2C_{\star}\varepsilon^2}|s_i - s_{i0}|^2 + K |s_i - s_{i0}|  + \frac12 KC_{\star}\varepsilon^2 \\
        & = 3K |s_i - s_{i0}| + \frac{5}{2}KC_{\star}\varepsilon^2 + \frac{K}{2C_{\star}\varepsilon^2} |s_i - s_{i0}|^2.
    \end{split}
\end{align}
Finally we get:
\begin{align}
    \begin{split} \label{eq:proof-theorem1-21e}
        |s_i - s_{i0}| &= \Big||\tau(y_i - \mathbf{x}_i^{\top}\beta)| - |\tau_0(y_i-\mathbf{x}_i^{\top}\beta_0)|\Big| \\
        &\leq \Big|\tau(y_i - \mathbf{x}_i^{\top}\beta) - \tau_0(y_i-\mathbf{x}_i^{\top}\beta_0)\Big| \\
        &\leq |\tau-\tau_0||y_i-\mathbf{x}_i^{\top}\beta_0| + \left(|\tau-\tau_0| + \tau_0\right)|\mathbf{x}_i^{\top}(\beta-\beta_0)| \\
        &\leq \frac{|\tau^2-\tau_0^2|}{\tau_0}|y_i-\mathbf{x}_i^{\top}\beta_0| + \left(\frac{|\tau^2-\tau_0^2|}{\tau_0} + \tau_0\right)|\mathbf{x}_i^{\top}(\beta-\beta_0)|,
    \end{split}
\end{align}
and:
\begin{align}
    \begin{split} \label{eq:proof-theorem1-21f}
        |s_i - s_{i0}|^2 &\leq \frac{|\tau^2-\tau_0^2|^2}{\tau_0^2}|y_i-\mathbf{x}_i^{\top}\beta_0|^2 + \left(\frac{|\tau^2-\tau_0^2|}{\tau_0} + \tau_0\right)^2|\mathbf{x}_i^{\top}(\beta-\beta_0)|^2\\
        & \qquad\qquad \qquad +
        2\frac{|\tau^2-\tau_0^2|}{\tau_0}\left(\frac{|\tau^2-\tau_0^2|}{\tau_0} + \tau_0\right)|y_i-\mathbf{x}_i^{\top}\beta_0||\mathbf{x}_i^{\top}(\beta-\beta_0)|.
    \end{split}
\end{align}
Putting \eqref{eq:proof-theorem1-21e} and \eqref{eq:proof-theorem1-21f} in \eqref{eq:proof-theorem1-21d} and then back in \eqref{eq:proof-theorem1-21a}, we get:
\begin{align}
    \begin{split} \label{eq:proof-theorem1-21g}
        \Delta_2 &= \sum_{i \in [n]}\Delta_{2,i} \\
        &\leq 3Kn \frac{|\tau^2-\tau_0^2|}{\tau_0}\mathcal{E}_1(y, \mathbf{X}, \beta_0) + 3Kn\left(\frac{|\tau^2-\tau_0^2|}{\tau_0} + \tau_0\right)\lVert\mathbf{X}\rVert_{2,\infty} \lVert\beta-\beta_0\rVert_2 \\
        & + \frac{5}{2} nKC_{\star}\varepsilon^2 +  \frac{Kn}{2C_{\star}\varepsilon^2}\frac{|\tau^2-\tau_0^2|^2}{\tau_0^2}\mathcal{E}_2(y, \mathbf{X}, \beta_0)  \\
        &+  \frac{Kn}{2C_{\star}\varepsilon^2}\left(\frac{|\tau^2-\tau_0^2|}{\tau_0} + \tau_0\right)^2\lVert\mathbf{X}\rVert_{2,\infty}^2 \lVert\beta-\beta_0\rVert_2^2 \\
        & + \frac{Kn}{C_{\star}\varepsilon^2}\frac{|\tau^2-\tau_0^2|}{\tau_0}\left(\frac{|\tau^2-\tau_0^2|}{\tau_0} + \tau_0\right)\lVert\mathbf{X}\rVert_{2,\infty} \lVert\beta-\beta_0\rVert_2\mathcal{E}_1(y, \mathbf{X}, \beta_0).
    \end{split}
\end{align}
Using \eqref{eq:proof-theorem1-17} and \eqref{eq:proof-theorem1-21g} in \eqref{eq:proof-theorem1-16}, we obtain, for any $\theta$ with $|\tau^2-\tau_0^2|\leq\tau_0^2/2$:
\begin{align}\label{eq:proof-theorem1-23}
\begin{split}
|\Delta(\theta, \theta_0)|
&\leq |\Delta_1| + |\widetilde \Delta_2|
\\
&\leq 4Kn \frac{|\tau^2-\tau_0^2|}{\tau_0}\mathcal{E}_1(y, \mathbf{X}, \beta_0)
+ 4Kn\left(\frac{|\tau^2-\tau_0^2|}{\tau_0} + \tau_0\right)
\lVert \mathbf{X}\rVert_{2,\infty}\lVert \beta-\beta_0\rVert_2
\\
&\quad + \frac{5}{2}\,nK  C_{\star}\varepsilon^2
+ \frac{Kn}{2 C_{\star}\varepsilon^2}\frac{|\tau^2-\tau_0^2|^2}{\tau_0^2}
\mathcal{E}_2(y, \mathbf{X}, \beta_0) + \frac{n}{\tau_0^2}|\tau^2-\tau_0^2|
\\
&\quad + \frac{Kn}{2 C_{\star}\varepsilon^2}
\left(\frac{|\tau^2-\tau_0^2|}{\tau_0}+\tau_0\right)^2
\lVert \mathbf{X}\rVert_{2,\infty}^2 \lVert \beta-\beta_0\rVert_2^2
\\
&\quad + \frac{Kn}{C_{\star}\varepsilon^2}
\frac{|\tau^2-\tau_0^2|}{\tau_0}
\left(\frac{|\tau^2-\tau_0^2|}{\tau_0}+\tau_0\right)
\lVert \mathbf{X}\rVert_{2,\infty}\lVert \beta-\beta_0\rVert_2
\mathcal{E}_1(y, \mathbf{X}, \beta_0).
\end{split}
\end{align}

By application of Markov's inequality along with Assumption \ref{ass:3} in the main manuscript, with $\mathbb{P}_{\theta_0}$-probability at least $1-\varepsilon$, we have $\mathcal{E}_{k}(y, \mathbf{X}, \beta_0) \leq \mathscr{E}_k \tau_0^{-k}\varepsilon^{-1}$ for each $k \in \{1,2\}$. Hence, from \eqref{eq:proof-theorem1-23}, with $\mathbb{P}_{\theta_0}$-probability at least $1-2\varepsilon$:
\begin{align}\label{eq:proof-theorem1-24}
\begin{split}
|\Delta(\theta, \theta_0)|
&\leq \frac{4Kn}{\tau_0^2\varepsilon}|\tau^2-\tau_0^2|\mathscr{E}_1
+ 4Kn\left(\frac{|\tau^2-\tau_0^2|}{\tau_0} + \tau_0\right)
\lVert \mathbf{X}\rVert_{2,\infty}\lVert \beta-\beta_0\rVert_2
\\
&\quad + \frac{5}{2}\,nK  C_{\star}\varepsilon^2
+ \frac{Kn}{2 C_{\star}\varepsilon^3}\frac{|\tau^2-\tau_0^2|^2}{\tau_0^4}\mathscr{E}_2 + \frac{n}{\tau_0^2}|\tau^2-\tau_0^2|
\\
&\quad + \frac{Kn}{2 C_{\star}\varepsilon^2}
\left(\frac{|\tau^2-\tau_0^2|}{\tau_0}+\tau_0\right)^2
\lVert \mathbf{X}\rVert_{2,\infty}^2 \lVert \beta-\beta_0\rVert_2^2
\\
&\quad + \frac{Kn}{C_{\star} \varepsilon^2}
\frac{|\tau^2-\tau_0^2|}{\tau_0}
\left(\frac{|\tau^2-\tau_0^2|}{\tau_0}+\tau_0\right)
\lVert \mathbf{X}\rVert_{2,\infty}\lVert \beta-\beta_0\rVert_2
\frac{\mathscr{E}_1}{\tau_0 \varepsilon}
\\
&=: \Xi_{\Delta(\theta, \theta_0)}
\\
&\leq n\varepsilon^2,\quad \text{if }
\lVert \beta-\beta_0\rVert_2 \leq \frac{\varepsilon^2}{\mathsf{Q}(\mathbf{X}, \tau_0)}
\text{ and }
|\tau^2 - \tau_0^2| \leq \frac{\varepsilon^3}{\mathsf{Q}(\mathbf{X}, \tau_0)},
\end{split}
\end{align}
for $\mathsf{Q}(\mathbf{X}, \tau_0) = \max\{1, 50K\}\max\{\tau_0^2, \tau_0^{-4}\}\max\{1, \mathscr{E}_1, \mathscr{E}_2\}\max\{1, \lVert \mathbf X \rVert_{2, \infty}^2\}$ and $C_{\star} = \mathsf{Q}(\mathbf{X}, \tau_0)^{-1}$.  Here, the restriction $|\tau^2-\tau_0^2|\leq\tau_0^2/2$ required in \eqref{eq:proof-theorem1-17} and \eqref{eq:proof-theorem1-23} holds, since $\tau_0^2\mathsf{Q}(\mathbf{X}, \tau_0)\geq \tau_0^2\max\{\tau_0^2, \tau_0^{-4}\}\geq 1$ gives $|\tau^2-\tau_0^2| \leq \varepsilon^3/\mathsf{Q}(\mathbf{X}, \tau_0)\leq \tau_0^2\varepsilon^3< \tau_0^2/27$. Moreover, under these two conditions, each of the seven terms of $\Xi_{\Delta(\theta, \theta_0)}$ in \eqref{eq:proof-theorem1-24} is at most $\tfrac{4}{50}$, $\tfrac{112}{1350}$, $\tfrac{5}{100}$, $\tfrac{\varepsilon}{100}$, $\varepsilon$, $\tfrac{1}{100}\big(\tfrac{28}{27}\big)^2$, and $\tfrac{28}{1350}$ times $n\varepsilon^2$, respectively, so that $\Xi_{\Delta(\theta, \theta_0)} < 0.6\,n\varepsilon^2\leq n\varepsilon^2$ for $\varepsilon\in(0,\tfrac{1}{3})$.
Now, \eqref{eq:proof-theorem1-24} implies $n^{-1}\tilde{D}\left(p(y\mid \mathbf{X}, \theta_0) \parallel \varphi(y\mid \mathbf{X}, \theta, \tilde{\xi})\right) \leq \varepsilon^2$. Also, since: $$V\left(p(y\mid \mathbf{X}, \theta_0) \parallel \varphi(y\mid \mathbf{X}, \theta, \tilde{\xi})\right) = nV\left(p(y_1\mid \mathbf{x}_1, \theta_0) \parallel \varphi(y_1\mid \mathbf{x}_1, \theta, \tilde{\xi})\right),$$ we use the single observation to obtain:
\begin{align}\label{eq:proof-theorem1-25}
    \begin{split}
        \Delta_1(\theta, \theta_0) \leq \frac{\Xi_{\Delta(\theta, \theta_0)}}{n}.
    \end{split}
\end{align}
Therefore:
\begin{align}\label{eq:proof-theorem1-26}
    \begin{split}
        \frac{1}{n}V\left(p(y\mid \mathbf{X}, \theta_0) \parallel \varphi(y\mid \mathbf{X}, \theta, \tilde{\xi})\right) \leq \varepsilon^2.
    \end{split}
\end{align}
From $n^{-1}\tilde{D}\left(p(y\mid \mathbf{X}, \theta_0) \parallel \varphi(y\mid \mathbf{X}, \theta, \tilde{\xi})\right) \leq \varepsilon^2$ and \eqref{eq:proof-theorem1-26}, finally we obtain:
\begin{align}\label{eq:proof-theorem1-27}
    \begin{split}
        &-\log \pi(\mathcal{B}_n(\theta_0, \varepsilon)) \leq -\log \pi\left(\lVert \beta-\beta_0\rVert_2\leq \mathsf{Q}^{-1}(\mathbf{X}, \tau_0)\varepsilon^2, |\tau^2-\tau_0^2| \leq \mathsf{Q}^{-1}(\mathbf{X}, \tau_0)\varepsilon^{3}\right)\\
        &\stackrel{\#}{\leq} -\log\left(\frac{C(a,b,\tau_0)\varepsilon^3}{\mathsf{Q}(\mathbf{X}, \tau_0)}\right)  + \log\left[2^{\frac{p}{2}}\Gamma\left(\frac{p}{2}+1\right)\left(\lambda_{\mathrm{max}}(\Sigma)\right)^{\frac{p}{2}}\right]- \frac{p}{2}\log t_{-}\\
        &\qquad + p\log\left(\frac{\mathsf{Q}(\mathbf{X}, \tau_0)}{\varepsilon^2}\right) + \frac{t_{+}}{2}\left(\Delta^2 + \frac{\varepsilon^4}{\mathsf{Q}^{2}(\mathbf{X}, \tau_0)\lambda_{\max}(\Sigma)}\right),
    \end{split}
\end{align}
where the inequality in $\#$ above follows from Lemma \ref{lemma:prior-inequality}.
Also, $t_+ = \tau_0^2 + \mathsf{Q}^{-1}(\mathbf{X}, \tau_0)\varepsilon^3$, $t_{-} = \tau_0^2 - \mathsf{Q}^{-1}(\mathbf{X}, \tau_0)\varepsilon^3$, $\Delta^{2} = \lVert \Sigma^{-\frac{1}{2}}(\beta_0-\mu)\rVert_{2}^{2}$, and $C(a,b,\tau_0) = \frac{b}{2\Gamma\left(\frac{a}{2}\right)}\left(\frac{b \tau_0^2}{2}\right)^{\frac{a}{2}-1}\exp\left\{-\frac{b\tau_0^2}{2}\right\}$. Note that, $t_{-} > 0$ since $\tau_0^{2}\mathsf{Q}(\mathbf{X}, \tau_0) \geq \max\{\tau_0^{4}, \tau_0^{-2}\}\geq 1 \geq \tfrac{1}{27} > \varepsilon^{3}$.

Finally we substitute \eqref{eq:proof-theorem1-27} in \eqref{eq:proof-theorem1-7} to obtain, for any $\varepsilon\in \left(0, \frac{1}{3}\right)$ with $\mathbb{P}_{\theta_0}$-probability at least $(1-3\varepsilon) - [(D-1)^2n\varepsilon^2]^{-1}$:
\begin{align}\label{eq:proof-theorem1-30}
    \begin{split}
        &(1-\alpha)\int_{\theta\in \Theta}D_{\alpha}(\theta, \theta_0)\pi_{\alpha}(\theta\mid \mathcal{D}_n, \xi^{\star})d\theta\\
        &\leq D\alpha \varepsilon^2 + \frac{p}{n}\log\left(\frac{\mathsf{Q}(\mathbf{X}, \tau_0)}{\varepsilon^2}\right) - \frac{1}{n}\log\left(
        \frac{C(a, b, \tau_0)\varepsilon^3}{\mathsf{Q}(\mathbf{X}, \tau_0)}\right) - \frac{p}{2n}\log t_{-}\\
        &+ \frac{1}{n}\log\left(2^{\frac{p}{2}}\Gamma\left(\frac{p}{2}+1\right)\left(\lambda_{\mathrm{max}}(\Sigma)\right)^{\frac{p}{2}}\right) + \frac{t_{+}}{2n}\left(\Delta^2 + \frac{\varepsilon^4}{\mathsf{Q}^{2}(\mathbf{X}, \tau_0)\lambda_{\max}(\Sigma)}\right) + \frac{1}{n}\log\left(\frac{1}{\varepsilon}\right).
    \end{split}
\end{align}
We simplify the bound in \eqref{eq:proof-theorem1-30} as:
\begin{align}\label{eq:proof-theorem1-final-TypeI-bound}
\boxed{(1-\alpha)\int_{\theta\in \Theta}D_{\alpha}(\theta, \theta_0)\pi_{\alpha}(\theta\mid \mathcal{D}_n, \xi^{\star})d\theta
\leq D\alpha \varepsilon^2 + \frac{p\log p}{n} + \frac{C_1 p}{n}\log\left(\frac{1}{\varepsilon}\right)},
\end{align}
for arbitrary $D>1$ and constant $C_1 \in \mathbb{R}^{+}$, with $\log^{+}x:=\max\{\log x, 0\}$, satisfying:
\begin{align}\label{eq:proof-theorem1-bound-constants-type-I}
\begin{gathered}
C_1:=6+2\log\mathsf{Q}(\mathbf{X},\tau_0)+\log^{+}\!\left(\frac{1}{C(a,b,\tau_0)}\right)+\\\log^{+}\!\left(\frac{\sqrt{4\lambda_{\max}(\Sigma)}}{\tau_0}\right)+\frac{1}{2}\left(\tau_0^2+\mathsf{Q}^{-1}(\mathbf{X},\tau_0)\right)\left(\Delta^2+\frac{\mathsf{Q}^{-2}(\mathbf{X},\tau_0)}{\lambda_{\max}(\Sigma)}\right)
\end{gathered}
\end{align}
which depends only on prior hyperparameters $(\mu, \Sigma, a, b)$, design matrix $\mathbf{X}$, and true parameter $\theta_0$. 
\end{proof}

\subsection{Proof of Variational Risk Bound under \texorpdfstring{$\alpha$-R\'{e}nyi}{alpha-Renyi} Divergence for Type II \texorpdfstring{$\mathsf{SSG}$}{SSG} Likelihoods (\texorpdfstring{Theorem~\ref{main-theorem-alpha-less-1-variational-risk-bound-Type-II}}{Theorem 4})}\label{app:proof-theorem-vb-alpha-less-1-type-2}

\begin{proof}[Proof of Theorem~\ref{main-theorem-alpha-less-1-variational-risk-bound-Type-II}]
The Type II $\mathsf{SSG}$ likelihoods are of the form:
\begin{equation}\label{eq:proof-theorem2-1}
p(y\mid \mathbf{X}, \beta) \propto \prod_{i\in [n]}\left(\exp\left\{a_i\mathbf{x}_{i}^{\top}\beta\right\}\left[1 + \exp\left\{\mathbf{x}_{i}^{\top}\beta\right\}\right]^{- b_i}\right).
\end{equation}

By Lemma \ref{lemma:majorization-Renyi}, the variational risk under $\alpha$-R\'{e}nyi divergence satisfies:
\begin{align}\label{eq:proof-theorem2-2}
        n(1-\alpha)\int_{\beta\in \mathbb{R}^{p}}D_{\alpha}(\beta, \beta_{0})\pi_{\alpha}(\beta \mid \mathcal{D}_n, \xi^{\star})d\beta
        \leq Dn\alpha\varepsilon^2 - \log \pi(\mathcal{B}_n(\beta_0, \varepsilon)) + \log\left(\frac{1}{\varepsilon}\right),
\end{align}
with $\mathbb{P}_{\beta_0}$-probability at least $1-\varepsilon-[(D-1)^2n\varepsilon^2]^{-1}$.

We now derive an upper bound for the local Bayesian complexity, $-\log \pi(\mathcal{B}_n(\beta_0, \varepsilon))$, by starting with an upper bound for the $\log$-pseudo-likelihood ratio:
\begin{align}\label{eq:proof-theorem1-32}
    \begin{split}
        \Delta(\beta, \beta_0) &:= \log \frac{\varphi(y\mid \mathbf{X}, \beta, \tilde{\xi})}{p(y\mid \mathbf{X}, \beta_0)}\\
        &= \left(\log p(y\mid \mathbf{X}, \beta) - \log p(y\mid \mathbf{X}, \beta_0)\right) + \left(\log \varphi(y\mid \mathbf{X}, \beta, \tilde{\xi}) - \log p(y\mid \mathbf{X}, \beta)\right)\\
        &= \Delta_1 + \widetilde \Delta_2.
    \end{split}
\end{align}
Observe that, $\Delta_1$ in \eqref{eq:proof-theorem1-32} above can be upper bounded as:
\begin{align}\label{eq:proof-theorem1-33}
\begin{split}
        \Delta_1 &:= \log p(y\mid \mathbf{X}, \beta) - \log p(y\mid \mathbf{X}, \beta_{0})\\
        &= \sum_{i \in [n]}\left(a_i\mathbf{x}_i^{\top}(\beta-\beta_0) + b_i\left[\log\left(1 + \exp\left\{\mathbf{x}_i^\top\beta_0\right\}\right) - \log\left(1 + \exp\left\{\mathbf{x}_i^\top\beta\right\}\right)\right]\right)\\
        &\leq n\lVert \mathbf{X}\rVert_{2, \infty}\lVert \beta-\beta_0\rVert_2\left(\frac{\sum_{i\in [n]} (a_i+b_i)}{n}\right),
\end{split}
\end{align}
where the last inequality in \eqref{eq:proof-theorem1-33} follows from the fact that $\log(1 + \exp\{t\})$ is $1$-Lipschitz.
Now we consider the negative Jensen's gap $\widetilde \Delta_2$ in \eqref{eq:proof-theorem1-32}, which can be upper bounded as follows:
\begin{equation}\label{eq:proof-theorem1-34}
    \begin{split}
        \widetilde \Delta_2 &:= \log \varphi(y\mid \mathbf{X}, \beta, \tilde{\xi}) - \log p(y\mid \mathbf{X}, \beta)\\
        &= \sum_{i\in [n]} b_i\left(h'(\tilde{\xi}_i^2) (\mathbf{x}_i^{\top}\beta)^2 + \gamma(\tilde{\xi}_i) - h((\mathbf{x}_i^{\top}\beta)^2)\right)\\
        &= \sum_{i\in [n]}b_i\left(h'(\tilde{\xi}_i^2)\left[(\mathbf{x}_i^{\top}\beta)^2 - \tilde{\xi}_i^2\right] + h(\tilde{\xi}_i^2) - h((\mathbf{x}_i^{\top}\beta)^2)\right),
    \end{split}
\end{equation}
where the last equality is obtained using $\gamma(\tilde{\xi}_i) = h(\tilde{\xi}_i^2) - \tilde{\xi}_i^2h'(\tilde{\xi}_i^2)$. Taking $\tilde{\xi}_i := |\mathbf{x}_i^{\top}\beta_0|$ and applying Taylor's theorem in \eqref{eq:proof-theorem1-34} for $\delta_i^{\star}$ between $(\mathbf{x}_i^{\top}\beta)^2$ and $(\mathbf{x}_i^{\top}\beta_0)^2$:
\begin{align}\label{eq:proof-theorem1-35}
    \begin{split}
        |\widetilde \Delta_2| =\Delta_2 &= \sum_{i\in [n]}\frac{b_i}{2}h''(\delta_i^{\star})\left((\mathbf{x}_i^{\top}\beta)^2 - (\mathbf{x}_i^{\top}\beta_0)^2\right)^2\\
        &\leq \frac{1}{2}\sum_{i\in [n]}\left(b_i h''(\delta_i^{\star})\left[\left(\mathbf{x}_i^{\top}(\beta-\beta_0)\right)^2\left(\mathbf{x}_i^{\top}(\beta-\beta_0) + 2\mathbf{x}_i^{\top}\beta_0\right)^2\right]\right)\\
        &\leq \frac{1}{2}\sum_{i\in [n]}\left(b_i\left[\left(\mathbf{x}_i^{\top}(\beta-\beta_0)\right)^2\left(\mathbf{x}_i^{\top}(\beta-\beta_0) + 2\mathbf{x}_i^{\top}\beta_0\right)^2\right]\right),
    \end{split}
\end{align}
where the last inequality in \eqref{eq:proof-theorem1-35} above follows from the fact that $0<h''(t)\leq \frac{1}{96} \leq 1$ for all $t\geq 0$, where $h(t) = -\log (2\cosh(\sqrt{t}/2))$. Finally, the upper bound of $\Delta_2$ is obtained as:
\begin{align}\label{eq:proof-theorem1-36}
    \begin{split}
        |\widetilde \Delta_2| &\stackrel{\dagger}{\leq} \frac{1}{2}\sum_{i\in [n]}b_i\left(2\left(\mathbf{x}_i^{\top}(\beta-\beta_0)\right)^4 + 8(\mathbf{x}_i^{\top}\beta_0)^2\left(\mathbf{x}_i^{\top}(\beta-\beta_0)\right)^2\right)\\
        &\leq\left(\frac{1}{n}\sum_{i\in [n]}b_i\right)\left(n\lVert \mathbf{X}\rVert_{2, \infty}^{4}\lVert \beta-\beta_0\rVert_{2}^{4} + 4n\lVert \mathbf{X}\rVert_{2, \infty}^{4}\lVert \beta_0\rVert_{2}^{2}\lVert \beta-\beta_0\rVert_{2}^{2}\right),
    \end{split}
\end{align}
where the inequality in $\dagger$ above is obtained from $(a+b)^2\leq 2(a^2+b^2)$. We now probabilistically bound two quantities viz., $n^{-1}\sum_{i\in [n]}b_i$ and $n^{-1}\sum_{i\in [n]}(a_i+b_i)$.

In particular, for Type II $\mathsf{SSG}$ likelihood with Negative-Binomial distribution, i.e., $\mathrm{NB}(m_i, p_i)$ with $p_i = \exp\{\mathbf{x}_i^{\top}\beta_0\}(1 + \exp\{\mathbf{x}_i^{\top}\beta_0\})^{-1}$ (having parameterization as in \eqref{eq:param-negbin}):
\begin{align}\label{eq:proof-theorem1-37}
    \begin{split}
        &\frac{\sum_{i\in[n]} a_i}{n} = \frac{\sum_{i\in [n]}m_i}{n} \leq \bar{m}\\
        &\frac{\sum_{i\in [n]}b_i}{n} = \frac{\sum_{i\in [n]}(y_i+m_i)}{n} \leq \bar{m}\left(1 + \frac{\exp\left\{\lVert \mathbf{X}\rVert_{2, \infty}\lVert \beta_0\rVert_2\right\}}{\varepsilon}\right),
    \end{split}
\end{align}
where the bound on $n^{-1}\sum_{i\in [n]}b_i$ above in \eqref{eq:proof-theorem1-37} is due to Markov's inequality and holds with $\mathbb{P}_{\beta_0}$-probability at least $1-\varepsilon$.

Now, for Type II $\mathsf{SSG}$ likelihood with Binomial distribution, i.e., $\mathrm{Bin}(m_i, p_i)$ with $p_i = \exp\{\mathbf{x}_i^{\top}\beta_0\}(1 + \exp\{\mathbf{x}_i^{\top}\beta_0\})^{-1}$ (having parameterization as in \eqref{eq:param-binom}), we have:
\begin{align}\label{eq:proof-theorem1-38}
    \begin{split}
        &\frac{\sum_{i\in[n]} a_i}{n} = \frac{\sum_{i\in [n]}y_i}{n} \leq \frac{\sum_{i\in [n]}m_i}{n} \leq \bar{m}\\
        &\frac{\sum_{i\in [n]}b_i}{n} = \frac{\sum_{i\in [n]}m_i}{n} \leq \bar{m}.
    \end{split}
\end{align}
Therefore, for both Negative-Binomial and Binomial Type II $\mathsf{SSG}$ likelihoods, using \eqref{eq:proof-theorem1-37} and \eqref{eq:proof-theorem1-38} in \eqref{eq:proof-theorem1-33} and \eqref{eq:proof-theorem1-36}, we bound $\Delta(\beta, \beta_0)$ in \eqref{eq:proof-theorem1-32} as:
\begin{align}\label{eq:proof-theorem1-39}
\begin{split}
|\Delta(\beta, \beta_{0})| &\leq |\Delta_1| + |\widetilde \Delta_2|\\
&\leq [n\lVert\mathbf{X}\rVert_{2, \infty}\lVert \beta-\beta_0\rVert_{2} + n\lVert\mathbf{X}\rVert_{2, \infty}^{4}\lVert\beta-\beta_0\rVert_{2}^{4}\\
&\quad+ 4n\lVert \mathbf{X}\rVert_{2, \infty}^{4}\lVert\beta_0\rVert_{2}^{2}\lVert\beta-\beta_0\rVert_{2}^{2}] \bar{m}\left(2 + \exp\left\{\lVert\mathbf{X}\rVert_{2, \infty}\lVert\beta_0\rVert_{2}\right\}/\varepsilon\right).
\end{split}
\end{align}
If $\lVert\beta - \beta_0\rVert_2 \leq \varepsilon^3 \mathsf{Q}^{-1}( \mathbf{X}, \beta_0)$, where $\mathsf{Q}(\mathbf{X}, \beta_0)$ is taken to be:
\begin{align}\label{eq:proof-theorem1-40}
    \begin{split}
        \mathsf{Q}(\mathbf{X}, \beta_0) = \max\left\{\max\left\{4\lVert\mathbf{X}\rVert_{2, \infty}, 8\lVert\mathbf{X}\rVert_{2, \infty}^{2}\lVert\beta_0\rVert_{2}\right\}\bar{m}\left(1 + \exp\{\lVert\mathbf{X}\rVert_{2, \infty}\lVert\beta_0\rVert_{2}\}\right),1\right\},
    \end{split}
\end{align}
then $|\Delta(\beta, \beta_0)| \leq n\varepsilon^2$. Following similar arguments as in the case of Type I $\mathsf{SSG}$ likelihoods above, we conclude that:
\begin{equation}\label{eq:proof-theorem1-41}
\begin{split}
&-\log (\pi\left(\mathcal{B}_n(\beta_0, \varepsilon)\right))
\leq -\log \pi\left(\lVert\beta - \beta_0\rVert_2 \leq \varepsilon^3\mathsf{Q}^{-1}(\mathbf{X}, \beta_0)\right)\\
&\leq \frac{1}{2}\left(\Delta^2 + \frac{\varepsilon^6}{\lambda_{\mathrm{max}}(\Sigma)\mathsf{Q}^{2}(\mathbf{X}, \beta_0)}\right) + p\log\left(\frac{\mathsf{Q}(\mathbf{X},\beta_0)}{\varepsilon^3}\right) + \log\left(2^{\frac{p}{2}}\Gamma\left(\frac{p}{2}+1\right)(\lambda_{\mathrm{max}}(\Sigma))^{\frac{p}{2}}\right)
\end{split}
\end{equation}
where the last inequality follows from \eqref{eq:Ball-application-2} with $t=1$ in Lemma \ref{lemma:prior-inequality} and $\Delta^{2} = \lVert \Sigma^{-\frac{1}{2}}(\beta_0-\mu)\rVert_{2}^{2}$. 

Finally, using \eqref{eq:proof-theorem1-41} in \eqref{eq:proof-theorem2-2}, for any $\varepsilon \in \left(0,\frac{1}{3}\right)$ with $\mathbb{P}_{\beta_0}$-probability at least $(1 - 2\varepsilon) - [(D-1)^2n\varepsilon^2]^{-1}$:
\begin{equation}\label{eq:proof-theorem1-42}
\begin{split}
    &(1-\alpha)\int_{\beta\in \mathbb{R}^p} D_{\alpha}(\beta, \beta_0)\pi_{\alpha}(\beta| \mathcal{D}_n, \xi^{\star})d\beta\\
&\leq D\alpha\varepsilon^2 + \frac{p}{n}\left[\log \left(\frac{\mathsf{Q}( \mathbf{X}, \beta_0)}{\varepsilon^3}\right) + \frac{1}{2} \log 2 + \frac{1}{2}\log \lambda_{\max}(\Sigma)\right] + \frac{1}{n}\log \left(\Gamma\left(\frac{p}{2} + 1\right)\right)\\
&\quad + \frac{1}{2n}\left(\Delta^2 + \frac{\varepsilon^6}{\lambda_{\mathrm{max}}(\Sigma)\mathsf{Q}^{2}(\mathbf{X}, \beta_0)}\right) + \frac{1}{n}\log\left(\frac{1}{\varepsilon}\right).
\end{split}
\end{equation}
We simplify the bound in \eqref{eq:proof-theorem1-42} as:
\begin{equation}\label{eq:proof-theorem1-final-TypeII-bound}
\boxed{(1-\alpha)\int_{\beta\in \mathbb{R}^p} D_{\alpha}(\beta, \beta_0)\pi_{\alpha}(\beta\mid \mathcal{D}_n, \xi^{\star})d\beta \leq D\alpha\varepsilon^2 + \frac{p \log p}{n} + \frac{C_2 p}{n}\log\left(\frac{1}{\varepsilon}\right)},
\end{equation}
for arbitrary $D>1$ and constant $C_2\in \mathbb{R}^{+}$, with $\log^+ x = \max\{\log x,0\}$, satisfying:
\begin{align}\label{eq:proof-theorem1-bound-constants-type-II}
C_2 := \frac{1}{2}\left(8 + \log 2 + \log^+ \lambda_{\max}(\Sigma) + 2\log \mathsf{Q}( \mathbf{X}, \beta_0) + \Delta^2 + \left(\lambda_{\max}(\Sigma)\mathsf{Q}^2(\mathbf{X}, \beta_0)\right)^{-1}\right),
\end{align}
which depends only on prior hyperparameters $(\mu, \Sigma)$, design matrix $\mathbf{X}$, true parameter $\beta_0$, and $\bar{m}$. 
\end{proof}

\subsection{\texorpdfstring{$\tssg$}{TAVIE-SSG} for Compact Restriction of Parameter Space}\label{subsec:compact-parameter-space}

In this section, we develop an extension of the $\tssg$ framework to settings where the model parameters are restricted to a compact subset of the full parameter space.

Let $\varphi(y_i \mid \mathbf{x}_i, \theta, \xi_i)$ be the tangent minorant of an $\ssg$ likelihood $p(y_i \mid \mathbf{x}_i, \theta)$ for observation $i \in [n]$, having the form (see Proposition~\ref{main-prop:tangent-lower-bound} of the main manuscript):
\begin{equation}\label{eq:compact-restrictin-0}
\varphi(y_i \mid \mathbf{x}_i, \theta, \xi_i) := r_i \, \exp \left\{\,s_i \zeta_i + t_i\,A(\xi_i)\zeta_i^2 + t_i\,\gamma(\xi_i)\right\} \leq p(y_i \mid \mathbf{x}_i, \theta),
\end{equation}
where $\zeta_i = u_i \mathbf{x}_i^{\top}\beta + v_i$ and $\theta\in \Theta$ contains the model parameters including the regression coefficient vector $\beta \in \mathbb{R}^{p}$ and (possibly) dispersion parameter $\tau$ along with $h(\cdot), r_i, t_i, s_i, u_i, v_i$ are as in Definition~\ref{main-def:SSG} of the main manuscript in general and in Table~\ref{main-tab:ssg-types} of the main manuscript in particular for Type I and Type II $\ssg$ likelihoods. Define $\varphi_{\alpha}(y \mid \mathbf{X}, \theta, \xi) := \prod_{i \in [n]} \left[\varphi(y_i \mid \mathbf{x}_i, \theta, \xi_i)\right]^{\alpha}$ to be the tangent minorant of the $\alpha$-fractional likelihood $p_{\alpha}(y \mid \mathbf{X}, \theta)$. 

Let $\pi(\theta)$ be the unrestricted prior on $\theta$. Recall that, the unrestricted $\alpha$-fractional posterior of $\theta$ and the minorant of the marginal likelihood are given by (see~\eqref{main-eq:alpha-frac-posterior} of the main manuscript):
\begin{equation}\label{eq:compact-restriction-1}
\pi_{\alpha}(\theta \mid \mathcal{D}_n, \xi) := \frac{\varphi_{\alpha}(y \mid \mathbf{X}, \theta, \xi) \pi(\theta)}{\varphi_{\alpha}(y \mid \mathbf{X}, \xi)},
\qquad \varphi_{\alpha}(y \mid \mathbf{X}, \xi) := \int_{\theta \in \Theta} \varphi_{\alpha}(y \mid \mathbf{X}, \theta, \xi) \pi(\theta) d\theta.
\end{equation}

Now, we consider a compact subset $\mathcal{C} \subset \Theta$. Define $\mathcal{Z}_{\mathcal{C}}: \mathcal{P}_{\Theta} \rightarrow \mathbb{R}^{+}$ as $\mathcal{Z}_{\mathcal{C}}(q) := \int_{\mathcal{C}} q(\theta) d\theta$. Therefore, the restriction of the prior $\pi$ to $\mathcal{C}$ is:
\begin{align}\label{eq:compact-restriction-prior}
    \widehat{\pi}(\theta) = \frac{\pi(\theta)\mathds{1}_{\mathcal{C}}(\theta)}{\mathcal{Z}_{\mathcal{C}}(\pi)}.
\end{align} 
With the restricted prior in \eqref{eq:compact-restriction-prior}, the minorant of the marginal likelihood becomes:
\begin{equation}
\label{eq:compact-restriction-2}
\begin{split}
\widehat{\varphi}_{\alpha}(y \mid \mathbf{X}, \xi) &:= \int_{\theta \in \mathcal{C}} \frac{\varphi_{\alpha}(y \mid \mathbf{X}, \theta, \xi) \pi(\theta)}{\mathcal{Z}_{\mathcal{C}}(\pi)} d\theta = \int_{\theta \in \mathcal{C}} \frac{\varphi_{\alpha}(y \mid \mathbf{X}, \xi) \pi_{\alpha}(\theta \mid \mathcal{D}_n, \xi)}{\mathcal{Z}_{\mathcal{C}}(\pi)} d\theta\\
&= \frac{\varphi_{\alpha}(y \mid \mathbf{X}, \xi)}{\mathcal{Z}_{\mathcal{C}}(\pi)}\int_{\theta \in \mathcal{C}} { \pi_{\alpha}(\theta \mid \mathcal{D}_n, \xi)} d\theta 
= \varphi_{\alpha}(y \mid \mathbf{X}, \xi)\frac{\mathcal{Z}_{\mathcal{C}}(\pi_{\alpha}(\cdot \mid \mathcal{D}_n, \xi))}{\mathcal{Z}_{\mathcal{C}}(\pi)},
\end{split}
\end{equation}
and the restricted $\alpha$-fractional posterior is:
\begin{equation}
\label{eq:compact-restriction-3}
\begin{split}
\widehat{\pi}_{\alpha}(\theta \mid \mathcal{D}_n, \xi) &:= \frac{\varphi_{\alpha}(y \mid \mathbf{X}, \theta, \xi) \widehat{\pi}(\theta)}{\widehat{\varphi}_{\alpha}(y \mid \mathbf{X}, \xi)} = \frac{\varphi_{\alpha}(y \mid \mathbf{X}, \theta, \xi) \pi(\theta)}{\varphi_{\alpha}(y \mid \mathbf{X}, \xi)\mathcal{Z}_{\mathcal{C}}(\pi_{\alpha}(\cdot \mid \mathcal{D}_n, \xi))}\mathds{1}_{\mathcal{C}}(\theta) 
\\&
= \frac{\pi_{\alpha}(\theta \mid \mathcal{D}_n, \xi)}{\mathcal{Z}_{\mathcal{C}}(\pi_{\alpha}(\cdot \mid \mathcal{D}_n, \xi))}\mathds{1}_{\mathcal{C}}(\theta).
\end{split}
\end{equation}
Thus, under this {restricted} setting, variational posterior is simply the same as that in the unrestricted setting, but restricted to $\mathcal{C}$. Now, the profile objective (up to constants) to be maximized is: 
\begin{equation}
\label{eq:compact-restriction-3a}
\widehat{\mathsf{L}}(\xi) := \log \widehat{\varphi}_{\alpha}(y \mid \mathbf{X}, \xi) = \mathsf{L}(\xi) + \log \mathcal{Z}_{\mathcal{C}}(\pi_{\alpha}(\cdot \mid \mathcal{D}_n, \xi)),
\end{equation}
where $\mathsf{L}(\xi) = \log \varphi_{\alpha}(y \mid \mathbf{X}, \xi)$ is the {unrestricted} objective as defined in~\eqref{main-eq:ELBO-general} of the main manuscript. 

The final step is the derivation of the $\tssg$ algorithm for the {restricted} setting. Let $\mathbb{E}_{\xi}$ denote the expectation with respect to the variational posterior distribution, $\widehat{\pi}_{\alpha}(\theta \mid \mathcal{D}_n, \xi)$, for any $\xi \in \mathbb{R}^n_{+}$. Then, the general EM surrogate function for the $(t+1)$th step is:
\begin{align}\label{eq:compact-restriction-4}
\begin{split}
\mathcal{Q}(\xi^{\dagger} \mid \xi^{(t)}) 
&:= \mathbb{E}_{\xi^{(t)}} \left[\log \widehat{\pi}(\theta)\right] + \alpha \sum_{i\in[n]} \mathbb{E}_{\xi^{(t)}} \left[\log \varphi(y_i \mid \mathbf{x}_i, \theta, \xi^{\dagger}_i) \right]\\
&= \alpha \sum_{i\in[n]} t_i\left\{\,A(\xi_i^{\dagger}) \mathbb{E}_{\xi^{(t)}} \left[\zeta_i^2\right] + \gamma(\xi_i^{\dagger})\right\} + \mathfrak{C}(\xi^{(t)}),
\end{split}
\end{align}
where $\mathfrak{C}(\xi^{(t)})$ collects the terms independent of $\xi^{\dagger}$. Note that, $\mathcal{Q}(\xi^{\dagger} \mid \xi^{(t)})$ is maximized at:
\begin{align}
\label{eq:compact-restriction-5}
(\xi_i^{(t+1)})^2 = \kappa_i(\xi^{(t)}) = \mathbb{E}_{\xi^{(t)}}\left[\zeta_i^2\right]. 
\end{align}
\begin{remark}
Observe that, $\mathbb{E}_{\xi^{(t)}}\left[\zeta_i^2\right]$ need not admit a closed-form expression under the {restricted} compact space. However, the purpose of this derivation is to establish that a $\tssg$ algorithm exists in principle. Also, if we take $\mathcal{C}$ to be large, $\mathbb{E}_{\xi^{(t)}}\left[\zeta_i^2\right]$ approaches its corresponding unrestricted counterpart.
\end{remark}

\subsection{Proof of Variational Risk Bound under Hellinger Distance for Type I and II \texorpdfstring{$\mathsf{SSG}$}{SSG} Likelihoods (\texorpdfstring{Theorems~\ref{main-theorem-alpha-equals-1-variational-risk-bound-Type-I}}{Theorems 5} and \texorpdfstring{\ref{main-theorem-alpha-equals-1-variational-risk-bound-Type-II}}{6})}\label{app:proof-theorem-vb-hellinger-type-1-2}

\begin{proof}[Proofs of Theorems~\ref{main-theorem-alpha-equals-1-variational-risk-bound-Type-I} and~\ref{main-theorem-alpha-equals-1-variational-risk-bound-Type-II}]
As $\Theta$ is compact and we operate with the squared Hellinger distance $\mathcal{H}^{2}(\theta \parallel \theta_0)$, we invoke Lemma~\ref{lemma:hell-risk-assumption} to obtain the sieve set $\mathcal{F}_{n, \varepsilon} \subset \Theta$ and test functions $\phi_{n, \varepsilon}$ satisfying \eqref{eq:hell-risk-assumption-1}-\eqref{eq:hell-risk-assumption-3} by~\cite{ghosal2007convergence}. Denote the $\log$-likelihood function as $\ell_n(\theta) := \log p(y\mid \mathbf{X}, \theta)$ and the $\log$-likelihood ratio function as $\ell_n(\theta, \theta_0):= \ell_n(\theta) - \ell_n(\theta_0)$. Clearly:
\begin{align}\label{eq:hell-risk-lemma-1}
    \mathbb{E}_{\theta_0}\left[\exp\left\{\ell_n(\theta, \theta_0)\right\}\right] = 1.
\end{align}
The type II error bound \eqref{eq:hell-risk-assumption-3} in Lemma \ref{lemma:hell-risk-assumption} implies for fixed $\varepsilon \geq \varepsilon_n$ and any $\theta \in \mathcal{F}_{n, \varepsilon}$:
\begin{align}\label{eq:hell-risk-lemma-2}
    \mathbb{E}_{\theta_0}\left[\exp\left\{\ell_n(\theta, \theta_0)\right\}(1 - \phi_{n, \varepsilon})\right] \leq \exp\left\{ - c n \mathcal{H}^{2}(\theta\parallel \theta_0) \mathds{1}(\mathcal{H}^{2}(\theta\parallel \theta_0) \geq \varepsilon^2)\right\}.
\end{align}
Thus, for any $\eta \in (0, 1)$, we have:
\begin{align}\label{eq:hell-risk-lemma-3}
    \begin{split}
        \mathbb{E}_{\theta_0}\left[\exp\left\{\ell_n(\theta, \theta_0) + c n  \mathcal{H}^{2}(\theta\parallel\theta_0)\mathds{1}(\mathcal{H}^{2}(\theta\parallel \theta_0) \geq \varepsilon^2)  - \log\left(\frac{1}{\eta}\right)\right\}(1 - \phi_{n, \varepsilon})\right] \leq \eta.
    \end{split}
\end{align}
Let $\Pi_{\mathcal{F}_{n, \varepsilon}}(\cdot) = \tfrac{\Pi(\cdot \cap \mathcal{F}_{n, \varepsilon})}{\Pi(\mathcal{F}_{n, \varepsilon})}$ be the restriction of the prior measure $\Pi$ on $\mathcal{F}_{n, \varepsilon}$ and the corresponding restricted density be $\pi_{\mathcal{F}_{n, \varepsilon}}(\cdot) = \tfrac{\pi(\cdot) \mathds{1}(\cdot \in \mathcal{F}_{n, \varepsilon})}{\Pi(\mathcal{F}_{n, \varepsilon})}$. Integrating both sides of \eqref{eq:hell-risk-lemma-3} with respect to $\pi_{\mathcal{F}_{n, \varepsilon}}$ and interchanging the integrals using Fubini's theorem, we obtain:
\begin{multline}\label{eq:hell-risk-lemma-4}
        \mathbb{E}_{\theta_0}\Bigg[
            (1-\phi_{n,\varepsilon})
            \int_{\mathcal{F}_{n,\varepsilon}}
            \exp\left\{
                \ell_{n}(\theta,\theta_0)
                + c n \mathcal{H}^{2}(\theta\parallel\theta_0)
                  \mathds{1}\!\left(\mathcal{H}^{2}(\theta\parallel\theta_0)\ge \varepsilon^2\right)
                - \log\left(\frac{1}{\eta}\right)
            \right\}\\ 
            \pi_{\mathcal{F}_{n,\varepsilon}}(\theta)d\theta
        \Bigg]
        \le \eta.
\end{multline}
Now, Lemma \ref{lemma-variational-inequality} implies for any $\rho \ll \pi_{\mathcal{F}_{n, \varepsilon}}$:
\begin{multline}\label{eq:hell-risk-lemma-5}
    \mathbb{E}_{\theta_0}\Bigg[(1-\phi_{n, \varepsilon}) \exp\Bigg\{\int_{\mathcal{F}_{n, \varepsilon}}\left(\ell_n(\theta, \theta_0) + c n \mathcal{H}^{2}(\theta\parallel \theta_0)\mathds{1}(\mathcal{H}^{2}(\theta\parallel \theta_0) \geq \varepsilon^2) - \log\left(\frac{1}{\eta}\right)\right)\rho(\theta)d\theta\\
    -\mathrm{KL}(\rho\parallel \pi_{\mathcal{F}_{n, \varepsilon}})\Bigg\}\Bigg] \leq \eta.
\end{multline}
Take $\rho$ to be the restriction $q^{\star}_{\mathcal{F}_{n, \varepsilon}}$ of $q^{\star}\equiv \pi_1(\cdot\mid \mathcal{D}_n, \xi^{\star})\equiv \pi(\cdot\mid \mathcal{D}_n, \xi^{\star})$ over $\mathcal{F}_{n, \varepsilon}$, and let the corresponding probability measures be $Q^{\star}_{\mathcal{F}_{n, \varepsilon}}$ and $Q^{\star}$, to obtain:
\begin{multline}\label{eq:hell-risk-lemma-6}
    \mathbb{E}_{\theta_0}\Bigg[(1-\phi_{n, \varepsilon}) \exp\Bigg\{\frac{1}{Q^{\star}(\mathcal{F}_{n, \varepsilon})}\int_{\mathcal{F}_{n, \varepsilon}}\Bigg(\ell_n(\theta, \theta_0) + c n \mathcal{H}^{2}(\theta\parallel \theta_0)\mathds{1}(\mathcal{H}^{2}(\theta\parallel \theta_0) \geq \varepsilon^2)\\ - \log\left(\frac{1}{\eta}\right)\Bigg)q^{\star}(\theta)d\theta
    -\mathrm{KL}(q^{\star}_{\mathcal{F}_{n, \varepsilon}}\parallel \pi_{\mathcal{F}_{n, \varepsilon}})\Bigg\}\Bigg] \leq \eta.
\end{multline}
By applying Markov's inequality, we further obtain that with $\mathbb{P}_{\theta_0}$-probability at least $(1 - \eta^{\tfrac{1}{2}})$:
\begin{multline}\label{eq:hell-risk-lemma-7}
    (1-\phi_{n, \varepsilon}) \exp\Bigg\{\frac{1}{Q^{\star}(\mathcal{F}_{n, \varepsilon})}\int_{\mathcal{F}_{n, \varepsilon}}\Bigg(\ell_n(\theta, \theta_0) + c n \mathcal{H}^{2}(\theta \parallel \theta_0) \mathds{1}(\mathcal{H}^{2}(\theta\parallel \theta_0) \geq \varepsilon^2)\\
    -\log\left(\frac{1}{\eta}\right)\Bigg)q^{\star}(\theta)d\theta - \mathrm{KL}(q^{\star}_{\mathcal{F}_{n, \varepsilon}}\parallel \pi_{\mathcal{F}_{n, \varepsilon}})\Bigg\} \leq \eta^{\tfrac{1}{2}}.
\end{multline}
Denote the exponential term in \eqref{eq:hell-risk-lemma-7} above as $\mathfrak{A}_n$, then:
\begin{align}\label{eq:hell-risk-lemma-8}
    (1-\phi_{n, \varepsilon})\mathfrak{A}_n \leq \eta^{\tfrac{1}{2}}.
\end{align}
The type I error bound \eqref{eq:hell-risk-assumption-2} in Lemma \ref{lemma:hell-risk-assumption} implies, by Markov's inequality, that $\phi_{n, \varepsilon} \leq \exp\left\{-\tfrac{c n \varepsilon_n^2}{2}\right\}$ holds with $\mathbb{P}_{\theta_0}$-probability at least $\left(1 - \exp\left\{-\tfrac{c n \varepsilon_n^2}{2}\right\}\right)$, obtaining:
\begin{align}\label{eq:hell-risk-lemma-9}
    \phi_{n, \varepsilon}\mathfrak{A}_n \leq \exp\left\{-\frac{c n \varepsilon_n^2}{2}\right\}\mathfrak{A}_n.
\end{align}
Combining \eqref{eq:hell-risk-lemma-8} and \eqref{eq:hell-risk-lemma-9}, we obtain with $\mathbb{P}_{\theta_0}$-probability at least $\left(1 - 2\exp\left\{-\tfrac{c n \varepsilon_n^{2}}{2}\right\}\right)$ (taking $\eta = \exp\left\{-c n \varepsilon_n^2\right\}$):
\begin{align}\label{eq:hell-risk-lemma-10}
    \mathfrak{A}_n = (1 - \phi_{n, \varepsilon})\mathfrak{A}_n + \phi_{n, \varepsilon}\mathfrak{A}_n \leq \exp\left\{-\frac{c n \varepsilon_n^2}{2}\right\} + \exp\left\{-\frac{c n \varepsilon_n^2}{2}\right\}\mathfrak{A}_n,
\end{align}
leading to the following bound for $\mathfrak{A}_n$ whenever $c n \varepsilon_n^2 \geq 2\log 2$ (otherwise $2\exp\{-\tfrac{c n \varepsilon_n^2}{2}\} > 1$, the probability lower bounds below are negative, and there is nothing to prove):
\begin{align}\label{eq:hell-risk-lemma-11}
    \mathfrak{A}_n \leq \frac{1}{1 - \exp\left\{-\frac{c n \varepsilon_n^2}{2}\right\}}\exp\left\{-\frac{c n \varepsilon_n^2}{2}\right\} \leq 2 \exp\left\{-\frac{c n \varepsilon_n^2}{2}\right\}.
\end{align}
Consequently, using the definition of $\mathfrak{A}_n$, we get:
\begin{multline}\label{eq:hell-risk-lemma-12}
    \frac{1}{Q^{\star}(\mathcal{F}_{n, \varepsilon})}\int_{\mathcal{F}_{n, \varepsilon}}\left(\ell_n(\theta, \theta_0) + c n \mathcal{H}^{2}(\theta\parallel \theta_0) \mathds{1}(\mathcal{H}^{2}(\theta \parallel \theta_0) \geq \varepsilon^2) - \log \left(\frac{1}{\eta}\right)\right)q^{\star}(\theta)d\theta\\
    -\mathrm{KL}(q^{\star}_{\mathcal{F}_{n, \varepsilon}}\parallel \pi_{\mathcal{F}_{n, \varepsilon}}) \leq -\frac{c n \varepsilon_n^2}{2} + \log 2.
\end{multline}
Rearranging terms, we obtain:
\begin{multline}\label{eq:hell-risk-lemma-13}
    c n \int_{\theta \in \mathcal{F}_{n, \varepsilon}, \mathcal{H}^{2}(\theta \parallel \theta_0) \geq \varepsilon^2}\mathcal{H}^{2}(\theta\parallel \theta_0) q^{\star}(\theta)d\theta - Q^{\star}(\mathcal{F}_{n, \varepsilon})\mathrm{KL}(q^{\star}_{\mathcal{F}_{n, \varepsilon}}\parallel \pi_{\mathcal{F}_{n, \varepsilon}})\\
    \leq \int_{\mathcal{F}_{n, \varepsilon}}-\ell_{n}(\theta, \theta_0)q^{\star}(\theta)d\theta + \left[\frac{c n \varepsilon_n^2}{2} + \log 2\right]Q^{\star}(\mathcal{F}_{n, \varepsilon}).
\end{multline}
Similarly, for each $\theta \in \mathcal{F}_{n, \varepsilon}^{c}$, from the identity $\mathbb{E}_{\theta_0}\left[\exp\left\{\ell_{n}(\theta, \theta_0)\right\}\right] = 1$ and Lemma \ref{lemma-variational-inequality}, we can obtain for any measure $\rho \ll \pi_{\mathcal{F}_{n, \varepsilon}^{c}}$:
\begin{align}\label{eq:hell-risk-lemma-14}
\mathbb{E}_{\theta_0}\left[\exp\left\{\int_{\mathcal{F}_{n, \varepsilon}^{c}}\left(\ell_n(\theta, \theta_0) - \log \left(\frac{1}{\eta}\right)\right)\rho(\theta)d\theta - \mathrm{KL}(\rho \parallel \pi_{\mathcal{F}^{c}_{n, \varepsilon}})\right\}\right] \leq \eta.
\end{align}
Take $\rho$ to be the restriction $q^{\star}_{\mathcal{F}^{c}_{n, \varepsilon}}$ of $q^{\star}$ over $\mathcal{F}^{c}_{n, \varepsilon}$, and let the corresponding probability measures be $Q^{\star}_{\mathcal{F}^{c}_{n, \varepsilon}}$ and $Q^{\star}$. Applying Markov's inequality to \eqref{eq:hell-risk-lemma-14} exactly as in \eqref{eq:hell-risk-lemma-7}, with $\eta = \exp\left\{-c n \varepsilon_n^2\right\}$, we get with $\mathbb{P}_{\theta_0}$-probability at least $\left(1 - \exp\left\{-\tfrac{c n \varepsilon_n^2}{2}\right\}\right)$:
\begin{align}\label{eq:hell-risk-lemma-15}
    \frac{1}{Q^{\star}(\mathcal{F}_{n, \varepsilon}^{c})}\int_{\mathcal{F}_{n, \varepsilon}^{c}}\ell_{n}(\theta, \theta_0)q^{\star}(\theta)d\theta -\mathrm{KL}(q^{\star}_{\mathcal{F}_{n, \varepsilon}^{c}}\parallel \pi_{\mathcal{F}_{n, \varepsilon}^{c}}) \leq \frac{c n \varepsilon_n^2}{2}.
\end{align}
Multiplying \eqref{eq:hell-risk-lemma-15} by $Q^{\star}(\mathcal{F}_{n, \varepsilon}^{c})$, rearranging, and adding the nonnegative term $Q^{\star}(\mathcal{F}_{n, \varepsilon}^{c})\log 2$ to the right hand side (so that the constant matches that in \eqref{eq:hell-risk-lemma-13}), we obtain the following, which also holds trivially when $Q^{\star}(\mathcal{F}_{n, \varepsilon}^{c}) = 0$:
\begin{align}\label{eq:hell-risk-lemma-16}
    0 \leq \int_{\mathcal{F}_{n, \varepsilon}^{c}}-\ell_n(\theta, \theta_0)q^{\star}(\theta)d\theta + Q^{\star}(\mathcal{F}_{n, \varepsilon}^{c})\mathrm{KL}(q^{\star}_{\mathcal{F}_{n, \varepsilon}^{c}}\parallel \pi_{\mathcal{F}_{n, \varepsilon}^{c}}) + \left[\frac{c n \varepsilon_n^2}{2}+\log 2\right]Q^{\star}(\mathcal{F}_{n, \varepsilon}^{c}).
\end{align}
Finally, by combining \eqref{eq:hell-risk-lemma-13} and \eqref{eq:hell-risk-lemma-16}, and using the identity:
\begin{align}\label{eq:hell-risk-lemma-17}
    \begin{split}
        \mathrm{KL}(q^{\star}\parallel \pi) &= \int q^{\star}(\theta) \log \frac{q^{\star}(\theta)}{\pi(\theta)}d\theta\\
        &= Q^{\star}(\mathcal{F}_{n, \varepsilon})\int_{\mathcal{F}_{n, \varepsilon}}q^{\star}_{\mathcal{F}_{n, \varepsilon}}(\theta)\log \frac{q^{\star}_{\mathcal{F}_{n, \varepsilon}}(\theta)}{\pi_{\mathcal{F}_{n, \varepsilon}}(\theta)}d\theta\\
        &\quad + Q^{\star}(\mathcal{F}_{n, \varepsilon}^{c})\int_{\mathcal{F}_{n, \varepsilon}^{c}}q^{\star}_{\mathcal{F}_{n, \varepsilon}^{c}}(\theta)\log \frac{q^{\star}_{\mathcal{F}_{n, \varepsilon}^c}(\theta)}{\pi_{\mathcal{F}_{n, \varepsilon}^{c}}(\theta)}d\theta\\
        &\quad + Q^{\star}(\mathcal{F}_{n, \varepsilon}^{c})\log \frac{Q^{\star}(\mathcal{F}_{n, \varepsilon}^{c})}{\Pi(\mathcal{F}_{n, \varepsilon}^{c})} + (1 - Q^{\star}(\mathcal{F}_{n, \varepsilon}^{c}))\log\frac{1 - Q^{\star}(\mathcal{F}_{n, \varepsilon}^{c})}{1 - \Pi(\mathcal{F}_{n, \varepsilon}^{c})},
    \end{split}
\end{align}
we have, on the intersection of the events in \eqref{eq:hell-risk-lemma-8}, \eqref{eq:hell-risk-lemma-9}, and \eqref{eq:hell-risk-lemma-15}, i.e., by a union bound with $\mathbb{P}_{\theta_0}$-probability at least $\left(1 - 3\exp\left\{-\tfrac{c n \varepsilon_n^2}{2}\right\}\right)$:
\begin{align}
\label{eq:hell-risk-lemma-18}
\begin{split}
    &c n \int_{\theta \in \mathcal{F}_{n, \varepsilon}, \mathcal{H}^{2}(\theta \parallel \theta_0) \geq \varepsilon^2}\mathcal{H}^{2}(\theta\parallel \theta_0)q^{\star}(\theta)d\theta + Q^{\star}(\mathcal{F}_{n, \varepsilon}^{c})\log \frac{Q^{\star}(\mathcal{F}_{n, \varepsilon}^{c})}{\Pi(\mathcal{F}_{n, \varepsilon}^{c})}\\
    &\qquad \qquad + (1 - Q^{\star}(\mathcal{F}_{n, \varepsilon}^{c}))\log\frac{1 - Q^{\star}(\mathcal{F}_{n, \varepsilon}^{c})}{1 - \Pi(\mathcal{F}_{n, \varepsilon}^{c})}\\
    &\leq \int - \ell_{n}(\theta, \theta_0)q^{\star}(\theta)d\theta + \mathrm{KL}(q^{\star}\parallel \pi) + \frac{c n \varepsilon_n^2}{2} + \log 2\\
    &= -\int \log \frac{p(y\mid \mathbf{X}, \theta)}{p(y\mid \mathbf{X}, \theta_0)}q^{\star}(\theta)d\theta + \mathrm{KL}(q^{\star}\parallel \pi) + \frac{c n \varepsilon_n^2}{2} + \log 2\\
    &\stackrel{\#}{\leq} -\int \log \frac{\varphi(y\mid \mathbf{X}, \theta, \xi^{\star})}{p(y\mid \mathbf{X}, \theta_0)}q^{\star}(\theta)d\theta + \mathrm{KL}(q^{\star}\parallel \pi) + \frac{c n \varepsilon_n^2}{2} + \log 2\\
    &= \mathcal{I}(q^{\star}, \xi^{\star}) +\mathrm{KL}(q^{\star}\parallel \pi) + \frac{c n \varepsilon_n^2}{2} + \log 2,
\end{split}
\end{align}
where the inequality \# in \eqref{eq:hell-risk-lemma-18} above follows from the fact that $\varphi(y\mid \mathbf{X}, \theta, \xi^{\star})$ is a minorizer of $p(y\mid \mathbf{X}, \theta)$. Using Lemma \ref{lemma-variational-optimizer} for $\alpha = 1$, for any choice of $\tilde{q}$ and $\tilde{\xi} \in \mathbb{R}_{+}^n$, we have:
\begin{align}\label{eq:hell-risk-lemma-19}
\mathcal{I}(q^{\star}, \xi^{\star}) + \mathrm{KL}(q^{\star} \parallel \pi) \leq \mathcal{I}(\tilde{q}, \tilde{\xi}) + \mathrm{KL}(\tilde{q} \parallel \pi).   
\end{align}
Taking $\tilde{q}$ as the following (with $\tilde{\xi}$ being arbitrary) as in the proof of Lemma \ref{lemma:majorization-Renyi}:
%
\begin{align}\label{eq:hell-risk-lemma-20}
    \begin{split}
        \tilde{q}(\theta) := \frac{\pi(\theta)}{\pi\left(\mathcal{B}_n(\theta_0, \varepsilon_n)\right)}\mathds{1}_{\mathcal{B}_n(\theta_0, \varepsilon_n)}(\theta),\quad \text{for all}\hspace{1mm}\theta\in \Theta,
    \end{split}
\end{align}
we obtain:
\begin{equation}\label{eq:hell-risk-lemma-21}
\mathrm{KL}(\tilde{q}\parallel\pi) = - \log \pi(\mathcal{B}_n(\theta_0, \varepsilon_n)),
\end{equation}
and with $\mathbb{P}_{\theta_0}$-probability at least $1 - [(D-1)^2n\varepsilon_n^2]^{-1}$, $\mathcal{I}(\tilde{q} , \tilde{\xi})$ in \eqref{eq:hell-risk-lemma-19} 
satisfies the following inequality:
\begin{align}\label{eq:hell-risk-lemma-22}
    \mathcal{I}(\tilde{q} , \tilde{\xi}) = -\int_{\theta\in \Theta}\tilde{q}(\theta)\log\frac{\varphi(y\mid \mathbf{X}, \theta, \tilde{\xi})}{p(y\mid \mathbf{X}, \theta_0)}d\theta \leq Dn\varepsilon_n^2,
\end{align}
for an arbitrary constant $D>1$.
%
Setting $\mathcal{F}_{n, \varepsilon} = \Theta$, we have $\mathcal{F}_{n, \varepsilon}^{c} = \emptyset$ and $Q^{\star}(\mathcal{F}_{n, \varepsilon}^{c}) = \Pi(\mathcal{F}_{n, \varepsilon}^{c}) = 0$. Hence \eqref{eq:hell-risk-lemma-16} holds deterministically, the event in \eqref{eq:hell-risk-lemma-15} is not needed, the two binary $\mathrm{KL}$ terms on the left hand side of \eqref{eq:hell-risk-lemma-18} vanish (with the convention $0\log\frac{0}{0} = 0$), and \eqref{eq:hell-risk-lemma-18} holds with $\mathbb{P}_{\theta_0}$-probability at least $\left(1 - 2\exp\left\{-\tfrac{c n \varepsilon_n^2}{2}\right\}\right)$. Using this along with \eqref{eq:hell-risk-lemma-19}, \eqref{eq:hell-risk-lemma-21}, and \eqref{eq:hell-risk-lemma-22}, we get:
\begin{align}\label{eq:hell-risk-lemma-23}
\begin{split}
    c\int_{\Theta}\mathcal{H}^{2}(\theta \parallel \theta_0) q^{\star}(\theta)d\theta &\leq c\varepsilon_n^2 + c\int_{\mathcal{H}^{2}(\theta\parallel \theta_0) \geq \varepsilon_n^2}\mathcal{H}^{2}(\theta\parallel \theta_0)q^{\star}(\theta)d\theta\\
    &\leq \left(D + \frac{3c}{2}\right)\varepsilon_n^{2} - \frac{1}{n}\log \pi(\mathcal{B}_n(\theta_0, \varepsilon_n)) + \frac{1}{n}\log 2,
\end{split}
\end{align}
with $\mathbb{P}_{\theta_0}$-probability at least $\left(1 - [(D-1)^2n\varepsilon_n^2]^{-1} - 2\exp\left\{-\tfrac{c n \varepsilon_n^2}{2}\right\}\right)$. 
We now derive an upper bound for the local Bayesian complexity, $-\log \pi(\mathcal{B}_n(\theta_0, \varepsilon_n))$ in \eqref{eq:hell-risk-lemma-23}, separately for Type I and Type II $\ssg$ likelihoods.

\textit{For Type I $\ssg$ likelihoods}. From the proof of Theorem~\ref{main-theorem-alpha-less-1-variational-risk-bound-Type-I} in Section \ref{app:proof-theorem-vb-alpha-less-1-type-1}, setting $\tilde{\xi}_i = \max\{s_{i0}, C_{\star}\varepsilon^2\}$ for $i\in [n]$, we have:
\begin{align}\label{eq:hell-risk-lemma-24}
    \begin{split}
        &-\log \pi(\mathcal{B}_n(\theta_0, \varepsilon_n)) 
        \leq -\log\left(\frac{C(a,b,\tau_0)\varepsilon_n^3}{\mathsf{Q}(\mathbf{X}, \tau_0)}\right)  + \log\left[2^{\frac{p}{2}}\Gamma\left(\frac{p}{2}+1\right)\left(\lambda_{\mathrm{max}}(\Sigma)\right)^{\frac{p}{2}}\right]\\
        &\qquad - \frac{p}{2}\log t_{-} + p\log\left(\frac{\mathsf{Q}(\mathbf{X}, \tau_0)}{\varepsilon_n^2}\right) + \frac{t_{+}}{2}\left(\Delta^2 + \frac{\varepsilon_n^4}{\mathsf{Q}^{2}(\mathbf{X}, \tau_0)\lambda_{\max}(\Sigma)}\right),
    \end{split}
\end{align}
where { $\mathsf{Q}(\mathbf{X}, \tau_0) = \max\{1, 50K\}\max\{\tau_0^2, \tau_0^{-4}\}\max\{1, \mathscr{E}_1, \mathscr{E}_2\}\max\{1, \lVert \mathbf X \rVert_{2, \infty}^2\}$}, 
$t_+ = \tau_0^2 + \mathsf{Q}^{-1}(\mathbf{X}, \tau_0)\varepsilon_n^3$, $t_{-} = \max\{\tau_0^2 - \mathsf{Q}^{-1}(\mathbf{X}, \tau_0)\varepsilon_n^3, 0\}$, $\Delta^{2} = \lVert \Sigma^{-\frac{1}{2}}(\beta_0-\mu)\rVert_{2}^{2}$, and $C(a,b,\tau_0) = \frac{b}{2\Gamma\left(\frac{a}{2}\right)}\left(\frac{b \tau_0^2}{2}\right)^{\frac{a}{2}-1}\exp\left\{-\frac{b\tau_0^{2}}{2}\right\}$. Note that, $t_{-} > 0$ since $\tau_0^{2}\mathsf{Q}(\mathbf{X}, \tau_0) \geq \max\{\tau_0^{4}, \tau_0^{-2}\}\geq 1 \geq \tfrac{1}{27} > \varepsilon_n^{3}$. Hence, using \eqref{eq:hell-risk-lemma-24} in \eqref{eq:hell-risk-lemma-23}, we obtain with $\mathbb{P}_{\theta_0}$-probability at least $\left(1 - 3\varepsilon_n - [(D-1)^2n\varepsilon_n^2]^{-1} - 2\exp\left\{-\tfrac{c n \varepsilon_n^2}{2}\right\}\right)$:
\begin{equation}\label{eq:hell-risk-lemma-25}
\boxed{c\int_{\Theta}\mathcal{H}^{2}(\theta \parallel \theta_0) q^{\star}(\theta)d\theta \leq \left(D + \frac{3c}{2}\right)\varepsilon_n^2 + \frac{p\log p}{n} + \frac{C_3 p}{n}\log\left(\frac{1}{\varepsilon_n}\right)},
\end{equation}
for arbitrary constant $D>1$, some $c\in \mathbb{R}^{+}$, and $C_3\in \mathbb{R}^{+}$, where for $\log^{+}x:=\max\{\log x, 0\}$, $C_3$ is given by:
\begin{equation}\label{eq:hell-risk-lemma-26}
\begin{gathered}
C_3:=5+\log 2+2\log\mathsf{Q}(\mathbf{X},\tau_0)+\log^{+}\!\left(\frac{1}{C(a,b,\tau_0)}\right)+\log^{+}\!\left(\frac{\sqrt{4\lambda_{\max}(\Sigma)}}{\tau_0}\right)+\\
\frac{1}{2}\left(\tau_0^2+\mathsf{Q}^{-1}(\mathbf{X},\tau_0)\right)\left(\Delta^2+\frac{\mathsf{Q}^{-2}(\mathbf{X},\tau_0)}{\lambda_{\max}(\Sigma)}\right),
\end{gathered}
\end{equation}
which depends only on prior hyperparameters $(\mu, \Sigma, a, b)$, design matrix $\mathbf{X}$, and true parameter $\theta_0$. Note that, $C_{3}=C_1-(1-\log 2)<C_{1}$, where $C_{1}$ is as in \eqref{eq:proof-theorem1-bound-constants-type-I}.

\textit{For Type II $\ssg$ likelihoods}. Here, $\theta = \beta$. From the proof of Theorem~\ref{main-theorem-alpha-less-1-variational-risk-bound-Type-II} in Section \ref{app:proof-theorem-vb-alpha-less-1-type-2}, setting $\tilde{\xi}_i = |\mathbf{x}_i^{\top}\beta_0|$ for $i\in [n]$, we have:
\begin{equation}\label{eq:hell-risk-lemma-27}
\begin{split}
-\log (\pi\left(\mathcal{B}_n(\beta_0, \varepsilon_n)\right))
&\leq \frac{1}{2}\left(\Delta^2 + \frac{\varepsilon_n^6}{\lambda_{\mathrm{max}}(\Sigma)\mathsf{Q}^{2}(\mathbf{X}, \beta_0)}\right) + p\log\left(\frac{\mathsf{Q}(\mathbf{X},\beta_0)}{\varepsilon_n^3}\right)\\
&\qquad + \log\left(2^{\frac{p}{2}}\Gamma\left(\frac{p}{2}+1\right)(\lambda_{\mathrm{max}}(\Sigma))^{\frac{p}{2}}\right),
\end{split}
\end{equation}
where 
$\mathsf{Q}(\mathbf{X}, \beta_0) = \max\left\{\max\left\{4\lVert\mathbf{X}\rVert_{2, \infty}, 8\lVert\mathbf{X}\rVert_{2, \infty}^{2}\lVert\beta_0\rVert_{2}\right\}\bar{m}\left(1 + \exp\{\lVert\mathbf{X}\rVert_{2, \infty}\lVert\beta_0\rVert_{2}\}\right), 1\right\}$ and $\Delta^{2} = \lVert \Sigma^{-\frac{1}{2}}(\beta_0-\mu)\rVert_{2}^{2}$. Hence, using \eqref{eq:hell-risk-lemma-27} in \eqref{eq:hell-risk-lemma-23}, we obtain with $\mathbb{P}_{\theta_0}$-probability at least $\left(1 - 2\varepsilon_n - [(D-1)^2n\varepsilon_n^2]^{-1} - 2\exp\left\{-\tfrac{c n \varepsilon_n^2}{2}\right\}\right)$:
\begin{equation}\label{eq:hell-risk-lemma-28}
\boxed{c\int_{\Theta}\mathcal{H}^{2}(\theta \parallel \theta_0) q^{\star}(\theta)d\theta \leq \left(D + \frac{3c}{2}\right)\varepsilon_n^2 + \frac{p\log p}{n} + \frac{C_4 p}{n}\log\left(\frac{1}{\varepsilon_n}\right)},
\end{equation}
for arbitrary constant $D>1$, some $c\in \mathbb{R}^{+}$, and $C_4\in \mathbb{R}^{+}$, where $C_4$ is given by:
\begin{align}\label{eq:hell-risk-lemma-29}
C_4 := \frac{1}{2}\left(6 + 3\log 2 + \log^+ \lambda_{\max}(\Sigma) + 2\log \mathsf{Q}( \mathbf{X}, \beta_0) + \Delta^2 + \left(\lambda_{\max}(\Sigma)\mathsf{Q}^2(\mathbf{X}, \beta_0)\right)^{-1}\right),
\end{align}
which depends only on prior hyperparameters $(\mu, \Sigma)$, design matrix $\mathbf{X}$, true parameter $\beta_0$, and $\bar{m}$.
\end{proof}

\subsection{Fixed-Dimensional Control of the Variational Risk Bound Constants}
\label{subsec:fixed-p-constant-control}

The constants $C_1,C_2,C_3$, and $C_4$ in \eqref{eq:proof-theorem1-bound-constants-type-I}, \eqref{eq:proof-theorem1-bound-constants-type-II}, \eqref{eq:hell-risk-lemma-26}, and~\eqref{eq:hell-risk-lemma-29}, appearing in Theorems~\ref{main-theorem-alpha-less-1-variational-risk-bound-Type-I}-\ref{main-theorem-alpha-equals-1-variational-risk-bound-Type-II} of the main manuscript, as summarized in Table~\ref{tab:risk-bound-constants}, depend on the prior hyperparameters, the true parameter, and the design matrix. Through $\mathsf{Q}(\mathbf{X},\tau_0)$ and $\mathsf{Q}(\mathbf{X},\beta_0)$, they also depend on the Type I likelihood quantities $K$, $\mathscr{E}_1$, and $\mathscr{E}_2$, and on the Type II trial sizes through $\bar{m}$, respectively. Here in Lemmata~\ref{lemma:control-C1-C3} and \ref{lemma:control-C2-C4}, we make this dependence explicit and show that these constants remain uniformly controlled under the fixed-dimensional ($p$) regime in the asymptotic limit of $n\to \infty$, provided the rows of the design matrix are uniformly bounded. Throughout, $\log^{+}x:=\max\{\log x, 0\}$ for $x>0$.

\begin{table}[H]
\centering
\caption{\footnotesize{Constants appearing in the variational risk bounds of Theorems~\ref{main-theorem-alpha-less-1-variational-risk-bound-Type-I}-\ref{main-theorem-alpha-equals-1-variational-risk-bound-Type-II} of the main manuscript.}}
\label{tab:risk-bound-constants}
\footnotesize
\renewcommand{\arraystretch}{1.6}
\begin{tabular}{@{}c p{0.20\textwidth} p{0.70\textwidth}@{}}
\toprule
\toprule
Constant & Theorem & Explicit form \\
\midrule

$C_1$
&
Theorem~\ref{main-theorem-alpha-less-1-variational-risk-bound-Type-I}
&
$
\begin{aligned}[t]
&\textstyle 6+2\log\mathsf{Q}(\mathbf{X},\tau_0)
+\log^{+}\!\left(
\frac{1}{C(a,b,\tau_0)}
\right)
+\log^{+}\!\left(
\frac{\sqrt{4\lambda_{\max}(\Sigma)}}{\tau_0}
\right)\\
&\textstyle +
\frac{1}{2}
\left(
\tau_0^2+\mathsf{Q}^{-1}(\mathbf{X},\tau_0)
\right)
\left(
\Delta^2+\frac{\mathsf{Q}^{-2}(\mathbf{X},\tau_0)}{\lambda_{\max}(\Sigma)}
\right)
\end{aligned}
$
\\

$C_2$
&
Theorem~\ref{main-theorem-alpha-less-1-variational-risk-bound-Type-II}
&
$
\frac{1}{2}\left\{
8+\log 2+\log^{+}\lambda_{\max}(\Sigma)
+2\log\mathsf{Q}(\mathbf{X},\beta_0)
+\Delta^2
+
\left[
\lambda_{\max}(\Sigma)
\mathsf{Q}^2(\mathbf{X},\beta_0)
\right]^{-1}
\right\}
$
\\

$C_3$
&
Theorem~\ref{main-theorem-alpha-equals-1-variational-risk-bound-Type-I}
&
$
\begin{aligned}[t]
&\textstyle 5+\log 2+2\log\mathsf{Q}(\mathbf{X},\tau_0)
+\log^{+}\!\left(
\frac{1}{C(a,b,\tau_0)}
\right)
+\log^{+}\!\left(
\frac{\sqrt{4\lambda_{\max}(\Sigma)}}{\tau_0}
\right)\\
&\textstyle +
\frac{1}{2}
\left(
\tau_0^2+\mathsf{Q}^{-1}(\mathbf{X},\tau_0)
\right)
\left(
\Delta^2+\frac{\mathsf{Q}^{-2}(\mathbf{X},\tau_0)}{\lambda_{\max}(\Sigma)}
\right)
\end{aligned}
$
\\

$C_4$
&
Theorem~\ref{main-theorem-alpha-equals-1-variational-risk-bound-Type-II}
&
$
\frac{1}{2}\left\{
6+3\log 2+\log^{+}\lambda_{\max}(\Sigma)
+2\log\mathsf{Q}(\mathbf{X},\beta_0)
+\Delta^2
+
\left[
\lambda_{\max}(\Sigma)
\mathsf{Q}^2(\mathbf{X},\beta_0)
\right]^{-1}
\right\}
$
\\

\bottomrule
\bottomrule
\end{tabular}
\end{table}

\begin{lemma}[Uniform control of variational risk bounds for Type I $\ssg$ likelihoods]
\label{lemma:control-C1-C3}
Consider the dimension $p$, the prior hyperparameters $(\mu, \Sigma, a, b)$, the true parameter $\theta_0 = (\beta_0^{\top}, \tau_0^{2})^{\top}$, and the Type I $\ssg$ likelihood to be fixed as $n\to \infty$. Further, uniformly over $n\in \mathbb{N}$, assume that $\|\mathbf{X}\|_{2, \infty} \leq B_{X}$ for a positive constant $B_X$.
Then the constants $C_1$ and $C_3$ appearing in Theorems~\ref{main-theorem-alpha-less-1-variational-risk-bound-Type-I} and~\ref{main-theorem-alpha-equals-1-variational-risk-bound-Type-I} of the main manuscript are of $\mathcal{O}(1)$ with respect to $n$.
\end{lemma}

\begin{proof}
We first note that $\Delta^{2} = \|\Sigma^{-\frac{1}{2}}(\beta_0-\mu)\|_{2}^{2}$ and $\lambda_{\max}(\Sigma)$ are positive constants that do not depend on $n$; hence, so is $\log^{+}\big(\sqrt{4\lambda_{\max}(\Sigma)}/\tau_0\big)$.

Next recall that:
$$\mathsf{Q}(\mathbf{X},\tau_0) = \max\{1,50K\}\max\{\tau_0^2,\tau_0^{-4}\}\max\{1,\mathscr{E}_1,\mathscr{E}_2\}\max\{1,\lVert\mathbf X\rVert_{2,\infty}^2\} \geq 1.$$
Therefore, $\log\mathsf{Q}(\mathbf{X},\tau_0)\geq 0$ and $\mathsf{Q}^{-2}(\mathbf{X},\tau_0) \leq \mathsf{Q}^{-1}(\mathbf{X},\tau_0) \leq 1$.
Moreover, the Lipschitz constant $K$ in Assumption~\ref{ass:1}(iv) and the moments $\mathscr{E}_1,\mathscr{E}_2$ in Assumption~\ref{ass:3} of the main manuscript depend only on the fixed function $h$. Hence, since $\|\mathbf{X}\|_{2,\infty}\leq B_X$, uniformly over $n\in\mathbb{N}$:
\begin{align*}
1\leq \mathsf{Q}(\mathbf{X},\tau_0)\leq \overline{\mathsf{Q}}_{\mathrm{I}}:=\max\{1,50K\}\max\{\tau_0^2,\tau_0^{-4}\}\max\{1,\mathscr{E}_1,\mathscr{E}_2\}\max\{1,B_X^2\}.
\end{align*}
Finally, $C(a, b, \tau_0) = \tfrac{b}{2\Gamma(a/2)}\left(\tfrac{b\tau_0^2}{2}\right)^{\tfrac a2-1}\exp\left\{-\tfrac{b\tau_0^2}{2}\right\}$ is the $\mathcal{G}(a/2, b/2)$ density with the shape-rate parameterization (see Section~\ref{sec:parameterization-details}) evaluated at $\tau_0^2>0$. Hence, $C(a,b,\tau_0)$ is a strictly positive constant that does not depend on $n$, and so is $\log^{+}\big(1/C(a,b,\tau_0)\big)$.

Combining the preceding observations, uniformly over $n\in\mathbb{N}$:
\begin{align*}
C_1\leq 6+2\log\overline{\mathsf{Q}}_{\mathrm{I}}+\log^{+}\!\left(\frac{1}{C(a,b,\tau_0)}\right)+\log^{+}\!\left(\frac{\sqrt{4\lambda_{\max}(\Sigma)}}{\tau_0}\right)+\frac{1}{2}\left(\tau_0^2+1\right)\left(\Delta^2+\frac{1}{\lambda_{\max}(\Sigma)}\right),
\end{align*}
and $C_3=C_1-(1-\log 2)<C_1$. Hence, we conclude that both $C_1$ and $C_3$ are of $\mathcal{O}(1)$ with respect to $n$.
\end{proof}

\begin{lemma}[Uniform control of variational risk bounds for Type II $\ssg$ likelihoods]
\label{lemma:control-C2-C4}
Consider the dimension $p$, the prior hyperparameters $(\mu, \Sigma)$, and the true parameter $\beta_0$ to be fixed as $n\to \infty$. Further, uniformly over $n\in \mathbb{N}$, assume that $\|\mathbf{X}\|_{2, \infty} \leq B_{X}$ for positive constant $B_X$.
Then the constants $C_2$ and $C_4$ appearing in Theorems~\ref{main-theorem-alpha-less-1-variational-risk-bound-Type-II} and~\ref{main-theorem-alpha-equals-1-variational-risk-bound-Type-II} of the main manuscript are of $\mathcal{O}(1)$ with respect to $n$.
\end{lemma}

\begin{proof}
As in the proof of Lemma~\ref{lemma:control-C1-C3}, $\Delta^{2} = \|\Sigma^{-\frac{1}{2}}(\beta_0-\mu)\|_{2}^{2}$ and $\lambda_{\max}(\Sigma)$ are positive constants that do not depend on $n$; hence, so is $\log^{+}\lambda_{\max}(\Sigma)$. Also, by definition, we have:
$$\mathsf{Q}(\mathbf{X}, \beta_0) = \max\left\{\max\left\{4\lVert\mathbf{X}\rVert_{2, \infty}, 8\lVert\mathbf{X}\rVert_{2, \infty}^{2}\lVert\beta_0\rVert_{2}\right\}\bar{m}\left(1 + \exp\{\lVert\mathbf{X}\rVert_{2, \infty}\lVert\beta_0\rVert_{2}\}\right),1\right\} \geq 1,$$
so that $\log\mathsf{Q}(\mathbf{X},\beta_0)\geq 0$ and $\left[\lambda_{\max}(\Sigma)\mathsf{Q}^{2}(\mathbf{X},\beta_0)\right]^{-1}\leq \lambda_{\max}^{-1}(\Sigma)$.

Thus, it only remains to provide a uniform upper bound on $\log \mathsf{Q}(\mathbf{X}, \beta_0)$, which follows from the uniform upper bounds $\bar{m}$ on $m_i$ and $\|\mathbf X\|_{2,\infty}$, since $\mathsf{Q}(\mathbf{X},\beta_0)$ is non-decreasing in both. That is, uniformly over $n\in\mathbb{N}$:
\begin{align*}
1\leq \mathsf{Q}(\mathbf{X},\beta_0)\leq \overline{\mathsf{Q}}_{\mathrm{II}}:=\max\left\{\max\left\{4B_X, 8B_X^{2}\lVert\beta_0\rVert_{2}\right\}\bar{m}\left(1+\exp\{B_X\lVert\beta_0\rVert_{2}\}\right),1\right\}.
\end{align*}
Hence, uniformly over $n\in\mathbb{N}$:
\begin{align*}
C_2&\leq \frac{1}{2}\left\{8+\log 2+\log^{+}\lambda_{\max}(\Sigma)+2\log\overline{\mathsf{Q}}_{\mathrm{II}}+\Delta^2+\lambda_{\max}^{-1}(\Sigma)\right\},\\
C_4&\leq \frac{1}{2}\left\{6+3\log 2+\log^{+}\lambda_{\max}(\Sigma)+2\log\overline{\mathsf{Q}}_{\mathrm{II}}+\Delta^2+\lambda_{\max}^{-1}(\Sigma)\right\},
\end{align*}
and therefore, the constants $C_2$ and $C_4$ are of $\mathcal{O}(1)$ with respect to $n$.
\end{proof}

\newpage

\section{Empirical Analysis of Variational Risk Bounds under \texorpdfstring{$\alpha$-R\'{e}nyi}{alpha-Renyi} Divergence}\label{sec:empirical-evidence-risk-bounds}

We conduct an empirical investigation of the variational risk bounds derived in Sections \ref{app:proof-theorem-vb-alpha-less-1-type-1} and \ref{app:proof-theorem-vb-alpha-less-1-type-2} under the $\alpha$-R\'{e}nyi divergence framework for both Type I and Type II $\mathsf{SSG}$ likelihoods. As representative examples of these two classes, we consider the Laplace and Negative-Binomial models, corresponding to their respective formulations in \eqref{eq:param-laplace} and \eqref{eq:param-negbin}. Recall that, the $\alpha$-R\'{e}nyi divergence, for $\alpha \in (0,1)$, is given as:
\begin{align}\label{eq:alpha-renyi-divergence}
\begin{split}
    D_{\alpha}(\theta, \theta_0) &:= \frac{1}{n(\alpha-1)}\log \int \left\{p(y\mid \mathbf{X}, \theta)\right\}^{\alpha}\left\{p(y\mid \mathbf{X}, \theta_0)\right\}^{1-\alpha}dy\\
    &= \frac{1}{n}\sum_{i\in [n]}D_{\alpha}(p_{\theta, i}, p_{\theta_0, i}).
\end{split}
\end{align}

\subsection{Type I Laplace \texorpdfstring{$\mathsf{SSG}$}{SSG} Likelihood}

Consider two Laplace distributions with densities $p_{\theta, i}$ and $p_{\theta_0, i}$, independently over $i\in [n]$, both having the parameterization in \eqref{eq:param-laplace}, where $\theta = (\beta^{\top}, \tau^{2})^{\top}$ and $\theta_0 = (\beta_0^{\top}, \tau_0^{2})^{\top}$. Following~\cite{GIL2013124}, $D_{\alpha}(p_{\theta, i}, p_{\theta_0, i})$ in \eqref{eq:alpha-renyi-divergence}, for $\alpha \in (0,1)$, is obtained as:
\begin{align}\label{eq:alpha-renyi-divergence-single}
    \begin{split}
        D_{\alpha}(p_{\theta, i}, p_{\theta_0, i}) &:= \frac{1}{\alpha-1}\log \left[\left(\frac{\tau}{2}\right)^{\alpha}\left(\frac{\tau_0}{2}\right)^{1-\alpha}\left(\Psi_1(\mathbf{x}_i, \theta, \theta_0, \alpha) + \Psi_2(\mathbf{x}_i, \theta, \theta_0, \alpha)\right)\right],
    \end{split}
\end{align}
where $\Psi_1(\mathbf{x}_i, \theta, \theta_0, \alpha)$ and $\Psi_2(\mathbf{x}_i, \theta, \theta_0, \alpha)$ are:
\begin{align}\label{eq:alpha-renyi-divergence-single-components}
    \begin{split}
        \Psi_1(\mathbf{x}_i, \theta, \theta_0, \alpha) &:= \frac{\exp\left\{-\tau_0(1-\alpha)|\mathbf{x}_i^{\top}(\beta-\beta_0)|\right\} + \exp\left\{-\tau\alpha |\mathbf{x}_i^{\top}(\beta-\beta_0)|\right\}}{\alpha \tau + (1-\alpha)\tau_0}\\
        \Psi_2(\mathbf{x}_i, \theta, \theta_0, \alpha) &:= \frac{\exp\left\{-\tau_0(1-\alpha)|\mathbf{x}_i^{\top}(\beta-\beta_0)|\right\} - \exp\left\{-\tau\alpha |\mathbf{x}_i^{\top}(\beta-\beta_0)|\right\}}{\alpha \tau - (1-\alpha)\tau_0}.
    \end{split}
\end{align}

From \eqref{eq:param-laplace}, for an unknown precision parameter $\tau \in \mathbb{R}^{+}$, the Laplace $\mathsf{SSG}$ likelihood for $y_i\in \mathbb{R}$, denoted as $y_i\mid \theta \sim \mathrm{Laplace}(\mu_i = \mathbf{x}_i^{\top}\beta, \tau)$, is:
\begin{align}\label{eq:param-Laplace-theory-validation}
    p(y_i\mid \mathbf{x}_i, \theta) = \tfrac{\tau}{2}\exp\left\{-\tau|y_i - \mathbf{x}_i^{\top}\beta|\right\},
\end{align}
independently for $i\in [n]$, where $\theta = (\beta^{\top}, \tau^{2})^{\top}$. The dataset $\mathcal{D}_n := \{(\mathbf{x}_i, y_i): i\in [n]\}$ is simulated repeatedly for $\texttt{nreps}=50$ replications using the model parameters in \eqref{eq:param-Laplace-theory-validation}, configured as: the true precision is fixed at $\tau_0^2 = 8$, the true regression coefficients are drawn as $\beta_0 \sim \mathcal{N}_{p}(0, I_{p})$, and the design matrix $\mathbf{X} = (\mathbf{x}_1, \ldots, \mathbf{x}_n)^{\top} \in \mathbb{R}^{n\times p}$ has entries $x_{i1}=1$ (intercept) and $x_{ij}$'s generated independently from the standard Gaussian distribution for $j=2, \ldots, p$.

For a fixed number of features $p=9$, data are simulated for two choices of sample size, $n=2000$ and $n = 10000$. For each configuration and $\alpha \in \{0.2, 0.3, 0.4, 0.6, 0.8, 0.95\}$, the $\tssg$ algorithm in Algorithm~\ref{main-alg:tavie-em} of the main manuscript is applied to obtain the optimal variational parameter estimate $\xi^{\star}$, which is then used to evaluate the integrated variational risk, corresponding to the left hand side of \eqref{eq:proof-theorem1-final-TypeI-bound}, based on expressions in \eqref{eq:alpha-renyi-divergence}-\eqref{eq:param-Laplace-theory-validation}. The integral with respect to the optimal variational posterior distribution is approximated using Monte Carlo simulations with $\texttt{n\_MC} = 100$ samples.

\begin{figure}[H]
    \centering
    \includegraphics[width=0.8\linewidth]{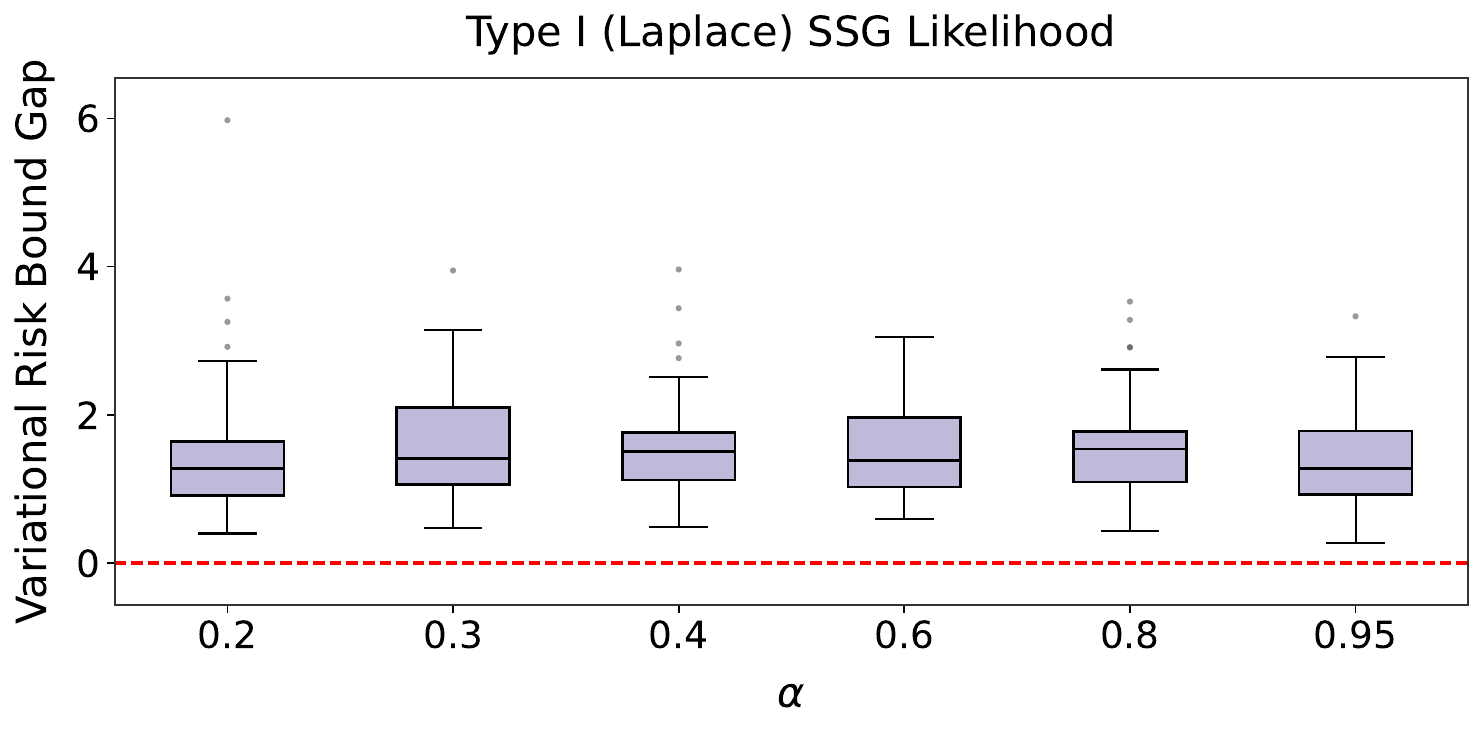}
    \caption{\footnotesize{Variational risk bound gap under $\alpha$-R\'{e}nyi divergence for Laplace Type I $\mathsf{SSG}$ likelihood ($n = 2000$).}}
    \label{fig:variational_risk_bound_gap_Laplace_n_2000_p_8}
\end{figure}

\begin{figure}[H]
    \centering
    \includegraphics[width=0.8\linewidth]{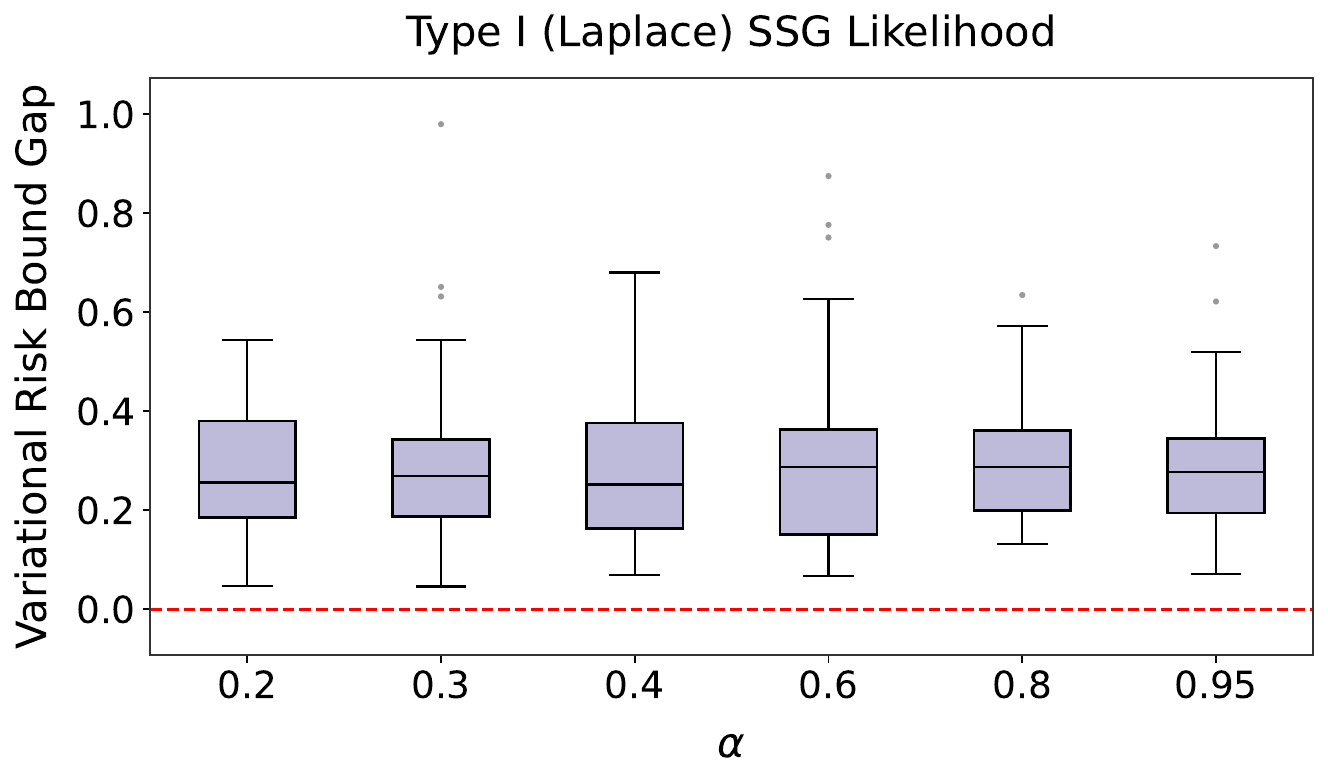}
    \caption{\footnotesize{Variational risk bound gap under $\alpha$-R\'{e}nyi divergence for Laplace Type I $\mathsf{SSG}$ likelihood ($n = 10000$).}}
    \label{fig:variational_risk_bound_gap_Laplace_n_10000_p_8}
\end{figure}

The right hand side of the variational risk bound in \eqref{eq:proof-theorem1-final-TypeI-bound} is computed accordingly, and the variational risk bound gap, defined as the difference between the right and left hand sides of \eqref{eq:proof-theorem1-final-TypeI-bound}, is evaluated. Boxplots of the empirical bound gap across two sample sizes are presented in Figures \ref{fig:variational_risk_bound_gap_Laplace_n_2000_p_8} and \ref{fig:variational_risk_bound_gap_Laplace_n_10000_p_8}, where the positive gap values demonstrate empirical validation of the theoretical variational risk bound in \eqref{eq:proof-theorem1-final-TypeI-bound}.

\subsection{Type II Negative-Binomial \texorpdfstring{$\mathsf{SSG}$}{SSG} Likelihood}

Consider two Negative-Binomial distributions with probability mass functions $p_{\beta, i}$ and $p_{\beta_0, i}$, independently over $i\in [n]$, both having the parameterization in \eqref{eq:param-negbin}. $D_{\alpha}(p_{\beta, i}, p_{\beta_0, i})$ in \eqref{eq:alpha-renyi-divergence}, for $\alpha \in (0,1)$, is obtained as:
\begin{align}\label{eq:alpha-renyi-divergence-single-nb}
    \begin{split}
        D_{\alpha}(p_{\beta, i}, p_{\beta_0, i}) &:= \frac{m_i}{\alpha-1}\log \left[\frac{\sigma(\mathbf{x}_i^{\top}\beta)^{\alpha}\sigma(\mathbf{x}_i^{\top}\beta_0)^{1-\alpha}}{1 - (1 - \sigma(\mathbf{x}_i^{\top}\beta))^{\alpha}(1 - \sigma(\mathbf{x}_i^{\top}\beta_0))^{1-\alpha}}\right],
    \end{split}
\end{align}
where $\sigma(t) = [1 + \exp\left\{-t\right\}]^{-1}$ is the sigmoid function. From \eqref{eq:param-negbin}, for a fixed size parameter $m_i > 0$, the Negative-Binomial $\mathsf{SSG}$ likelihood for $y_i\in \mathbb{N} \cup \{0\}$, denoted as $y_i\mid \mathbf{x}_{i},\beta \sim \mathrm{NB}(m_i, p_i = \sigma(\mathbf{x}_i^{\top}\beta))$, is:
\begin{align}\label{eq:param-NB-theory-validation}
    p(y_i\mid \mathbf{x}_i, \beta) = \binom{y_i + m_i - 1}{y_i}p_i^{m_i}(1-p_i)^{y_i},
\end{align}
independently for $i\in [n]$. The dataset $\mathcal{D}_n := \{(\mathbf{x}_i, y_i): i\in [n]\}$ is simulated repeatedly for $\texttt{nreps}=50$ replications using the model parameters in \eqref{eq:param-NB-theory-validation}, configured as: $m_i = m = 10$ for $i \in [n]$, the true regression coefficients are drawn as $\beta_0 \sim \mathcal{N}_{p}(0, 0.5 I_{p})$, and the design matrix $\mathbf{X} = (\mathbf{x}_1, \ldots, \mathbf{x}_n)^{\top} \in \mathbb{R}^{n\times p}$ has entries $x_{i1}=1$ (intercept) and $x_{ij}$'s generated independently from the standard Gaussian distribution for $j=2, \ldots, p$.

For a fixed number of features $p=9$, data are simulated for two choices of sample size, $n=2000$ and $n = 10000$. For each configuration and $\alpha \in \{0.2, 0.3, 0.4, 0.6, 0.8, 0.95\}$, the $\tssg$ algorithm in Algorithm~\ref{main-alg:tavie-em} of the main manuscript is applied to obtain the optimal variational parameter estimate $\xi^{\star}$, which is then used to evaluate the integrated variational risk, corresponding to the left hand side of \eqref{eq:proof-theorem1-final-TypeII-bound}, based on expressions in \eqref{eq:alpha-renyi-divergence}, \eqref{eq:alpha-renyi-divergence-single-nb}, and  \eqref{eq:param-NB-theory-validation}. The integral with respect to the optimal variational posterior distribution is approximated using Monte Carlo simulations with $\texttt{n\_MC} = 100$ samples.

The right hand side of the variational risk bound in \eqref{eq:proof-theorem1-final-TypeII-bound} is computed accordingly, and the variational risk bound gap, defined as the difference between the right and left hand sides of \eqref{eq:proof-theorem1-final-TypeII-bound}, is evaluated. Boxplots of the empirical bound gap across two sample sizes are presented in Figures \ref{fig:variational_risk_bound_gap_NB_n_2000_p_8} and \ref{fig:variational_risk_bound_gap_NB_n_10000_p_8}, where the positive gap values demonstrate empirical validation of the theoretical variational risk bound in \eqref{eq:proof-theorem1-final-TypeII-bound}.

\begin{figure}[H]
    \centering
    \includegraphics[width=0.8\linewidth]{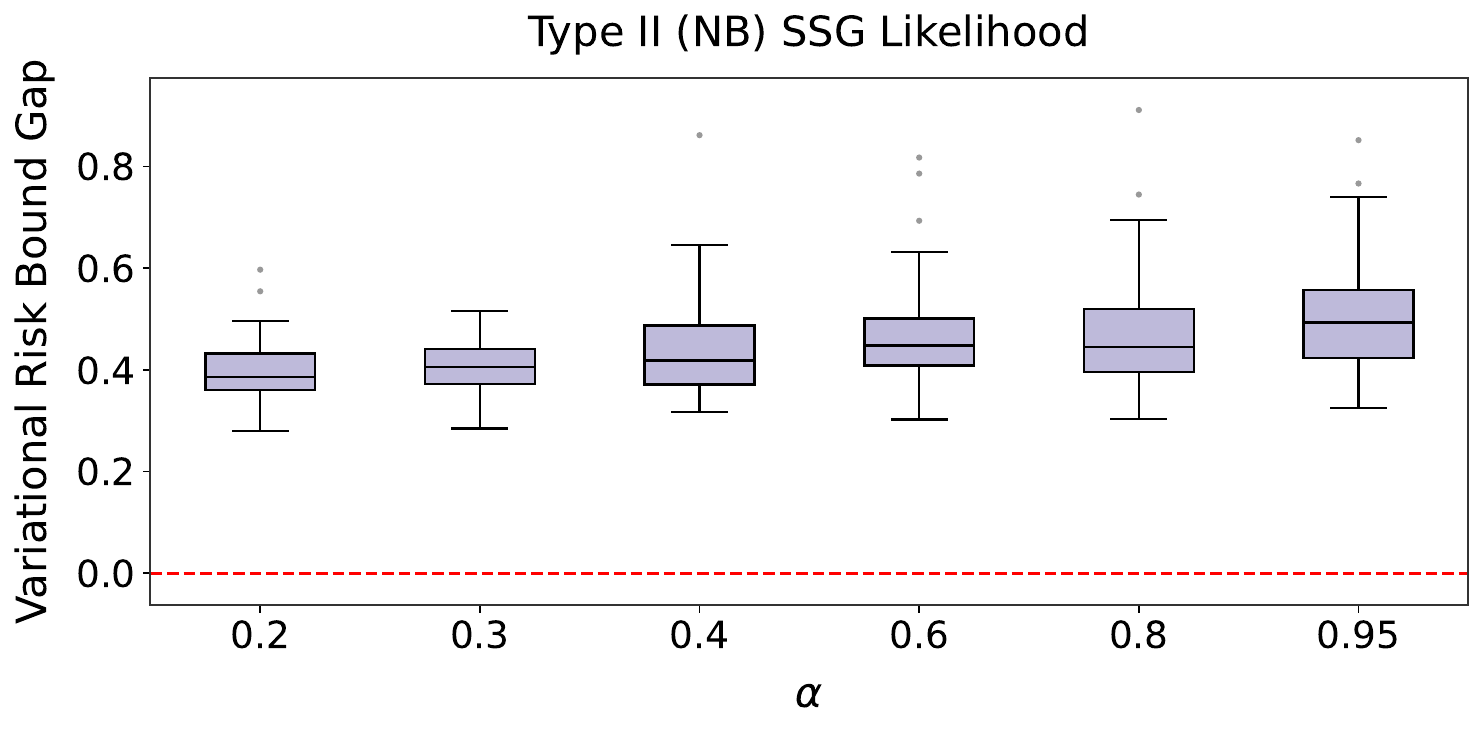}
    \caption{\footnotesize{Variational risk bound gap under $\alpha$-R\'{e}nyi divergence for Negative-Binomial Type II $\mathsf{SSG}$ likelihood ($n = 2000$).}}
    \label{fig:variational_risk_bound_gap_NB_n_2000_p_8}
\end{figure}

\begin{figure}[H]
    \centering
    \includegraphics[width=0.8\linewidth]{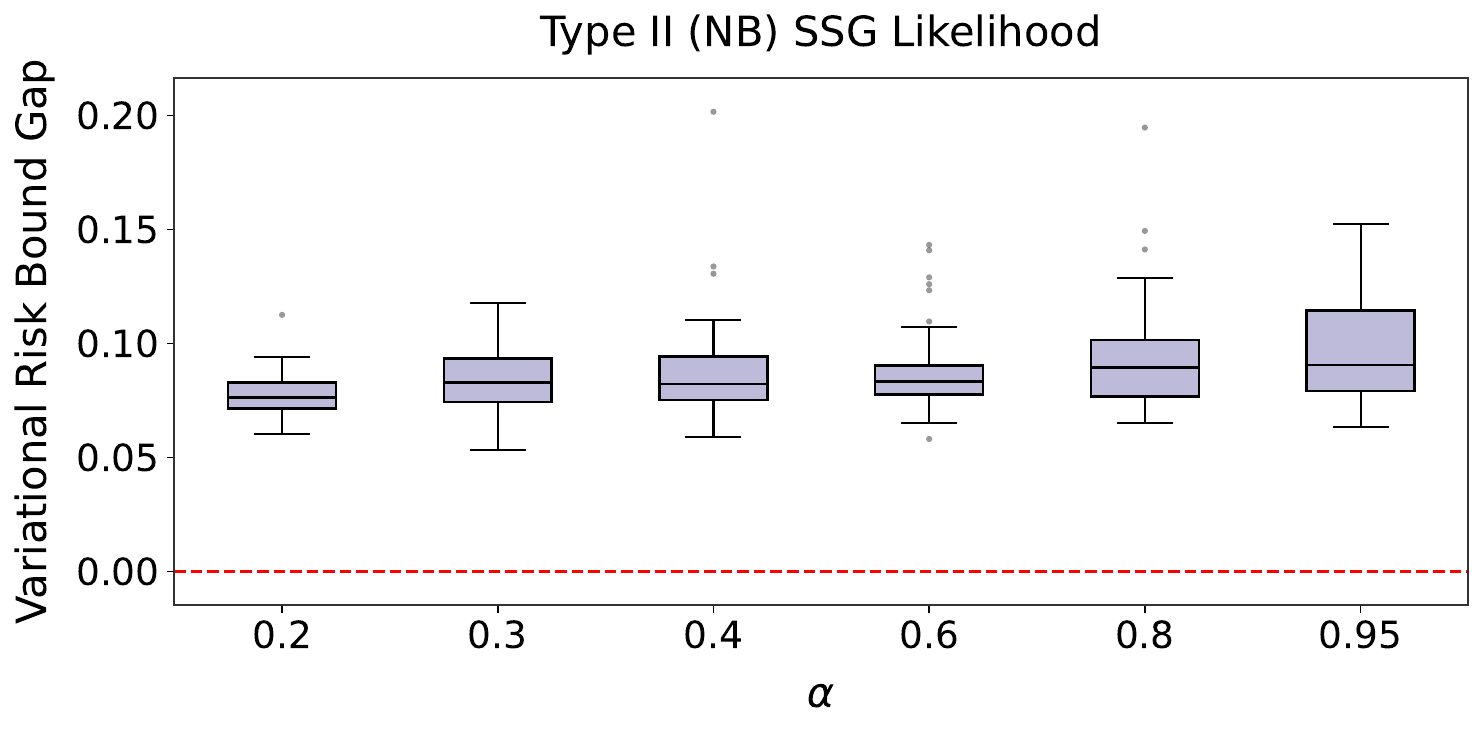}
    \caption{\footnotesize{Variational risk bound gap under $\alpha$-R\'{e}nyi divergence for Negative-Binomial Type II $\mathsf{SSG}$ likelihood ($n = 10000$).}}
    \label{fig:variational_risk_bound_gap_NB_n_10000_p_8}
\end{figure}

\newpage

\section{Experimental Settings of Competing Methods}\label{app-competing-methods-details}

This section outlines the implementation details and experimental configurations of the competing variational inference (VI) and Monte Carlo (MC) algorithms used for comparison with $\tssg$ in Section~\ref{main-sec:TAVIE-SSG-applications} of the main manuscript.

\begin{itemize}
    \item \textbf{Automatic Differentiation Variational Inference} (ADVI): We employ two variants of this black-box variational inference (BBVI) algorithm viz., ADVI mean-field (MF) and ADVI full-rank (FR)~\citep{advi}. A Gaussian variational family is specified for the regression parameter $\beta \in \mathbb{R}^{p}$, and a Gamma variational family for the precision parameter $\tau^{2} = \sigma^{-2} \in \mathbb{R}^{+}$, i.e., $q(\tau^2) \equiv \mathcal{G}(a, b)$ (whenever applicable to the $\mathsf{SSG}$ likelihood, such as Laplace, Student's-$t$, or scale-mixture of Gaussian models). In the ADVI FR variant, a full covariance structure is adopted with $q(\beta) \equiv \mathcal{N}_{p}(\mu, LL^{\top})$, where $L$ is the Cholesky factor. The ADVI MF variant constrains $q(\beta)$ to have a diagonal covariance. Both variants use \texttt{AdamW} optimization~\citep{AdamW} with a learning rate of $10^{-3}$ and weight decay regularization. Prior hyperparameters are set as $\mu = 0$, $\Sigma = I$, $a=0.05$, and $b = 0.05$, respectively. Gradients are estimated via the reparameterization trick, and optimization terminates when the $\mathsf{ELBO}$ improvement falls below a tolerance of $10^{-9}$ for a patience of $10^{4}$ iterations. Particularly, for Bayesian quantile regression on U.S. Census 2000 data in Section~\ref{main-subsec:Census-data-study} of the main manuscript, the \texttt{AdamW} learning rate is set to $5\times 10^{-3}$ and maximum iterations for convergence is increased to $\texttt{max\_iter} = 5\times 10^4$.

    \item \textbf{Deterministic ADVI} (DADVI): The DADVI algorithm~\citep{dadvi} is executed with its default configuration using $N=30$ gradient samples per iteration. The implementation is based on the \texttt{PyMC} framework~\citep{pymc,pymc2023peerj} and customized for each $\mathsf{SSG}$ model under study\footnote{The \texttt{Python} implementation of DADVI is available at \url{https://github.com/martiningram/dadvi}. The corresponding \texttt{PyMC} model specifications for the $\mathsf{SSG}$ likelihoods are provided in the $\tssg$ GitHub repository, \url{https://github.com/Roy-SR-007/TAVIE-SSG}.}. We increase the number of gradient samples per iteration to $N=100$ for performing Bayesian quantile regression on the U.S. Census 2000 data in Section~\ref{main-subsec:Census-data-study} of the main manuscript.

    \item \textbf{Mean-field Variational Inference} (MFVI): For the Student's-$t$ Type I $\mathsf{SSG}$ likelihood, we benchmark against the MFVI algorithm for Student's-$t$ linear regression proposed by~\cite{mfvi-student}. The algorithm employs Gaussian and Inverse-Gamma variational families under mean-field assumptions for the regression and variance parameters, respectively. Prior hyperparameters are set as $\mu = 0$, $\Sigma = I$, $a=2$, and $b=2$. Convergence is declared when the relative change in parameter estimates falls below a tolerance of $10^{-9}$ or when the maximum iteration count ($\texttt{max\_iter}=500$) is reached.

    \item \textbf{\texttt{PyMC}'s No-U-Turn Monte Carlo Sampling} (NUTS): As a non-variational baseline, we employ NUTS \citep{Hoffman-Gelman-NUTS} implemented in \texttt{PyMC}~\citep{pymc,pymc2023peerj}. The same customized \texttt{PyMC} $\mathsf{SSG}$ model used for DADVI is used here, with all NUTS parameters configured at their default settings.

    \item \textbf{\texttt{statsmodels}}: Quantile regression on the U.S. Census 2000 data in Section~\ref{main-subsec:Census-data-study} of the main manuscript is performed using the \texttt{QuantReg} module of \texttt{statsmodels}~\citep{seabold2010statsmodels}. Each quantile is modeled independently with the Powell optimization method, using a maximum of $2000$ iterations and a tolerance of $10^{-6}$.
\end{itemize}

\newpage

\section{Auxiliary Results for \texorpdfstring{Student's-$t$}{Student's-t} Type I \texorpdfstring{$\mathsf{SSG}$}{SSG} Likelihood in \texorpdfstring{Section~\ref{main-subsec:sim-exp-student}}{Section 4.1}}\label{app:additional-student-results}

Here, we present additional results comparing the performance of $\tssg$ under the Student’s-$t$ $\ssg$ likelihood ($\nu=5$) with competing methods across experiments E1 and E2, as outlined in Section~\ref{main-subsec:sim-exp-student} of the main manuscript. The reported metrics summarize the median mean-squared errors (MSEs) of $\widehat{\beta}$ and $\widehat{\tau}^{2}$, and computational runtimes over $100$ data repetitions, as well as the estimated sliced Wasserstein (SW) distance~\citep{bonneel2015sliced,kolouri2016sliced} between the variational posterior and the true posterior estimated using \texttt{PyMC} (NUTS) over the first $10$ repetitions under experiment E1. We also provide the first ($Q_1$) and third ($Q_3$) quartiles of the corresponding MSEs and runtimes. In addition, we include the $\mathsf{ELBO}$ convergence trajectories for $\tssg$ and the ADVI (MF/FR) variants to illustrate their optimization behavior.

\subsection{Performance Metrics for \texorpdfstring{Student's-$t$ $\mathsf{SSG}$}{Student's-t SSG} Likelihood under Experiment E1}\label{app-avg-metrics-student-singlep-multin}

{\scriptsize
\setlength{\tabcolsep}{3pt}
\renewcommand{\arraystretch}{1.15}%
\begin{longtable}{llcccc}
\caption{\footnotesize{Comparison of $\tssg$ against MFVI, DADVI, ADVI (MF/FR), and \texttt{PyMC} (NUTS). Performance [median over $100$ repetitions for MSE and runtime; median over first $10$ repetitions for estimated SW distance; quartile range $(Q_1, Q_3)$ in parentheses] for the Student's-$t$ $\ssg$ likelihood $(\nu=5)$ under experiment E1. The top $3$ best-performing algorithms are in \textbf{bold}.}}
\label{tab:student_singlep_multin}\\
\toprule
\toprule
 &  & \multicolumn{4}{c}{Sample Size ($n$)} \\
\cmidrule(lr){3-6}
Metric & Method & 200 & 500 & 1000 & 2000 \\
\midrule
\endfirsthead

\toprule
 &  & \multicolumn{4}{c}{Sample Size ($n$)} \\
\cmidrule(lr){3-6}
Metric & Method & 200 & 500 & 1000 & 2000 \\
\midrule
\endhead

\midrule
\multicolumn{6}{r}{Continued on next page} \\
\midrule
\endfoot

\bottomrule
\endlastfoot

\multirow{6}{*}{MSE of $\widehat{\beta}$}
& $\tssg$
& \makecell[c]{\textbf{2.380e-03}\\(1.571e-03, 3.013e-03)}
& \makecell[c]{\textbf{8.982e-04}\\(6.168e-04, 1.245e-03)}
& \makecell[c]{\textbf{3.753e-04}\\(2.907e-04, 5.395e-04)}
& \makecell[c]{\textbf{1.978e-04}\\(1.440e-04, 2.793e-04)} \\

& MFVI
& \makecell[c]{\textbf{2.377e-03}\\(1.574e-03, 3.105e-03)}
& \makecell[c]{\textbf{8.900e-04}\\(6.042e-04, 1.242e-03)}
& \makecell[c]{\textbf{3.795e-04}\\(2.880e-04, 5.325e-04)}
& \makecell[c]{\textbf{1.990e-04}\\(1.438e-04, 2.778e-04)} \\

& DADVI
& \makecell[c]{2.503e-03\\(1.687e-03, 3.095e-03)}
& \makecell[c]{9.487e-04\\(6.465e-04, 1.329e-03)}
& \makecell[c]{4.050e-04\\(3.146e-04, 5.508e-04)}
& \makecell[c]{2.064e-04\\(1.470e-04, 2.840e-04)} \\

& ADVI (MF)
& \makecell[c]{2.396e-03\\(1.640e-03, 3.175e-03)}
& \makecell[c]{9.351e-04\\(6.314e-04, 1.264e-03)}
& \makecell[c]{3.960e-04\\(3.010e-04, 5.745e-04)}
& \makecell[c]{2.100e-04\\(1.524e-04, 3.152e-04)} \\

& ADVI (FR)
& \makecell[c]{2.552e-03\\(1.644e-03, 3.312e-03)}
& \makecell[c]{9.209e-04\\(6.155e-04, 1.305e-03)}
& \makecell[c]{4.210e-04\\(3.127e-04, 5.545e-04)}
& \makecell[c]{2.101e-04\\(1.584e-04, 3.149e-04)} \\

& \texttt{PyMC} (NUTS)
& \makecell[c]{\textbf{2.349e-03}\\(1.551e-03, 3.024e-03)}
& \makecell[c]{\textbf{9.173e-04}\\(6.072e-04, 1.240e-03)}
& \makecell[c]{\textbf{3.706e-04}\\(2.893e-04, 5.262e-04)}
& \makecell[c]{\textbf{1.992e-04}\\(1.437e-04, 2.789e-04)} \\
\hline

\multirow{6}{*}{MSE of $\widehat{\tau}^{2}$}
& $\tssg$
& \makecell[c]{\textbf{9.653e-02}\\(2.293e-02, 2.440e-01)}
& \makecell[c]{\textbf{3.058e-02}\\(8.707e-03, 9.737e-02)}
& \makecell[c]{\textbf{1.692e-02}\\(4.087e-03, 4.041e-02)}
& \makecell[c]{\textbf{6.677e-03}\\(2.091e-03, 1.773e-02)} \\

& MFVI
& \makecell[c]{3.401e-01\\(1.690e-01, 7.403e-01)}
& \makecell[c]{1.040e-01\\(1.630e-02, 2.380e-01)}
& \makecell[c]{4.918e-02\\(8.589e-03, 1.032e-01)}
& \makecell[c]{1.677e-02\\(3.339e-03, 5.053e-02)} \\

& DADVI
& \makecell[c]{\textbf{9.666e-02}\\(2.598e-02, 2.198e-01)}
& \makecell[c]{\textbf{3.806e-02}\\(6.140e-03, 1.007e-01)}
& \makecell[c]{\textbf{1.588e-02}\\(2.815e-03, 4.366e-02)}
& \makecell[c]{\textbf{5.319e-03}\\(1.414e-03, 1.730e-02)} \\

& ADVI (MF)
& \makecell[c]{3.278e-01\\(1.041e-01, 7.127e-01)}
& \makecell[c]{7.374e-02\\(1.540e-02, 2.784e-01)}
& \makecell[c]{9.919e-02\\(3.982e-02, 1.852e-01)}
& \makecell[c]{1.508e-01\\(1.032e-01, 2.522e-01)} \\

& ADVI (FR)
& \makecell[c]{4.235e-01\\(1.322e-01, 8.145e-01)}
& \makecell[c]{5.985e-02\\(1.533e-02, 1.802e-01)}
& \makecell[c]{1.009e-01\\(4.517e-02, 1.669e-01)}
& \makecell[c]{1.389e-01\\(8.533e-02, 2.494e-01)} \\

& \texttt{PyMC} (NUTS)
& \makecell[c]{\textbf{7.733e-02}\\(2.410e-02, 2.121e-01)}
& \makecell[c]{\textbf{4.244e-02}\\(7.833e-03, 9.288e-02)}
& \makecell[c]{\textbf{1.476e-02}\\(2.863e-03, 4.121e-02)}
& \makecell[c]{\textbf{6.006e-03}\\(1.914e-03, 1.830e-02)} \\
\hline

\multirow{6}{*}{Runtime (s)}
& $\tssg$
& \makecell[c]{\textbf{4.197e-03}\\(3.837e-03, 4.828e-03)}
& \makecell[c]{\textbf{7.338e-03}\\(5.092e-03, 1.030e-02)}
& \makecell[c]{\textbf{6.414e-03}\\(6.053e-03, 7.364e-03)}
& \makecell[c]{\textbf{9.371e-03}\\(9.100e-03, 1.017e-02)} \\

& MFVI
& \makecell[c]{\textbf{1.997e-02}\\(9.201e-03, 3.146e-02)}
& \makecell[c]{\textbf{3.709e-02}\\(2.904e-02, 5.346e-02)}
& \makecell[c]{\textbf{3.629e-02}\\(3.211e-02, 4.248e-02)}
& \makecell[c]{\textbf{4.791e-02}\\(4.425e-02, 5.531e-02)} \\

& DADVI
& \makecell[c]{\textbf{8.732e-01}\\(8.531e-01, 8.933e-01)}
& \makecell[c]{\textbf{9.757e-01}\\(9.441e-01, 1.037e+00)}
& \makecell[c]{\textbf{9.810e-01}\\(9.662e-01, 1.010e+00)}
& \makecell[c]{\textbf{1.126e+00}\\(1.115e+00, 1.149e+00)} \\

& ADVI (MF)
& \makecell[c]{1.363e+01\\(1.280e+01, 1.409e+01)}
& \makecell[c]{1.483e+01\\(1.421e+01, 1.570e+01)}
& \makecell[c]{1.605e+01\\(1.532e+01, 1.648e+01)}
& \makecell[c]{1.654e+01\\(1.609e+01, 1.695e+01)} \\

& ADVI (FR)
& \makecell[c]{1.945e+01\\(1.849e+01, 2.043e+01)}
& \makecell[c]{2.178e+01\\(2.013e+01, 2.254e+01)}
& \makecell[c]{2.233e+01\\(2.147e+01, 2.324e+01)}
& \makecell[c]{2.351e+01\\(2.292e+01, 2.384e+01)} \\

& \texttt{PyMC} (NUTS)
& \makecell[c]{5.119e+00\\(4.748e+00, 5.547e+00)}
& \makecell[c]{5.477e+00\\(4.738e+00, 6.010e+00)}
& \makecell[c]{5.499e+00\\(5.250e+00, 6.013e+00)}
& \makecell[c]{6.043e+00\\(5.934e+00, 6.921e+00)} \\
\hline

\multirow{5}{*}{$\widehat{\mathrm{SW}}$ distance}
& $\tssg$
& \makecell[c]{\textbf{1.002e-02}\\(9.274e-03, 1.058e-02)}
& \makecell[c]{\textbf{5.453e-03}\\(5.189e-03, 5.595e-03)}
& \makecell[c]{\textbf{3.503e-03}\\(3.377e-03, 3.598e-03)}
& \makecell[c]{\textbf{2.478e-03}\\(2.353e-03, 2.524e-03)} \\

& MFVI
& \makecell[c]{\textbf{2.031e-02}\\(1.899e-02, 2.195e-02)}
& \makecell[c]{\textbf{8.409e-03}\\(8.200e-03, 9.014e-03)}
& \makecell[c]{\textbf{5.080e-03}\\(4.958e-03, 5.303e-03)}
& \makecell[c]{\textbf{3.035e-03}\\(2.958e-03, 3.087e-03)} \\

& DADVI
& \makecell[c]{\textbf{1.110e-02}\\(7.727e-03, 1.196e-02)}
& \makecell[c]{\textbf{7.040e-03}\\(4.828e-03, 7.888e-03)}
& \makecell[c]{\textbf{3.854e-03}\\(3.344e-03, 4.140e-03)}
& \makecell[c]{\textbf{3.021e-03}\\(2.729e-03, 3.597e-03)} \\

& ADVI (MF)
& \makecell[c]{3.571e-01\\(2.813e-01, 7.031e-01)}
& \makecell[c]{4.663e-01\\(2.890e-01, 7.782e-01)}
& \makecell[c]{4.914e-01\\(2.294e-01, 6.221e-01)}
& \makecell[c]{3.606e-01\\(1.813e-01, 6.746e-01)} \\

& ADVI (FR)
& \makecell[c]{3.814e-01\\(2.608e-01, 4.260e-01)}
& \makecell[c]{6.496e-01\\(3.229e-01, 7.621e-01)}
& \makecell[c]{3.561e-01\\(2.022e-01, 4.319e-01)}
& \makecell[c]{2.879e-01\\(2.212e-01, 5.062e-01)} \\
\hline

\end{longtable}
}

\newpage

\subsection{Performance Metrics for \texorpdfstring{Student's-$t$ $\mathsf{SSG}$}{Student's-t SSG} Likelihood under Experiment E2}\label{app-avg-metrics-student-singlen-multip}

{\scriptsize
\setlength{\tabcolsep}{3pt} 
\renewcommand{\arraystretch}{1.15}%
\begin{longtable}{llcccc}
\caption{\footnotesize{Comparison of $\tssg$ against MFVI, DADVI, ADVI (MF/FR), and \texttt{PyMC} (NUTS). Performance [median over $100$ repetitions; quartile range ($Q_1, Q_3$)  in parentheses] for the Student's-$t$ $\mathsf{SSG}$ likelihood ($\nu=5$) under experiment E2. { Top $3$ best performing algorithms are in \textbf{bold}.}}}
\label{tab:student_singlen_multip}\\\\
\toprule
\toprule
 &  & \multicolumn{4}{c}{Dimension ($p$)} \\
\cmidrule(lr){3-6}
Metric & Method & 4 & 9 & 16 & 21 \\
\midrule
\endfirsthead
\toprule
 &  & \multicolumn{4}{c}{Dimension ($p$)} \\
\cmidrule(lr){3-6}
Metric & Method & 4 & 9 & 16 & 21 \\
\midrule
\endhead
\midrule
\multicolumn{6}{r}{Continued on next page} \\
\midrule
\endfoot
\bottomrule
\endlastfoot

\multirow{6}{*}{MSE of $\widehat{\beta}$}
& $\tssg$ & \makecell[c]{\textbf{3.244e-04}\\(1.775e-04, 5.976e-04)} & \makecell[c]{\textbf{3.774e-04}\\(2.872e-04, 5.312e-04)} & \makecell[c]{4.204e-04\\(3.300e-04, 5.421e-04)} & \makecell[c]{\textbf{4.721e-04}\\(3.779e-04, 5.649e-04)} \\
& MFVI & \makecell[c]{\textbf{3.341e-04}\\(1.781e-04, 6.078e-04)} & \makecell[c]{\textbf{3.833e-04}\\(2.943e-04, 5.331e-04)} & \makecell[c]{\textbf{4.147e-04}\\(3.306e-04, 5.642e-04)} & \makecell[c]{\textbf{4.670e-04}\\(3.740e-04, 5.658e-04)} \\
& DADVI & \makecell[c]{3.568e-04\\(2.085e-04, 6.767e-04)} & \makecell[c]{4.032e-04\\(3.170e-04, 5.489e-04)} & \makecell[c]{\textbf{4.123e-04}\\(3.293e-04, 5.456e-04)} & \makecell[c]{4.807e-04\\(4.038e-04, 6.050e-04)} \\
& ADVI (MF) & \makecell[c]{4.020e-04\\(2.298e-04, 7.339e-04)} & \makecell[c]{4.181e-04\\(3.119e-04, 5.729e-04)} & \makecell[c]{4.439e-04\\(3.368e-04, 5.606e-04)} & \makecell[c]{4.862e-04\\(3.983e-04, 5.557e-04)} \\
& ADVI (FR) & \makecell[c]{4.169e-04\\(2.387e-04, 6.988e-04)} & \makecell[c]{4.052e-04\\(3.141e-04, 5.656e-04)} & \makecell[c]{4.346e-04\\(3.428e-04, 5.609e-04)} & \makecell[c]{4.828e-04\\(3.878e-04, 6.043e-04)} \\
& \texttt{PyMC} (NUTS) & \makecell[c]{\textbf{3.350e-04}\\(1.737e-04, 5.771e-04)} & \makecell[c]{\textbf{3.761e-04}\\(2.874e-04, 5.242e-04)} & \makecell[c]{\textbf{4.174e-04}\\(3.249e-04, 5.485e-04)} & \makecell[c]{\textbf{4.721e-04}\\(3.794e-04, 5.671e-04)} \\
\hline

\multirow{6}{*}{MSE of $\widehat{\tau}^{2}$}
& $\tssg$ & \makecell[c]{\textbf{1.345e-02}\\(2.301e-03, 4.277e-02)} & \makecell[c]{\textbf{4.099e-02}\\(4.973e-03, 5.930e-02)} & \makecell[c]{\textbf{1.497e-02}\\(1.789e-03, 4.661e-02)} & \makecell[c]{4.784e-02\\(3.377e-02, 1.174e-01)} \\
& MFVI & \makecell[c]{3.834e-02\\(7.493e-03, 1.131e-01)} & \makecell[c]{4.904e-02\\(8.612e-03, 1.034e-01)} & \makecell[c]{3.076e-02\\(8.434e-03, 7.777e-02)} & \makecell[c]{4.918e-02\\(1.948e-02, 1.199e-01)} \\
& DADVI & \makecell[c]{\textbf{1.142e-02}\\(2.485e-03, 4.509e-02)} & \makecell[c]{\textbf{1.591e-02}\\(2.813e-03, 4.366e-02)} & \makecell[c]{\textbf{1.514e-02}\\(4.577e-03, 3.238e-02)} & \makecell[c]{\textbf{1.137e-02}\\(3.221e-03, 3.611e-02)} \\
& ADVI (MF) & \makecell[c]{7.109e-02\\(1.809e-02, 2.376e-01)} & \makecell[c]{5.746e-02\\(1.346e-02, 1.289e-01)} & \makecell[c]{7.288e-02\\(2.666e-02, 1.514e-01)} & \makecell[c]{1.338e-02\\(3.048e-03, 4.238e-02)} \\
& ADVI (FR) & \makecell[c]{9.497e-02\\(2.163e-02, 2.965e-01)} & \makecell[c]{5.620e-02\\(1.589e-02, 1.035e-01)} & \makecell[c]{5.825e-02\\(2.279e-02, 1.347e-01)} & \makecell[c]{\textbf{1.141e-02}\\(3.793e-03, 4.037e-02)} \\
& \texttt{PyMC} (NUTS) & \makecell[c]{\textbf{1.124e-02}\\(2.757e-03, 4.309e-02)} & \makecell[c]{\textbf{1.488e-02}\\(2.783e-03, 4.141e-02)} & \makecell[c]{\textbf{1.600e-02}\\(4.590e-03, 3.224e-02)} & \makecell[c]{\textbf{1.041e-02}\\(2.863e-03, 3.279e-02)} \\
\hline

\multirow{6}{*}{Runtime (s)}
& $\tssg$ & \makecell[c]{\textbf{5.323e-03}\\(5.053e-03, 5.826e-03)} & \makecell[c]{\textbf{6.184e-03}\\(5.986e-03, 6.599e-03)} & \makecell[c]{\textbf{7.675e-03}\\(7.361e-03, 8.520e-03)} & \makecell[c]{\textbf{1.038e-02}\\(9.632e-03, 1.106e-02)} \\
& MFVI & \makecell[c]{\textbf{2.715e-02}\\(2.380e-02, 3.232e-02)} & \makecell[c]{\textbf{3.742e-02}\\(3.305e-02, 4.617e-02)} & \makecell[c]{\textbf{4.414e-02}\\(3.848e-02, 5.316e-02)} & \makecell[c]{\textbf{6.303e-02}\\(5.476e-02, 7.381e-02)} \\
& DADVI & \makecell[c]{\textbf{9.786e-01}\\(9.599e-01, 1.007e+00)} & \makecell[c]{\textbf{1.007e+00}\\(9.925e-01, 1.040e+00)} & \makecell[c]{\textbf{1.113e+00}\\(1.080e+00, 1.127e+00)} & \makecell[c]{\textbf{1.137e+00}\\(1.103e+00, 1.160e+00)} \\
& ADVI (MF) & \makecell[c]{1.381e+01\\(1.280e+01, 1.484e+01)} & \makecell[c]{1.624e+01\\(1.530e+01, 1.658e+01)} & \makecell[c]{1.673e+01\\(1.648e+01, 1.695e+01)} & \makecell[c]{1.693e+01\\(1.672e+01, 1.712e+01)} \\
& ADVI (FR) & \makecell[c]{1.946e+01\\(1.779e+01, 2.062e+01)} & \makecell[c]{2.289e+01\\(2.195e+01, 2.377e+01)} & \makecell[c]{2.397e+01\\(2.339e+01, 2.428e+01)} & \makecell[c]{2.434e+01\\(2.414e+01, 2.455e+01)} \\
& \texttt{PyMC} (NUTS) & \makecell[c]{4.919e+00\\(4.811e+00, 5.076e+00)} & \makecell[c]{5.470e+00\\(5.327e+00, 5.594e+00)} & \makecell[c]{5.955e+00\\(5.546e+00, 6.105e+00)} & \makecell[c]{6.221e+00\\(5.537e+00, 6.354e+00)} \\
\hline
\end{longtable}
}

\subsection{\texorpdfstring{$\mathsf{ELBO}$}{ELBO} History for \texorpdfstring{$\tssg$}{TAVIE-SSG}, ADVI (MF), and ADVI (FR) under the \texorpdfstring{Student's-$t$ $\mathsf{SSG}$}{Student's-t SSG} (Type I) Likelihood}\label{app-convergence-ELBO-student}

\begin{figure}[!htp]
    \centering
    \includegraphics[width=\linewidth]{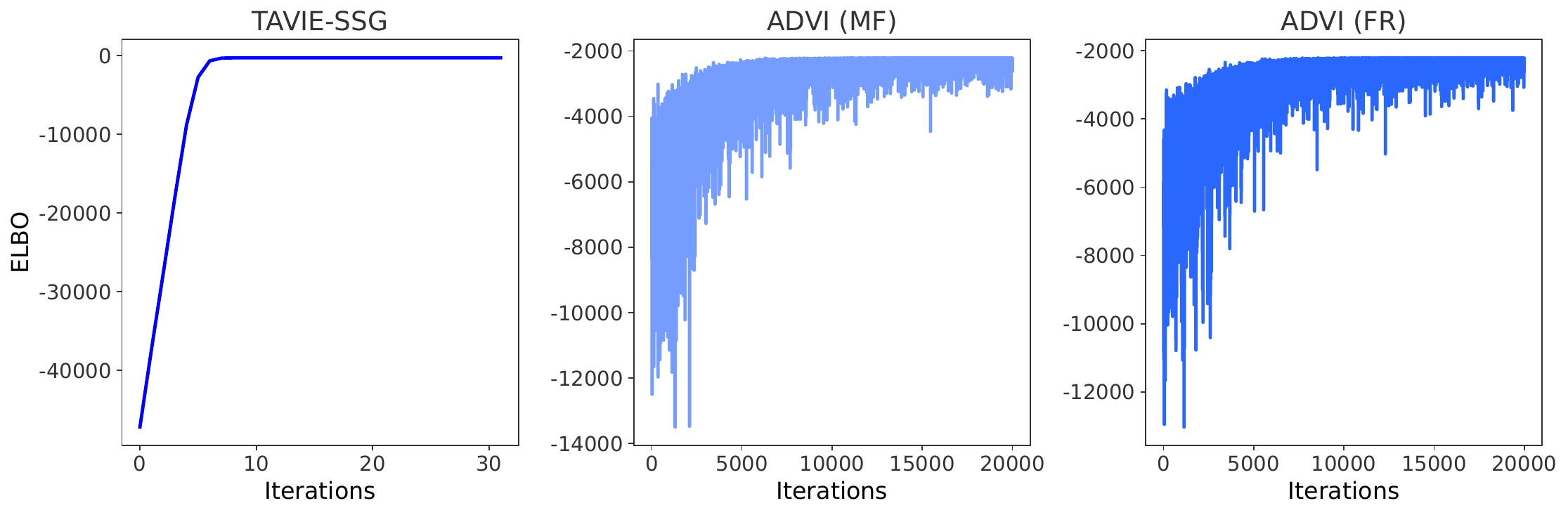}
    \caption{\footnotesize{Convergence diagnostics ($\mathsf{ELBO}$ monitoring) of $\tssg$, ADVI (MF), and ADVI (FR) for $(n,p)=(2000, 9)$ under the Student's-$t$ $\mathsf{SSG}$ likelihood ($\nu=5$). { For $\tssg$, the $\mathsf{ELBO}$ corresponds to $\mathsf{L}(\xi)$ in~\eqref{main-eq:ELBO-general} of the main manuscript, while for ADVI (MF/FR) the MC approximation of the true $\mathsf{ELBO}$ is tracked.}}}
    \label{fig:convergence_ELBO_student}
\end{figure}

\newpage

\section{Additional Simulation Studies}\label{app:additional-simulation-studies}

\subsection{Laplace Type I \texorpdfstring{$\mathsf{SSG}$}{SSG} Likelihood}
\label{app:additional-laplace-results}

Following \eqref{eq:param-laplace}, the Laplace $\ssg$ likelihood for $y_i\in \mathbb{R}$, denoted as $y_i\mid \mathbf{x}_i, \theta \sim \text{Laplace}(\mu_i = \mathbf{x}_i^{\top}\beta, \tau)$, is:
\begin{align}\label{eq:student-SSG-model}
     p(y_i\mid \mathbf{x}_i, \theta)\propto \tau\exp\left\{-\tau|y_i-\mathbf{x}_i^{\top}\beta|\right\},
\end{align}
independently for $i\in [n]$, where $\theta = (\beta^{\top}, \tau^{2})^{\top}$. Data, $\mathcal{D}_n := \{(\mathbf{x}_i, y_i): i\in [n]\}$, is simulated using the true model parameters in \eqref{eq:student-SSG-model} configured as: $\tau_0^2=3$, $\beta_0\sim \mathcal{N}_{p}(0, I_{p})$, and $\mathbf{X}=(\mathbf{x}_1, \ldots, \mathbf{x}_{n})^{\top}\in \mathbb{R}^{n\times p}$ with $x_{i1}=1$ (intercept) and $x_{ij}$'s generated independently and identically from the standard Gaussian distribution for $j=2, \ldots, p$.

We consider two sets of experiments: (E1) increasing sample sizes, $n\in \{200, 500, 1000, 2000\}$ with fixed number of features, $p=9$; and (E2) increasing features, $p\in \{4, 9, 16, 21\}$ with a fixed sample size, $n=1000$. Under both experimental specifications, $\tssg$ is implemented for the standard likelihood setup with $\alpha=1$ and the prior hyperparameters set as, $(\mu, \Sigma, a, b) = (0, I_p, 0.05, 0.05)$. Convergence is assessed in Algorithm~\ref{main-alg:tavie-em} of the main manuscript through a tolerance level of $\texttt{tol}=10^{-9}$. The competing methods (DADVI, ADVI (MF/FR), \texttt{PyMC} (NUTS)) are executed under experimental settings outlined in Section \ref{app-competing-methods-details}. The performance of $\mathsf{TAVIE}$-$\mathsf{SSG}$ and competing methods is assessed based on statistical accuracy, quantified by the {mean-squared error} (MSE) between the estimated and true parameter values, $p^{-1}\lVert \beta_0 - \widehat{\beta}\rVert_{2}^{2}$ and $(\tau_0^{2} - \widehat{\tau}^{2})^2$; computational efficiency, measured by the {runtimes} (in seconds) of each algorithm; and in experiment E1, sliced Wasserstein (SW) distance~\citep{bonneel2015sliced,kolouri2016sliced} between the true and variational posteriors. The corresponding numerical results for MSE and runtime are tabulated in Tables \ref{tab:laplace_singlep_multin} and \ref{tab:laplace_singlen_multip} as well as visually represented in Figures \ref{fig:laplace_singlep_multin}, \ref{fig:runtime_laplace} ($100$ repetitions of the dataset $\mathcal{D}_n$ for MSE and runtime analyses). Figure~\ref{fig:laplace_SW_distance} illustrates the estimated SW distance between the variational posterior and the true posterior estimated using \texttt{PyMC} (NUTS) over first $10$ repetitions, under experiment E1.

Overall, the results collectively demonstrate the robust empirical performance of $\tssg$ across varying sample sizes and feature dimensions. In experiment E1, for increasing sample sizes, as illustrated in Figure \ref{fig:laplace_singlep_multin}, $\tssg$, DADVI, and \texttt{PyMC} (NUTS) compare well and deliver accurate posterior estimates for the regression coefficients, achieving lower MSEs of $\widehat{\beta}$ for estimating $\beta$ compared to the ADVI (MF/FR) variants.
For estimating $\tau^{2}$, $\tssg$, DADVI, and \texttt{PyMC} (NUTS) have similar performance.
Moreover, the posterior distance comparisons in Figure~\ref{fig:laplace_SW_distance} show that $\tssg$ and DADVI produce closer approximations to the true posterior than ADVI variants.
When the feature dimension increases under experiment E2 (see Figure \ref{fig:laplace_singlen_multip}), $\tssg$ continues to deliver competitive accuracy of posterior estimates relative to DADVI and \texttt{PyMC} (NUTS). As evidenced by Figure \ref{fig:runtime_laplace}, across both experiments E1 and E2, $\tssg$ consistently demonstrates superior computational efficiency, achieving runtimes that are orders of magnitude lower than all competing methods. 

We further investigate the convergence behavior of ADVI (MF/FR) variants by tracking $\mathsf{ELBO}$ trajectories over iterations and comparing them with $\tssg$ (see Figure \ref{fig:convergence_ELBO_laplace}). The results indicate that $\tssg$ converges rapidly to highly accurate posterior variational estimates of $\beta$ and $\tau^2$.

{\scriptsize
\setlength{\tabcolsep}{3pt} 
\renewcommand{\arraystretch}{1.15}%
\begin{longtable}{llcccc}
\caption{\footnotesize{Comparison of $\tssg$ against DADVI, ADVI (MF/FR), and \texttt{PyMC} (NUTS). Performance [median over $100$ repetitions; quartile range ($Q_1, Q_3$)  in parentheses] for the Laplace $\mathsf{SSG}$ likelihood under experiment E1. { Top $3$ best performing algorithms are in \textbf{bold}.}}}
\label{tab:laplace_singlep_multin}\\
\toprule
\toprule
 &  & \multicolumn{4}{c}{Sample Size ($n$)} \\
\cmidrule(lr){3-6}
Metric & Method & 200 & 500 & 1000 & 2000 \\
\midrule
\endfirsthead
\toprule
 &  & \multicolumn{4}{c}{Sample Size ($n$)} \\
\cmidrule(lr){3-6}
Metric & Method & 200 & 500 & 1000 & 2000 \\
\midrule
\endhead
\midrule
\multicolumn{6}{r}{Continued on next page} \\
\midrule
\endfoot
\bottomrule
\endlastfoot

\multirow{5}{*}{MSE of $\widehat{\beta}$}
& $\tssg$ & \makecell[c]{\textbf{1.942e-03}\\(1.271e-03, 2.808e-03)} & \makecell[c]{\textbf{6.709e-04}\\(4.707e-04, 9.619e-04)} & \makecell[c]{\textbf{3.472e-04}\\(2.313e-04, 4.866e-04)} & \makecell[c]{\textbf{1.640e-04}\\(1.200e-04, 2.333e-04)} \\
& DADVI & \makecell[c]{2.156e-03\\(1.344e-03, 2.925e-03)} & \makecell[c]{7.294e-04\\(5.107e-04, 9.775e-04)} & \makecell[c]{\textbf{3.539e-04}\\(2.424e-04, 5.289e-04)} & \makecell[c]{\textbf{1.778e-04}\\(1.285e-04, 2.433e-04)} \\
& ADVI (MF) & \makecell[c]{\textbf{2.073e-03}\\(1.501e-03, 3.029e-03)} & \makecell[c]{7.455e-04\\(5.471e-04, 1.076e-03)} & \makecell[c]{4.937e-04\\(2.548e-04, 5.625e-04)} & \makecell[c]{1.860e-04\\(1.326e-04, 2.474e-04)} \\
& ADVI (FR) & \makecell[c]{2.177e-03\\(1.584e-03, 2.847e-03)} & \makecell[c]{\textbf{7.142e-04}\\(5.599e-04, 1.023e-03)} & \makecell[c]{3.881e-04\\(2.562e-04, 5.399e-04)} & \makecell[c]{1.841e-04\\(1.332e-04, 2.712e-04)} \\
& \texttt{PyMC} (NUTS) & \makecell[c]{\textbf{1.858e-03}\\(1.293e-03, 2.738e-03)} & \makecell[c]{\textbf{6.768e-04}\\(4.685e-04, 9.772e-04)} & \makecell[c]{\textbf{3.359e-04}\\(2.295e-04, 4.901e-04)} & \makecell[c]{\textbf{1.631e-04}\\(1.164e-04, 2.331e-04)} \\
\hline

\multirow{5}{*}{MSE of $\widehat{\tau}^{2}$}
& $\tssg$ & \makecell[c]{\textbf{1.992e-01}\\(7.869e-02, 4.184e-01)} & \makecell[c]{7.107e-02\\(2.492e-02, 1.285e-01)} & \makecell[c]{\textbf{1.844e-02}\\(3.774e-03, 6.791e-02)} & \makecell[c]{\textbf{9.466e-03}\\(2.118e-03, 2.316e-02)} \\
& DADVI & \makecell[c]{\textbf{8.999e-02}\\(1.254e-02, 3.064e-01)} & \makecell[c]{\textbf{2.821e-02}\\(5.432e-03, 8.045e-02)} & \makecell[c]{\textbf{2.189e-02}\\(7.219e-03, 5.612e-02)} & \makecell[c]{\textbf{5.571e-03}\\(1.693e-03, 2.234e-02)} \\
& ADVI (MF) & \makecell[c]{4.710e-01\\(1.358e-01, 9.991e-01)} & \makecell[c]{4.890e-02\\(6.977e-03, 9.014e-02)} & \makecell[c]{4.859e-02\\(1.685e-02, 1.687e-01)} & \makecell[c]{1.215e-01\\(6.424e-02, 2.092e-01)} \\
& ADVI (FR) & \makecell[c]{3.371e-01\\(1.367e-01, 8.357e-01)} & \makecell[c]{\textbf{3.530e-02}\\(7.254e-03, 1.239e-01)} & \makecell[c]{4.308e-02\\(1.128e-02, 1.310e-01)} & \makecell[c]{1.029e-01\\(4.128e-02, 1.803e-01)} \\
& \texttt{PyMC} (NUTS) & \makecell[c]{\textbf{9.121e-02}\\(1.513e-02, 2.935e-01)} & \makecell[c]{\textbf{2.789e-02}\\(7.147e-03, 8.254e-02)} & \makecell[c]{\textbf{2.295e-02}\\(6.442e-03, 5.360e-02)} & \makecell[c]{\textbf{7.313e-03}\\(1.389e-03, 2.002e-02)} \\
\hline

\multirow{5}{*}{Runtime (s)}
& $\tssg$ & \makecell[c]{\textbf{7.201e-03}\\(6.727e-03, 7.873e-03)} & \makecell[c]{\textbf{1.028e-02}\\(9.386e-03, 1.146e-02)} & \makecell[c]{\textbf{1.401e-02}\\(1.312e-02, 1.516e-02)} & \makecell[c]{\textbf{2.291e-02}\\(2.202e-02, 2.461e-02)} \\
& DADVI & \makecell[c]{\textbf{1.497e+00}\\(1.264e+00, 1.542e+00)} & \makecell[c]{\textbf{2.348e+00}\\(2.218e+00, 2.396e+00)} & \makecell[c]{\textbf{3.728e+00}\\(3.669e+00, 3.771e+00)} & \makecell[c]{\textbf{5.638e+00}\\(5.581e+00, 5.704e+00)} \\
& ADVI (MF) & \makecell[c]{1.257e+01\\(1.190e+01, 1.322e+01)} & \makecell[c]{1.374e+01\\(1.326e+01, 1.441e+01)} & \makecell[c]{1.479e+01\\(1.420e+01, 1.522e+01)} & \makecell[c]{1.551e+01\\(1.514e+01, 1.567e+01)} \\
& ADVI (FR) & \makecell[c]{1.820e+01\\(1.733e+01, 1.904e+01)} & \makecell[c]{1.986e+01\\(1.908e+01, 2.094e+01)} & \makecell[c]{2.186e+01\\(2.070e+01, 2.227e+01)} & \makecell[c]{2.243e+01\\(2.216e+01, 2.270e+01)} \\
& \texttt{PyMC} (NUTS) & \makecell[c]{\textbf{4.928e+00}\\(4.667e+00, 5.143e+00)} & \makecell[c]{\textbf{4.716e+00}\\(4.363e+00, 4.861e+00)} & \makecell[c]{\textbf{5.064e+00}\\(4.277e+00, 5.180e+00)} & \makecell[c]{\textbf{4.846e+00}\\(4.747e+00, 5.063e+00)} \\
\hline
\end{longtable}
}

{\scriptsize
\setlength{\tabcolsep}{3pt} 
\renewcommand{\arraystretch}{1.15}%
\begin{longtable}{llcccc}
\caption{\footnotesize{Comparison of $\tssg$ against DADVI, ADVI (MF/FR), and \texttt{PyMC} (NUTS). Performance [median over $100$ repetitions; quartile range ($Q_1, Q_3$)  in parentheses] for the Laplace $\mathsf{SSG}$ likelihood under experiment E2. { Top $3$ best performing algorithms are in \textbf{bold}.}}} 
\label{tab:laplace_singlen_multip}\\
\toprule
\toprule
 &  & \multicolumn{4}{c}{Dimension ($p$)} \\
\cmidrule(lr){3-6}
Metric & Method & 4 & 9 & 16 & 21 \\
\midrule
\endfirsthead
\toprule
 &  & \multicolumn{4}{c}{Dimension ($p$)} \\
\cmidrule(lr){3-6}
Metric & Method & 4 & 9 & 16 & 21 \\
\midrule
\endhead
\midrule
\multicolumn{6}{r}{Continued on next page} \\
\midrule
\endfoot
\bottomrule
\endlastfoot

\multirow{5}{*}{MSE of $\widehat{\beta}$}
& $\tssg$ & \makecell[c]{\textbf{2.876e-04}\\(1.925e-04, 5.205e-04)} & \makecell[c]{\textbf{3.603e-04}\\(2.418e-04, 4.938e-04)} & \makecell[c]{\textbf{3.699e-04}\\(2.886e-04, 4.899e-04)} & \makecell[c]{\textbf{3.998e-04}\\(3.148e-04, 5.113e-04)} \\
& DADVI & \makecell[c]{\textbf{3.018e-04}\\(1.983e-04, 5.484e-04)} & \makecell[c]{\textbf{3.521e-04}\\(2.459e-04, 5.364e-04)} & \makecell[c]{\textbf{3.679e-04}\\(2.986e-04, 4.706e-04)} & \makecell[c]{4.068e-04\\(3.342e-04, 5.264e-04)} \\
& ADVI (MF) & \makecell[c]{4.865e-04\\(2.171e-04, 6.225e-04)} & \makecell[c]{4.034e-04\\(2.539e-04, 5.725e-04)} & \makecell[c]{4.030e-04\\(3.008e-04, 4.997e-04)} & \makecell[c]{\textbf{4.036e-04}\\(3.269e-04, 5.311e-04)} \\
& ADVI (FR) & \makecell[c]{3.592e-04\\(2.006e-04, 6.035e-04)} & \makecell[c]{4.996e-04\\(2.574e-04, 5.869e-04)} & \makecell[c]{3.924e-04\\(2.906e-04, 4.909e-04)} & \makecell[c]{4.066e-04\\(3.371e-04, 5.301e-04)} \\
& \texttt{PyMC} (NUTS) & \makecell[c]{\textbf{2.896e-04}\\(1.884e-04, 5.164e-04)} & \makecell[c]{\textbf{3.439e-04}\\(2.266e-04, 4.878e-04)} & \makecell[c]{\textbf{3.597e-04}\\(2.875e-04, 4.724e-04)} & \makecell[c]{\textbf{3.857e-04}\\(3.095e-04, 5.071e-04)} \\
\hline

\multirow{5}{*}{MSE of $\widehat{\tau}^{2}$}
& $\tssg$ & \makecell[c]{\textbf{1.499e-02}\\(4.142e-03, 6.087e-02)} & \makecell[c]{\textbf{1.976e-02}\\(3.289e-03, 7.129e-02)} & \makecell[c]{\textbf{2.193e-02}\\(5.989e-03, 6.220e-02)} & \makecell[c]{9.581e-02\\(4.721e-02, 1.733e-01)} \\
& DADVI & \makecell[c]{\textbf{1.522e-02}\\(6.435e-03, 6.094e-02)} & \makecell[c]{\textbf{2.187e-02}\\(7.142e-03, 5.670e-02)} & \makecell[c]{\textbf{1.654e-02}\\(3.080e-03, 3.956e-02)} & \makecell[c]{\textbf{1.561e-02}\\(4.019e-03, 5.266e-02)} \\
& ADVI (MF) & \makecell[c]{9.595e-02\\(3.286e-02, 2.601e-01)} & \makecell[c]{7.049e-02\\(1.360e-02, 1.741e-01)} & \makecell[c]{4.324e-02\\(1.114e-02, 1.200e-01)} & \makecell[c]{\textbf{1.359e-02}\\(2.388e-03, 5.029e-02)} \\
& ADVI (FR) & \makecell[c]{1.130e-01\\(3.235e-02, 2.665e-01)} & \makecell[c]{5.189e-02\\(9.952e-03, 1.340e-01)} & \makecell[c]{2.663e-02\\(8.285e-03, 7.636e-02)} & \makecell[c]{1.961e-02\\(3.408e-03, 4.936e-02)} \\
& \texttt{PyMC} (NUTS) & \makecell[c]{\textbf{1.677e-02}\\(3.942e-03, 6.138e-02)} & \makecell[c]{\textbf{2.341e-02}\\(6.561e-03, 5.394e-02)} & \makecell[c]{\textbf{1.393e-02}\\(4.197e-03, 4.060e-02)} & \makecell[c]{\textbf{1.782e-02}\\(4.193e-03, 5.348e-02)} \\
\hline

\multirow{5}{*}{Runtime (s)}
& $\tssg$ & \makecell[c]{\textbf{1.195e-02}\\(1.116e-02, 1.360e-02)} & \makecell[c]{\textbf{1.595e-02}\\(1.457e-02, 1.828e-02)} & \makecell[c]{\textbf{2.066e-02}\\(1.740e-02, 2.506e-02)} & \makecell[c]{\textbf{2.456e-02}\\(2.284e-02, 2.646e-02)} \\
& DADVI & \makecell[c]{\textbf{2.051e+00}\\(2.003e+00, 2.104e+00)} & \makecell[c]{\textbf{3.747e+00}\\(3.660e+00, 3.838e+00)} & \makecell[c]{\textbf{6.968e+00}\\(6.754e+00, 7.379e+00)} & \makecell[c]{\textbf{8.900e+00}\\(7.968e+00, 9.019e+00)} \\
& ADVI (MF) & \makecell[c]{1.265e+01\\(1.203e+01, 1.349e+01)} & \makecell[c]{1.504e+01\\(1.452e+01, 1.557e+01)} & \makecell[c]{1.563e+01\\(1.535e+01, 1.617e+01)} & \makecell[c]{1.575e+01\\(1.545e+01, 1.593e+01)} \\
& ADVI (FR) & \makecell[c]{1.850e+01\\(1.708e+01, 1.986e+01)} & \makecell[c]{2.200e+01\\(2.101e+01, 2.257e+01)} & \makecell[c]{2.291e+01\\(2.239e+01, 2.366e+01)} & \makecell[c]{2.290e+01\\(2.255e+01, 2.319e+01)} \\
& \texttt{PyMC} (NUTS) & \makecell[c]{\textbf{4.293e+00}\\(3.911e+00, 4.524e+00)} & \makecell[c]{\textbf{4.782e+00}\\(4.364e+00, 5.200e+00)} & \makecell[c]{\textbf{5.557e+00}\\(5.059e+00, 6.242e+00)} & \makecell[c]{\textbf{6.627e+00}\\(6.376e+00, 6.784e+00)} \\
\hline
\end{longtable}
}

\begin{figure}[!htp]
    \centering
    \includegraphics[width=\linewidth]{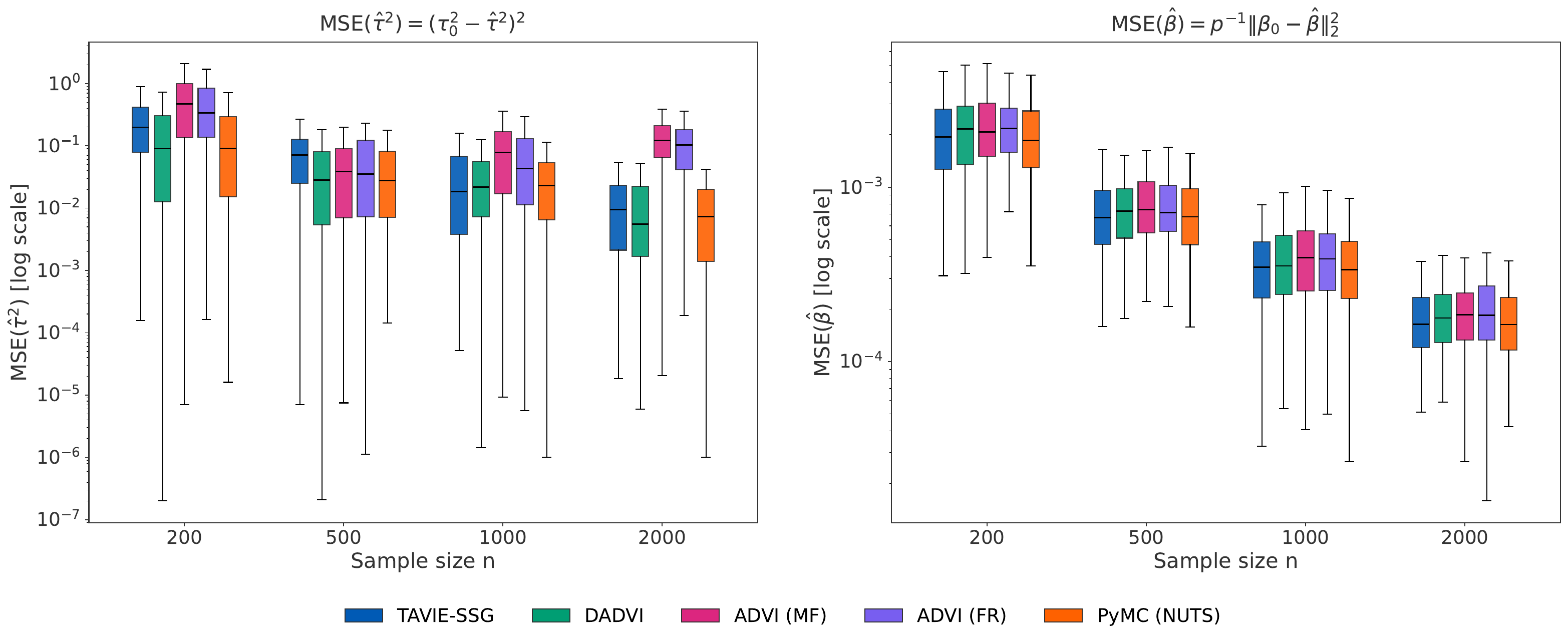}
    \caption{\footnotesize{MSEs of $(\widehat{\beta}, \widehat{\tau}^2)$ (in $\log$-scale) 
    across $100$ data repetitions of $\tssg$ and competitors for the Laplace $\mathsf{SSG}$ likelihood under experiment E1: $n\in \{200, 500, 1000, 2000\},\; p=9$.}}
    \label{fig:laplace_singlep_multin}
\end{figure}

\begin{figure}[!htp]
    \centering
    \includegraphics[width=0.6\linewidth]{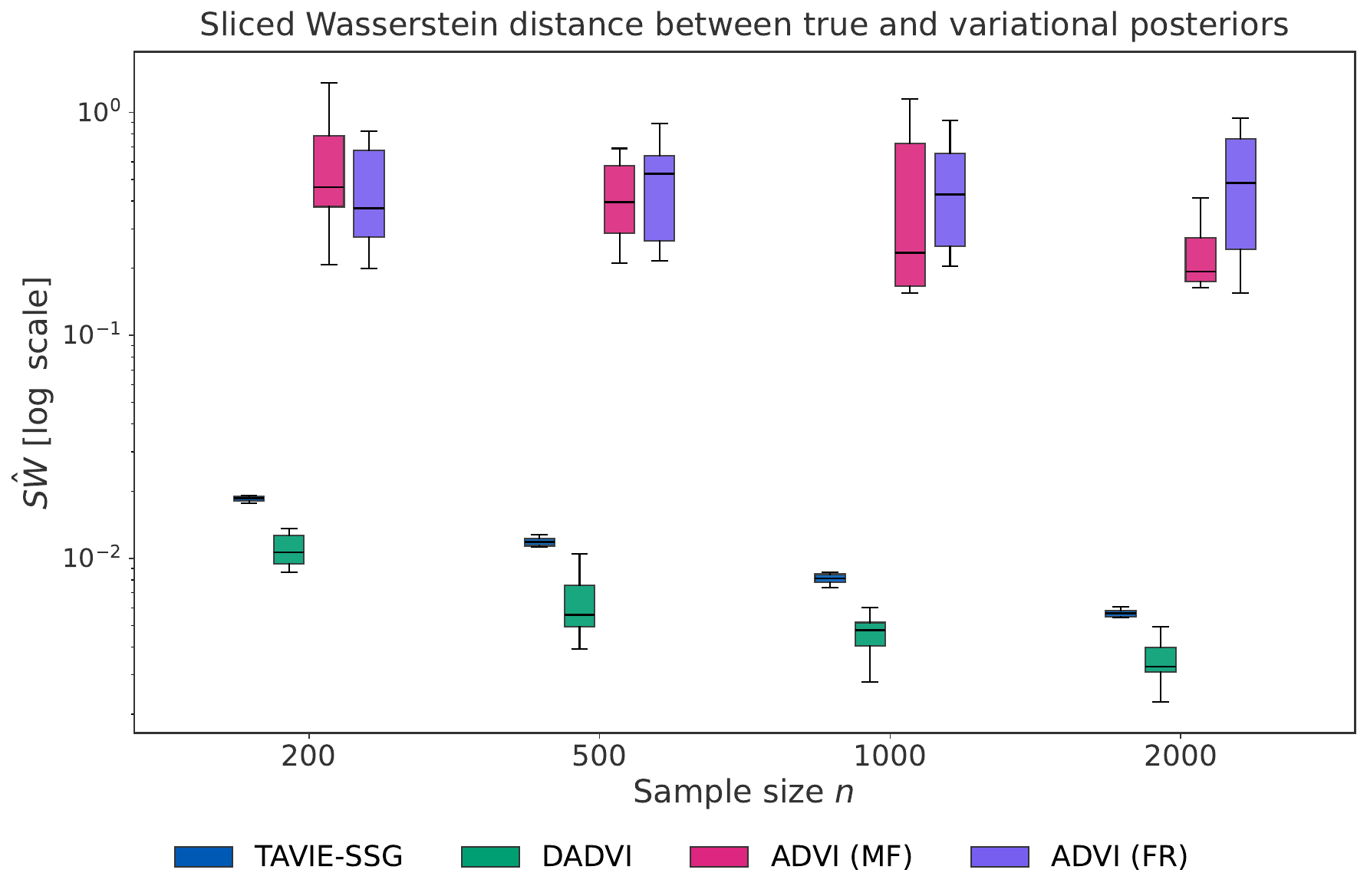}
    \caption{\footnotesize{Monte Carlo (MC) estimated sliced Wasserstein (SW) distance (in $\log$-scale) between the fractional true posterior (approximated using \texttt{PyMC} (NUTS)) and the variational posteriors of $\tssg$ and competitors under the Laplace Type I $\ssg$ model, with $p=9$ and $\alpha=1$. For each sample size $n\in \{200, 500, 1000, 2000\}$, the boxplots summarize results over $10$ independent replications. The SW distance is computed in the joint space $(\beta,\log \tau^2)$ using random one-dimensional projections.}}
    \label{fig:laplace_SW_distance}
\end{figure}

\begin{figure}[!htp]
    \centering
    \includegraphics[width=\linewidth]{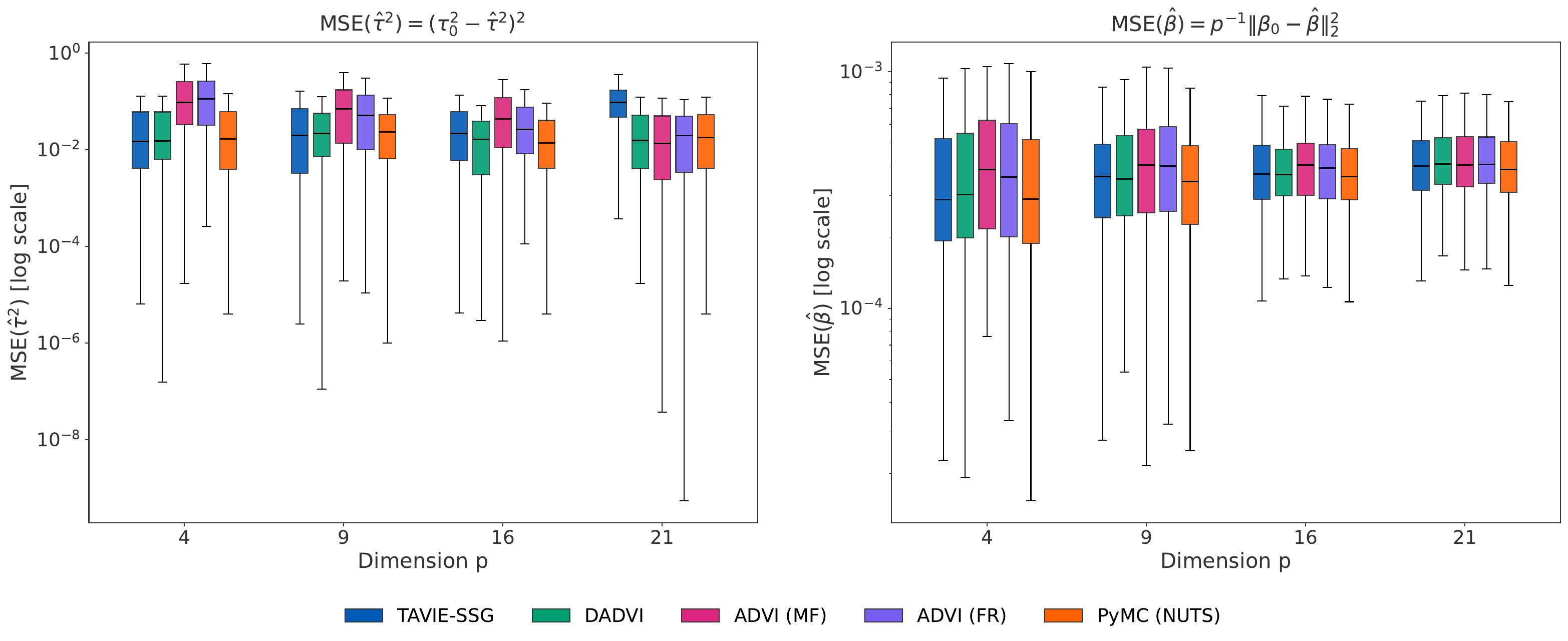}
    \caption{\footnotesize{MSEs of $(\widehat{\beta}, \widehat{\tau}^2)$ (in $\log$-scale) 
    across $100$ data repetitions of $\tssg$ and competitors for the Laplace $\mathsf{SSG}$ likelihood under experiment E2: $p\in \{4, 9, 16, 21\},\; n=1000$.}}
    \label{fig:laplace_singlen_multip}
\end{figure}

\begin{figure}[!htp]
    \centering
    \includegraphics[width=1.0\linewidth]{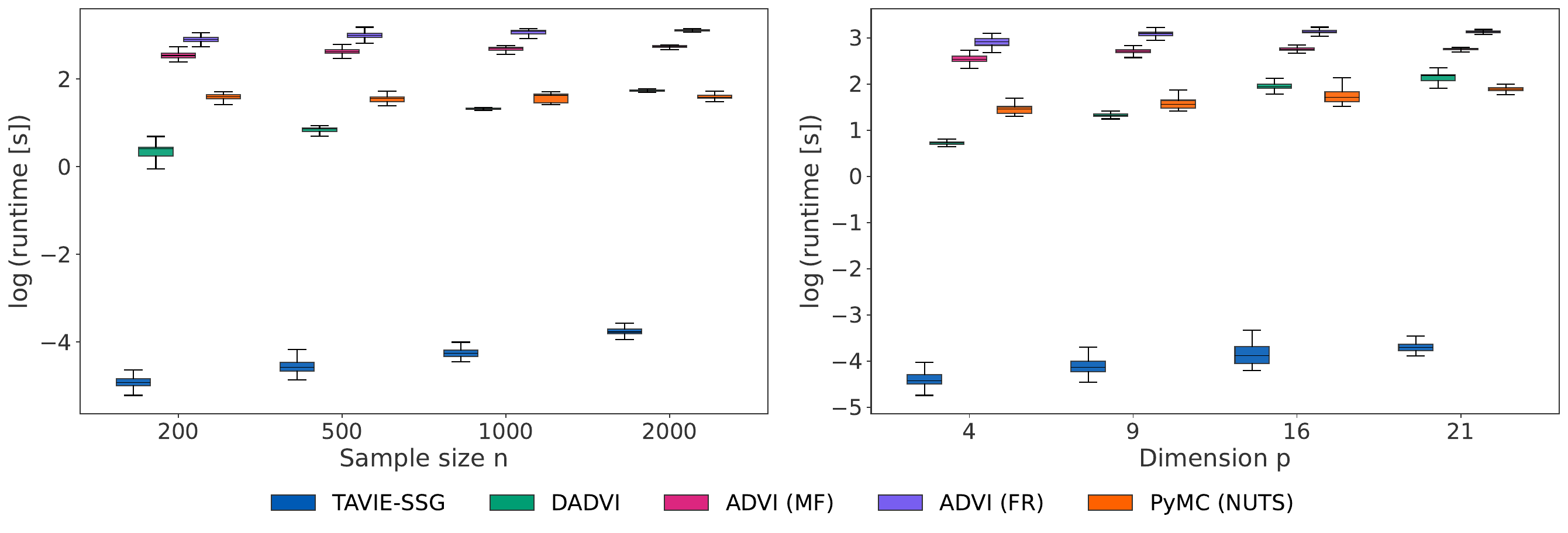}
    \caption{\footnotesize{Runtimes (in $\log$-scale) across $100$ data repetitions of $\tssg$ and competitors for the Laplace $\ssg$ likelihood under experiments E1 and E2.}}
    \label{fig:runtime_laplace}
\end{figure}

\begin{figure}[!htp]
    \centering
    \includegraphics[width=1.0\linewidth]{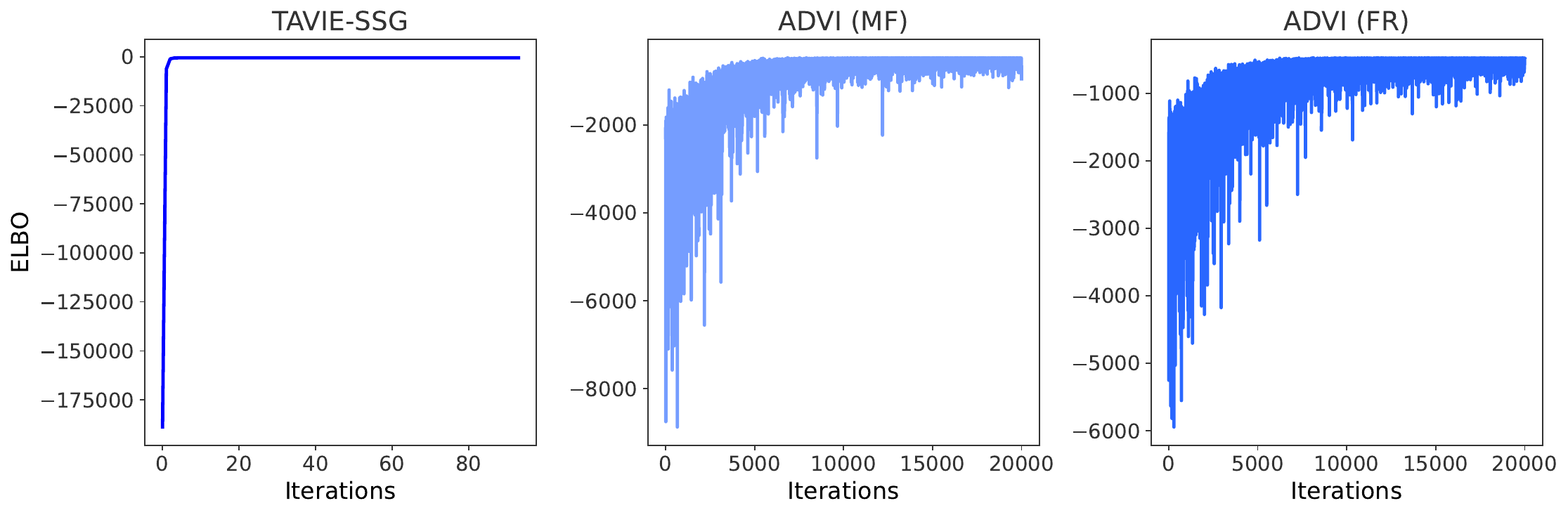}
    \caption{\footnotesize{Convergence diagnostics ($\mathsf{ELBO}$ monitoring) of $\tssg$, ADVI (MF), and ADVI (FR) for $(n,p)=(1000, 9)$ under the Laplace $\mathsf{SSG}$ likelihood. { For $\tssg$, the $\mathsf{ELBO}$ corresponds to $\mathsf{L}(\xi)$ in~\eqref{main-eq:ELBO-general} of the main manuscript, while for ADVI (MF/FR) the MC approximation of the true $\mathsf{ELBO}$ is tracked.}}}
    \label{fig:convergence_ELBO_laplace}
\end{figure}

\newpage

\subsection{Negative-Binomial Type II \texorpdfstring{$\mathsf{SSG}$}{SSG} Likelihood}\label{subsec:sim-exp-negbin}
The Negative-Binomial $\mathsf{SSG}$ model uses the count form with sizes $m_i>0$ (fixed) and success probability $p_i = \sigma(\mathbf{x}_i^{\top}\beta)$ with $\sigma(t) = [1 + \exp\{-t\}]^{-1}$ being the sigmoid function, where the per-observation probability mass function (for $y_i\in \mathbb{N}\cup \{0\}$) following \eqref{eq:param-negbin}, is:
\begin{align}\label{eq:negbin-SSG-model}
    y_i\mid \mathbf{x}_i, \beta \sim \mathrm{NB}(m_i, p_i), \quad p(y_i\mid \mathbf{x}_i, \beta) = \binom{y_i+m_i-1}{y_i}(1-p_i)^{y_i}p_{i}^{m_i},
\end{align}
independently for $i\in [n]$. Data, $\mathcal{D}_n:=\{(\mathbf{x}_i, y_i): i\in [n]\}$, is simulated using the true model parameters in \eqref{eq:negbin-SSG-model} configured as: $m_i=10$ for $i\in [n]$, $\beta_0 \sim \mathcal{N}_{p}(0, 0.5 I_{p})$, and $\mathbf{X} = (\mathbf{x}_1, \ldots, \mathbf{x}_n)^{\top}\in \mathbb{R}^{n\times p}$ with $x_{i1}=1$ (intercept) and $x_{ij}$'s generated independently and identically from the standard Gaussian distribution for $j=2,\ldots, p$.

We consider the same two sets of experiments: (E1) increasing sample sizes with a fixed number of features ($n\in \{200, 500, 1000, 2000\}, p = 9$) and (E2) increasing dimensionality with a fixed sample size ($p\in \{4, 9, 16, 21\}, n = 1000$). Under both experimental specifications, $\tssg$ is implemented for the standard likelihood setup with $\alpha=1$ and the prior hyperparameters set as, $(\mu, \Sigma) = (0, I_{p})$. Convergence is assessed in Algorithm~\ref{main-alg:tavie-em} of the main manuscript through a tolerance level of $\texttt{tol}=10^{-9}$. The competing methods (DADVI, ADVI (MF/FR), \texttt{PyMC} (NUTS)) are executed under experimental settings outlined in Section \ref{app-competing-methods-details}. The performance of $\tssg$ and competing methods is assessed based on statistical accuracy, quantified by the mean-squared error (MSE) between the estimated and true parameter values $p^{-1}\lVert \beta_0 - \widehat{\beta}\rVert_{2}^{2}$; computational efficiency, measured by the runtimes (in seconds) of each algorithm; and in experiment E1, sliced Wasserstein (SW) distance~\citep{bonneel2015sliced,kolouri2016sliced} between true and variational posteriors. The corresponding numerical results for MSE and runtime are tabulated in Tables \ref{tab:negbin_singlep_multin} and \ref{tab:negbin_singlen_multip} as well as visually represented in Figures \ref{fig:negbin_singlep_multin}, \ref{fig:runtime_negbin} ($100$ repetitions of the dataset $\mathcal{D}_n$ for MSE and runtime analyses). Figure~\ref{fig:neg-bin_SW_distance} illustrates the estimated SW distance between the variational posterior and the true posterior estimated using \texttt{PyMC} (NUTS) over first $10$ repetitions, under experiment E1.

For experiment E1, the results in Figure \ref{fig:negbin_singlep_multin} and Table \ref{tab:negbin_singlep_multin} highlight the superior performance of $\tssg$ comparing well with DADVI and \texttt{PyMC} (NUTS) in accurately estimating $\beta$. In contrast, both ADVI (MF/FR) produce less accurate posterior estimates of $\beta$ with comparatively higher MSE, reflecting poorer black-box optimization as the sample size increases. Also, the posterior distance comparisons in Figure~\ref{fig:neg-bin_SW_distance} show that $\tssg$ and DADVI produce closer approximations to the true posterior than ADVI variants.

The instability of the black-box optimization in ADVI (MF/FR) when estimating $\beta$ persists across varying (increasing) feature dimensions with a fixed sample size, i.e., in experiment E2, as portrayed through the results in Figure \ref{fig:negbin_singlen_multip} and Table \ref{tab:negbin_singlen_multip}, respectively. $\tssg$ exhibits improved performance over DADVI and \texttt{PyMC} (NUTS) in $\beta$-estimation, while also maintaining strong computational advantages, achieving runtimes an order of magnitude lower than the competing algorithms across all Negative-Binomial simulation studies (refer to Figure~\ref{fig:runtime_negbin}). Additional convergence diagnostics, showing the $\mathsf{ELBO}$ history over iterations for $\tssg$ and ADVI (MF/FR) are provided in Figure \ref{fig:convergence_ELBO_neg-bin}.

{\scriptsize
\setlength{\tabcolsep}{3pt} 
\renewcommand{\arraystretch}{1.15}%
\begin{longtable}{llcccc}
\caption{\footnotesize{Comparison of $\tssg$ against DADVI, ADVI (MF/FR), and \texttt{PyMC} (NUTS). Performance [median over $100$ repetitions; quartile range ($Q_1, Q_3$)  in parentheses] for the Negative-Binomial $\mathsf{SSG}$ likelihood under experiment E1. { Top $3$ best performing algorithms are in \textbf{bold}.}}} \label{tab:negbin_singlep_multin}\\
\toprule
\toprule
 &  & \multicolumn{4}{c}{Sample Size ($n$)} \\
\cmidrule(lr){3-6}
Metric & Method & 200 & 500 & 1000 & 2000 \\
\midrule
\endfirsthead
\toprule
 &  & \multicolumn{4}{c}{Sample Size ($n$)} \\
\cmidrule(lr){3-6}
Metric & Method & 200 & 500 & 1000 & 2000 \\
\midrule
\endhead
\midrule
\multicolumn{6}{r}{Continued on next page} \\
\midrule
\endfoot
\bottomrule
\endlastfoot

\multirow{5}{*}{MSE of $\widehat{\beta}$}
& $\tssg$ & \makecell[c]{\textbf{1.110e-03}\\(7.542e-04, 1.583e-03)} & \makecell[c]{\textbf{5.028e-04}\\(3.510e-04, 6.714e-04)} & \makecell[c]{\textbf{2.326e-04}\\(1.671e-04, 3.165e-04)} & \makecell[c]{\textbf{1.279e-04}\\(9.403e-05, 1.772e-04)} \\
& DADVI & \makecell[c]{\textbf{1.080e-03}\\(8.092e-04, 1.690e-03)} & \makecell[c]{\textbf{5.070e-04}\\(3.724e-04, 6.997e-04)} & \makecell[c]{\textbf{2.436e-04}\\(1.708e-04, 3.205e-04)} & \makecell[c]{\textbf{1.307e-04}\\(9.823e-05, 1.734e-04)} \\
& ADVI (MF) & \makecell[c]{1.429e-03\\(9.713e-04, 2.072e-03)} & \makecell[c]{6.106e-04\\(4.164e-04, 8.447e-04)} & \makecell[c]{2.791e-04\\(2.051e-04, 3.956e-04)} & \makecell[c]{1.789e-04\\(1.238e-04, 2.527e-04)} \\
& ADVI (FR) & \makecell[c]{1.385e-03\\(9.088e-04, 2.166e-03)} & \makecell[c]{6.110e-04\\(4.233e-04, 8.934e-04)} & \makecell[c]{2.800e-04\\(2.017e-04, 3.646e-04)} & \makecell[c]{1.432e-04\\(9.815e-05, 1.962e-04)} \\
& \texttt{PyMC} (NUTS) & \makecell[c]{\textbf{1.112e-03}\\(7.403e-04, 1.580e-03)} & \makecell[c]{\textbf{5.055e-04}\\(3.568e-04, 6.703e-04)} & \makecell[c]{\textbf{2.357e-04}\\(1.680e-04, 3.185e-04)} & \makecell[c]{\textbf{1.279e-04}\\(9.441e-05, 1.764e-04)} \\
\hline

\multirow{5}{*}{Runtime (s)}
& $\tssg$ & \makecell[c]{\textbf{3.709e-03}\\(3.434e-03, 4.485e-03)} & \makecell[c]{\textbf{4.459e-03}\\(4.083e-03, 5.576e-03)} & \makecell[c]{\textbf{5.513e-03}\\(4.949e-03, 7.461e-03)} & \makecell[c]{\textbf{8.233e-03}\\(7.327e-03, 1.003e-02)} \\
& DADVI & \makecell[c]{\textbf{7.744e-01}\\(7.262e-01, 8.346e-01)} & \makecell[c]{\textbf{9.946e-01}\\(9.484e-01, 1.041e+00)} & \makecell[c]{\textbf{1.224e+00}\\(1.171e+00, 1.291e+00)} & \makecell[c]{\textbf{1.499e+00}\\(1.454e+00, 1.557e+00)} \\
& ADVI (MF) & \makecell[c]{\textbf{4.008e+00}\\(3.937e+00, 4.070e+00)} & \makecell[c]{\textbf{4.075e+00}\\(3.977e+00, 4.147e+00)} & \makecell[c]{\textbf{4.061e+00}\\(3.978e+00, 4.162e+00)} & \makecell[c]{\textbf{4.207e+00}\\(4.114e+00, 4.311e+00)} \\
& ADVI (FR) & \makecell[c]{1.276e+01\\(1.206e+01, 1.351e+01)} & \makecell[c]{1.461e+01\\(1.372e+01, 1.490e+01)} & \makecell[c]{1.489e+01\\(1.443e+01, 1.525e+01)} & \makecell[c]{1.516e+01\\(1.493e+01, 1.545e+01)} \\
& \texttt{PyMC} (NUTS) & \makecell[c]{6.652e+00\\(6.403e+00, 6.943e+00)} & \makecell[c]{6.995e+00\\(6.465e+00, 7.411e+00)} & \makecell[c]{7.798e+00\\(7.227e+00, 8.426e+00)} & \makecell[c]{1.011e+01\\(9.515e+00, 1.121e+01)} \\
\hline
\end{longtable}
}

{\scriptsize
\setlength{\tabcolsep}{3pt} 
\renewcommand{\arraystretch}{1.15}%
\begin{longtable}{llcccc}
\caption{\footnotesize{Comparison of $\tssg$ against DADVI, ADVI (MF/FR), and \texttt{PyMC} (NUTS). Performance [median over $100$ repetitions; quartile range ($Q_1, Q_3$)  in parentheses] for the Negative-Binomial $\mathsf{SSG}$ likelihood under experiment E2. { Top $3$ best performing algorithms are in \textbf{bold}.}}} 
\label{tab:negbin_singlen_multip}\\
\toprule
\toprule
 &  & \multicolumn{4}{c}{Dimension ($p$)} \\
\cmidrule(lr){3-6}
Metric & Method & 4 & 9 & 16 & 21 \\
\midrule
\endfirsthead
\toprule
 &  & \multicolumn{4}{c}{Dimension ($p$)} \\
\cmidrule(lr){3-6}
Metric & Method & 4 & 9 & 16 & 21 \\
\midrule
\endhead
\midrule
\multicolumn{6}{r}{Continued on next page} \\
\midrule
\endfoot
\bottomrule
\endlastfoot

\multirow{5}{*}{MSE of $\widehat{\beta}$}
& $\tssg$ & \makecell[c]{\textbf{1.471e-04}\\(1.035e-04, 2.732e-04)} & \makecell[c]{\textbf{1.652e-04}\\(1.126e-04, 2.204e-04)} & \makecell[c]{\textbf{1.795e-04}\\(1.383e-04, 2.020e-04)} & \makecell[c]{\textbf{2.415e-04}\\(1.772e-04, 2.939e-04)} \\
& DADVI & \makecell[c]{\textbf{1.694e-04}\\(1.094e-04, 2.756e-04)} & \makecell[c]{\textbf{1.730e-04}\\(1.176e-04, 2.400e-04)} & \makecell[c]{\textbf{1.801e-04}\\(1.458e-04, 2.151e-04)} & \makecell[c]{\textbf{2.437e-04}\\(1.828e-04, 2.921e-04)} \\
& ADVI (MF) & \makecell[c]{3.575e-04\\(1.720e-04, 6.735e-04)} & \makecell[c]{2.794e-04\\(1.820e-04, 4.092e-04)} & \makecell[c]{2.593e-03\\(1.782e-03, 3.646e-03)} & \makecell[c]{2.100e-02\\(1.101e-02, 4.009e-02)} \\
& ADVI (FR) & \makecell[c]{3.530e-04\\(1.856e-04, 7.168e-04)} & \makecell[c]{1.741e-04\\(1.296e-04, 2.567e-04)} & \makecell[c]{1.833e-04\\(1.547e-04, 2.292e-04)} & \makecell[c]{2.549e-04\\(1.995e-04, 3.117e-04)} \\
& \texttt{PyMC} (NUTS) & \makecell[c]{\textbf{1.510e-04}\\(1.028e-04, 2.761e-04)} & \makecell[c]{\textbf{1.645e-04}\\(1.137e-04, 2.168e-04)} & \makecell[c]{\textbf{1.814e-04}\\(1.377e-04, 2.017e-04)} & \makecell[c]{\textbf{2.423e-04}\\(1.769e-04, 2.909e-04)} \\
\hline

\multirow{5}{*}{Runtime (s)}
& $\tssg$ & \makecell[c]{\textbf{3.634e-03}\\(3.177e-03, 4.488e-03)} & \makecell[c]{\textbf{9.286e-03}\\(8.433e-03, 1.044e-02)} & \makecell[c]{\textbf{4.726e-02}\\(4.075e-02, 5.987e-02)} & \makecell[c]{\textbf{8.083e-02}\\(5.691e-02, 1.095e-01)} \\
& DADVI & \makecell[c]{\textbf{1.384e+00}\\(1.349e+00, 1.446e+00)} & \makecell[c]{\textbf{1.542e+00}\\(1.495e+00, 1.593e+00)} & \makecell[c]{\textbf{1.636e+00}\\(1.599e+00, 1.715e+00)} & \makecell[c]{\textbf{1.705e+00}\\(1.658e+00, 1.750e+00)} \\
& ADVI (MF) & \makecell[c]{\textbf{4.030e+00}\\(3.963e+00, 4.084e+00)} & \makecell[c]{\textbf{4.156e+00}\\(4.077e+00, 4.226e+00)} & \makecell[c]{\textbf{4.145e+00}\\(4.077e+00, 4.223e+00)} & \makecell[c]{\textbf{4.123e+00}\\(4.067e+00, 4.179e+00)} \\
& ADVI (FR) & \makecell[c]{1.223e+01\\(1.114e+01, 1.317e+01)} & \makecell[c]{1.517e+01\\(1.489e+01, 1.538e+01)} & \makecell[c]{1.525e+01\\(1.496e+01, 1.542e+01)} & \makecell[c]{1.518e+01\\(1.499e+01, 1.532e+01)} \\
& \texttt{PyMC} (NUTS) & \makecell[c]{6.560e+00\\(6.338e+00, 6.989e+00)} & \makecell[c]{8.407e+00\\(8.127e+00, 8.794e+00)} & \makecell[c]{9.336e+00\\(9.053e+00, 9.762e+00)} & \makecell[c]{9.480e+00\\(9.132e+00, 1.004e+01)} \\
\hline
\end{longtable}
}

\begin{figure}[!htp]
    \centering
    \includegraphics[width=0.9\linewidth]{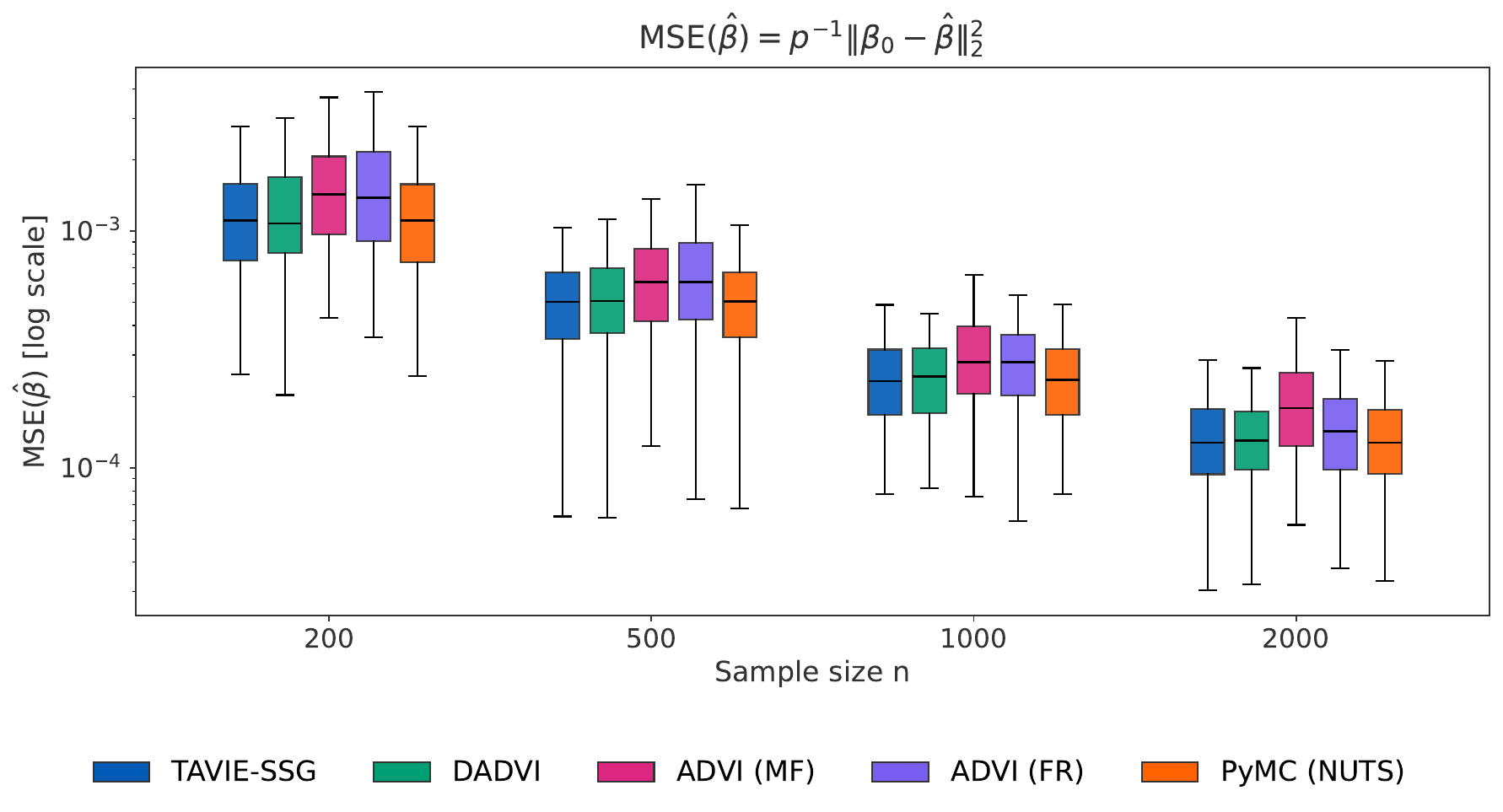}
    \caption{\footnotesize{MSE of $\widehat{\beta}$ (in $\log$-scale) 
    across $100$ data repetitions of $\tssg$ and competitors for the Negative-Binomial $\mathsf{SSG}$ likelihood under experiment E1: $n\in \{200, 500, 1000, 2000\},\; p=9$.}}
    \label{fig:negbin_singlep_multin}
\end{figure}

\begin{figure}[!htp]
    \centering
    \includegraphics[width=0.8\linewidth]{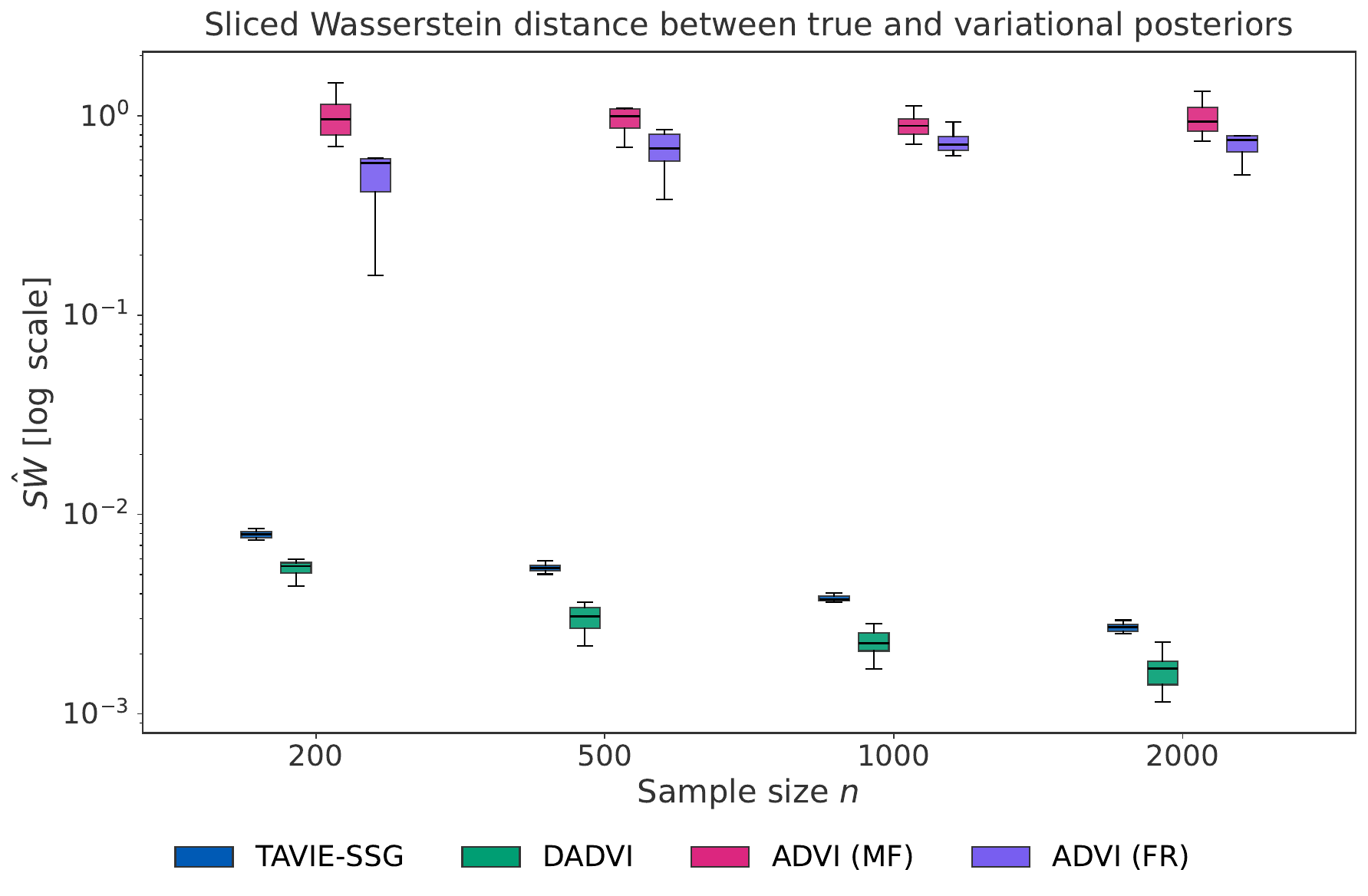}
    \caption{\footnotesize{Monte Carlo (MC) estimated sliced Wasserstein (SW) distance (in $\log$-scale) between the fractional true posterior (approximated using \texttt{PyMC} (NUTS)) and the variational posteriors of $\tssg$ and competitors under the Negative-Binomial Type II $\ssg$ model, with $p=9$ and $\alpha=1$. For each sample size $n\in \{200, 500, 1000, 2000\}$, the boxplots summarize results over $10$ independent replications.}}
    \label{fig:neg-bin_SW_distance}
\end{figure}

\begin{figure}[!htp]
    \centering
    \includegraphics[width=\linewidth]{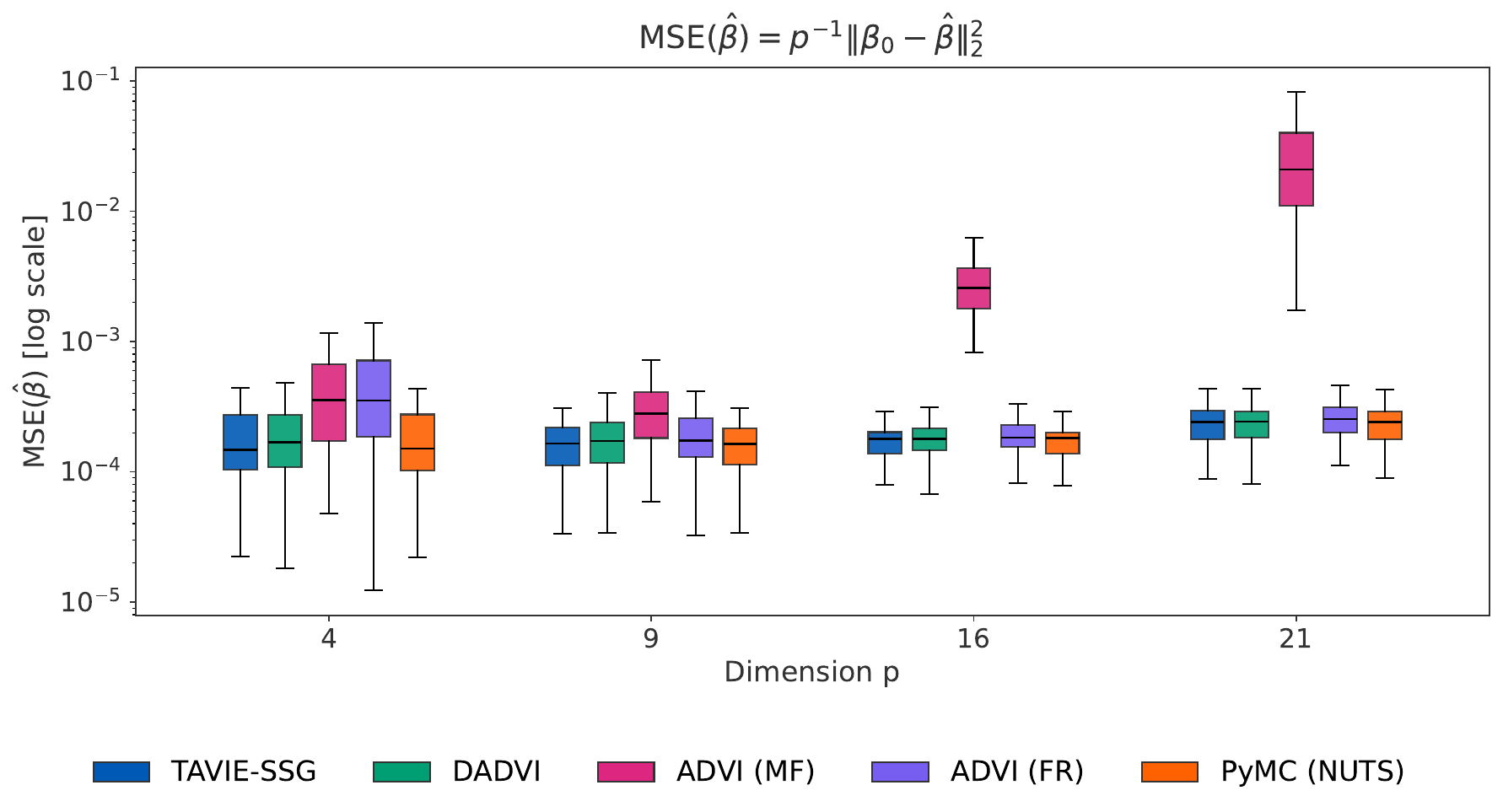}
    \caption{\footnotesize{MSE of $\widehat{\beta}$ (in $\log$-scale) across $100$ data repetitions of $\tssg$ and competitors for the Negative-Binomial $\mathsf{SSG}$ likelihood under experiment E2: $p\in \{4, 9, 16, 21\}, n=1000$.}}
    \label{fig:negbin_singlen_multip}
\end{figure}

\begin{figure}[!htp]
    \centering
    \includegraphics[width=1.0\linewidth]{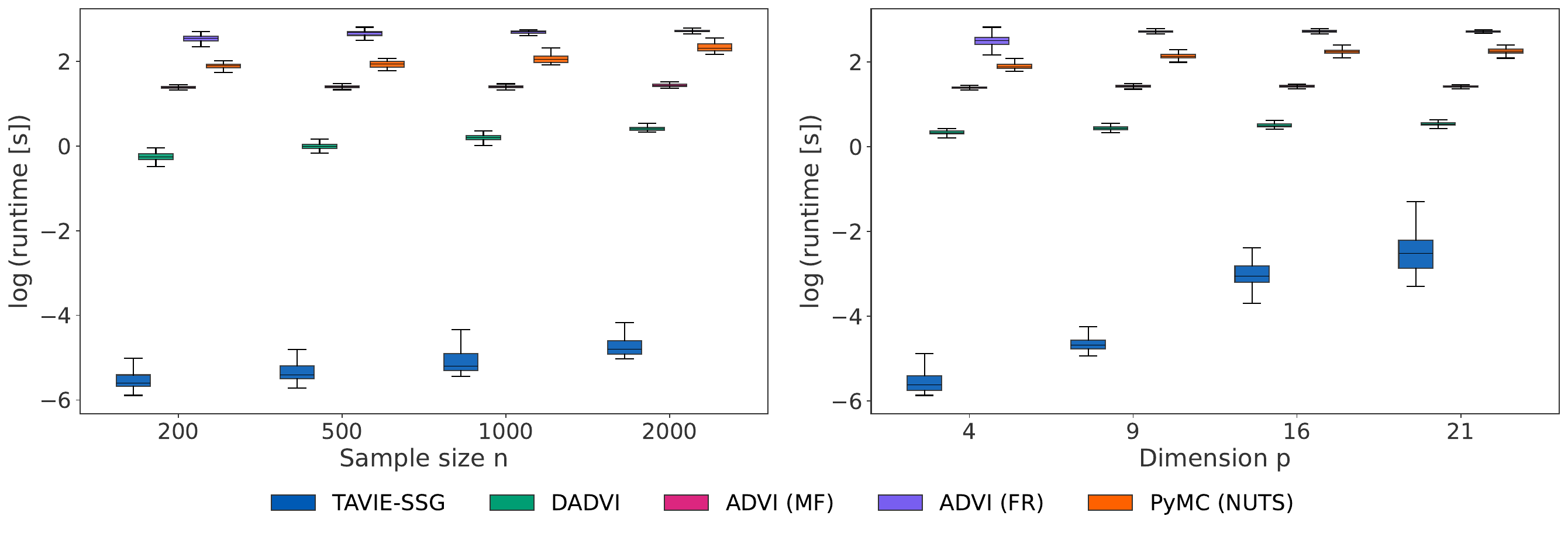}
    \caption{\footnotesize{Runtimes (in $\log$-scale) across $100$ data repetitions of $\tssg$ and competitors for the Negative-Binomial $\ssg$ likelihood under experiments E1 and E2.}}
    \label{fig:runtime_negbin}
\end{figure}

\newpage

\begin{figure}[!htp]
    \centering
    \includegraphics[width=1.0\linewidth]{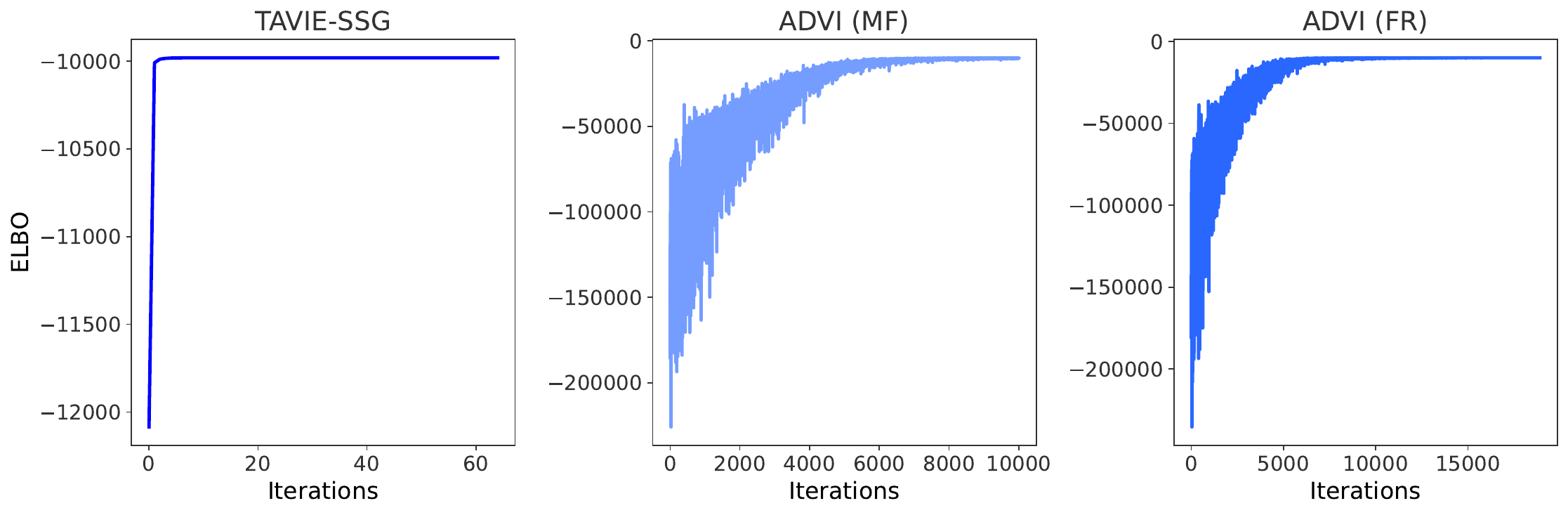}
    \caption{\footnotesize{Convergence diagnostics ($\mathsf{ELBO}$ monitoring) of $\tssg$, ADVI (MF), and ADVI (FR) for $(n,p)=(1000, 9)$ under the Negative-Binomial $\mathsf{SSG}$ likelihood. { For $\tssg$, the $\mathsf{ELBO}$ corresponds to $\mathsf{L}(\xi)$ in~\eqref{main-eq:ELBO-general} of the main manuscript, while for ADVI (MF/FR) the MC approximation of the true $\mathsf{ELBO}$ is tracked.}}}
    \label{fig:convergence_ELBO_neg-bin}
\end{figure}

\newpage

\subsection{Extended Evaluation of \texorpdfstring{$\tssg$}{TAVIE-SSG} under Different \texorpdfstring{$\alpha$}{alpha}}

In addition to the results presented in Sections \ref{app:additional-student-results}, \ref{app:additional-laplace-results}, and \ref{subsec:sim-exp-negbin}, we further evaluate the performance of $\tssg$ under Student's-$t$ ($\nu=5$), Laplace, and Negative-Binomial $\ssg$ likelihoods over a range of tempering parameters $\alpha \in \{0.20,\, 0.40,\, 0.60,\, 0.80,\, 0.95,\, 1.00\}$. For each likelihood tempered with the aforementioned choices of $\alpha$, $\tssg$ is implemented using the same experimental settings as in the corresponding simulation setups of Sections~\ref{main-subsec:sim-exp-student} (in the main manuscript), \ref{app:additional-laplace-results}, and \ref{subsec:sim-exp-negbin}, respectively. All experiments are conducted at sample size and feature dimension $(n, p) = (2000, 9)$ for the specified values of $\alpha$.


\subsubsection{\texorpdfstring{Student's-$t$}{Student's-t} Type I \texorpdfstring{$\ssg$}{SSG} Likelihood}

{\scriptsize
\setlength{\tabcolsep}{6pt} 
\renewcommand{\arraystretch}{1.15}%
\begin{longtable}{lccc}
\caption{\footnotesize{$\tssg$ under different choices of the likelihood tempering parameter $\alpha$. Performance [median over $100$ repetitions of the simulated dataset with $(n, p) = (2000, 9)$; quartile range $(Q_1, Q_3)$ in parentheses] for the Student's-$t$ $\ssg$ likelihood ($\nu=5$).}} \label{tab:Student-different-alpha}\\
\toprule
\toprule
$\alpha$ & MSE of $\widehat{\beta}$ & MSE of $\widehat{\tau}^{2}$ & Runtime (s) \\
\midrule
\endfirsthead

\toprule
$\alpha$ & MSE of $\widehat{\beta}$ & MSE of $\widehat{\tau}^{2}$ & Runtime (s) \\
\midrule
\endhead

\midrule
\multicolumn{4}{r}{Continued on next page} \\
\midrule
\endfoot

\bottomrule
\bottomrule
\endlastfoot

0.20 & 
  \makecell[c]{2.074e-04\\(1.492e-04, 3.004e-04)} &
  \makecell[c]{1.038e-01\\(5.693e-02, 1.356e-01)} &
  \makecell[c]{7.745e-03\\(7.572e-03, 7.964e-03)} \\

0.40 & 
  \makecell[c]{2.059e-04\\(1.474e-04, 2.786e-04)} &
  \makecell[c]{3.083e-02\\(7.879e-03, 5.140e-02)} &
  \makecell[c]{8.070e-03\\(7.924e-03, 8.298e-03)} \\

0.60 & 
  \makecell[c]{1.993e-04\\(1.487e-04, 2.770e-04)} &
  \makecell[c]{1.704e-02\\(3.443e-03, 3.055e-02)} &
  \makecell[c]{8.188e-03\\(8.110e-03, 8.357e-03)} \\

0.80 & 
  \makecell[c]{1.982e-04\\(1.477e-04, 2.773e-04)} &
  \makecell[c]{1.196e-02\\(2.989e-03, 2.426e-02)} &
  \makecell[c]{8.262e-03\\(8.152e-03, 8.415e-03)} \\

0.95 & 
  \makecell[c]{1.970e-04\\(1.471e-04, 2.769e-04)} &
  \makecell[c]{1.001e-02\\(2.434e-03, 2.040e-02)} &
  \makecell[c]{8.276e-03\\(8.136e-03, 8.369e-03)} \\

1.00 & 
  \makecell[c]{1.971e-04\\(1.470e-04, 2.770e-04)} &
  \makecell[c]{9.789e-03\\(2.427e-03, 1.951e-02)} &
  \makecell[c]{8.335e-03\\(8.143e-03, 8.393e-03)} \\

\end{longtable}
}

\begin{figure}[H]
    \centering
    \includegraphics[width=0.8\linewidth]{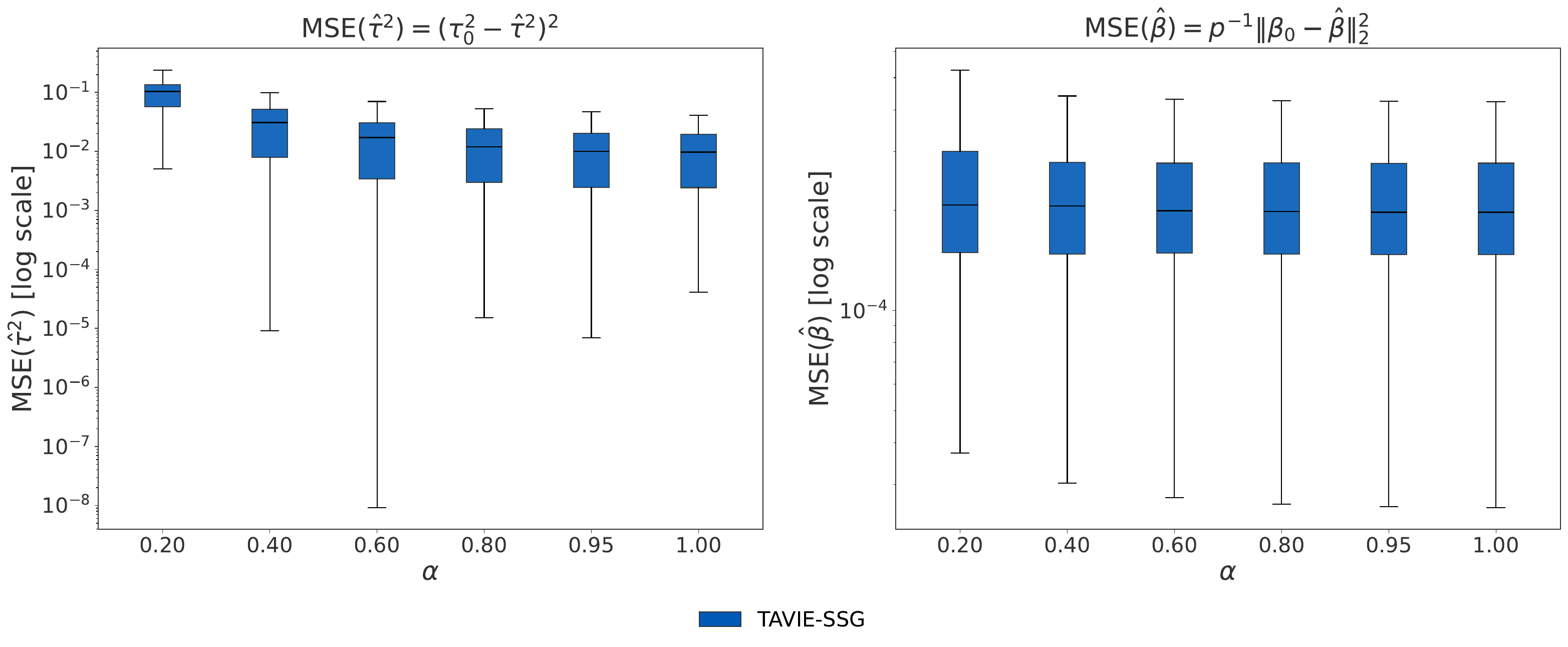}
    \caption{\footnotesize{MSEs of $(\widehat{\beta}, \widehat{\tau}^2)$ (in $\log$-scale) across $100$ data repetitions with $(n, p) = (2000, 9)$ for $\tssg$ under Student's-$t$ $\ssg$ likelihood ($\nu=5$) across different choices of the likelihood tempering parameter $\alpha$.}}
    \label{fig:Student_alpha_MSE_boxplots}
\end{figure}


\subsubsection{Laplace Type I \texorpdfstring{$\ssg$}{SSG} Likelihood}

{\scriptsize
\setlength{\tabcolsep}{6pt} 
\renewcommand{\arraystretch}{1.15}%
\begin{longtable}{lccc}
\caption{\footnotesize{$\tssg$ under different choices of the likelihood tempering parameter $\alpha$. Performance [median over $100$ repetitions of the simulated dataset with $(n, p) = (2000, 9)$; quartile range $(Q_1, Q_3)$ in parentheses] for the Laplace $\ssg$ likelihood.}} \label{tab:Laplace-different-alpha}\\
\toprule
\toprule
$\alpha$ & MSE of $\widehat{\beta}$ & MSE of $\widehat{\tau}^{2}$ & Runtime (s) \\
\midrule
\endfirsthead

\toprule
$\alpha$ & MSE of $\widehat{\beta}$ & MSE of $\widehat{\tau}^{2}$ & Runtime (s) \\
\midrule
\endhead

\midrule
\multicolumn{4}{r}{Continued on next page} \\
\midrule
\endfoot

\bottomrule
\bottomrule
\endlastfoot

0.20 &
  \makecell[c]{1.609e-04\\(1.276e-04, 2.438e-04)} &
  \makecell[c]{1.227e-01\\(8.567e-02, 1.748e-01)} &
  \makecell[c]{1.651e-02\\(1.509e-02, 1.803e-02)} \\

0.40 &
  \makecell[c]{1.609e-04\\(1.172e-04, 2.355e-04)} &
  \makecell[c]{3.431e-02\\(1.422e-02, 6.922e-02)} &
  \makecell[c]{1.938e-02\\(1.768e-02, 2.098e-02)} \\

0.60 &
  \makecell[c]{1.614e-04\\(1.181e-04, 2.377e-04)} &
  \makecell[c]{1.677e-02\\(5.419e-03, 4.202e-02)} &
  \makecell[c]{2.104e-02\\(1.963e-02, 2.302e-02)} \\

0.80 &
  \makecell[c]{1.627e-04\\(1.182e-04, 2.325e-04)} &
  \makecell[c]{1.291e-02\\(3.910e-03, 3.094e-02)} &
  \makecell[c]{2.306e-02\\(2.059e-02, 2.528e-02)} \\

0.95 &
  \makecell[c]{1.646e-04\\(1.170e-04, 2.314e-04)} &
  \makecell[c]{1.084e-02\\(2.924e-03, 2.653e-02)} &
  \makecell[c]{2.405e-02\\(2.158e-02, 2.670e-02)} \\

1.00 &
  \makecell[c]{1.652e-04\\(1.173e-04, 2.314e-04)} &
  \makecell[c]{1.019e-02\\(2.803e-03, 2.526e-02)} &
  \makecell[c]{2.485e-02\\(2.251e-02, 2.783e-02)} \\

\end{longtable}
}

\begin{figure}[H]
    \centering
    \includegraphics[width=\linewidth]{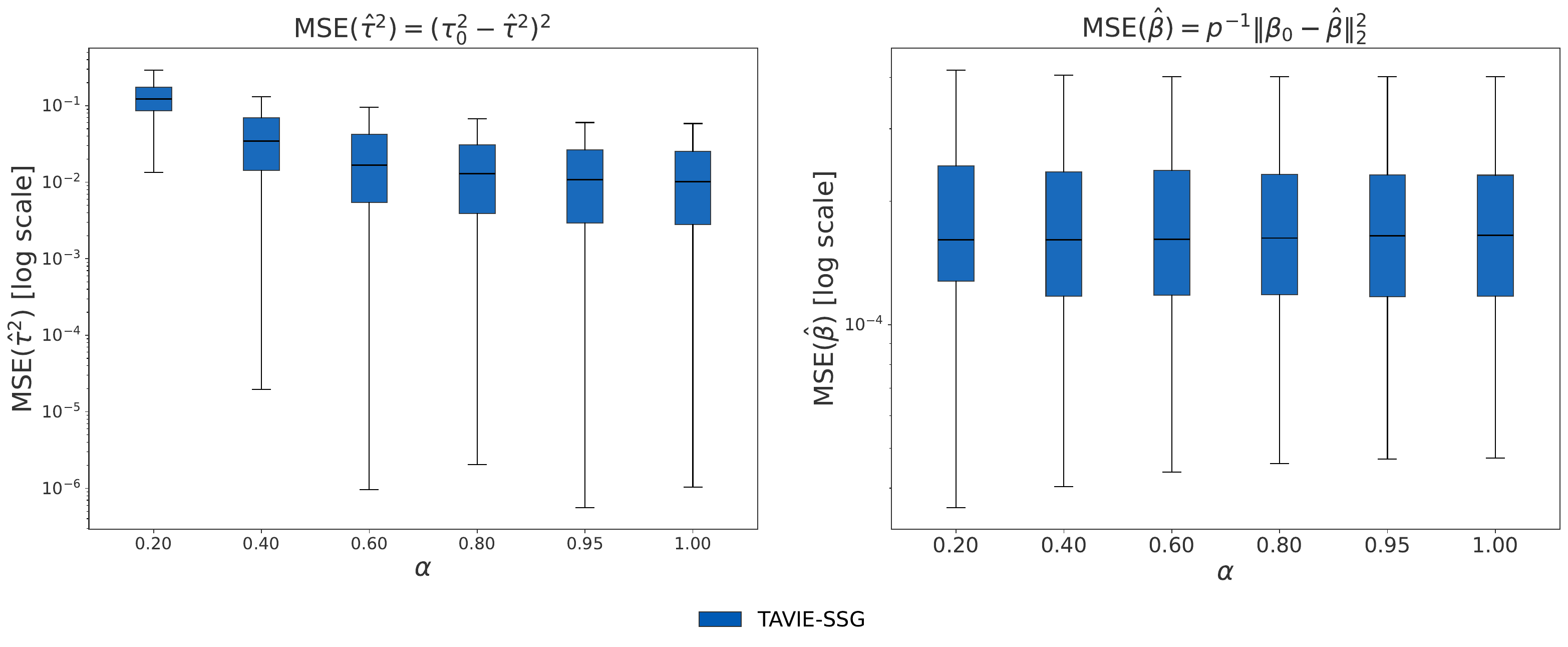}
    \caption{\footnotesize{MSEs of $(\widehat{\beta}, \widehat{\tau}^2)$ (in $\log$-scale) across $100$ data repetitions with $(n, p) = (2000, 9)$ for $\tssg$ under Laplace $\ssg$ likelihood across different choices of the likelihood tempering parameter $\alpha$.}}
    \label{fig:Laplace_alpha_MSE_boxplots}
\end{figure}

\newpage

\subsubsection{Negative-Binomial Type II \texorpdfstring{$\ssg$}{SSG} Likelihood}

{\scriptsize
\setlength{\tabcolsep}{6pt} 
\renewcommand{\arraystretch}{1.15}%
\begin{longtable}{lccc}
\caption{\footnotesize{$\tssg$ under different choices of the likelihood tempering parameter $\alpha$. Performance [median over $100$ repetitions of the simulated dataset with $(n, p) = (2000, 9)$; quartile range $(Q_1, Q_3)$ in parentheses] for the Negative-Binomial $\ssg$ likelihood.}} \label{tab:NegBin-different-alpha}\\
\toprule
\toprule
$\alpha$ & MSE of $\widehat{\beta}$ & Runtime (s) \\
\midrule
\endfirsthead

\toprule
$\alpha$ & MSE of $\widehat{\beta}$ & Runtime (s) \\
\midrule
\endhead

\midrule
\multicolumn{3}{r}{Continued on next page} \\
\midrule
\endfoot

\bottomrule
\bottomrule
\endlastfoot

0.20 &
  \makecell[c]{8.763e-05\\(5.344e-05, 1.245e-04)} &
  \makecell[c]{1.768e-02\\(1.663e-02, 1.959e-02)} \\

0.40 &
  \makecell[c]{8.802e-05\\(5.310e-05, 1.238e-04)} &
  \makecell[c]{1.794e-02\\(1.670e-02, 2.031e-02)} \\

0.60 &
  \makecell[c]{8.815e-05\\(5.300e-05, 1.236e-04)} &
  \makecell[c]{1.874e-02\\(1.718e-02, 2.062e-02)} \\

0.80 &
  \makecell[c]{8.822e-05\\(5.295e-05, 1.235e-04)} &
  \makecell[c]{1.811e-02\\(1.709e-02, 2.044e-02)} \\

0.95 &
  \makecell[c]{8.825e-05\\(5.292e-05, 1.234e-04)} &
  \makecell[c]{1.784e-02\\(1.645e-02, 1.964e-02)} \\

1.00 &
  \makecell[c]{8.826e-05\\(5.292e-05, 1.234e-04)} &
  \makecell[c]{1.802e-02\\(1.699e-02, 2.014e-02)} \\

\end{longtable}
}

\begin{figure}[H]
    \centering
    \includegraphics[width=\linewidth]{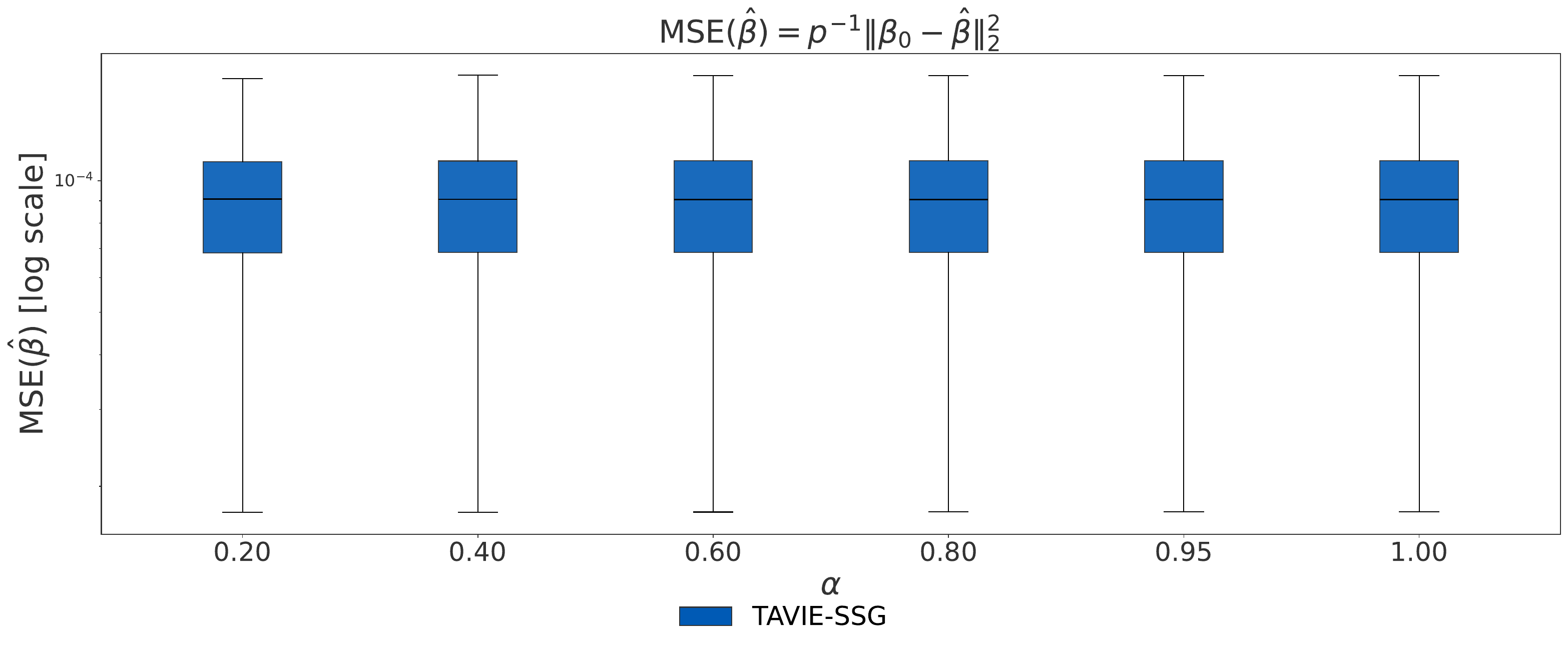}
    \caption{\footnotesize{MSE of $\widehat{\beta}$ (in $\log$-scale) across $100$ data repetitions with $(n, p) = (2000, 9)$ for $\tssg$ under Negative-Binomial $\ssg$ likelihood across different choices of the likelihood tempering parameter $\alpha$.}}
    \label{fig:NegBin_alpha_MSE_boxplots}
\end{figure}

\newpage

{

\subsection{Effect of the Dimension-to-Sample-Size Ratio on \texorpdfstring{$\tssg$}{TAVIE-SSG}}
\label{subsec:regimes-p-comparable-n}
Here we investigate the empirical behavior of $\tssg$ in a regime where the sample size $n$ is held fixed while the ambient dimension $p$ increases from small to moderate values, eventually becoming comparable to $n$. For illustration, we consider the Student's-$t$ Type I (the dimension-to-sample-size ratio increases up to $0.80$, i.e., $\tfrac{p}{n} \uparrow 0.80$) and Negative-Binomial Type II (the dimension-to-sample-size ratio increases up to $0.50$, i.e., $\tfrac{p}{n} \uparrow 0.50$) $\ssg$ models under the same experimental configurations as outlined in Section~\ref{main-subsec:sim-exp-student} of the main manuscript and Section~\ref{subsec:sim-exp-negbin}, respectively.

Figure~\ref{fig:mse_boxplots_scaling_p_student_t} reports the mean-squared errors (MSEs) of $\widehat{\beta}$ and $\widehat{\tau}^2$, i.e., $p^{-1}\lVert \beta_0 - \widehat{\beta}\rVert_{2}^{2}$ and $(\tau_0^{2} - \widehat{\tau}^{2})^2$, across $10$ independently regenerated datasets for the Student's-$t$ Type I $\ssg$ model. Here we fix $n=1000$ and consider $p\in \{51, 101, 201, 501, 801\}$ (including intercept). The left panel shows that $\tssg$ remains highly accurate for estimating $\beta$ over a broad range of dimensions, with only mild deterioration up to $p=201$. A more noticeable increase in MSE appears at $p=501$, and the degradation becomes pronounced at $p=801$. Thus, recovery of $\beta$ remains robust in moderate dimensions, but worsens once the model becomes highly parameter-rich relative to the sample size.
The right panel shows a stronger sensitivity for estimation of the precision parameter $\tau^2$. While the MSE of $\widehat{\tau}^{2}$ is small for low dimensions, it increases steadily with $p$, becoming substantially larger from $p=101$ onward and particularly pronounced at $p=501$ and $p=801$. This suggests that scale estimation is more sensitive to growing dimensionality than $\beta$-estimation, which is expected since reliable recovery of the residual scale depends on accurate estimation of a large number of regression coefficients.

\begin{figure}[!htp]
    \centering
    \includegraphics[width=\linewidth]{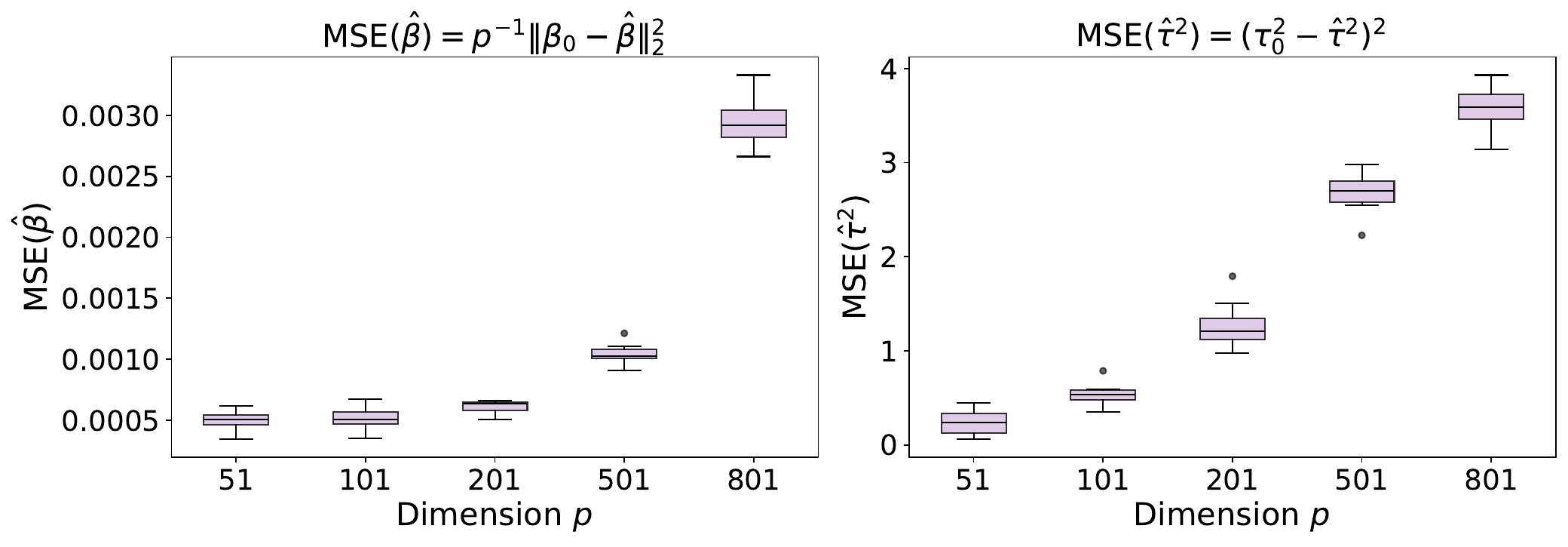}
    \caption{\footnotesize{Performance of $\tssg$ for the Student's-$t$ Type I $\ssg$ model as the dimension-to-sample-size ratio $\tfrac{p}{n}$ increases (up to $0.80$). MSEs of $\widehat{\beta}$ and $\widehat{\tau}^{2}$ over $10$ independent data regenerations, with fixed $n=1000$ and $p\in \{51, 101, 201, 501, 801\}$.}}
    \label{fig:mse_boxplots_scaling_p_student_t}
\end{figure}

For the Negative-Binomial experiment, we fix $n=400$ and consider $p\in \{11, 51, 101, 151, 201\}$ (including intercept). Figure~\ref{fig:mse_boxplots_scaling_p_negbin} reports the MSE of $\widehat{\beta}$ across $10$ independently regenerated datasets. The same qualitative pattern emerges as $\tfrac{p}{n}$ increases: for $p=11$ and $p=51$, the MSE remains very small, indicating strong performance of $\tssg$ in low- to moderate-dimensional count regression settings; however, once $p=101$, the MSE increases markedly, with further deterioration at $p=151$ and $p=201$.

\begin{figure}[!htp]
    \centering
    \includegraphics[width=0.7\linewidth]{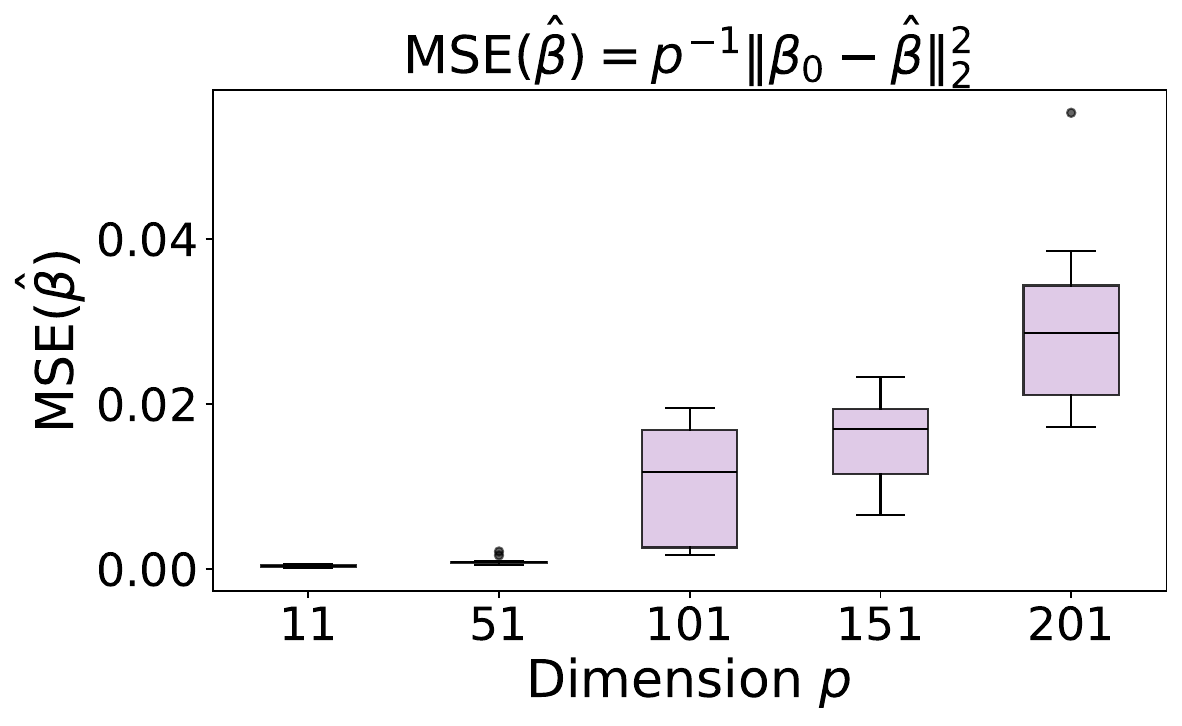}
    \caption{\footnotesize{Performance of $\tssg$ for the Negative-Binomial Type II $\ssg$ model as the dimension-to-sample-size ratio $\tfrac{p}{n}$ increases (up to $0.50$). MSEs of $\widehat{\beta}$ over $10$ independent data regenerations, with fixed $n=400$ and $p\in \{11, 51, 101, 151, 201\}$.}}
    \label{fig:mse_boxplots_scaling_p_negbin}
\end{figure}

Overall, the above experiments show that $\tssg$ remains accurate when $\tfrac{p}{n}$ is moderate, but its performance deteriorates as the ambient dimension $p$ grows large relative to the sample size $n$. This behavior is a natural manifestation of the curse of dimensionality and falls beyond the primary scope of the current $\tssg$ framework. A more refined treatment of such large dimensional choices would likely require sparsity-inducing priors or structural sparsity assumptions on the design matrix $\mathbf{X}$, both of which are promising directions deferred to future work.
}

\newpage

\section{Auxiliary Results for Bayesian Quantile Regression on U.S. Census 2000 Data}\label{sec:BQR-additional}

\subsection{Adaptation of \texorpdfstring{$\tssg$}{TAVIE-SSG} Methodology to Bayesian Quantile Regression}\label{subsec:TAVIE-BQR-setup}

Recall, the joint likelihood for quantile regression with $\tau = \tau_0$ and quantile of interest $u\in (0, 1)$, is given by:
\begin{align}\label{eq:TAVIE-QR-joint-likelihood}
    p(y\mid \mathbf{X}, \beta) = \prod_{i\in [n]}p_{\mathrm{ALD}}(y_i - \mathbf{x}_i^{\top}\beta\mid \tau_0, u).
\end{align}
Following Proposition~\ref{main-prop:tangent-lower-bound} of the main manuscript, the joint likelihood in \eqref{eq:TAVIE-QR-joint-likelihood} is minorized by:
\begin{align}
\label{eq:BQR-minorizer}
\begin{split}
    &\varphi(y \mid \mathbf{X}, \beta, \xi) \propto \exp\Big\{-\tau_0\tilde{u}\mathbf{1}_n^{\top}(y - \mathbf{X}\beta) + \tau_0(y - \mathbf{X}\beta)^{\top}\mathcal{A}(\xi)(y - \mathbf{X}\beta) + \tau_0\sum_{i \in [n]}\gamma(\xi_i)\Big\},
\end{split}
\end{align}
where $\tilde{u} := 2u-1$, $\xi := (\xi_1,\ldots,\xi_n)^{\top} \in \mathbb{R}^{n}_{+}$, $\mathcal{A}(\xi) := \mathrm{diag}(h'(\xi_1^2),\ldots,h'(\xi_n^2))$, $h(t) := -\sqrt{t}$, and $\gamma(t) := h(t^2) - t^2h'(t^2) = -\tfrac{|t|}{2}$. Analogous to Section~\ref{main-subsec:prior-posterior} of the main manuscript, the regression parameter vector $\beta$ is endowed upon with a conjugate Gaussian prior, $\beta \sim \mathcal{N}_p(\mu, {\Sigma})$, yielding a multivariate Gaussian $\alpha$-fractional variational posterior, $\pi_{\alpha}(\beta \mid \mathcal{D}_n, \xi) \equiv \mathcal{N}_p(\mu_{\alpha}(\xi), \Sigma_{\alpha}(\xi))$, with parameters:
\begin{align}\label{eq:QR-variational-posterior-parameters}
\begin{split}
    \Sigma_{\alpha}^{-1}(\xi) = \Sigma^{-1} - 2\tau_{\alpha} \mathbf{X}^{\top}\mathcal{A}(\xi)\mathbf{X},\quad \mu_{\alpha}(\xi) = \Sigma_{\alpha}(\xi)\left[\Sigma^{-1}\mu - 2\tau_{\alpha}\mathbf{X}^{\top}\mathcal{A}(\xi)y + \tau_{\alpha} \tilde{u}\mathbf{X}^{\top}\mathbf{1}_n\right],
\end{split}
\end{align}
where $\tau_{\alpha} = \tau_0\alpha$. The optimal variational parameter $\xi^{\star} \in \mathbb{R}^{n}_{+}$ is obtained by maximizing the $\mathsf{ELBO}$ $\mathsf{L}(\xi)$, which up to constants is:
\begin{align}\label{eq:ELBO-BQR}
\begin{split}
    \mathsf{L}(\xi) &:= \log\left[\frac{\varphi_{\alpha}(y\mid \mathbf{X}, \beta, \xi)\pi(\beta)}{\pi_{\alpha}(\beta\mid \mathcal{D}_n, \xi)}\right] \\&= \frac{1}{2}\mu_{\alpha}(\xi)^{\top}\Sigma^{-1}_{\alpha}(\xi)\mu_{\alpha}(\xi) + \frac{1}{2}\log|\Sigma_{\alpha}(\xi)| + \tau_{\alpha}\sum_{i \in [n]} \gamma(\xi_i) + \tau_{\alpha}y^{\top}\mathcal{A}(\xi)y.
\end{split}
\end{align}

\textbf{Construction of the EM surrogate}. In this case, following the derivations in Section~\ref{main-subsec:tavie-algo} of the main manuscript, the EM surrogate function for the $(t+1)$th step is:
\begin{equation}
\label{eq:BQR-EM-surrogate}
\begin{split}
\mathcal{Q}(\xi^{\dagger} \mid \xi^{(t)}) &:= \mathbb{E}_{\xi^{(t)}}[\log \pi(\beta)] + \alpha\sum_{i \in [n]} \mathbb{E}_{\xi^{(t)}}\left[\log \varphi(y_i \mid \mathbf{x}_i, \beta, \xi_i^{\dagger})\right]
\\ &= \tau_\alpha \sum_{i \in [n]} \left[A(\xi_i^{\dagger})\kappa_i(\xi^{(t)}) + \gamma(\xi_{i}^{\dagger})\right] + \mathfrak{C}(\xi^{(t)}),
\end{split}
\end{equation}
where $A(s) := h'(s^2)$, $\mathfrak{C}(\xi^{(t)})$ is a constant independent of $\xi^{\dagger}$, and $\kappa(\xi) := (\kappa_1(\xi),\ldots, \kappa_n(\xi))^{\top} \in \mathbb{R}^{n}_{+}$, with:
\begin{align}\label{eq:kappa-BQR}
    \kappa_i(\xi) = 
        \mathbf{x}_i^{\top}\Sigma_{\alpha}(\xi) \mathbf{x}_i + \left(y_i - \mathbf{x}_{i}^{\top}\mu_{\alpha}(\xi)\right)^2,
\end{align}
for $i\in [n]$. Next, we present the optimization of the EM surrogate function $\mathcal{Q}(\xi^{\dagger}\mid \xi^{(t)})$ in \eqref{eq:BQR-EM-surrogate} with respect to $\xi^{\dagger}$.

\textbf{Maximization of the surrogate}. As in Section~\ref{main-subsec:tavie-algo} of the main manuscript, the  EM update for maximization of $\mathsf{L}(\xi)$ in \eqref{eq:ELBO-BQR} is performed as:
\begin{equation}
\label{eq:TAVIE-xi-update-BQR}
{\xi}_i^{(t+1)}
= \arg\max_{{\xi}_i^{\dagger} > 0}\;\mathbb{E}_{\xi^{(t)}} \left[\log \varphi(y_i \mid \mathbf{x}_i, \beta, {\xi}_i^{\dagger}) \right] = \sqrt{\kappa_i(\xi^{(t)})},\quad i\in [n],
\end{equation}
where the last equality in \eqref{eq:TAVIE-xi-update-BQR} above follows from $\gamma'(t) = -t^2 A'(t)$, for $t\in \mathbb{R}^{+}$, which uses the definitions of $\gamma(t)$ and $A(t)$ in Proposition~\ref{main-prop:tangent-lower-bound} of the main manuscript. The final EM algorithm is presented in Algorithm \ref{alg:tavie-em-BQR} below.

{
\renewcommand{\baselinestretch}{1.0}\normalsize
\begin{algorithm}[H]
\caption{The $\tssg$ EM Algorithm for Bayesian Quantile Regression}
\label{alg:tavie-em-BQR}
\DontPrintSemicolon

\KwIn{Data $\mathcal{D}_n$, prior hyperparameters, tempering parameter $\alpha$, scale parameter $\tau_0$, quantile level $u$, tolerance $\texttt{tol}$.}
\KwOut{Variational parameters $\xi^\star$ and variational posterior hyperparameters.}

\textbf{Initialize}: Set $t \gets 0$ and initialize $\xi^{(0)} \in \mathbb{R}_{+}^{n}$. \;

\Repeat{$\lVert \xi^{(t)} - \xi^{(t-1)} \rVert_2 \le$ {\texttt{tol}}}{
  \tcc{Update variational posterior hyperparameters}
    update $(\mu_{\alpha}(\xi^{(t)}),\, \Sigma_{\alpha}(\xi^{(t)}))$ via~\eqref{eq:QR-variational-posterior-parameters}\;

  \tcc{Update variational parameters (coordinate-wise)}
  \For{$i \in [n]$}{
    $\xi_i^{(t+1)} \gets \sqrt{\kappa_i(\xi^{(t)})}$ \quad where $\kappa_i(\xi)$ is defined in~\eqref{eq:kappa-BQR}\;
  }

  $t \gets t+1$\;
}

\end{algorithm}
}

\newpage

\subsection{Additional Plots}\label{subsec:BQR-additional-plots}

\begin{figure}[!htp]
    \centering
    \includegraphics[width=0.95\linewidth]{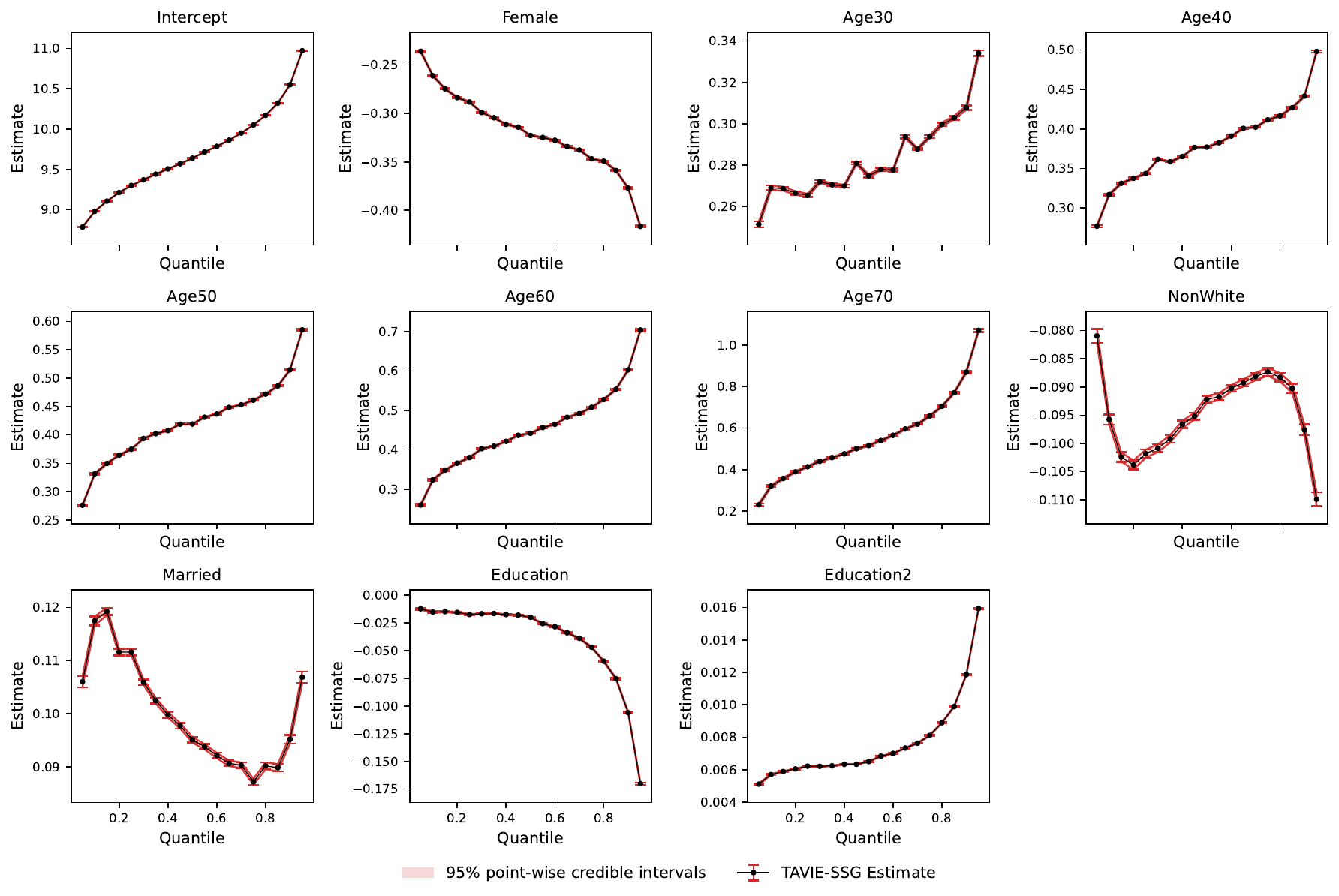}
    \caption{\footnotesize{$\tssg$ variational estimates and $95\%$ point-wise credible intervals for all features in U.S. Census 2000 full data.}}
    \label{fig:tavie_QR_estimates_census_all}
\end{figure}

\begin{figure}[!htp]
    \centering
    \includegraphics[width=\linewidth]{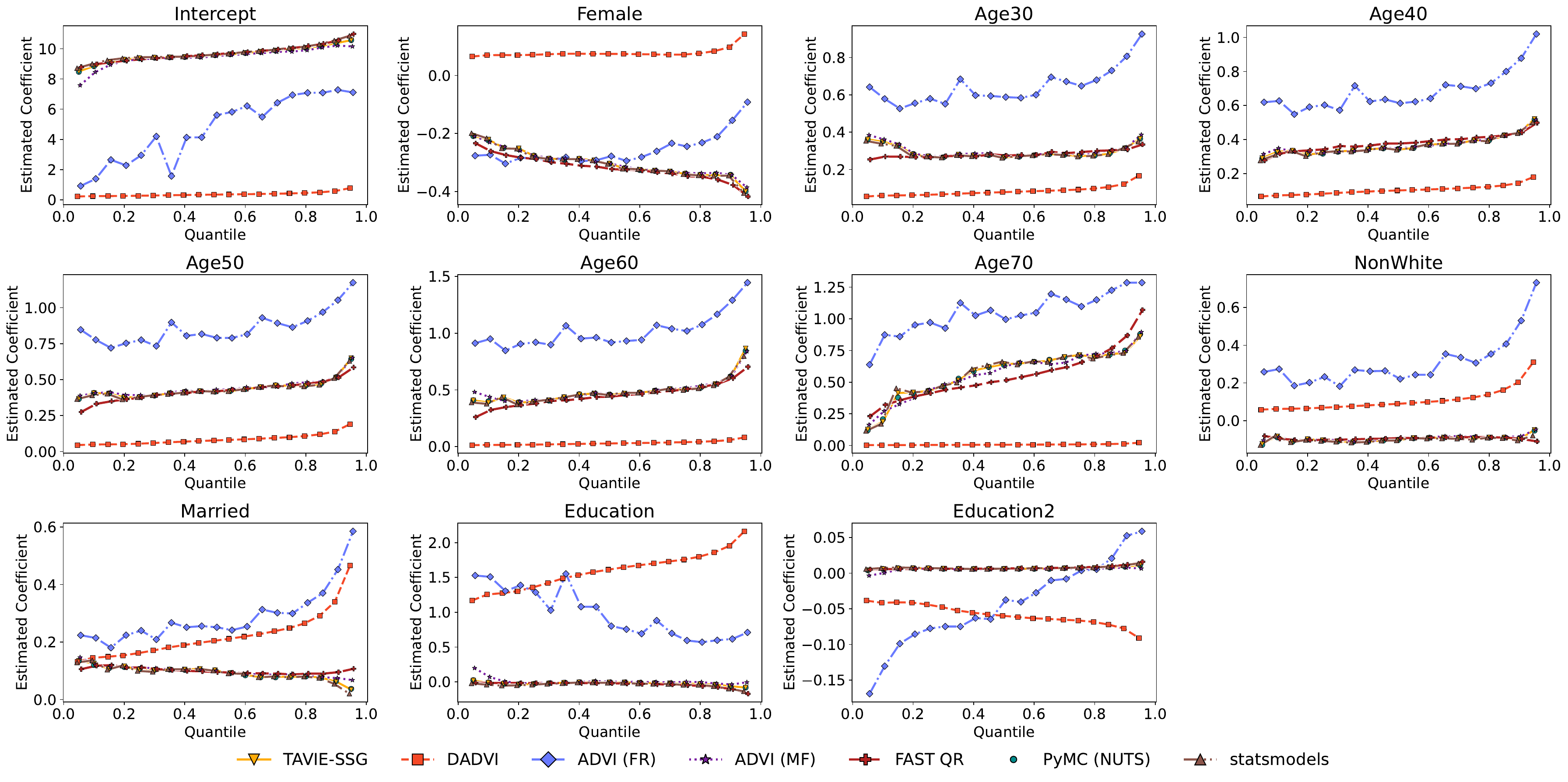}
    \caption{\footnotesize{Complete comparison results of $\tssg$ variational estimates on sub-sampled ($n=10^4$) U.S. Census 2000 data with FAST QR, DADVI, ADVI (MF/FR), PyMC (NUTS), and \texttt{statsmodels}.}}
    \label{fig:census_estimates_competing_methods_n_10000_all}
\end{figure}

\begin{figure}[!htp]
    \centering
    \includegraphics[width=\linewidth]{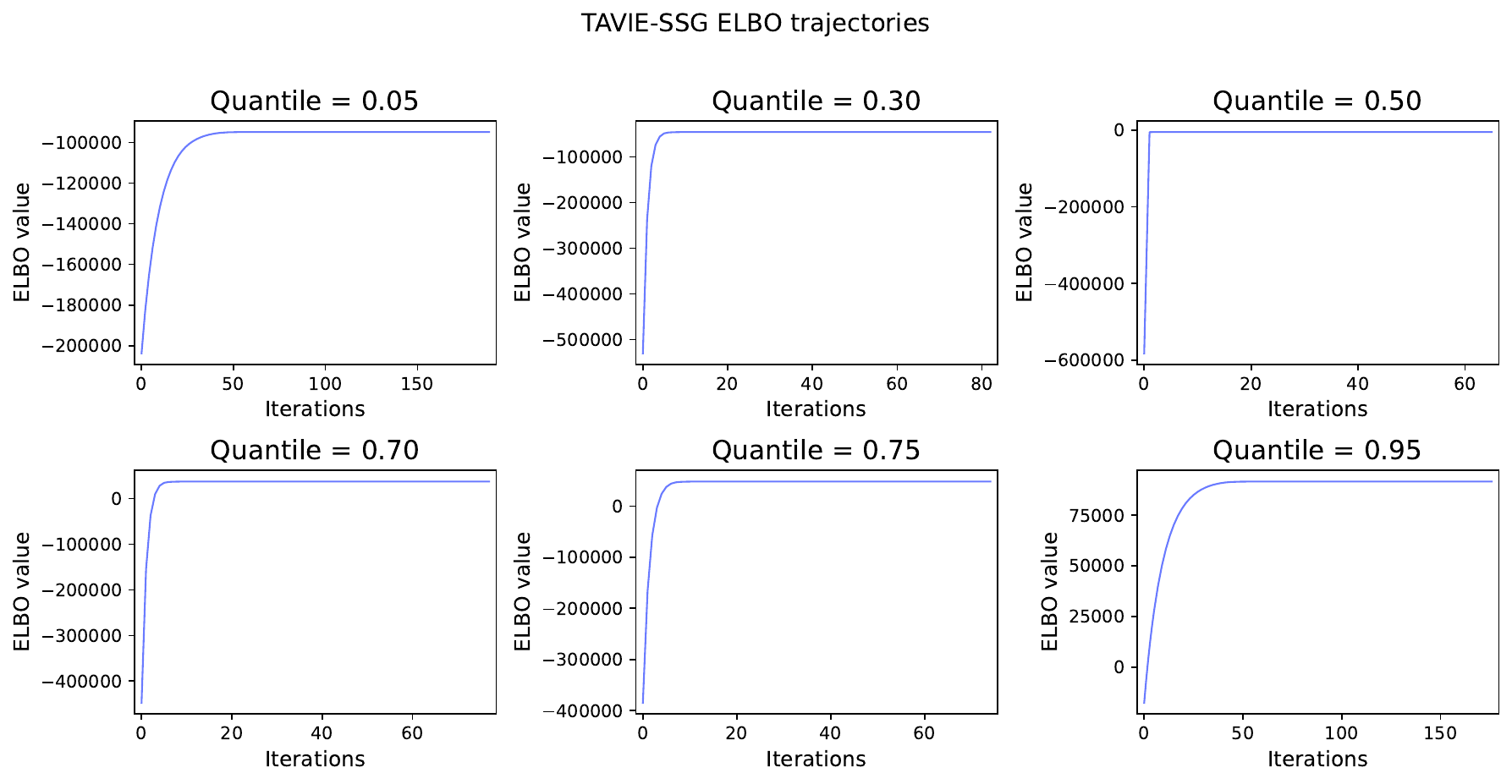}
    \caption{\footnotesize{$\mathsf{ELBO}$ of $\tssg$ plotted over iterations for various quantiles using the sub-sampled ($n=10^4$) U.S. Census 2000 data, demonstrating monotonic ascent and convergence. { For $\tssg$, the $\mathsf{ELBO}$ corresponds to $\mathsf{L}(\xi)$ in~\eqref{eq:ELBO-BQR}.}}}
    \label{fig:census_tavie_elbo}
\end{figure}

\begin{figure}[!htp]
    \centering
    \includegraphics[width=\linewidth]{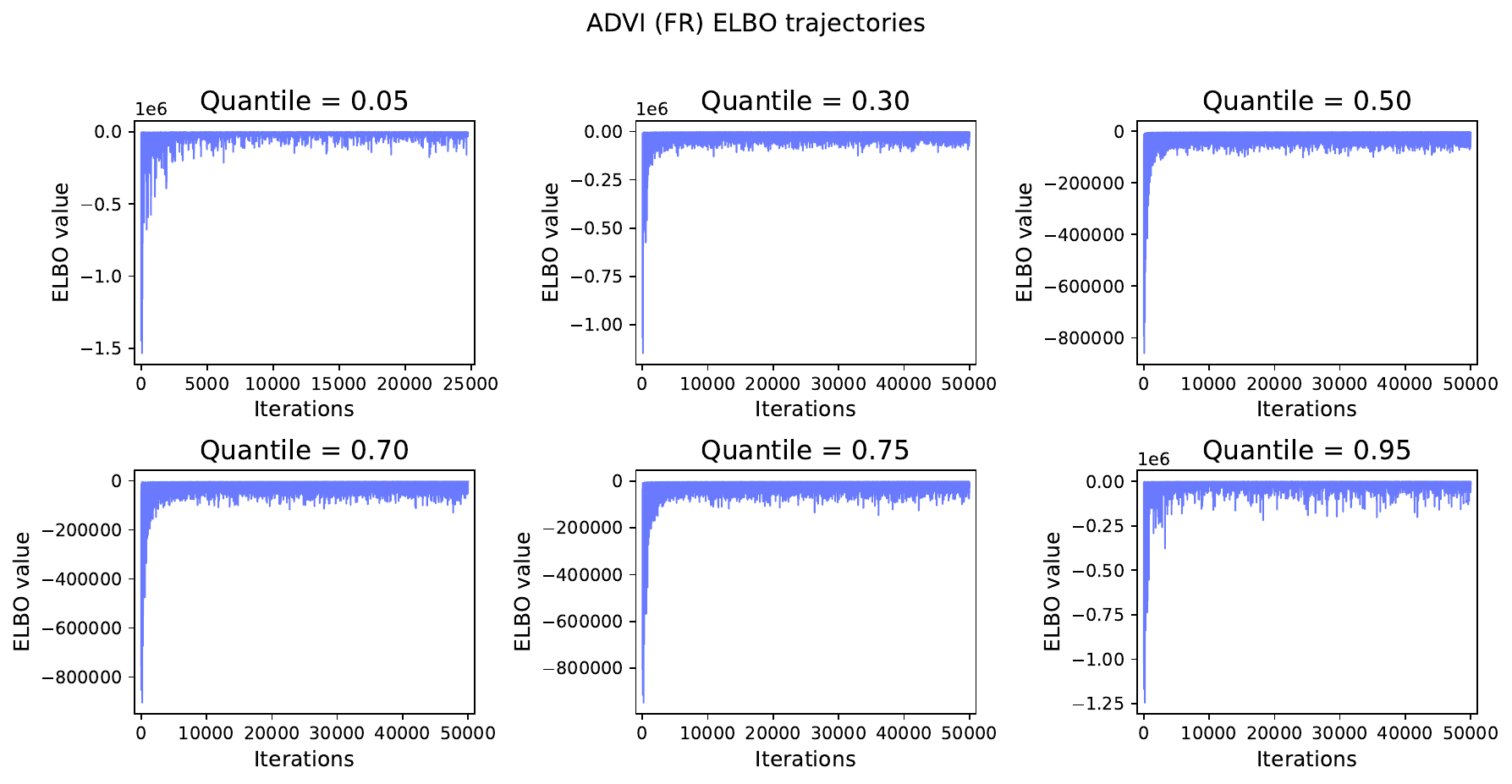}
    \caption{\footnotesize{$\mathsf{ELBO}$ of ADVI (FR) plotted over iterations for various quantiles using the sub-sampled ($n=10^4$) U.S. Census 2000 data, demonstrating convergence and stochastic behavior. { For ADVI (FR), the MC approximation of the true $\mathsf{ELBO}$ is tracked.}}}
    \label{fig:census_advi_fr_elbo}
\end{figure}

\newpage

\section{Auxiliary Results for STARmap Data Analysis}\label{sec:STARmap-additional-results}

\begin{figure}[!htp]
    \centering

    \begin{subfigure}[t]{0.75\textwidth}
        \centering
        \includegraphics[width=\linewidth]{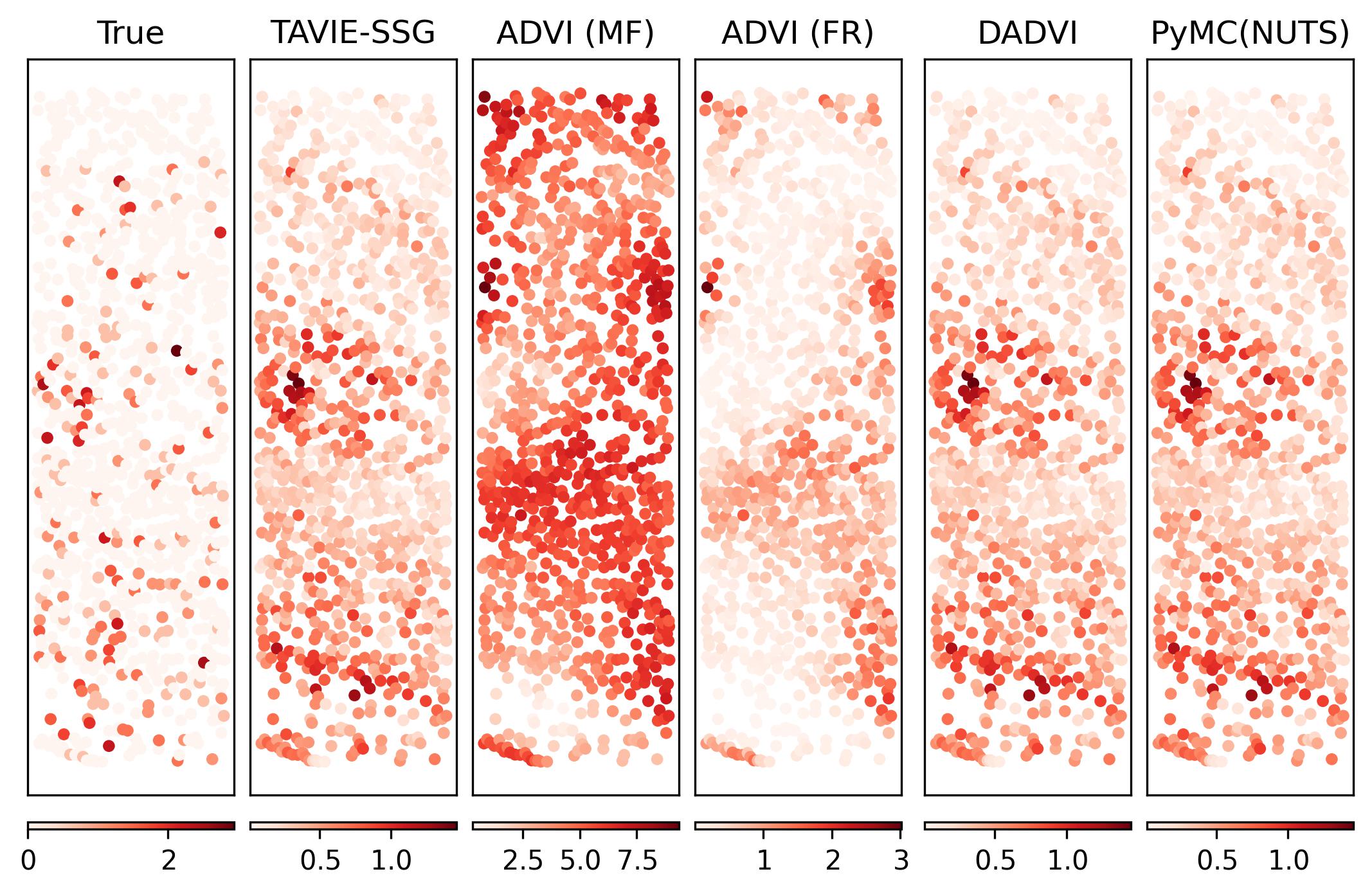}
        \caption{\footnotesize{\texttt{Lmo2} gene expressions.}}
        \label{fig:Lmo2}
    \end{subfigure}

    \vspace{0.75em}

    \begin{subfigure}[t]{0.75\textwidth}
        \centering
        \includegraphics[width=\linewidth]{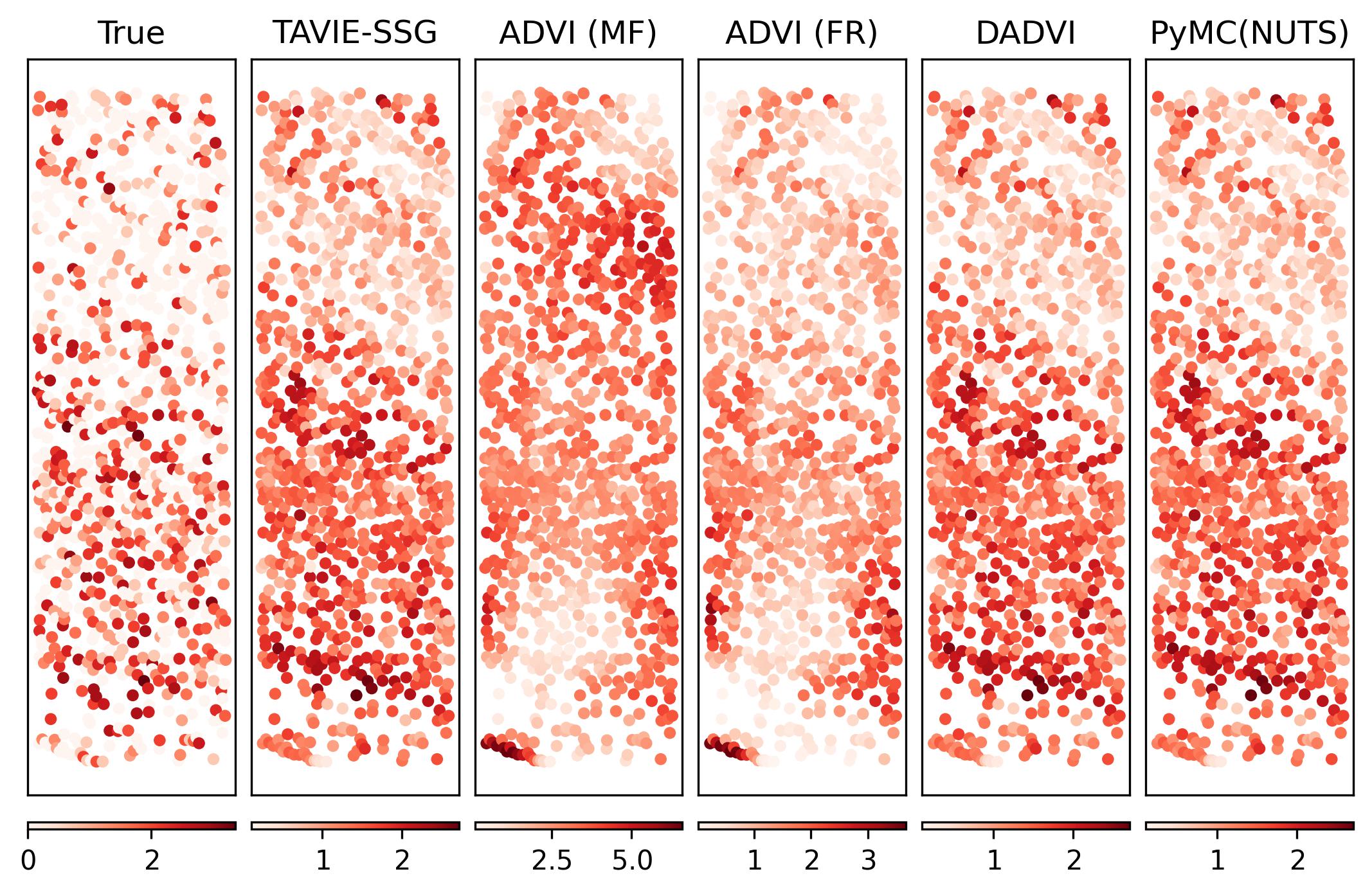}
        \caption{\footnotesize{\texttt{Chat} gene expressions.}}
        \label{fig:Chat}
    \end{subfigure}

    \caption{\footnotesize{Log-normalized true and predicted gene expression counts obtained from $\tssg$, ADVI (MF/FR), DADVI, and \texttt{PyMC} (NUTS).}}
    \label{fig:gene_pred}
\end{figure}

\begin{figure}[!htp]
    \ContinuedFloat
    \centering

    \begin{subfigure}[t]{0.75\textwidth}
        \centering
        \includegraphics[width=\linewidth]{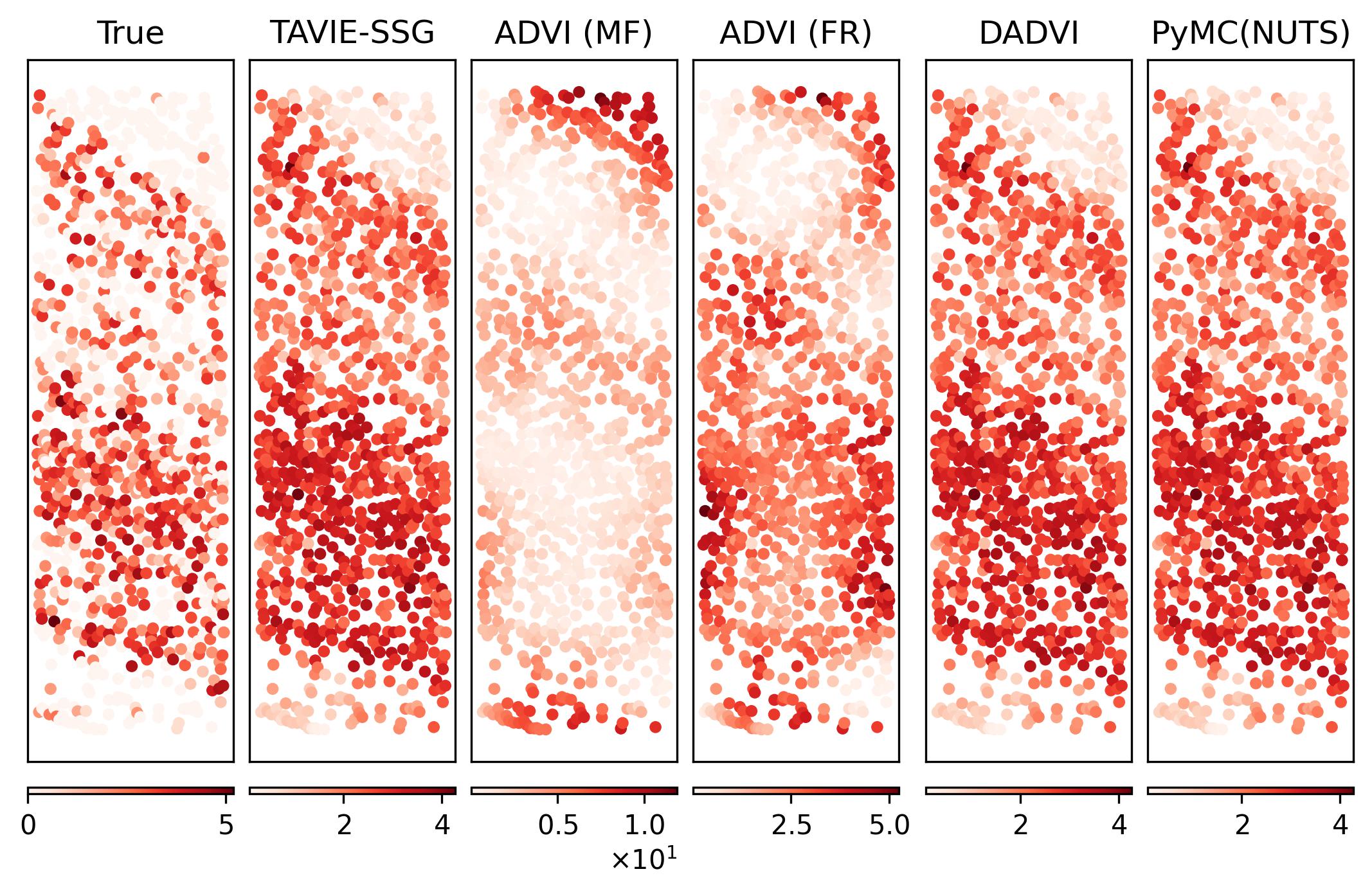}
        \caption{\footnotesize{\texttt{Egr1} gene expressions.}}
        \label{fig:Egr1}
    \end{subfigure}

    \vspace{0.75em}

    \begin{subfigure}[t]{0.75\textwidth}
        \centering
        \includegraphics[width=\linewidth]{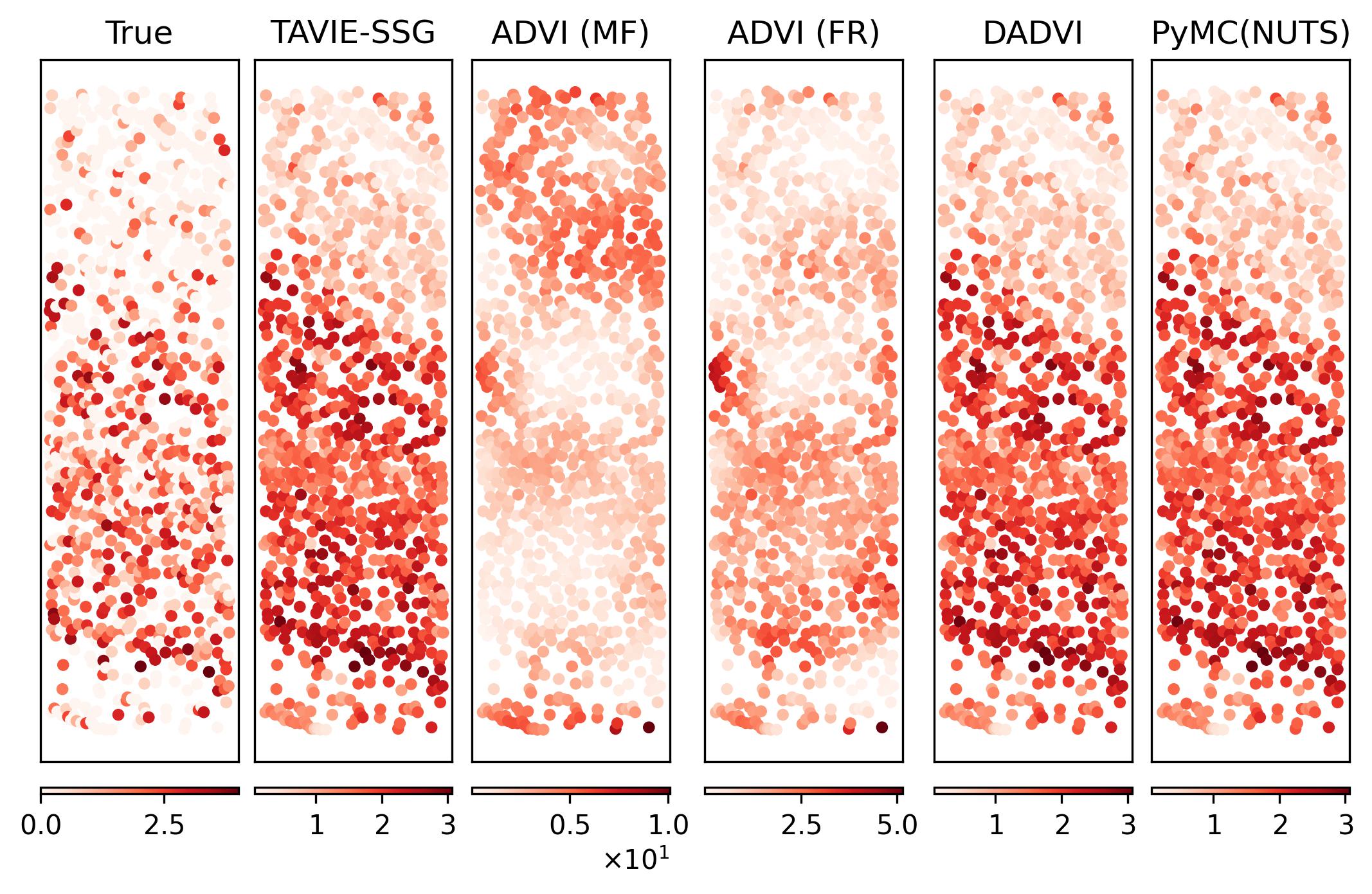}
        \caption{\footnotesize{\texttt{Fam19a1} gene expressions.}}
        \label{fig:Fam19a1}
    \end{subfigure}

    \caption{\footnotesize{Log-normalized true and predicted gene expression counts obtained from $\tssg$, ADVI (MF/FR), DADVI, and \texttt{PyMC} (NUTS) (continued).}}
\end{figure}

\begin{figure}[!htp]
    \ContinuedFloat
    \centering

    \begin{subfigure}[t]{0.75\textwidth}
        \centering
        \includegraphics[width=\linewidth]{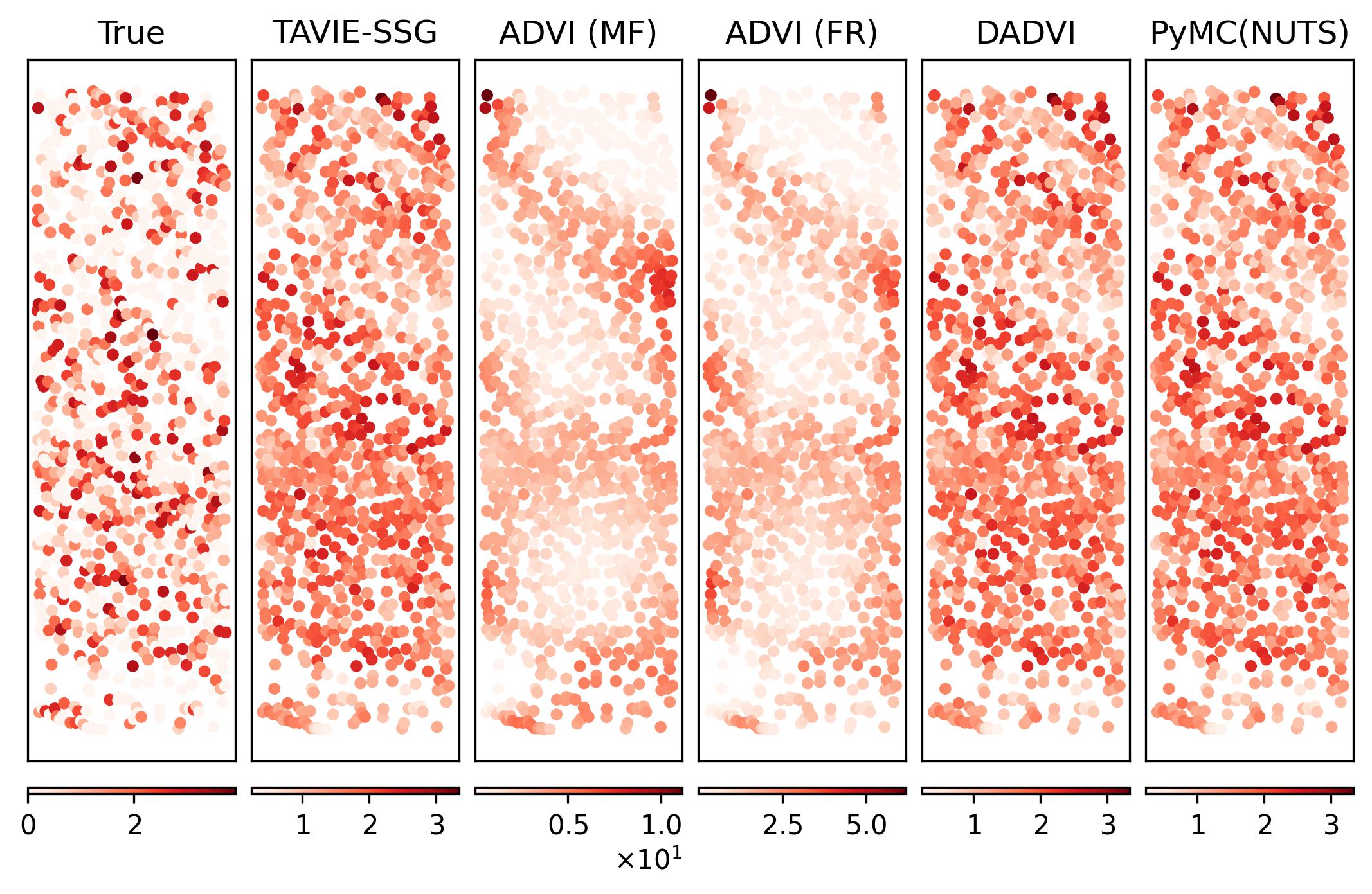}
        \caption{\footnotesize{\texttt{Mog} gene expressions.}}
        \label{fig:Mog}
    \end{subfigure}

    \vspace{0.75em}

    \begin{subfigure}[t]{0.75\textwidth}
        \centering
        \includegraphics[width=\linewidth]{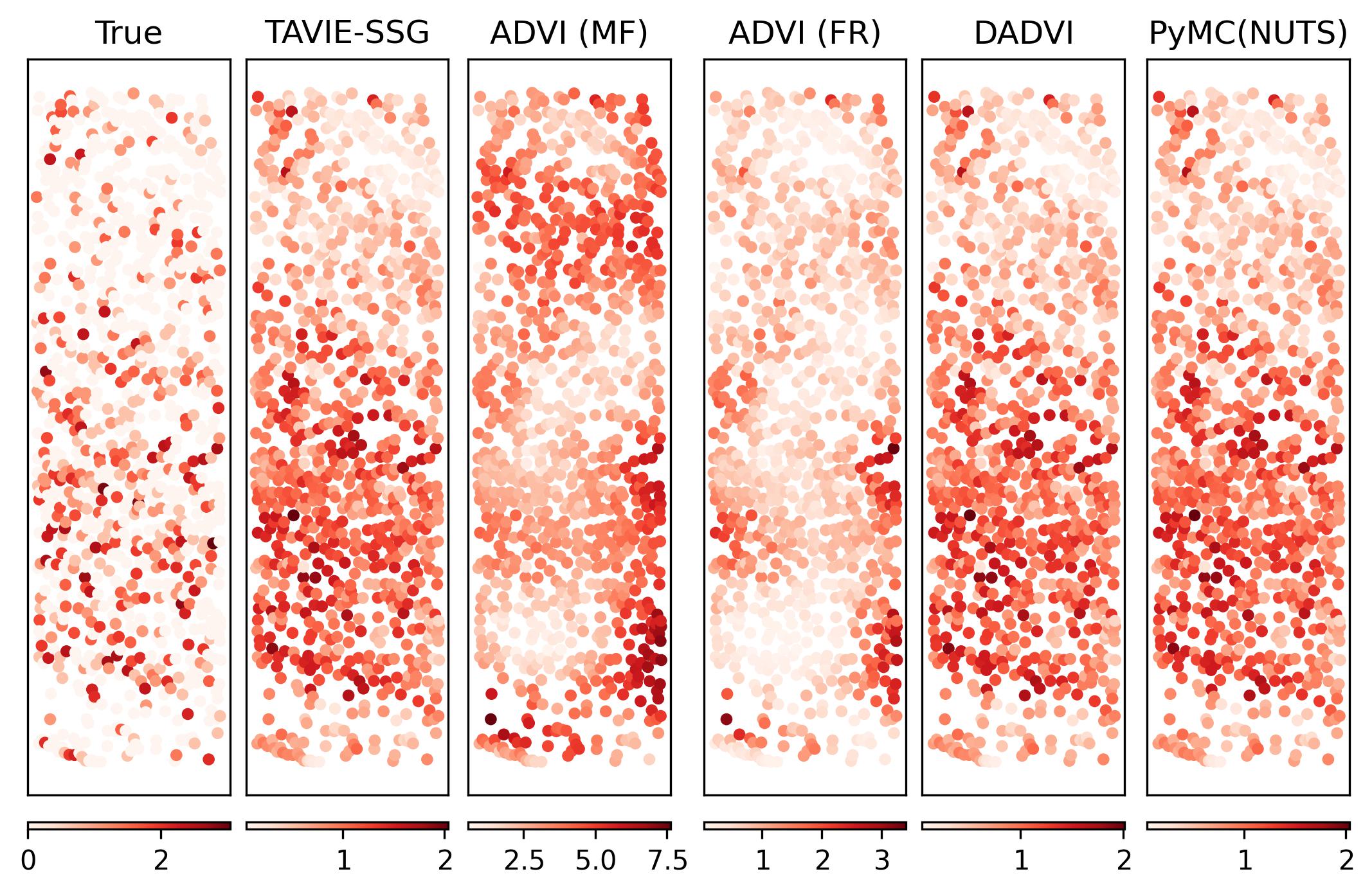}
        \caption{\footnotesize{\texttt{Gpc3} gene expressions.}}
        \label{fig:Gpc3}
    \end{subfigure}

    \caption{\footnotesize{Log-normalized true and predicted gene expression counts obtained from $\tssg$, ADVI (MF/FR), DADVI, and \texttt{PyMC} (NUTS) (continued).}}
\end{figure}

\newpage

\begin{figure}[!htp]
    \centering

    \begin{subfigure}[t]{0.75\textwidth}
        \centering
        \includegraphics[width=\linewidth]{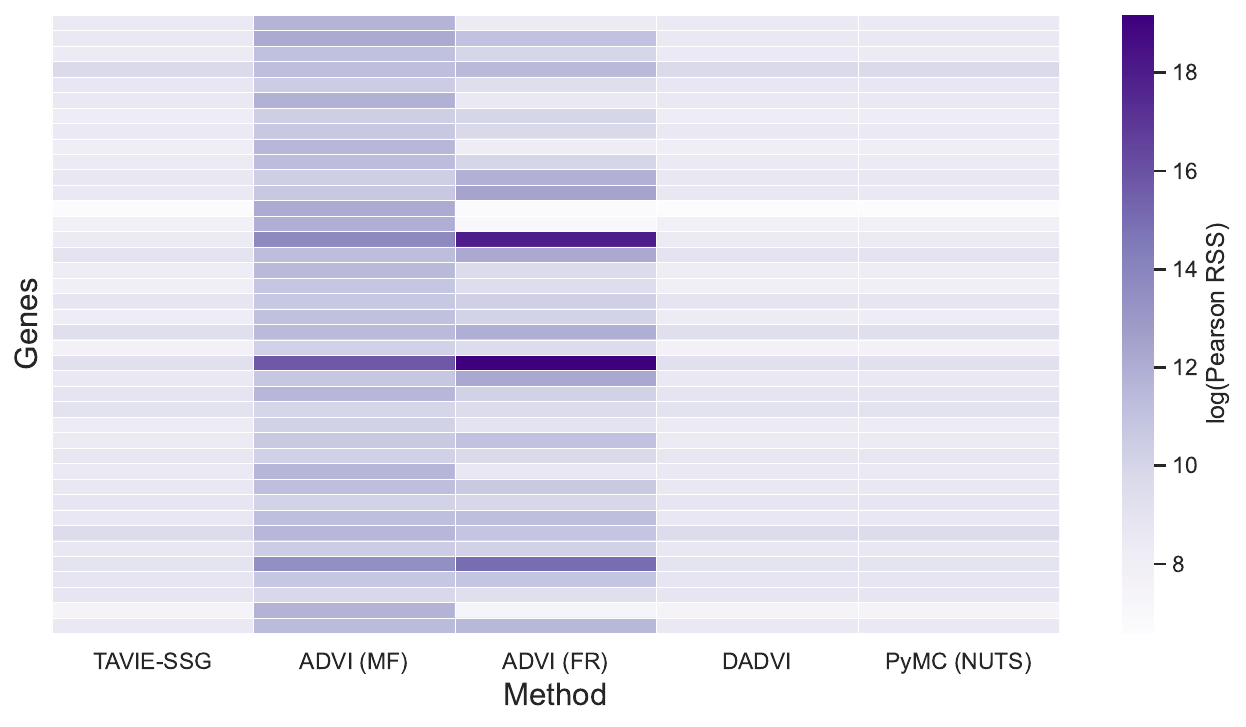}
        \caption{\footnotesize{Genes 1--40.}}
        \label{fig:gene1-40}
    \end{subfigure}

    \vspace{0.75em}

    \begin{subfigure}[t]{0.75\textwidth}
        \centering
        \includegraphics[width=\linewidth]{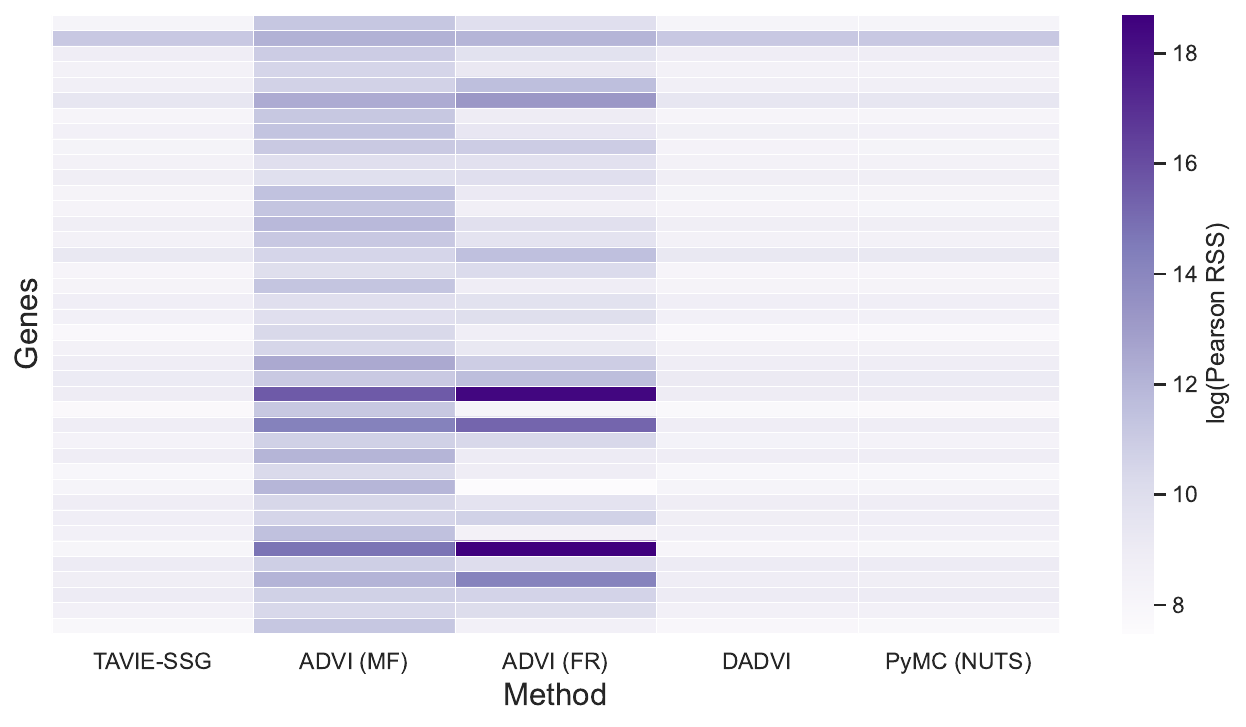}
        \caption{\footnotesize{Genes 41--80.}}
        \label{fig:gene41-80}
    \end{subfigure}

    \caption{\footnotesize{
    Heatmaps of the $\log$ Pearson residual sum of squares between observed and predicted gene expression counts across $\tssg$, ADVI (MF/FR), DADVI, and \texttt{PyMC} (NUTS), shown over four blocks of randomly selected genes.
    }}
    \label{fig:heatmap}
\end{figure}

\begin{figure}[!htp]
    \ContinuedFloat
    \centering

    \begin{subfigure}[t]{0.75\textwidth}
        \centering
        \includegraphics[width=\linewidth]{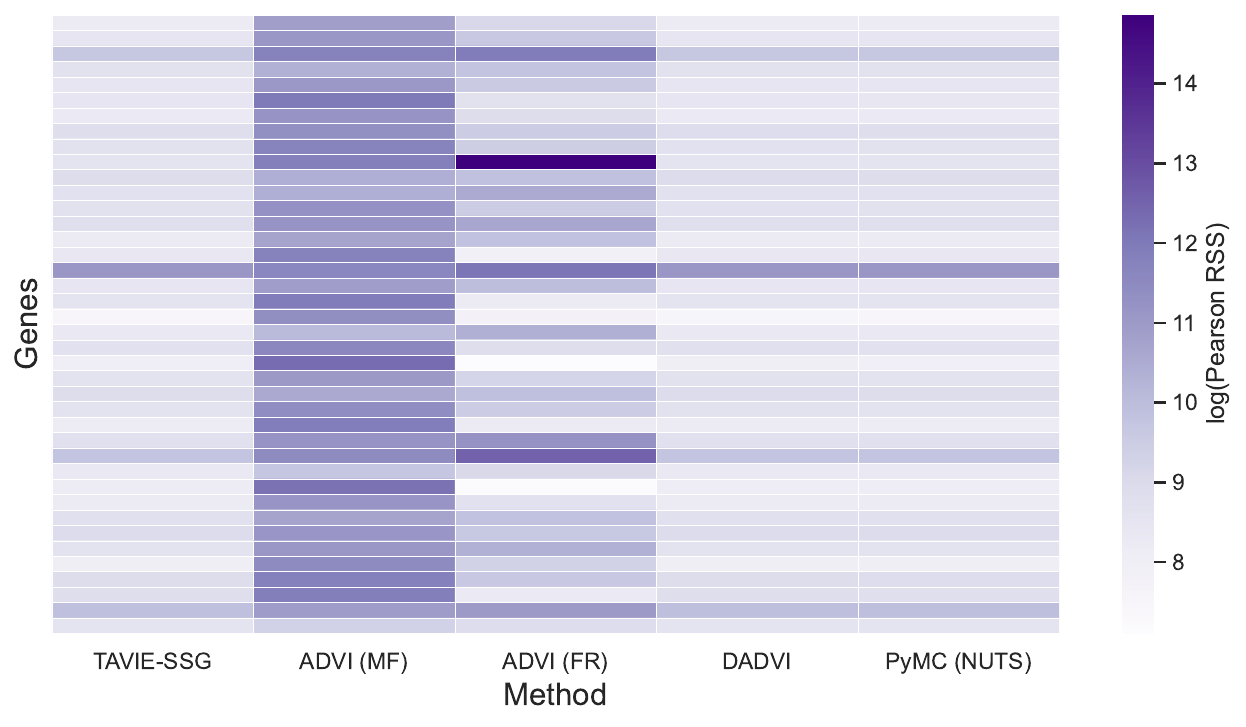}
        \caption{\footnotesize{Genes 81--120.}}
        \label{fig:gene81-120}
    \end{subfigure}

    \vspace{0.75em}

    \begin{subfigure}[t]{0.75\textwidth}
        \centering
        \includegraphics[width=\linewidth]{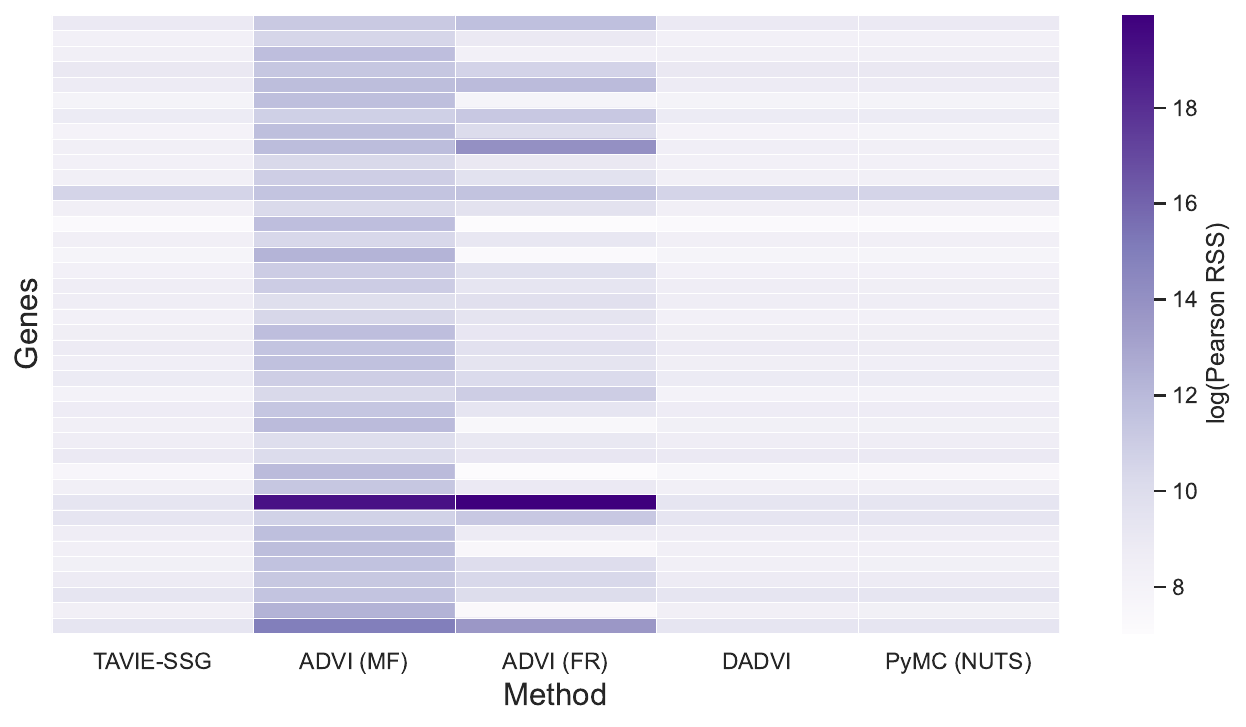}
        \caption{\footnotesize{Genes 121--160.}}
        \label{fig:gene121-160}
    \end{subfigure}

    \caption{\footnotesize{
    Heatmaps of the $\log$ Pearson residual sum of squares between observed and predicted gene expression counts across $\tssg$, ADVI (MF/FR), DADVI, and \texttt{PyMC} (NUTS), shown over four blocks of randomly selected genes (continued).
    }}
\end{figure}

\newpage

\section{Empirical Performance of \texorpdfstring{$\tssg$}{TAVIE-SSG} under Model Misspecification}
\label{app:model-misspecification-empirics}

We now investigate the empirical performance of $\tssg$ when the likelihood used for posterior inference differs from the data-generating distribution. The primary goal of the ensuing experiments is to determine whether, under likelihood misspecification, the $\tssg$ variational approximation preserves the predictive behavior of the corresponding $\alpha$-fractional posterior. To this end, we compare $\tssg$ with the fractional posterior approximated using \texttt{PyMC} (NUTS)~\citep{pymc,pymc2023peerj}, where both methods use the same fitted likelihood, prior distribution, and likelihood tempering parameter $(\alpha)$.

\subsection{Experimental Setup}
\label{app:model-misspecification-empirics-settings}

For the Type I $\ssg$ experiments, we set the sample size and number of features to $n=200$ and $p=9$, respectively. Responses are generated according to $y_i = \mathbf{x}_i^{\top}\beta_0 + \varepsilon_i$ for $i\in [n]$, where $\beta_0\sim \mathcal{N}_{p}(0, I_{p})$, and $x_{i1} = 1$ (intercept) and $x_{ij} \stackrel{\mathrm{i.i.d.}}{\sim}\mathcal{N}_{1}(0, 1)$ for $j=2, \ldots, p$. Note that, $\beta_0$ is generated once and held fixed across the experiments. The generating errors are centered and standardized to have the common variance $\tau_0^{-2}$, with $\tau_0^{2} = 3$.

We consider $5$ Type I $\ssg$ misspecification regimes: (i) a Laplace likelihood fitted to data with Gaussian errors; (ii) a Laplace likelihood fitted to Student's-$t$ errors with $\nu=3$; (iii) a Laplace likelihood fitted to Student's-$t$ errors with $\nu=8$; (iv) a Student's-$t$ likelihood with fixed degrees of freedom $\nu=5$ fitted to data with Gaussian errors; and (v) the same Student's-$t$ likelihood fitted to data with Laplace errors. The second experiment represents the most pronounced tail misspecification among these configurations, whereas the remaining experiments examine discrepancies between light and moderately heavy-tailed distributions.

For both $\tssg$ and \texttt{PyMC} (NUTS), we use the common Normal-Gamma prior, i.e., $\tau^{2} \sim \mathcal{G}(0.05/2, 0.05/2)$ and $\beta\mid \tau^{2} \sim \mathcal{N}_{p}(0, \tau^{-2}I_{p})$. The likelihood tempering parameter is varied over $\alpha \in \{0.20,0.50,0.80,0.95,1.00\}$. For the \texttt{PyMC} benchmark, NUTS is implemented using $2$ chains, with $1000$ tuning iterations followed by $4000$ retained draws per chain and a target acceptance probability of $0.9$.

We additionally consider a Type II $\ssg$ misspecification experiment in which a Negative-Binomial likelihood is fitted to Poisson-generated count data. Specifically, the responses are generated as $y_i \sim \mathrm{Poisson}(\lambda_i)$ with probability mass function, $p(y_i \mid \lambda_i) = \exp\{-\lambda_i\}\lambda_i^{y_i}/y_i!$ (for $y_i \in \mathbb{N} \cup \{0\}$), where $\lambda_i = m\exp\{-\mathbf{x}_i^{\top}\beta_0\}$, with the size parameter fixed at $m=10$. The fitted Negative-Binomial model uses $p_i = [1 + \exp\{-\mathbf{x}_i^{\top}\beta\}]^{-1}$. Therefore, $\mathbb{E}(y_i\mid \mathbf{x}_i,\beta) = m\frac{1-p_i}{p_i} = m\exp\{-\mathbf{x}_i^{\top}\beta\}$. Thus, the generating and fitted models have the same form of the conditional mean but different conditional variances, providing a controlled experiment in variance misspecification. The components of the true coefficient vector $\beta_0$ are generated independently from $\mathcal{N}_1(0, 0.5^{2})$ and subsequently held fixed, while both inference methods use the Gaussian prior $\beta \sim \mathcal{N}_{p}(0, I_{p})$.

Each experiment is repeated over $10$ independently regenerated datasets. Within each experiment and repetition, the same simulated dataset is reused for every value of $\alpha$ and for both $\tssg$ and \texttt{PyMC} (NUTS), facilitating a direct paired comparison. Performance is assessed using the in-sample response mean-squared error (MSE) viz., $\mathrm{MSE}(y, \widehat{y}) := n^{-1}\|y-\widehat{y}\|^{2}_{2}$. For the Type I $\ssg$ likelihoods, $\widehat{y}_{i} = \mathbf{x}_i^{\top}\widehat{\beta}$, where $\widehat{\beta}$ is the posterior mean of $\beta$. For the Negative-Binomial model, the plug-in conditional mean $\widehat{y}_i = m\exp\{-\mathbf{x}_i^{\top}\widehat{\beta}\}$ is used.

\subsection{Empirical Results}
\label{app:model-misspecification-empirics-results}

Figure~\ref{fig:misspecified_likelihood_comparisons} summarizes the response MSEs of $\tssg$ and the corresponding fractional posterior \texttt{PyMC} (NUTS) benchmark across the aforementioned likelihood misspecification regimes. For all $5$ Type I $\ssg$ experiments, the boxplots produced by the two methods exhibit almost complete overlap across the entire range of $\alpha$. In particular, across the $25$ Type I configurations obtained by combining the $5$ misspecification regimes with the $5$ values of $\alpha$, the largest absolute difference between the median MSEs of $\tssg$ and \texttt{PyMC} (NUTS) is approximately $5.8\times 10^{-4}$, corresponding to a relative difference below $0.18\%$. The between repetition variability is therefore substantially larger than the discrepancy between the two inference methods.

This agreement persists under both lighter- and heavier-tailed misspecification. When a Laplace likelihood is fitted to Gaussian or Student's-$t$ data, $\tssg$ closely reproduces the predictive performance of the NUTS-based fractional posterior. Notably, the same behavior is observed for Student's-$t$ errors with $\nu=3$, indicating that the variational approximation remains stable even under the most substantial tail mismatch considered here. Similarly, when the fitted Student's-$t$ likelihood with $\nu=5$ is applied to Gaussian or Laplace-generated data, the MSE distributions of $\tssg$ and \texttt{PyMC} (NUTS) remain nearly indistinguishable.

The Poisson-to-Negative-Binomial misspecification experiment in the final panel of Figure~\ref{fig:misspecified_likelihood_comparisons} yields a consistent conclusion under Type II variance misspecification. Across the $5$ values of $\alpha$, the median MSE of $\tssg$ ranges approximately $7.11$ to $7.18$, compared with approximately $7.19$ to $7.22$ for \texttt{PyMC} (NUTS). $\tssg$ has a slightly lower median MSE at each tempering level, although the largest relative difference is only about $1.3\%$, and the corresponding boxplots overlap substantially. Hence, there is no evidence that the $\tssg$ approximation introduces a meaningful loss in predictive accuracy when the fitted count likelihood specifies the conditional mean but overstates the conditional dispersion.

Taken together, these findings demonstrate that $\tssg$ closely preserves the predictive behavior of the corresponding $\alpha$-fractional posterior under all likelihood misspecification regimes considered here. The agreement holds across mismatches in tail behavior, distributional shape, and conditional variance, and remains stable over a broad range of likelihood tempering parameters. These experiments show that conditional on the chosen misspecified likelihood, the tractable $\tssg$ approximation incurs essentially no additional predictive degradation relative to the substantially more computationally intensive NUTS benchmark. A complete theoretical analysis of $\tssg$ under such model misspecification regimes is left for future work.
\begin{figure}[!t]
    \centering
    \includegraphics[width=\linewidth]{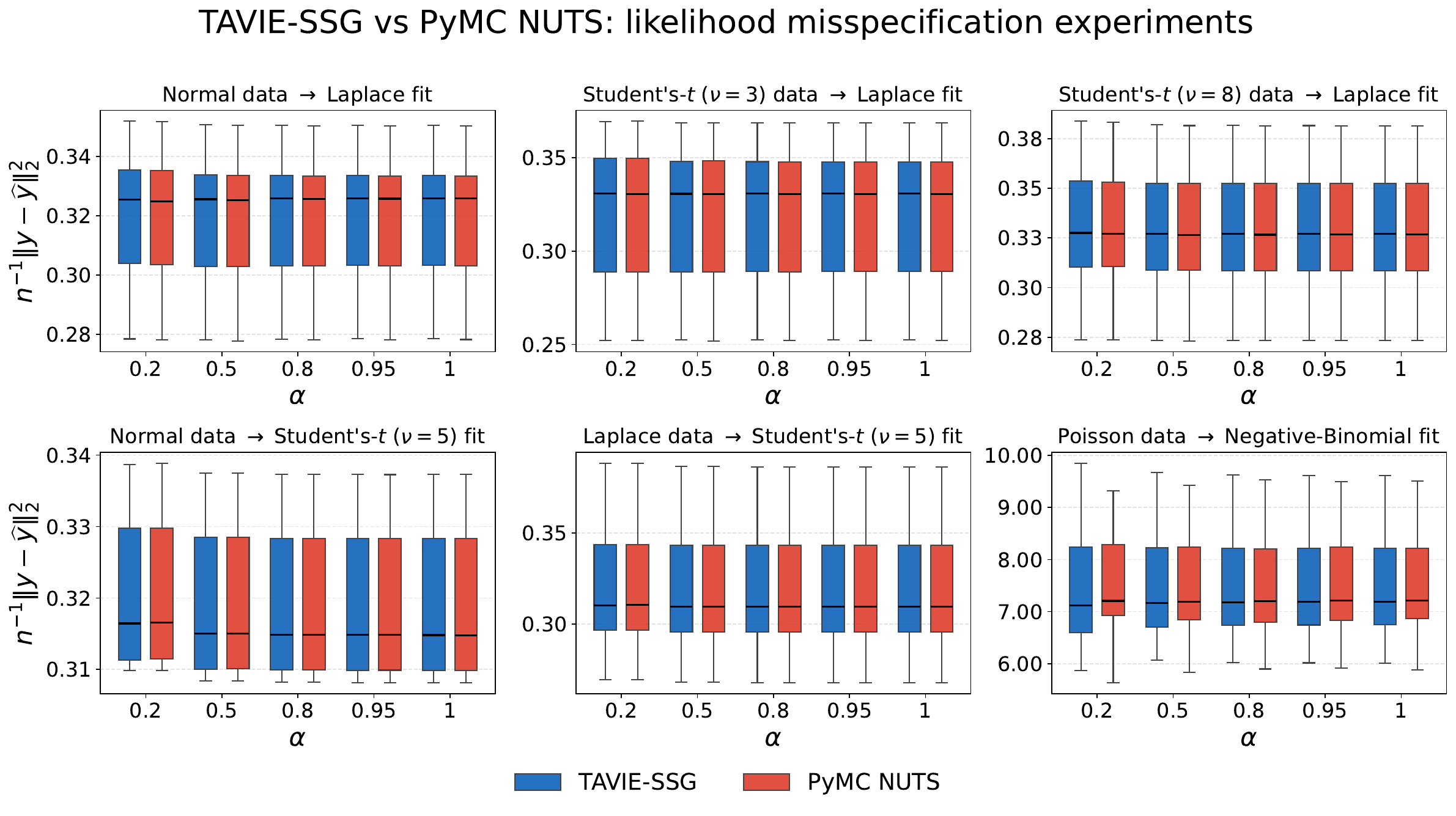}
    \caption{\footnotesize{MSEs of $\tssg$ and the corresponding $\alpha$-fractional posterior approximated using \texttt{PyMC} (NUTS) under various likelihood misspecification regimes. Results are shown for $(n, p) = (200, 9)$, $10$ independent repetitions, and $\alpha \in \{0.20,0.50,0.80,0.95,1.00\}$.}}
    \label{fig:misspecified_likelihood_comparisons}
\end{figure}

\newpage

\putbib[bibliography]
\end{bibunit}

\end{document}